\documentclass[11pt,twoside]{report}

\pdfoutput=1

\usepackage{lipsum}
\usepackage[paperwidth=6.25in,paperheight=9.25in,top=0.625in, bottom=0.625in, inner=1.125in, outer=0.625in,includefoot]{geometry}

\usepackage[numbers]{natbib}
\usepackage[affil-it]{authblk}
\usepackage{amsmath}
\usepackage{amsthm}
\usepackage{mathrsfs}
\usepackage[english]{babel}
\usepackage{graphicx}
\usepackage{amsfonts} 
\usepackage{amssymb} 
\usepackage{float}
\usepackage{times}
\usepackage{tensor}
\usepackage{mathtools}
\usepackage{color}
\usepackage[usenames,dvipsnames]{xcolor}
\usepackage{pdflscape}
\usepackage{rotating}
\usepackage{multirow}
\usepackage{array}
\usepackage{adjustbox}

\usepackage{layouts}

\usepackage{bm}% bold math
\usepackage{wrapfig}
\usepackage{overpic}

\usepackage[inline]{enumitem}
\usepackage{caption}
\usepackage{subfig}
\usepackage{incgraph}

\usepackage{setspace}
\usepackage[Conny]{fncychap}

\usepackage{import}

\usepackage{afterpage}

\newcommand\blankpage{%
    \null
    \thispagestyle{empty}%
    \addtocounter{page}{-1}%
    \newpage}

\hypersetup{colorlinks=true,linkcolor=Magenta,linktocpage=true,plainpages=false,allcolors=Magenta,hyperindex,pdftex}

\def\bra#1{\mathinner{\langle{#1}|}}
\def\ket#1{\mathinner{|{#1}\rangle}}
\newcommand{\braket}[2]{\langle #1|#2\rangle}

\def\rcurs{{\mbox{$\resizebox{.16in}{.08in}{\includegraphics{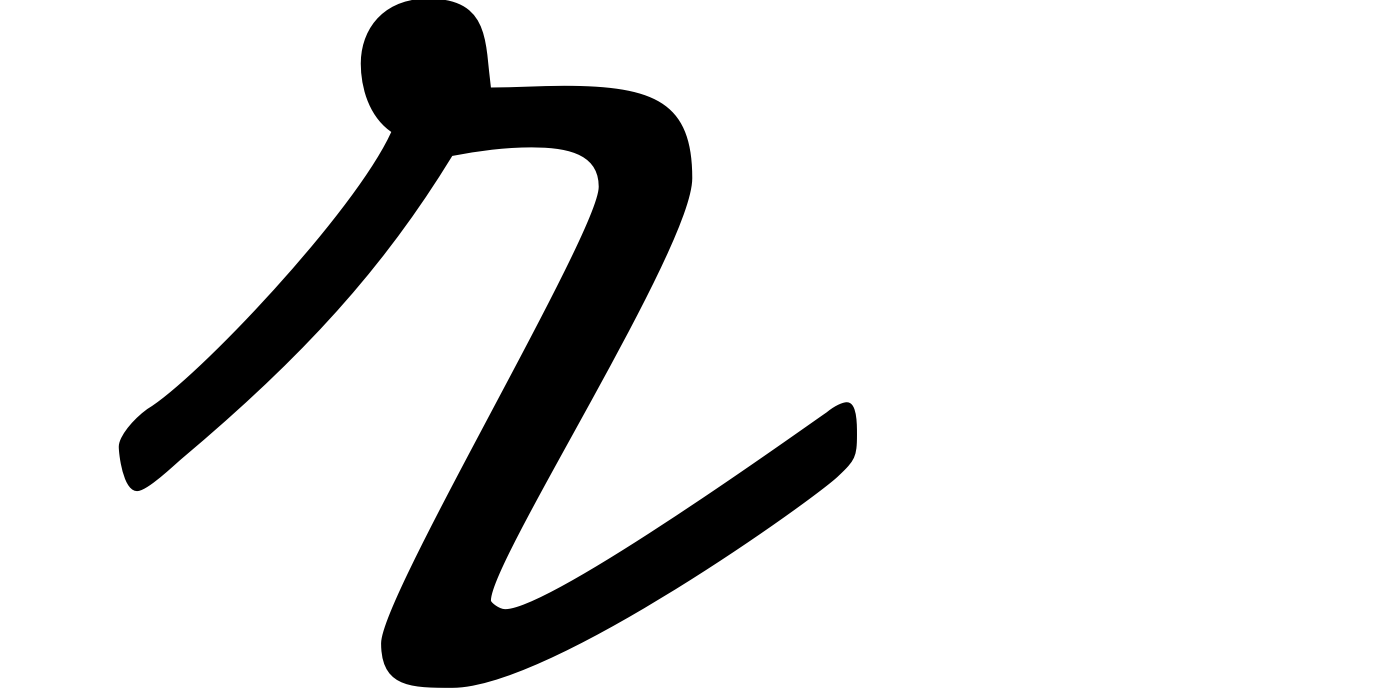}}$}}}

\DeclareMathOperator{\dee}{d}
\DeclareMathOperator{\Imag}{Im}
\DeclareMathOperator{\Real}{Re}
\newcommand{\myparagraph}[1]{\paragraph{#1}\mbox{}\\}

\newcounter{ingredientnumber}

\newenvironment{ingredient}{\refstepcounter{ingredientnumber}\par\medskip\noindent \textcolor{magenta}{\textbf{Ingredient \arabic{ingredientnumber}}}}{}

\newtheorem*{theorem}{Theorem} % note that the * means that the numbering is turned off....
\newtheorem*{conjecture}{Conjecture} % note that the * means that the numbering is turned off....

\begin{document}
\incgraph[paper=graphics]{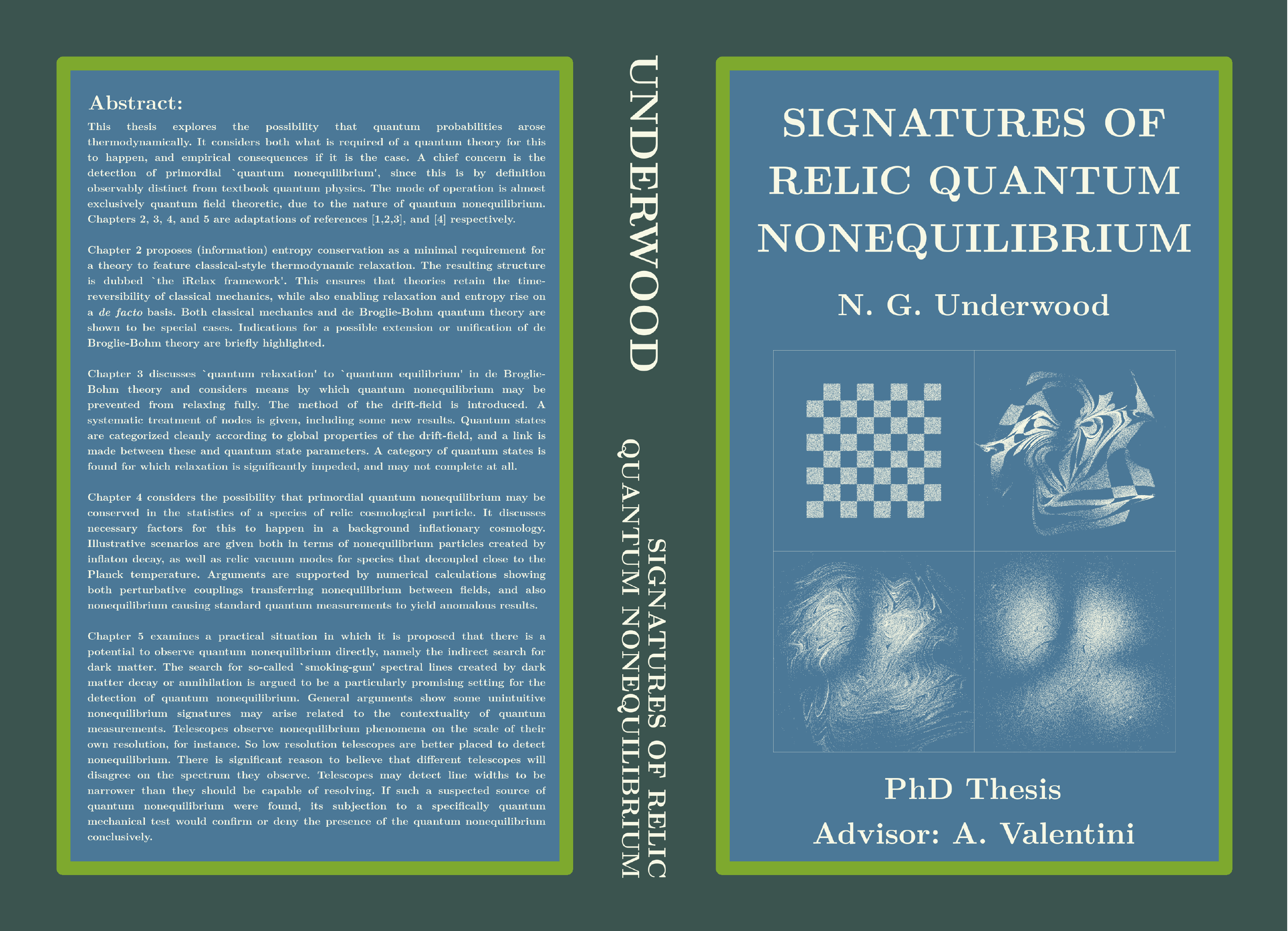}
\setcounter{page}{1}

\singlespacing
\vspace*{5mm}
\begin{center}
{\Huge  SIGNATURES OF RELIC QUANTUM\\ \vspace{5mm} NONEQUILIBRIUM}\\
\vspace{30mm}
{\huge Nicolas Graeme Underwood}\\
\vspace{30mm}
{\huge PhD thesis, Physics}\\
\vspace{30mm}
{\huge Advisor: Antony Valentini}
\end{center}
\thispagestyle{empty}

\clearpage

\chapter*{ABSTRACT}
\addcontentsline{toc}{chapter}{ABSTRACT}
This thesis explores the possibility that quantum probabilities arose thermodynamically.
It considers both what is required of a quantum theory for this to happen, and empirical consequences if it is the case.
A chief concern is the detection of primordial `quantum nonequilibrium', since this is by definition observably distinct from textbook quantum physics.
The mode of operation is almost exclusively quantum field theoretic, due to the nature of quantum nonequilibrium.
Chapters \ref{2}, \ref{3}, \ref{4}, and \ref{5}, are adaptations of references \cite{NU_In_preparation,NU18,UV15}, and \cite{UV16} respectively.

Chapter \ref{2} proposes (information) entropy conservation as a minimal requirement for a theory to feature classical-style thermodynamic relaxation.
The resulting structure is dubbed `the iRelax framework'.
This ensures that theories retain the time-reversibility of classical mechanics, while also enabling relaxation and entropy rise on a \emph{de facto} basis.
Both classical mechanics and de Broglie-Bohm quantum theory are shown to be special cases.
Indications for a possible extension or unification of de Broglie-Bohm theory are briefly highlighted.

Chapter \ref{3} discusses `quantum relaxation' to `quantum equilibrium' in de Broglie-Bohm theory and considers means by which quantum nonequilibrium may be prevented from relaxing fully.
The method of the drift-field is introduced.
A systematic treatment of nodes is given, including some new results.
Quantum states are categorized cleanly according to global properties of the drift-field, and a link is made between these and quantum state parameters.
A category of quantum states is found for which relaxation is significantly impeded, and may not complete at all.

Chapter \ref{4} considers the possibility that primordial quantum nonequilibrium may be conserved in the statistics of a species of relic cosmological particle.
It discusses necessary factors for this to happen in a background inflationary cosmology.
Illustrative scenarios are given both in terms of nonequilibrium particles created by inflaton decay, as well as relic vacuum modes for species that decoupled close to the Planck temperature. 
Arguments are supported by numerical calculations showing both perturbative couplings transferring nonequilibrium between fields, and also nonequilibrium causing standard quantum measurements to yield anomalous results.

Chapter \ref{5} examines a practical situation in which it is proposed that there is a potential to observe quantum nonequilibrium directly, namely the indirect search for dark matter.
The search for so called `smoking-gun' spectral lines created by dark matter decay or annihilation is argued to be a particularly promising setting for the detection of quantum nonequilibrium.
General arguments show some unintuitive nonequilibrium signatures may arise related to the contextuality of quantum measurements.
Telescopes observe nonequilibrium phenomena on the scale of their own resolution, for instance. 
So low resolution telescopes are better placed to detect nonequilibrium.
There is significant reason to believe that different telescopes will disagree on the spectrum they observe.
Telescopes may detect line widths to be narrower than they should be capable of resolving. 
If such a suspected source of quantum nonequilibrium were found, its subjection to a specifically quantum mechanical test would confirm or deny the presence of the quantum nonequilibrium conclusively. 

\chapter*{DEDICATION}
My gratitude is extended first and foremost to Antony Valentini, whom I met a lifetime ago at Imperial College London.
Antony's daring but scrupulous approach to science has been a constant inspiration.
If there are parts of this thesis that approach good scientific English, then a significant portion of the credit should be aimed in his direction.
(The other part goes to BBC Radio 4.)
I understand how stubborn I can be when I think I am in the right.
That Antony managed to deal with this with composure is a credit to his presence of mind.

During my time on the PhD program at Clemson University, apart from Antony, there were only a few people with whom I could meaningfully engage in scientific dialogue on fundamental physics.
These were Philipp Roser, Adithya `P.K.' Kandhadai, and Indrajit Sen.
All three have my gratitude, my best wishes for the future, and a standing invitation to collaborate. 
I extend my thanks also to Lucien Hardy, who hosted me during my stay at Perimeter Institute. 
I very much enjoyed my stay at PI. 
My horizon's were indeed extended by the experiences I had during my stay.
My conversations with Lucien and other residents have helped to provide nuance and context to my view of physics.
This has already proven to have left a lasting effect on the direction my research is taking.
Finally I thank Amanda Ellenberg. 
Amanda kept me afloat when times took a turn for the worse.

\addcontentsline{toc}{chapter}{DEDICATION}
%%%%%%%%%%%%%%%%%%%%   END OF STUFF FOR THE FIRST FEW PAGES %%%%%%%%%%%%%%%%%%%%%%%%%%%%%%%%%%%%%%%%%%%%%%%%%%%%%%%%%%%%%%%%%%%%%%%

\setcounter{secnumdepth}{1}
\setcounter{tocdepth}{2} 

\newpage
\tableofcontents
\addcontentsline{toc}{chapter}{TABLE OF CONTENTS}
\newpage
\listoffigures
\addcontentsline{toc}{chapter}{LIST OF FIGURES}
%\newpage
%\listoftables
%\addcontentsline{toc}{chapter}{LIST OF TABLES}
\clearpage
\afterpage{\blankpage}

\setcounter{page}{0}
\pagenumbering{arabic}
\chapter{INTRODUCTION}
\subsection{A didactic history of quantum physics}
Somewhere in the mid-1920s quantum physics changed beyond recognition.
In the old quantum theory \cite{H67} of the previous few decades, the word `quantum' had referred to some relatively anodyne proposals to make previously continuous quantities discrete. 
Planck had proposed quanta of electron excitations in a black body \cite{Planck}.
Einstein had proposed quanta of the electromagnetic field \cite{E1905_light_quanta}.
Bohr had proposed quanta of atomic excitations \cite{B1913}.
But by the early 1930s what had emerged was fundamentally different. 
Quantum physics had turned operational.
A theory of experiments.
And it was spectacularly powerful.

What emerged was not any specific theory, but a list of changes that could be made to a classical Hamiltonian theory in order to `make it quantum'.
To `quantize' it. 
This \emph{verb-ification} of the adjective quantum is perhaps most attributable to Paul Dirac, who, upon receiving a proof copy of Heisenberg's \emph{\"{U}ber quantentheoretischer Umdeutung} \cite{H25} from his PhD advisor in August 1925 \cite{Dirac_biography}, set to work searching for method by which the canonical variables and equations of classical mechanics could be translated to fit the emerging quantum mechanics. 
His solution was to instruct the canonical variables of classical mechanics to obey an algebra of non-commutating $q$-numbers. It was published only three months later, in November 1925 \cite{D25}.
Although there would be many more contributions over the next couple of years \cite{BV09,Beller}, Dirac had laid the foundations for what would eventually be called canonical quantization. 
Technical considerations aside, he had made it possible to convert any classical Hamiltonian theory into a quantum theory.
A result that would be popularized in his 1930 textbook \cite{Dirac}.

It is important to bear in mind that even in the early days, quantum mechanics was not a single theory, but a set of changes that could be made to a classical theory.
And in those early days the wave vs.\ particle debate was very much still alive. 
It was considered necessary not just to apply the new quantum mechanics to particles (electrons), but also the electromagnetic field.
After all, Einstein had proposed the existence of light quanta (later called photons) some 20 years earlier \cite{E1905_light_quanta}.
The first attempt to apply the nascent quantum mechanics to the electromagnetic field was due to Born, Heisenberg, and Jordan in 1926 \cite{BHJ26}, who treated the free field.
The first treatment of an interacting field theory was due to Dirac in February 1927 \cite{D27a}, eight months before the fifth Solvay conference. 
Dirac coupled the electromagnetic field to an atom and used this to calculate Einstein's $A$ and $B$ coefficients for the spontaneous emission and absorption of radiation.
Unlike the particle-based quantum mechanics, this field-based approach had the ability to describe the creation and destruction of particles (photons).
This was a fantastic outpost from which to develop further quantum field theories, but it would be left for others to do so. 
Dirac himself resisted the idea that a quantum field theory was necessary for any particle other than the photon \cite{W77}.
Instead he spent the next few years developing earlier suggestions from Walter Gordon \cite{G26} and Oscar Klein \cite{K27}, in order to develop a relativistic Schr\"{o}dinger equation for electrons \cite{D28a,D28b}.
Though this theory had some marked success, particularly in the prediction \cite{D30} of the anti-material positrons that would be discovered in 1932 \cite{A32}, the `sea model' he introduced had some conceptual difficulties that never fully convinced the community \cite{W77}.
The next steps towards finally ending the wave vs.\ particle debate were made in an important series of papers in the late 1920s by Jordan and Wigner \cite{JW28}, Heisenberg and Pauli \cite{HP29,HP30}, and Fermi \cite{F29}.
These introduced the idea that, just like photons, material particles could also be understood to be the quanta of their own individual fields \cite{W77}. 
This claim was further substantiated in 1934 by Furry and Oppenheimer \cite{FO34}, and Pauli and Weisskopf \cite{PW34}, who showed that quantum field theories naturally account for the existence of Dirac's antiparticles \cite{W77} without the interpretive problems of Dirac's sea model.
In a manner of speaking, this was the final realization of de Broglie's 1924 thesis \cite{dB25}.
Matter particles would henceforth be treated on the same footing as photons, as excitations (quanta) of their own respective fields. 

This branch of twentieth century fundamental physics would continue to be developed from much the same viewpoint.
Classical field theories were written down in the hope that, when quantized, they could describe the particle physics developments of the day.
Attempts would be guided both by the need to remove stray infinities through the process of renormalization and the group theoretic analysis of gauge-symmetries. 
Important milestones were the perfection of quantum electrodynamics (QED) in the late 1940s, the discovery of the Higgs mechanism in the mid-1960s, the development of electroweak theory (EW) in the late 1950s and 1960s, and the maturity of quantum chromodynamics (QCD) and the Standard Model (SM) in the 1970s.
The crowning achievement to date is the discovery of the Higgs particle by the ATLAS and CMS experiments at the Large Hadron Collider in 2012. 

To the authors mind at least, the historical development of quantum physics quite naturally evolves into field theory and particle physics. 
De Broglie's proposal to unify light and matter into the same framework had been achieved by the mid-1930s.
It was called quantum field theory.
And our contemporary understanding of elementary particles is rooted in this principle.
Nevertheless, the process of quantization had made it possible to describe non-relativistic point particles, and for a majority of physicists and chemists, this was sufficient.
The moral to this didactically-skewed account of the history of quantum physics is twofold.
First, quantum theories of physical systems are simply classical theories that have had a complicated procedure applied to them. So it seems natural that these classical theories guide research in quantum foundations.
Second, our best classical theories are currently field theories, so the mythical correct answer is more likely to resemble a field theory than the non-relativistic point particle. 
So although the non-relativistic point particle type theories may be useful for the purposes of explanation, it might be unwise to place too much faith in such approaches.

Of course the theories that emerged out of this history--this complicated quantization procedure--were not classical any more, but something radically different.
Although they retained some of the structure the format had changed entirely.
The classical realism that dealt with entities moving around with definite states had transformed to an operational theory of measurement inputs and measurement outputs. The next section will elaborate on this point.

\subsection{On operationalism and realism}
The discipline called quantum mechanics, that emerged out of 1920s, is an operational theory of experimental measurements.
Its inputs, the quantum state and the Hilbert space operators, correspond to laboratory procedures for how to run an experiment. 
Its output is the likelihood of obtaining any particular measurement outcome.
To an empiricist, this is a very functional form for a theory to take.
It certainly found a comfortable home in the particle physics that arose in the 1930s.
The bread and butter calculations of particle physics are quantum field theoretic perturbative scattering cross-sections.
These very naturally fit into the operational framework, possessing as they do very definite inputs (the particles colliding) and outputs (the product particles).   

Quite what quantum mechanics means in the absence of experiments is more problematic however. 
Traditionally, the presumption of an experiment is couched in the language of `observers', `observables', and `observations'.
All of which are loaded terms that may be attacked from multiple directions.
Without an observer, quantum mechanical theories seemingly do not make any predictions at all.
And while this may not cause any particular problem to an empiricist working a particle collider, nature is supposed to be fundamentally and universally quantum mechanical, so the rules of quantum mechanics should be able to be applied even in the absence of observers. 
Much early Universe cosmology, for example, relies upon quantum theoretic particle scattering cross-sections. But the interactions being described took place in the early Universe long before human observers. 
So the question over when the collapse of the quantum state took place becomes problematic.
Indeed to some, the presence of an observer inside a scientific theory goes against the spirit of the Copernican revolution, placing humans once more at center-stage. 
And of course, famously, the concept of an observer is not well defined in any case, as highlighted by Schr\"{o}dinger in his well known cat argument \cite{S35}. 

All these points tend to be grouped together under the heading `the measurement problem'.
The measurement problem is old news however.
As Matt Liefer put it recently, ``Any interpretation of quantum mechanics worth its salt has solved the measurement problem'' \cite{L_quote}.
There are of course many diverse interpretations of quantum mechanics, of varying levels of wackiness. 
And since a solution to the measurement problem is a key goal for these interpretations, it is solved through many diverse and sometimes wacky means.
In principle the solution is simple however.
The measurement problem is introduced as a consequence of the operational phrasing of conventional quantum mechanics. 
The language of experiments and measurements.
And arguably the easiest way to deal with the problem is not to introduce it in the first place.
At least that is what Ockham's razor appears to suggest.
If the predictions of quantum mechanics were reproduced with a classical-style realist theory, then in a sense, the measurement problem would not be introduced in the first place and so would need to be fixed. 

To be clear, the operational methodology does have distinct and powerful advantages when applied to certain problems. 
It has lead to fascinating and useful theorems like no-cloning \cite{WZ82,D82} and quantum teleportation \cite{quantum_teleportation,first_experimental_quantum_teleportation}. 
A major thrust of the field of quantum information science is predicated upon the development of a framework for computation that is independent of the actual implementation, ie.\ whether the qubits in question are trapped ions, quantum dots, NMR, squeezed light etc.
For this reason it is acutely operational.
This has given rise to an industry in quantum computation, with many small demonstration quantum computers to date.
A firm milestone in this effort would be a quantum computer capable of running Shor's algorithm to break private key encryption.
And a realistic application of this would require the construction of a machine that is five orders of magnitude larger and two orders of magnitude less error prone than those currently operational. (See for instance page S-3 of reference \cite{NAS_quantum_computing_progress_report}.)

In the foundations of quantum mechanics, it is very difficult to deny that researchers are guided by their own philosophical predilections. 
Certainly the founders of quantum mechanics made little attempt to hide this.
The struggle of the G\"{o}ttingen-Copenhagen physicists with positivism and the use of the rhetoric of anti-realism are well documented \cite{Beller}.
But in physics one comes to appreciate how the understanding of the same physical phenomena through different means and from different perspectives can be a great aid to insight. 
So any claim that operationalism and positivism are the only terms on which it is possible to understand quantum physics should be accompanied by an abundance of proof. 
Otherwise it is a philosophical assertion, not a scientific one.
There is some history on this point.
For years an argument by von Neumann in his 1932 textbook \cite{vN32} was thought to preclude such a realist account of quantum mechanics.
In actuality, as shown by Bell \cite{B66} and Kochen and Specker \cite{KS67} in the mid-1960s, it precluded only non-contextual realist theories.
It said that sets of measurements performed on a quantum system could not simply yield the values of pre-existing quantities.
The results of quantum measurements are contextual; they depend on the manner in which the measurements are performed.
To this day, the Kochen-Specker theorem (as it became known) is widely thought to describe a very non-classical sort of phenomenon.
From a non-operational perspective, however, it is far less surprising.
After all, in any serious non-operational theory, the operational predictions of conventional quantum mechanics must be reproduced. 
And in order to model a measurement non-operationally, it is necessary to include the measurement apparatus inside the model.
The subject of the measurement and the apparatus are treated as a combined quantum system. 
The measurement itself is a brief period of interaction in which the part of the model labeled apparatus interacts dynamically with the part labeled system, so that the two become entangled. 
Naturally the outcome of any measurement will depend on the finer details of the dynamical interaction between system and apparatus.
How could it not?

John Bell's work in this area was inspired by the apparent violation of von Neumann's no-go theorem by the de Broglie-Bohm theory \cite{B66}, which had been revived roughly a decade earlier in 1952 by David Bohm \cite{B52a,B52b}.
Another property that the de Broglie-Bohm theory had which interested John Bell was its nonlocality. 
And this lead him to ask whether this property should be expected of any realist theory \cite{B66}. 
His eponymous 1964 no-go theorem \cite{B64,CHSH69} said that, yes, realist theories must be nonlocal to reproduce standard quantum mechanical predictions. 
But it went much further than this.
For it showed that \emph{even in conventional quantum mechanics}, through entanglement, distant measurement outcomes may be correlated in a way not possible to explain without a superluminal connection.
It is a funny kind of superluminal connection however, as it disappears on the statistical level, when subject to standard quantum uncertainty. 
So although nonlocal correlations do exist in nature, standard quantum uncertainty serves to obscure them, making them impossible to use for any practical signaling. 
For this reason the nonlocality has been called uncontrollable \cite{S83,W05}.
(There are however other proposals and means by which the nonlocality may be put to use \cite{BCT99,BBT05,JW05}.)
This apparent `peaceful coexistence' \cite{S83}, or perhaps maybe `extraordinary conspiracy' \cite{AV92}, between nonlocal correlations and quantum probabilities is surely the central enigma of Bell's theorem. 

In this thesis, Bell's theorem is not tackled directly. The author does have a paper in preparation \cite{NU_straw_men_and_covert_assumptions} that aims to clarify a realist perspective on this matter, though it has not yet undergone the usual sanity checks and so is not included here.  As explained below, the author has wished to take this work in a different direction. 
Beforehand though, a word of caution on how to interpret Bell's theorem.
A strong prevailing opinion among a majority of physicists is that Bell's theorem rules out a realist account of quantum mechanics as such an account would be required to be nonlocal. 
In fact Bell's theorem shows that \emph{textbook quantum mechanics is nonlocal}, and Bell experiments show that \emph{nature is nonlocal}.
So it should not be a surprise that realist descriptions should be correspondingly nonlocal. 
While it is quite unsettling to deal in theories that appear to violate relativity, the blame for this should be squarely leveled at conventional textbook quantum mechanics. And for that matter nature. These correlations are after all a proven fact \cite{ADR82,BCP14}. 

\subsection{Thermodynamic origin of quantum probabilities}
One irksome feature of textbook quantum mechanics is that it gives no explanation for the existence or character of the Born rule of quantum probabilities. 
The Born rule holds the status of an axiom, and is therefore beyond question or explanation \cite{Dirac,vN32}.
Or at least that is how the argument is presented sometimes. 
It is important to bear in mind that in the entirety of science save quantum physics, there is a general assumption that probability distributions are associated with a lack of knowledge.
Perhaps the probabilities express the consequences of uncertain initial conditions. 
Perhaps they reflect the use of an effective theory that averages over underlying degrees of freedom.
Perhaps these degrees of freedom are too numerous or too chaotic to be practically modeled.
Perhaps the underlying dynamics are not yet known, and the probabilities reflect ignorance of these dynamics.
The late physicist and probabilist Edwin Jaynes, to whom several of the results of chapter \ref{2} owe a debt of gratitude, had a different opinion on the source of quantum mechanical probabilities.
He once wrote in polemic, 
\begin{quote}
``In current quantum theory, probabilities express our own ignorance due to our failure to search for the real causes of physical phenomena; and, worse, our failure even to think seriously about the problem. This ignorance may be unavoidable in practice, but in our present state of knowledge we do not know whether it is unavoidable in principle; the `central dogma' simply asserts this, and draws the conclusion that belief in causes, and searching for them, is philosophically na\"ive. If everybody accepted this and abided by it, no further advances in understanding of physical law would be made...'' \cite{Jaynesbook}
\end{quote}
The question of the origin of quantum probabilities has remained largely ignored by modern physics.
It may be that the question does not seem important from an operational perspective, but for a realist the issue is thrown into sharp relief.
Consider for instance the consequences of taking even a conservative realist approach, through conserving determinism. 
Of course any single quantum measurement can only ever yield a single outcome. 
And clearly realism and determinism mean that this outcome is predetermined by some unknown initial conditions.
But why the initial conditions should be such that the Born distribution is reproduced? 
Why not any initial conditions?
One intuitively imagines dividing up any particular ensemble to create new ensembles with new and arbitrary statistics.  

The first time a realist account of quantum mechanics was brought to a mass audience in the post-war era was in 1952 by the works of David Bohm \cite{B52a,B52b}.
In this formulation, the Born distribution was simply postulated. 
This assumption was subsequently and rightly criticized by Pauli \cite{P53} and others \cite{K53}.
In the following years Bohm and collaborators made attempts to address this concern, to find ways to remove this unnatural postulate from the theory \cite{B53,BV54}. 
The primary means explored was a sort of thermodynamic relaxation.
In such a view the Born distribution would represent a kind of thermodynamical equilibrium.
An ensemble selected to have an arbitrary $\rho(x)$, would then tend to relax towards $\rho(x)=|\psi(x)|^2$ as time progresses. 
In this way, the possibility of having arbitrary $\rho(x)$ could be reconciled with the large body of evidence for $\rho(x)=|\psi(x)|^2$.
Of course there were many such distributions in classical physics from which to take precedent.
If the Born distribution is akin to the Maxwell-Boltzmann distribution say, then quantum mechanics is akin to equilibrium thermodynamics, and realism (hidden variables) is akin to the underlying kinetic theory. 
Or so the idea went. But it was yet to be understood how such a relaxation could occur. 
In order to bring about this relaxation Bohm and collaborators attempted adding both random collisions \cite{B53} and later irregular fluid fluctuations \cite{BV54} to the dynamics.
But although these attempts may have been well motivated, neither is taken seriously today, as they have been shown to be unnecessary. 

So in the latter part of the 20th century, the de Broglie-Bohm theory had what was perceived to be a fairly major naturalness problem. 
And Bohm's attempts to address the concern were not considered compelling by many. 
Even a major textbook disregarded these attempts and instead postulated agreement with the Born rule outright \cite{Holl93}.
The big question still remained over how to remove this unnatural postulate and what to replace it with. 
In 1991 Valentini answered this question in a remarkable way. 
His answer was to remove the postulate, and replace it with \emph{nothing} \cite{AV91a,AV91b}.
Valentini showed that the thermodynamic relaxation to the Born distribution would take place without the need to add extra elements into the theory. 
The mechanism to cause the relaxation had been there all along, unrecognized in plain view.

The principles behind the relaxation are expounded in detail in chapter \ref{2}, which may be considered one half introduction to the subject, and one half love letter to thermodynamic relaxation. 
And the introduction to that chapter specifically warns of the nuances of particular classical thermodynamic analogies. 
Nonetheless, the relaxation is present in the foundations of classical mechanics too. 
So, if a short descriptive classical analogy is what is desired, the author recommends the following.

Suppose a classical gas is confined to some volume and is composed of $N$ non-interacting particles. 
As these are non-interacting, they may be considered to be an $N$-member ensemble of individual systems.
The state of each individual system, the analogy of the underlying realist (hidden variable) state, is represented by a phase space coordinate of the particle in question, $(q,p)$.
As each of these individual states may in principle be anything, the overall ensemble distribution of these phase-space points, $\rho(q,p)$, is similarly unconstrained. 
There is no reason for instance why an ensemble could be prepared in which all particles had very similar phase space coordinates, $\rho(q,p)\sim \delta(q-q_\text{average})\delta(p-p_\text{average})$. 
But as time evolves, experience says that the gas molecules will become spread out within the volume in which they are confined. 
This corresponds to $\rho(q,p)$ becoming uniform in position $q$.
And clearly an entropy rise will result.
However even this example hides the true cause of entropy rise, instead making an appeal to experience. 
And the state of maximum entropy is never truly reached, for in the case of classical mechanics, the equilibrium distribution (corresponding to entropy \eqref{differential_entropy}) is uniform in both $q$ and $p$.
And of course to be uniform in $p$ would be to be unbounded in energy, and so full relaxation is prevented by energetic considerations.
The exact mechanism is explained in depth with the use of a longer classical analogy in section \ref{2.3}.

A major theme in chapter \ref{2} is that even upon an abstract state space, with unknown laws, thermodynamics and entropy rise may be thought to derive from a very basic principle of entropy conservation. 
Relaxation appears to be contingent only upon this factor in both classical mechanics and de Broglie-Bohm quantum mechanics.
That entropy conservation could lead to entropy rising sounds like a contradiction in terms.
But this very mechanism exists in classical mechanics, where the conservation of entropy law is named Liouville's theorem.
And the classical H-theorem that proves entropy rise goes through with only reference to Liouville's theorem.
(As explained in sections \ref{2.3} and \ref{2.4}, strictly speaking it is only a \emph{de facto} practical notion of entropy that rises as a result of demanding entropy conservation.) 
Crucially, the tendency of entropy to rise does not appear contingent on the particular form of Hamilton's equations. 
Liouville's theorem ensures that the dynamics is incompressible, but goes no further than this. 
However it is only this incompressibility that is needed to ensure relaxation.
So it is natural to place Hamiltonian mechanics in a larger category of theories that all share Liouville's theorem, but not the exact form of Hamilton's equations.  
All such theories in this category would be expected to relax. 
In chapter \ref{2}, it is shown how this category of theories may be generalized to make them relax towards any arbitrary equilibrium.
This of course is motivated by the prospect of regarding the Born distribution as an equilibrium distribution.
One mode of reasoning in this direction is shown to lead to de Broglie-Bohm style quantum physics. 
And of course the de Broglie-Bohm theory relaxes to the Born distribution.
The insights that such an approach may provide are discussed, and generalizations of de Broglie-Bohm theory commented upon.

\subsection{Experimental evidence for quantum nonequilibrium}
By definition, quantum nonequilibrium produces distinct predictions to that of textbook quantum mechanics.
Clearly this raises the prospect that realist theories with this sort of relaxation property could be experimentally verified.
And the prospect of experimental evidence for an interpretation of quantum mechanics is so rare that it begs to be investigated\footnote{The other example is collapse theories \cite{G18,GRW86,P96}.}.
Indeed it becomes debatable whether relaxation theories should qualify as interpretations or independent theories in their own right.
The main purpose of this thesis is to contribute to the ongoing effort to detect quantum nonequilibrium.  
As the main thrust of this effort is being carried out using de Broglie Bohm theory, for the sake of consistency in the literature, it is helpful to follow suit. 
After all de Broglie-Bohm theory is the archetypal quantum mechanical relax theory%
\footnote{De Broglie Bohm theory is justified in ways that it has not been possible to cover in any great depth. One that comes to mind is its relation to Hamilton-Jacobi theory that was originally highlighted by Louis de Broglie \cite{deB28}, though is also covered in some depth in references \cite{AV92,Holl93}.}.
So although chapter \ref{2} develops relax theories on a general footing, the de Broglie-Bohm theory is employed in chapters \ref{3} through \ref{5}. 

Of course, if quantum nonequilibrium were found commonly in regular tabletop laboratory experiments, it would be known widely by now. 
Classical nonequilibrium on the other hand is an experience of daily life. 
To account for this discrepancy between the prevalence of quantum and classical nonequilibrium, one need only appeal to the rate at which quantum relaxation has been shown to take place numerically.
Although the theoretical background for quantum relaxation was published by Valentini in 1991 \cite{AV91a}, it took until 2005 in reference \cite{VW05} for it to be demonstrated numerically. 
Since then quantum relaxation has been demonstrated multiple times \cite{EC06,CS10,SC12,TRV12,ACV14}.
Its presence and efficacy are a matter of course for those who deal with numerical simulations of de Broglie-Bohm systems.
Figure \ref{quantum_relaxation} illustrates the process for a simple two-dimensional harmonic oscillator.
To put a ballpark numerical value to the relaxation timescale, in this figure, for a superposition of nine low energy modes, relaxation is achieved to a high degree with only the 19 wave function periods displayed.
This is the blink of an eye for a typical quantum system.

There is still however quite limited knowledge on whether to expect this sort of relaxation timescale generally. 
And there still remain significant outstanding questions over whether equilibrium will invariably be obtained for all systems.
While the formalism developed in section \ref{2} proves that relaxation should be expected to take place for all but the most trivial systems, it does not show whether these systems may always be expected to actually reach equilibrium.
It has already been mentioned that although classical systems relax, they are inevitably prevented from reaching their equilibrium distribution (uniform on phase space) by conservation of energy.
Energy conservation provides a barrier that prevents full relaxation in classical mechanics.
It appears natural to ask whether there exist analogous barriers to quantum relaxation.

Certainly for small systems with limited superpositions, and relatively mild initial nonequilibrium, relaxation to equilibrium has been shown to take place remarkably quickly \cite{AV92,AV01,VW05,EC06,TRV12,CS10,SC12,ACV14}. 
The speed of the relaxation appears to scale with the complexity of the superposition \cite{TRV12}, while not occurring at all for non-degenerate energy states \cite{Holl93}.
However, the study of such systems has historically been motivated by the wish to demonstrate the validity of quantum relaxation. 
And such systems are also the least computationally demanding. 
In light of these two facts, it could be argued that the study of such systems represents something of a selection bias \cite{NU18}.
It would be interesting to test the boundaries of what is known about relaxation.
Would equilibrium be reached for larger systems? Or for systems with only mild excitations or mild deviations from an energy state? Or for systems with more extreme forms of nonequilibrium?  
Is it possible that, even within the parameter space tested, there exist systems for which relaxation is prevented?
Because all of these possibilities might conceivably open the window a little wider for nonequilibrium to persist, the research field has rather changed its modus operandi.
In the first decade of this century, a major focus had been on demonstrating the process of relaxation. 
Now that this has been done, the focus has rather reversed, so that attention is focused upon various means by which full relaxation may be slowed or prevented, thereby illuminating possible ways for nonequilibrium to be experimentally detected. 
In recent years some authors have begun to make inroads on this question.  
For instance one collaboration concluded that for modest superpositions there could remain a small residual nonequilibrium that was unable to relax \cite{ACV14}. 
Another found that for states that were perturbatively close to the ground state, trajectories may be confined to individual subregions of the state space, a seeming barrier to total relaxation \cite{KV16}.
Chapter \ref{3} is a further contribution to this effort.

It should be noted that through the direct simulation of a quantum system, while it is possible to conclude that equilibrium is reached, it is seemingly difficult to conclude that the system will never reach equilibrium.
To clarify, with finite computational resources it will only ever be possible to numerically evolve a system through some finite time $t$.
Either the system reaches equilibrium in this time interval or it does not. 
If it does, then equilibrium has been shown to be reached. 
If it does not, then equilibrium could still be reached for intervals greater than $t$.
At least for the most studied system, the two-dimensional isotropic oscillator\footnote{The two-dimensional HO is also mathematically equivalent to a single uncoupled mode of a real scalar field \cite{AV07,AV08,ACV14,CV15,UV15,KV16}, granting it physical significance in studies concerning cosmological inflation scenarios or relaxation in high energy phenomena (relevant to potential avenues for experimental discovery of quantum nonequilibrium).}, state of the art calculations (for instance \cite{ACV14,KV16}) achieve evolution timescales of approximately $10^2$ to $10^4$ wave function periods and manage this only with considerable computational expense.

To provide a means by which this computational bottleneck may be avoided, the method of the `drift-field' is introduced in chapter \ref{3}.
This allows access to relaxation timescales that are far in excess of those previously available.
In principle these timescales could be unlimited. 
The logic may be summarized as follows. 
For the typical oscillator states that have been studied, the dynamics are complicated and chaotic within the bulk of the Born distribution, but are simple and regular outside of it (`in the tails') \cite{EKC07,ECT17,NU18}.
And while relaxation is generally expected to proceed efficiently in the bulk of the Born distribution, much less is known of what may happen to a nonequilibrium distribution that is concentrated far away in the tails.
Of course, in order for a distribution in the tails to reach equilibrium, it must first relocate to the bulk. 
So one might think to crudely divide the relaxation timescale for such a system into two parts, $t_\text{reach bulk}$ and $t_\text{efficient relaxation}$, so that the total relaxation timescale is $t_\text{total relaxation time}\sim t_\text{reach bulk}+t_\text{efficient relaxation}$.
But in all systems studied $t_\text{reach bulk}\gg t_\text{efficient relaxation}$, so that in rough order of magnitude estimates the latter may be reasonably neglected, thus $t_\text{total relaxation time}\sim t_\text{reach bulk}$.
Hence it becomes relevant to calculate the time taken for systems to reach the bulk. 
And to this end the drift field exploits the regular property of such trajectories to create a chart, as it were, of the slow migration or drift of trajectories through the space. 
As explained in detail in chapter \ref{3}, the drift field may be further used to cleanly categorize quantum states into those that feature a mechanism to bring about relaxation, and those that have a conspicuous absence of this mechanism.
For those quantum states in which the mechanism is absent, the timescale $t_\text{reach bulk}$ may be far larger than would otherwise be the case, and may even be unbounded.
The result is that for such cases quantum relaxation may be severely retarded, or possibly even prevented outright. 
Further work remains to be done to clarify this matter, but chapter \ref{3} may aid in this process by providing means through which the quantum state parameters of such relaxation retarding states may be calculated.

Despite the ongoing work on the mechanics of quantum relaxation, and the many unanswered questions, it seems reasonable to expect that quantum nonequilibrium will have decayed away for all but the most isolated systems.
After all, relaxation speed appears to increase with the complexity of the quantum state. 
And the particles that are available to routine experiment have a long, complex and violent astrophysical history, allowing them ample opportunity to relax fully.
So the question arises of where to look for effects of quantum nonequilibrium, and attention soon turns to cosmological matters.
One main thrust of current research entails nonequilibrium signatures that could have been frozen-in to the cosmic microwave background (CMB).
According to current understanding, the observed one part in $10^5$ anisotropy in the CMB was ultimately seeded by quantum fluctuations during an inflationary era. 
So these anisotropies provide an empirical window onto quantum physics in the early Universe. 
As shown by references \cite{AV07,AV08,AV10,CV13,CV15,AVbook}, on an expanding radiation-dominated background, relaxation may be suppressed at long (super-Hubble) wavelengths.
And a treatment of the Bunch-Davies vacuum indicates that relaxation is completely frozen during inflation itself \cite{AV07,AV10}.
Thus, in a cosmology with a radiation-dominated pre-inflationary phase \cite{VF82,L82,S82,PK07,WN08}, the standard CMB power-law is subject to large-scale, long-wavelength modifications \cite{AV07,AV08,AV10,CV13,CV15,VPV19}.
And this may be consistent with the observed large-scale CMB power deficit observed by the Planck team \cite{PlanckXV,PlanckXI}.
Potentially primordial quantum nonequilibrium also offers a single mechanism by which the CMB power deficit and the CMB statistical anisotropies may be explained \cite{AV15}. 
This scenario is currently being subjected to extensive statistical testing \cite{VPV19}.

There are also windows of opportunity for quantum nonequilibrium to be preserved for species of relic particles.
And this prospect is developed in chapter \ref{4}. 
For instance, as well as imprinting a power deficit onto the CMB, the decay of a nonequilibrium inflaton field would also transfer nonequilibrium to the particles created by the decay. 
And since such particles are thought to constitute almost all the matter content of the Universe, it seems conceivable that nonequilibrium could be preserved for some subset of these particles. 
Even if inflaton decay did create nonequilibrium particles though, in order to avoid subsequent relaxation resulting from interactions with other particles, these particles would have to be created after their corresponding decoupling time $t_\text{dec}$ (when the mean free time between collisions $t_\text{col}$ is larger than the expansion timescale $t_\text{exp}=a/ \dot{a}$).
Nevertheless, simple estimates suggest that it is in principle possible to satisfy these constraints for some models.
This scenario is illustrated for the particular case of the gravitino, the supersymmetric partner of the graviton that arises in theories that combine general relativity and supersymmetry.
In some models, gravitinos are copiously produced \cite{EHT06,KTY06,ETY07}, and make up a significant component of dark matter \cite{T08}. 
(For recent reviews of gravitinos as dark matter candidates see for example references \cite{EO13,C14}.)
As well as considering decay of a nonequilibrium inflaton, chapter \ref{4} also considers the prospect that nonequilibrium could be stored in conformally coupled vacuum modes. 
Particle physics processes taking place upon such a nonequilibrium background would have their statistics altered.
Particles created on such a background for instance would be expected to pick up some of the nonequilibrium.
The act of transferal of nonequilibrium from one quantum field to another is illustrated explicitly through the use of a simple model of two scalar fields taking inspiration from cavity quantum electrodynamics.
Finally, it is shown that the measured spectra of such fields may in general be expected to be altered by quantum nonequilibrium.

This final observation is the foundation upon which chapter \ref{5} is built. 
The purpose of chapter \ref{5} is to consider observable consequences and signatures of quantum nonequilibrium in the context of ongoing experiment. 
If the dark matter (whatever its nature) does indeed possess nonequilibrium statistics, then it may be that this could turn up in experiment.
But at present there is little known about what a nonequilibrium signature could look like, and so such signatures would likely be overlooked or misinterpreted. 
In order to begin the development of the theory of nonequilibrium measurement, a particularly simple example is taken.
In many dark matter models, particles may annihilate or decay into mono-energetic photons \cite{BS88,Rudaz89,B97,B04,Creview16,DW94,Abazajian01}, and a significant part of the indirect search for dark matter concerns the detection of the resulting `smoking-gun' line-spectra with telescopes capable of single photon measurements \cite{fermi3.7,fermi5.8,Pull07,Mack08,W12,SF12-1,Albert14,HESS16,Bulbul14,Boyarsky14,A15,Hitomi16}.
For the present purposes, this represents a particularly appealing scenario for detection of nonequilibrium.
For if the dark matter particles were in a state of quantum nonequilibrium prior to the decay/annihilation, this would likely be passed onto the mono-energetic photons produced in the decay.
And if these photons were sufficiently free-streaming on their journey to the telescope, then they could retain their nonequilibrium until they are ultimately measured. 
Hence the effects of nonequilibrium would be imposed on an otherwise extremely clean line spectrum, and it is hoped that this may aid in the process of detection.
In order to model the effects of quantum nonequilibrium upon such line spectra, one may consider the energy measurement of a nonequilibrium ensemble of mono-energetic photons.
General arguments show that it is not the physical spectrum of photons that is affected by nonequilibrium, but instead the interaction between each photon and the telescope.
It is natural, therefore, to think of the role of nonequilibrium to be to alter the telescope's energy dispersion function, $D(E|E_\gamma)$.
This is of course related to contextuality, and has some counter-intuitive outcomes. 
Quantum nonequilibrium will be more evident in telescopes of lower energy resolution for instance, for if the width of $D(E|E_\gamma)$ is larger than the intrinsic broadening of the line, then the effects would appear on this larger and more conspicuous lengthscale.
It also appears possible to observe a line that is narrower than the resolution of the telescope should allow for. 
And more generally, it is reasonable to expect different telescopes to react to the nonequilibrium in different ways.
Which could lead to situations in which telescopes could be seen to disagree on the shape or existence of a line. 
Chapter \ref{5} develops these ideas and, in order to provide explicit calculations displaying these effects, develops the first model of de Broglie-Bohm measurement on a continuous spectrum.
It concludes by summarizing results with reference to some recent controversies concerning purported discoveries of lines in the $\gamma$ and $X$-ray ranges. 

\subsection{Structure of this thesis}
The substance of this thesis is covered in chapters \ref{2} through \ref{5}, each of which are largely works in their own right.
In order to aid the reader, each chapter features a non-technical foreword that is intended to help place the chapter in the wider context of the thesis, and explain its relevance to quantum nonequilibrium. 
As each chapter is a work in its own right, each features its own internal conclusion. 
For completeness however, chapter \ref{6} concludes the thesis by summarizing the major conclusions and suggesting promising research directions. 
Chapter \ref{2} is intended to be an introduction to relax theories and to de Broglie-Bohm.
It explains the mechanics behind classical and quantum relaxation in depth.
It considers minimal conditions required to make dynamical theory relax in the usual manner expected in classical physics, and then applies these to quantum physics.
Chapter \ref{3} contributes to the effort to find systems for which total relaxation to equilibrium may be prevented. 
In particular it considers the possibility that quantum nonequilibrium may be prevented for what are defined as `extreme' forms of nonequilibrium.
It introduced the drift field, describes a category relaxation retarding quantum states, and provides possible the most systematic treatment of nodes in de Broglie-Bohm. 
Chapter \ref{4} considers the possibility that quantum nonequilibrium may stored in relic particle species.
Simple estimates performed upon illustrative scenarios suggest that quantum nonequilibrium could indeed have had a small chance to survive.
Chapter \ref{5} argues that the indirect search for dark matter, and in particular the search for smoking-gun line spectra, represents favorable conditions for a detection of quantum nonequilibrium to take place. 
It develops a measurement theory of quantum nonequilibrium in this experimental context and suggests possible signatures of nonequilibrium that could inform future line searches.

\chapter*{PREFACE TO CHAPTER 2}\addcontentsline{toc}{chapter}{PREFACE TO CHAPTER 2}

These days, it is known widely that de Broglie-Bohm quantum systems relax towards the Born distribution, $|\psi(x)|^2$.
But it was not always this way.
De Broglie-Bohm theory was not created to relax.
And relaxation was not added \emph{ex post facto}.
The mechanism that produces relaxation is not the devising of some unscrupulous theorist with a vested interest in bringing about relaxation.
Instead quantum relaxation hid its clandestine existence from scientists for some 70 years.  
It was there all along, hidden in plain sight.  
Neither Louis de Broglie nor David Bohm knew of it in their lifetime. 

In 1991 the existence of quantum relaxation in de Broglie-Bohm theory was \emph{discovered} by Valentini, \cite{AV91a}.
Valentini showed that de Broglie-Bohm theory satisfied a modified version of the classical $H$-theorem \cite{Ehrenfests,Daviesbook}, which proves entropy rise in classical Hamiltonian mechanics. 
So by extending the $H$-theorem to de Broglie-Bohm quantum mechanics, Valentini extended classical thermodynamics to the prototypical realist quantum theory.
Since then, quantum relaxation has been demonstrated many times over, and is a matter of course for those who work with de Broglie-Bohm numerically.
Indeed quantum relaxation appears to be so ubiquitous that one is hard pressed to find (non-trivial) situations in which it is prevented from completing. Chapter \ref{3} is a contribution to the effort to find such instances.
To the author's mind, there are two sincerely compelling aspects of the de Broglie-Bohm theory.
The first is the measurement theory created by Bohm. 
This very naturally produces contextuality as, in order to make a measurement, a dynamical interaction must take place between observer and observed.
It serves as the bridge between realism and universalism, and the operationalist theory of experiments. 
The second is quantum relaxation, which is covered at length in this chapter.

Chapter \ref{2} is intended as an introduction to de Broglie-Bohm theory and relaxation.
But it is a rather unusual one.
It proceeds by regarding the relaxation present in classical mechanics as the central guiding principle around which to build a general theory.
In this regard, relax theories are developed on a general footing, and de Broglie-Bohm theory is obtained as a special case. 
Paradoxically perhaps, the largest conceptual jump in chapter \ref{2} concerns not quantum theory, but some long accepted facts regarding relaxation in classical mechanics. 
Indeed, when attempting to explain the concept of quantum relaxation to the uninitiated, one often finds oneself defending relaxation in classical mechanics. 
Once classical relaxation is understood, the conceptual step to quantum relaxation is only relatively minor. 
For this reason the mechanics of classical relaxation are expounded at length in section \ref{2.3}.
One important point to bear in mind is that the classical $H$-theorem, and thus classical relaxation is contingent upon only one factor. 
Namely, Liouville's theorem.
But Liouville's theorem is supposed to represent conservation of (differential) entropy, which appears antithetical. 
How can the Liouville's theorem (conservation of entropy) be used to prove relaxation (entropy rise)?
Nevertheless the logic is sound.
Any doubts of this should be put to rest, first as this is the very mechanism by which classical system relax, and second by the many demonstrations that have been presented to date.

Finally a point on style.
When discussing probability, a choice must be made between two (and arguably more) contrasting viewpoints.
These are commonly called the `Bayesian' and `frequentist' viewpoints.
Although for the present purposes the difference is superficial, the Bayesian viewpoint, dealing as it does with inferences given prior information, lends itself more readily to a discussion on information in physical theories.
For this reason a Bayesian mode-of-speech is adopted for much of the following discussion.
This will almost certainly be off-putting for readers unfamiliar to non-ensemble based statistics.
Notably, it becomes natural to refer to information and entropy (its measure) interchangeably. 
Every instance in which information shall be said to be increasing, decreasing, or conserved shall refer only to changes in the entropy, its measure.
Readers unfamiliar with the Bayesian mode-of-speech are encouraged to translate the arguments presented into their frequentist analogs.
In particular, if the word `information' appears troublesome in some contexts, try replacing it in your mind with the word `entropy'.
It is really very straightforward reasoning.
Reference \cite{Jaynesbook} is a good resource to further explore this topic.

\chapter{RELAX THEORIES FROM AN INFORMATION PRINCIPLE}
\label{2}
Quantum equilibrium ($\rho=|\psi|^2$), quantum nonequilibrium ($\rho\neq|\psi|^2$) and quantum relaxation (${\rho\rightarrow|\psi|^2}$) are the central concepts of this thesis, and so deserve a correspondingly careful exposition.
Chapter \ref{2} is intended as an introduction to these topics, and to de Broglie-Bohm theory, which will form the backbone of later chapters.
It has also been viewed as an opportunity to expound the author's perspective on these topics.
For this reason, much of the material is new. 
This chapter asks the question, `What is required of a theory so that it retains the classical notion of thermodynamic relaxation, but instead relaxes to reproduce quantum probabilities?'
It answers, `Entropy conservation.'
Merely by imposing entropy conservation (Liouville's theorem) upon an abstract state space, it is possible to recover exactly the kind of relaxation that occurs in classical mechanics.
Time-symmetry is preserved, but yet entropy rise occurs regardless, on a \emph{de facto} basis.
The imposition of entropy conservation on an abstract state space results in a structure, that for the present purposes is dubbed `the iRelax framework'\footnote{A name was necessary for the purposes of utility. But it was difficult to find one that felt appropriate. The name changed several times in writing. In the present iteration the small `i' is intended to refer to information. The author invites suggestions for a better name.}.
This chapter proceeds by example, through repeated applications and generalizations of the iRelax framework.
Both classical mechanics and de Broglie-Bohm quantum mechanics are shown to be special cases of iRelax theories.

Before embarking on the substance of chapter \ref{2}, a point on style.  
Chapter \ref{2} makes liberal use of the word `information'.
This is symptomatic of a Bayesian approach to probability, and will almost certainly be off-putting to those readers more comfortable with a frequentist approach.
Information is of course a loaded term, and the reader is advised to avoid carrying over other notions of information from other areas of physics. 
In quantum information theory, for instance, information refers to that which is embodied by qubits. 
This is not what is meant here.
In the vast majority of instances where the word information is used in chapter \ref{2}, it may be replaced with the word `entropy' without changing the intended meaning.
The inclination to use the word information enters in the following way.
Consider a single abstract system with a single abstract state, $x$.
Suppose that this state is not known exactly, however. 
It will still be possible to represent the state of the system with a probability distribution, $\rho(x)$.
But this probability distribution does not represent the spread of some ensemble, as such an ensemble does not exist. 
Instead it represents a state of knowledge, or ignorance, or `information' regarding the state of the system.
If the distribution is tightly packed around some point, then the state of the system is known to reside only within that small region, and so the information is good.
If it is widely spread, then the information regarding the system's state is less good.
It is of course useful to put a number to the quality of the information on the system's state.
And of course the standard scalar measure of information is entropy. 
A low entropy corresponds to low uncertainty of the system's state. 
A high entropy corresponds to high uncertainty of the system's state.
The information principle upon which chapter \ref{2} is built is simply entropy conservation.
For this reason, use of the word information is really quite pedestrian. 
(And so the reader is encouraged not to dismiss this work for fear of a word.)
It is even quite natural to slip into using the phrase information conservation in the place of entropy conservation.
This is the justification for the name iRelax, wherein the `i' refers to information.
In a classical context the statement of information conservation is Liouville's theorem.
In the other contexts considered, the corresponding statements are simple generalizations of Liouville's theorem.

When explaining quantum nonequilibrium and quantum relaxation, it is not uncommon to encounter an air of suspicion. 
The principles at play are amongst the most primitive in statistical physics, but they are commonly omitted from standard statistical physics courses. 
Though some arguments used do appear in discussions on the emergence of the arrow of time \cite{Daviesbook}.
Nonetheless, even amongst researchers in the field, seemingly innocuous questions like `What is the cause of the relaxation?' can prompt diverse and sometimes misguided responses.
To this end, a couple of points should be stressed at the outset. Firstly, there is nothing `quantum' about `quantum relaxation', except its context.
The mechanism behind quantum relaxation is not unique to de Broglie-Bohm theory. It is a simple statistical effect. A guise of the second law of thermodynamics that is apposite to small dynamical systems.
In section \ref{2.3}, this is shown explicitly in a classical context.
Secondly, the same relaxation occurs spontaneously in any dynamics that conserves information.
Indeed, for this reason, information conservation is regarded as the central guiding principle of this chapter.
Unlike other physical quantities, information is not a property of a physical system itself, but rather a property of the ensemble or of the persons attempting to rationalize the system (depending on whether a frequentist or a Bayesian approach to probability theory is favored). 
This means that unlike other quantities used to measure other sorts of relaxation, information and its measure entropy, are universal.
They may be defined in every context.  
Exactly what is meant by information conservation is context dependent, however, and requires some elaboration.
Nevertheless, the definition is rigorous and useful. 
Indeed, the imposition of information conservation can constrain the space of possible laws of motion to such an extent that some physics may be all-but-derived. In section \ref{2.3}, Hamilton's equations are all-but-derived in this manner. A similar process is used to derive de Broglie-Bohm theory in section \ref{2.5}.

As the exact same type of relaxation does occur in classical physics, researchers in the field are often disposed to introduce quantum relaxation via an analogy with `thermal' relaxation.
Whilst the use of this thermal analogy is eminently defensible (especially when attempting a brief explanation), it can have the unfortunate effect of introducing certain misconceptions.
It should be acknowledged, for instance, that there are multiple notions of equilibrium in classical thermodynamics. 
The prefix `thermal' in the phrase thermal equilibrium is usually reserved for the type of equilibrium established by the zeroth thermodynamical law (a simple equality of temperatures). 
The temperature of an individual thermodynamic system, however, is only defined in the thermodynamic (macroscopic) limit or when the system is able to exchange energy with a heat bath (e.g.\ the canonical ensemble).
Hence, temperature is unsuitable for the small, isolated, quantum systems that are of primary concern here.
The equilibrium available to an isolated system is usually designated `thermodynamic' or `statistical' equilibrium (at least by authors wishing to make the distinction). 
Classical thermodynamics however, defines such equilibria in terms of the stability of macroscopic state parameters. A gas, for instance, is usually said to be in a state of internal thermodynamic equilibrium, when its macroscopic state parameters (temperature, pressure, chemical composition) are static with respect to time \cite{BB}. 
Clearly not just temperature, but all macroscopic state parameters are unsuited to the study of the small quantum systems of interest here.

To establish a notion of equilibrium based upon microscopic parameters, recourse is usually taken in the so-called `principle-of-indifference'\footnote{Sometimes referred to as the `equal a prioi probabilities' postulate.}, which stipulates that each possible (micro)state should be equally likely when in equilibrium. 
The principle-of-indifference may only be unequivocally applied to systems with discrete state-spaces, however.
In going from a discrete to a continuous state-space, the natural extension to `equal likelihoods' is of course the `uniform distribution'.
But there is always a choice of variables to describe a continuous system.
So the question soon arises as to which set of variables the probability distribution should be uniform with respect to.
This point is most famously illustrated by Bertrand's Paradox \cite{B1889,J73}.
As shall be discussed, the ambiguity in the application of the indifference principle to continuous spaces is key to the developments described in sections \ref{2.4} and \ref{2.5}.

The notion of a relaxation from nonequilibrium in de Broglie-Bohm theory was first suggested by David Bohm and collaborators \cite{B53,BV54,BH89}, shortly after Bohm's seminal 1952 works \cite{B52a,B52b}.
This was at least partially to address criticisms raised by Pauli \cite{P53} and others \cite{K53}, who argued that Bohm's original $\rho=|\psi|^2$ postulate was inconsistent in a theory aimed at giving a causal interpretation of quantum mechanics.
In order to bring about relaxation, Bohm and collaborators introduced both random collisions \cite{B53} and later irregular fluid fluctuations \cite{BV54}. 
In these early works, it appears the authors felt it necessary to introduce some kind of stochastic element into the theory in order to bring about relaxation in the otherwise deterministic dynamics.

The modern notions of relaxation and nonequilibrium in de Broglie-Bohm theory come from works of Valentini in 1991 \cite{AV91a,AV91b}. In reference \cite{AV91a}, Valentini argued that the stochastic mechanisms introduced by Bohm and collaborators were unnecessary. The theory relaxed even in their absence. He showed that with only a slight modification, the H-theorem of classical mechanics\footnote{The generalized H-theorem, not the earlier Boltzmann H-theorem. (See reference \cite{Daviesbook}.)} could be applied to de Broglie-Bohm theory. 
In essence, the mechanism that causes relaxation in classical mechanics was shown also to be present in standard de Broglie-Bohm theory.
This claim has since had considerable computational substantiation, first in 2005 by Valentini and Westman \cite{VW05}, and then later in other works \cite{EC06,CS10,SC12,TRV12,ACV14}.
Demonstration of the validity of Valentini's quantum relaxation has two important consequences for de Broglie-Bohm theory.
\begin{description}
\item[Firstly,] It removes the need for the $\rho=|\psi|^2$ postulate that Pauli and others found objectionable. (Without the need to invent another postulate with which to replace it.)
\item[Secondly,] It suggests that nonequilibrium distributions ($\rho\neq|\psi|^2$) are a part of de Broglie-Bohm, and hence that de Broglie-Bohm is experimentally distinguishable from canonical quantum theory.
\end{description}

In this chapter, Valentini's original account of relaxation in de Broglie-Bohm (reference \cite{AV91a}) is expanded in order to demonstrate its generality and in order to discuss other types of theories. Various informational aspects of the mechanism are highlighted in order to underline the role they play in the relaxation. 
Ingredients needed for a physical theory to feature this type of relaxation are discussed.
Relaxation is shown to be a feature of a category of physical theories that share the same underlying framework.
For the present purposes, this is dubbed the iRelax framework.
This minimal framework is shared by various de Broglie-Bohm type quantum theories, as well as standard classical mechanics.
Relaxation in classical mechanics is of course of perennial interest.
It is the archetypal example of time-asymmetry arising from a theory in which time-symmetry is hard-coded.
This is exactly the kind of relaxation that is captured by the iRelax framework.

The chapter is structured as follows.
In section \ref{2.1}, six ingredients for the iRelax framework are listed and briefly described.
This helps structure the discussion in the following sections, which shall proceed by example.
Section \ref{2.2} details the properties of discrete state-space iRelax theories, which do not suffer from the complexities involved in extending the principle of indifference to continuous state-spaces.
In section \ref{2.3}, the formalism is extended to continuous state-spaces by means of so-called differential entropy, equation \eqref{differential_entropy}.
Classical mechanics is presented as an iRelax theory, and the classical counterpart to quantum relaxation is described in detail and illustrated explicitly.
This emergence of the second thermodynamical law prompts a quick digression in order to remark on the arrow of time. 
Then, as coordinates become very important in the next section, the iRelax equivalent of a canonical transform is developed. 
The iRelax structure of classical field theories is also briefly treated.
Although it is undoubtedly useful, the differential entropy employed in section \ref{2.3} overlooks an ambiguity in extending the principle of indifference to continuous spaces.
This oversight is corrected for in section \ref{2.4} through the use of the Jaynes entropy, equation \eqref{Jaynes_entropy}. 
Although this prompts a (minor) correction to the formulation of classical mechanics in section \ref{2.3}, it has a more notable consequence. 
In the interest of internal consistency, it suggests the introduction of a density-of-states, $m$, upon the state-space.
The majority of section \ref{2.4} is spent adapting the formalism of section \ref{2.3} to account for this density of states.
Though a brief geometrical interpretation of the density of states is also highlighted.
Section \ref{2.5} considers how quantum probabilities could arise through the iRelax framework.
A category of de Broglie-Bohm type theories are shown to result from viewing equilibrium distribution $\rho_\text{eq}$ (and in doing so also the density of states $m$) as equal to $|\psi|^2$.
The corresponding entropy is the Valentini entropy, equation \eqref{Valentini_entropy}.
Canonical de Broglie-Bohm, a member of this category, results from a particular choice made in finding a law of evolution that corresponds to the Valentini entropy.
To highlight this fact, the effect of making a different choice is also treated.
The fit of de Broglie-Bohm theory into the iRelax framework is good, but not perfect.
It suggests a possible unification of conventional von Neumann entropy-based quantum theory with the Valentini entropy-based de Broglie-Bohm theory into an over-arching structure.
This possibility is briefly speculated upon.
Section \ref{2.6} then suggests two avenues down which to pursue further research.
Finally, the properties of the iRelax theories covered in this chapter are summarized in table \ref{IPRT_summary}.

\section{Ingredients for the iRelax framework}\label{2.1}
The iRelax framework is the result of imposing information (entropy) conservation upon an abstract state space. 
But in practical situations it is often more useful to infer the state-space and the statement of entropy conservation from other factors, like the distribution of maximum entropy\footnote{In section \ref{2.5}, for a particular choice of integration constants, de Broglie-Bohm theory is obtained through viewing $|\psi|^2$ as a equilibrium (maximum entropy) distribution.}.
For this reason, it is convenient to present the framework as a list of self-consistent ingredients that follow from the state space and information conservation.
The extent to which the different ingredients inter-depend and imply one another is discussed throughout the chapter.
\begin{ingredient}\label{Ing1}
\textbf{A state-space:} $\{x\}=\Omega$\\
The space of all possible states the `system' may hold. 
The word system could be considered to come with interpretive baggage of course, but it may at least be defined rigorously with reference to the state-space, \ref{Ing1}.
By a state-space, only the usual notion is meant.
A single coordinate $x$ in the state-space should correspond to a single system state.
Each individual system should occupy a single state $x$ at any time $t$.  
Given the law of evolution (\ref{Ing2}), a coordinate in the state-space should constitute sufficient information to entirely determine to future state of the system.
\end{ingredient}
\begin{ingredient}\label{Ing2}
\textbf{A law of evolution:} $\dot{x}$\\
In-fitting with \ref{Ing1}, the law of evolution should entirely determine the future evolution of the system, given its current state-space coordinate. In fact it turns out that this requirement of determinism may not be necessary. As proven in section \ref{2.2} for a discrete state-space, determinism and indeed time-reversibility are already implied by information conservation, \ref{Ing5}.
In sections \ref{2.3} and \ref{2.4} the continuous space equivalent will (very nearly\footnote{There is a missing link in the proof. This is due to a divergence of the standard entropy argument $x\log x$ for perfectly defined states $\rho(x)\rightarrow\delta^{(n)}(x-x')$ in continuous state-spaces. This is explained in sections \hyperref[sec:state_propagators]{2.3.2} and \hyperref[sec:modified_state_propagators]{2.4.2}. The missing link is also indicated by a question mark in figure \ref{logic_map}.}) be proven.
Namely that information conservation requires systems to traverse completely determined \emph{trajectories} in the state-space. That it forbids any probabilistic evolution.
\end{ingredient}
\begin{ingredient}\label{Ing3}
\textbf{A distribution of minimal information (maximum entropy):} $\rho_\text{eq}(x)$\\
As entropy is conserved (\ref{Ing5}), any distribution with a uniquely maximum entropy must be conserved by the law of evolution, \ref{Ing1}.
The distribution of maximum entropy/least information may therefore be regarded as a stable equilibrium distribution, hence the notation $\rho_\text{eq}(x)$.
It also lends itself to deduction via symmetry arguments, particularly the principle of indifference.  
In sections \ref{2.2} and \ref{2.3}, the indifference principle is used to define the maximum entropy distribution. 
Later sections feature maximum entropy distributions motivated by the shortcomings of the indifference principle, and so the particular form of $\rho_\text{eq}(x)$ must be otherwise motivated.
\end{ingredient}
\begin{ingredient}\label{Ing4}
\textbf{A measure for information (entropy):} $S$\\ 
Clearly \ref{Ing3} and \ref{Ing4} must be consistent with each other.
And in section \ref{2.4} and \ref{2.5}, formulae for entropy will be defined with respect to the distribution of maximum entropy, $\rho_\text{eq}(x)$, so that this is indeed the case.
\end{ingredient}
\begin{ingredient}\label{Ing5}
\textbf{A statement of information/entropy conservation}\\
Information conservation is defined with respect to information measure \ref{Ing4}. It must be consistent with the law of evolution \ref{Ing2}. It serves to place a constraint upon the permissible laws of evolution.
\end{ingredient}
\begin{ingredient}\label{Ing6}
\textbf{Reasonable mechanism by which \textit{de facto} entropy rises}\\
Since by \ref{Ing5} the exact entropy is conserved, it may only rise \textit{in effect} through some limitation of experiment. Various means by which this comes about are discussed below. \ref{Ing6} is included as it is useful to consider the qualitative causes of relaxation. This aids for example in the recognition of relaxation taking place, and the identification of any potential barriers to full relaxation.
\end{ingredient}\\
\\
To illustrate how these ingredients enable the process of relaxation, consider the following examples.
\section{Discrete state-space iRelax theories} \label{2.2}
Consider first a simple discrete state-space example as follows.
A system is confined to a state-space (\ref{Ing1}) consisting of only five distinct states, denoted A to E.
As the state-space is discrete, the law of evolution (\ref{Ing2}) is iterative.
For each unit of time that passes, the law of evolution moves the system from its current state to some next state.
Four example iterative laws are shown diagrammatically in figure \ref{discrete_state_space}.
To explore the effect of these four laws upon information, it is first necessary to define information. 

For the purposes of this chapter, information shall be taken to be synonymous with the specification of a probability distribution over the state-space, in this case $\left\{p_i|i\in\text{state-space}\right\}$. A Bayesian would view such a probability distribution as representing incomplete knowledge of the state in which a single system resides. A frequentist would view the probability distribution as representing the variety and relative frequency of states occupied by an ensemble of systems.
From both points of view, the distribution
\begin{align}
p_A=0,\quad p_B=\frac12,\quad p_C=\frac12,\quad p_D=0,\quad p_E=0
\end{align} 
represents knowledge that a system is either in state B or state C with equal likelihood. A frequentist would assume an ensemble of systems, half occupying state $B$, half state $C$. To the frequentist, the probability represents ignorance over which member of the ensemble would be picked, were one ensemble member chosen at random. To the Bayesian, the probability distribution more clearly represents knowledge or information on a single system. Of course it is possible to invent information-rich statements that don't fully specify a probability distribution. For instance, suppose that the above system were again known to reside either in state B or in state C. This time, however, suppose that the states were not known to be equally likely. Mathematically, this is $p_B\cup p_C=1$, not $p_B=p_C=1/2$, and so the $p_i$'s are not fully specified. When the distribution is not fully specified, the statistical convention is to adopt the least informationful distribution that is consistent with the statement. In this case, $p_B=p_C=1/2$. This is called the principle of maximum entropy \cite{Jaynesbook}, and is closely related to the principle of indifference.  

In order to consider the conservation of information, it is necessary to have a method of quantifying the information content of probability distributions $\{p_i\}$.  
Such a method would provide the means to compare the information content of distributions, and use expressions like `more', `less', or `equivalent' information.
Entropy (\ref{Ing4}),
\begin{align}\label{discrete_entropy}
S = -\sum_i p_i\log p_i,
\end{align}
is the usual scalar measure of the information that fulfills this role. The following conventions for entropy are observed throughout: a) Any multiplicative constants (e.g.\ the Boltzmann constant) are omitted as they serve no useful role presently. b) The base of the logarithm is left unspecified unless required.
c) The negative sign is included as for historical reasons entropy and information conventionally vary inversely to each other. 
For this reason, in the following, the word \emph{entropy} could be replaced with word \emph{ignorance} without issue.

Entropy \eqref{discrete_entropy} is a measure of the spread of the probability distribution $\{p_i\}$.
For instance, the maximum information possible corresponds to zero-uncertainty of the state in which the system resides (no spread).
That is, $p_i=\delta_{ij}$ for some $j$. This is the state of lowest entropy, $S=0$.
This contrasts with the state of minimum information, when nothing is known of the system state. This is represented by uniformly distributed likelihoods, $p_i=1/5\ \ \forall\ \ i$. This is the distribution of maximum entropy (\ref{Ing3}). The base of the logarithm is, for the present purposes, somewhat arbitrary. It may be left unspecified or chosen judiciously. A sensible choice for the 5 state scenario would be base 5, which results in an entropy that varies between 0 and 1.

\begin{figure}
\begin{center}
\begin{overpic}[tics=10,width=\textwidth]{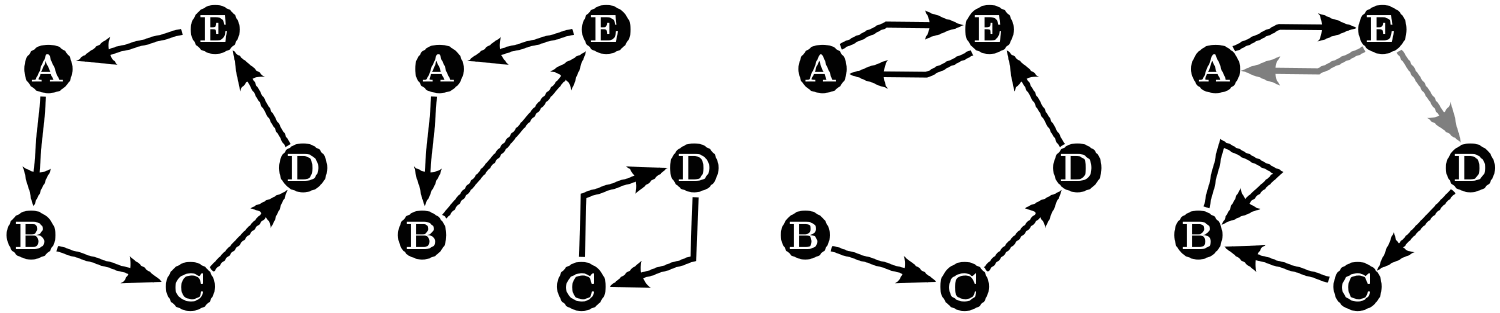}
\put (7,-3){Law 1}
\put (33,-3){Law 2}
\put (58.5,-3){Law 3}
\put (84.5,-3){Law 4}
\put (86,13.5){\textcolor{gray}{$T_{AE}$}}
\put (95,16){\textcolor{gray}{$T_{DE}$}}
\end{overpic}
\end{center}
\caption[Schematic diagram of four laws of evolution on a five-state discrete state-space]{This figure represents four discrete laws of evolution upon a five-state state-space. Letters A to E represent the five possible states of a system. The arrows between the states represent different iterative laws of evolution.
Laws 1 and 2 are one-to-one permutations between states. The effect of these laws upon a probability distribution $\{p_i\}$ is merely to permute the probabilities between the states.
For this reason, laws 1 and 2 leave entropy \eqref{discrete_entropy} invariant.
Law 3 on the other hand, is not one-to-one, as states A and D both are both mapped to state E.
After a few iterations of law 3, a probability distribution $\{p_i\}$ will evolve such that only $p_A$ and $p_E$ are non-zero.
This corresponds to an decrease in entropy, \eqref{discrete_entropy}, as the knowledge of which state the system occupies becomes more certain. 
Law 3 may be thought to violate information-conservation as a result of its many-to-one mapping.
Law 4 is intended to illustrate a one-to-many mapping. 
This possibility is also ruled out by information conservation as explained in section \hyperref[sec:transition_matrix]{2.2.1}.
 }\label{discrete_state_space}
\end{figure}
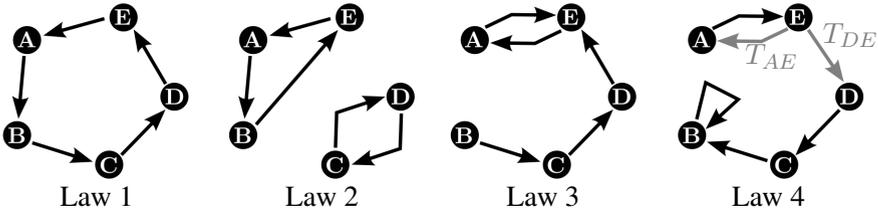
Now consider whether the four laws depicted in figure \ref{discrete_state_space} conserve information (entropy), \ref{Ing5}.
Laws 1 and 2 are simple one-to-one permutations between states.
Under an iteration of such a permutation law, the individual probabilities $p_i$ that comprise the probability distribution $\{p_i\}$ are permuted in the same manner as the states themselves. 
So for instance, an iteration of law 1 transforms the probabilities as $p_A\to p_B$, $p_B\to p_C$, $p_C\to p_D$ etc.
Clearly, such a permutation of probabilities leaves entropy \eqref{discrete_entropy} invariant. Hence laws 1 and 2 conserve information as a result of being one-to-one (injective) maps between states.
In contrast, laws 3 and 4 do not conserve information (in violation of \ref{Ing5}). 
Consider, for instance, the distribution of least information, $p_i=1/5\ \ \forall\ \ i$.
If law 3 were observed, then after a number of iterations of the law, this distribution would evolve to approximately $p_A=p_E=1/2$, $p_B=p_C=p_D=0$.
Hence, the distribution would become spread between states $A$ and $E$ only, rather than between all five possible states.
This constitutes better information regarding the state (the system is known to reside in one of two possible states rather than one of five possible states), and so entropy is decreased. 
In short, law 3 is in violation of \ref{Ing5} as the two-to-one mapping from states A and D onto state E lowers entropy, resulting in decreased uncertainty of the system state.
So one-to-one laws do conserve entropy, and many-to-one laws appear not to. Naturally the other case to consider is that of one-to-many laws.
Of course, for a law to map a single initial state to multiple possible final states, it must choose which final state to send each individual state to.
If there is some deterministic mechanism that facilitates this decision, then it must depend on factors that have not been taken into account in defining the state-space. 
The possibility of such a mechanism is disregarded for this reason.
Rather, if the state-space does indeed exhaust all factors relevant to the system evolution (as per the definition of \ref{Ing1}), then one-to-many laws must be fundamentally probabilistic. 
This is illustrated in figure \ref{discrete_state_space} by law 4, which is a near-inverse to law 3 and so contains a one-to-many component.
(A true inverse is impossible as law 3 is neither one-to-one nor onto.)
Using the notation below, the respective probabilities that state $E$ is iterated to states $A$ and $D$ are denoted $T_{AE}$ and $T_{DE}$.
As is proven below, such one-to-many probabilistic laws also necessarily violate information conservation (\ref{Ing5}).

Laws 3 and 4 serve to highlight an important distinction between determinism and time reversibility. 
Note for instance that law 3 is deterministic, despite not conserving information.
That is, given some present state of the system, law 3 perfectly determines the future state after any arbitrary number of iterations. 
It does not perfectly determine the previous evolution however.
There is no way for example, to tell whether a system in state E resided in state A or state D one iteration previously.
Although Law 3 determines future states perfectly, it rapidly loses it predictive power backwards in time. 
For this reason it is useful to distinguish between forwards-time determinism (many-to-one laws) and backwards-time determinism (one-to-many laws). The former category may be considered perfectly \emph{pre}dictive whilst the later are perfectly \emph{post}dictive.
One-to-one laws are of course forwards and backwards-time deterministic. Presently, this is used as the definition of time-reversibility.

\subsection{The transition matrix and the information conservation theorem}\label{sec:transition_matrix}
To formalize the above observations on information conservation, consider the generally probabilistic case of a law on $n$ discrete states in which every state has a probability of being iterated to any other state. Denote the probability that state $j$ is iterated to state $i$ by $T_{ij}$.
Then, under an iteration of the law, the probability distribution $p_i$ transforms as
\begin{align}\label{discrete_transition}
p_i\longrightarrow p_i' =\sum_j T_{ij}p_j,
\end{align}
which may be regarded as the matrix equation $p'=Tp$. \sloppy The $p'=(p_1',p_2',\dots,p_n')^\text{tr}$ and $p=(p_1,p_2,\dots,p_n)^\text{tr}$ are probability column vectors and $T$ is an $n\times n$ matrix with elements $T_{ij}$.
Of course, since the elements of matrix $T$ are indeed probabilities, they are restricted to real values between zero and one, $T_{ij}\in [0,1]\,\, \forall\,\, i,j$.
In addition, since a system in state $j$ must exist in some state after an iteration, it follows that the columns of $T$ sum to unity, $\sum_i T_{ij}=1\,\,\forall \,\,j$.
Without the imposition of further constraints however, these are the only two conditions satisfied by the matrix $T$.
Matrices with these properties are in fact well known and are widely used to study Markovian chains, where they are known as stochastic or Markov matrices.
Indeed the Markovian property, wherein the next state a system advances to is presumed only to depend upon the current state it occupies, has already been mentioned in passing whilst defining \ref{Ing1}. 
iRelax theories however, have the additional constraint of information (entropy) conservation (\ref{Ing5}), which changes the usual picture significantly.
To this end, consider the following theorem and proof.
\begin{theorem}[Information Conservation Theorem on discrete state-spaces]\quad\newline
On a discrete state-space $\Omega$,
\begin{align}\label{theorem1}
\begin{matrix}\text{The set of all} \\\text{information-conserving laws }\end{matrix}\quad&=\quad\begin{matrix}\text{ The set of all }\\\text{laws that are one-to-one}\end{matrix}.
\end{align}
\end{theorem}
\begin{proof}
As already discussed, the implication: a one-to-one law $\Rightarrow$ information-conservation follows trivially. In order to prove this theorem therefore, the reverse implication: a one-to-one law $\Leftarrow$ information conservation, must be sought.
Let $T_{ij}\in[0,1]$ be the elements of a stochastic matrix $T$ (with columns normalized to unity, $\sum_i T_{ij}=1$, in accordance with the discussion above). 
Under an iteration of the law, probabilities transform as $p_i\to \sum_jT_{ij}p_{j}$. So, in order to conserve entropy, $S=-\sum_i p_i\log p_i$, it must be the case that 
\begin{align}\label{theorem_misc_1}
\sum_i p_i \log p_i = \sum_i\left[\left(\sum_j T_{ij} p_j\right)\log\left(\sum_j T_{ij} p_j\right)\right]
\end{align}
for all possible probability distributions $\{p_i\}$. Substitution of a maximum information distribution $p_i=\delta_{ik}$ into equation \eqref{theorem_misc_1} gives 
\begin{align}\label{theorem_misc_2}
0=\sum_i T_{ik} \log T_{ik}.
\end{align}
On the \emph{open} interval $x\in (0,1)$ the function $x\log x$ is strictly less than zero.
It is equal to 0 on both boundaries of this interval however (at $x=1$ and in the limit $x\to 0$).
 The summation of equation \eqref{theorem_misc_2} cannot work, therefore, if any $T_{ik}$ is between 0 and 1. 
All matrix elements must be on the limits of this interval; either $T_{ik}=0$ or $T_{ik}=1$ $\forall$ $i,k$.
Since the columns of matrix $T$ are normalized, this statement may be taken to mean that each column of $T$ is populated with 0s except for a single 1.

This is determinism; every state is mapped to a single other. 
The law may yet be many-to-one however (there may be multiple 1s to a row of $T$). To rule out this possibility for the case of a state-space $\Omega$ with a finite number of states, substitute the minimum information distribution, $p_{i}=1/n\,\,\forall \,\,i$, into equation \eqref{theorem_misc_1} to give an expression that simplifies to 
\begin{align}
0=\frac{1}{n}\sum_i\left[\left(\sum_jT_{ij}\right)\log\left(\sum_j T_{ij}\right)\right].
\end{align}
This time the properties of $x\log x$ indicate that $\sum_j T_{ij}$ must equal 0 or 1 for each $i$. That is, each row may admit at most a single 1. This concludes the proof for a finite number of discrete states. To conclude the proof for a countably infinite $\Omega$ (ie. $n\rightarrow\infty$), consider a uniform distribution over a finite $m$-state subset, $\omega\subset\Omega$. Substituting this into equation \eqref{theorem_misc_1} gives an expression that simplifies to 
\begin{align}
0&=\frac{1}{m}\sum_{i\in\Omega}\left[\left(\sum_{j\in\omega}T_{ij}\right)\log\left(\sum_{j\in\omega}T_{ij}\right)\right].
\end{align}
Hence, $\sum_{j\in\omega}T_{ij}$ equals 0 or 1 for all finite subsets $\omega\in\Omega$; rows of the transition matrix $T$ contain at most a single element equal to 1 with all other elements equal to 0.
\end{proof}
\hspace{-\parindent}In summary, for discrete state-spaces%
\footnote{On a finite discrete state-space $\Omega$, one-to-one laws of evolution are permutations and are guaranteed to be onto (surjective). This is not necessarily the case if $\Omega$ is (countably) infinite. Suppose states were labeled $i=1,2,...,\infty$, for instance, and a law of evolution specified an iteration that moved each state to the state with double its current label, $i\rightarrow 2i$. Then each state would be mapped to exactly one other, but no state labeled with an odd number would be ever be mapped to.
Such laws, that are one-to-one but not onto, should be considered time-reversible as they do perfectly predict future and past evolution. Indeed they could be considered an interesting sub-class of laws with the curious property that, in a sense, they may be used to infer a \emph{beginning to time}.
For instance, suppose that at the present time (iteration) a system governed by the law above were known to exist in state $i=12$. Then two iterations of the law previously, the system must have existed in state $i=3$. As there is no $i=3/2$ state, however, the law could not have begun iterating prior to this. So in this sense, two iterations prior to the current time, corresponds to the beginning of time.}, 
\begin{quote}
\hspace{-6mm}Information conserving laws = one-to-one laws = time reversible laws $\subset$ deterministic laws.
\end{quote}
For the sake of later comparison, consider a statement of the one-to-one property as follows. Let $p_n(i)$ denote the time dependent probability distribution over states $i\in \Omega $ at iteration $n$.
Let $i_0$ denote a system's initial state, and $i_n$ denote the state it is moved to after $n$ iterations. 
Then the one-to-one property ensures that 
\begin{align}\label{discrete_information_conservation}
p_n(i_n)=p_0(i_0).
\end{align}
This is the discrete equivalent of Liouville's Theorem, $\text{d}\rho/\text{d}t=0$, the usual statement of information conservation, \ref{Ing5}, in classical mechanics.

The final ingredient, \ref{Ing6}, describes how the limitations of experiment cause a \textit{de facto} rise in entropy, despite its exact conservation (\ref{Ing5}). Although this ingredient tends to be more appropriate for systems with a continuous state-space, it could arise for the discrete case as follows. Suppose a discrete law were to be iterated a great number of times. After some large number of iterations, the count may be lost, so that there is uncertainty over whether the next iteration has already been made. A sensible option in this case would be to employ the principle of maximum entropy (see reference \cite{Jaynesbook})--that is, use the probability distribution of highest entropy that is consistent with the degraded information. This would cause an entropy rise.

%%%%%%%%%%%%%%%%%%%%%%%%%%%%%%%%%%%%%%%%%%%%%%%%%%%%%%%%%%%%%%%%%%%%%%%%%%%%%%%%%%%%%%%%%%%%%%%%%%%%%%%%%%%%%%%%%%%%%%%%%%%%%%%%%%%%%
\begin{figure}\label{long_figure_part_1}
\includegraphics[width=\textwidth/6]{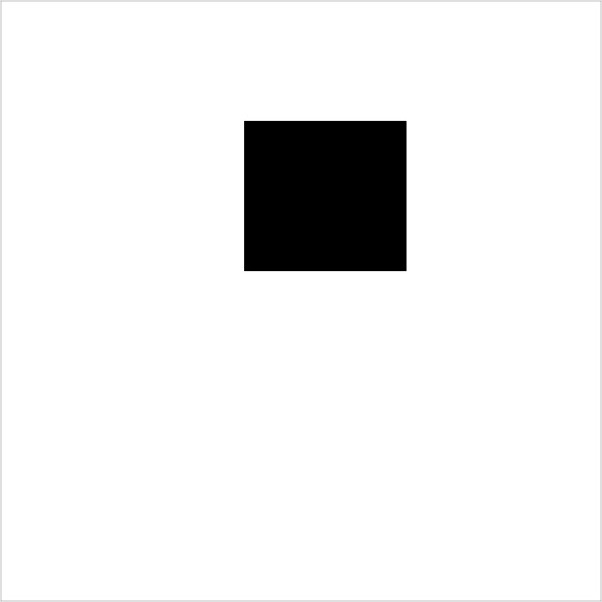}%
\includegraphics[width=\textwidth/6]{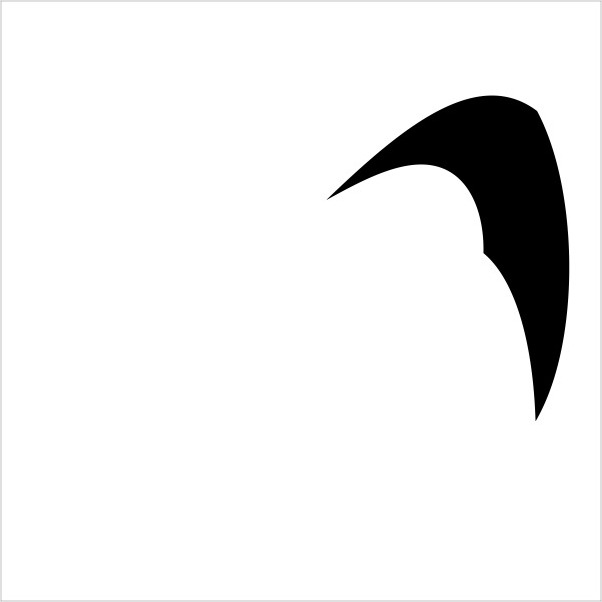}%
\includegraphics[width=\textwidth/6]{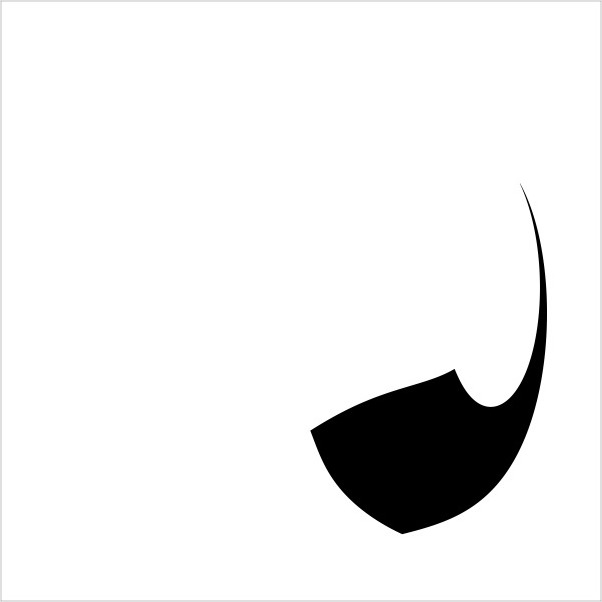}%
\includegraphics[width=\textwidth/6]{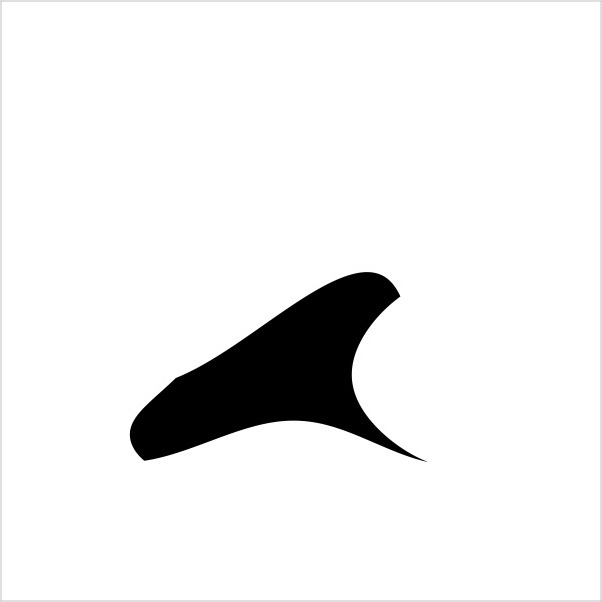}%
\includegraphics[width=\textwidth/6]{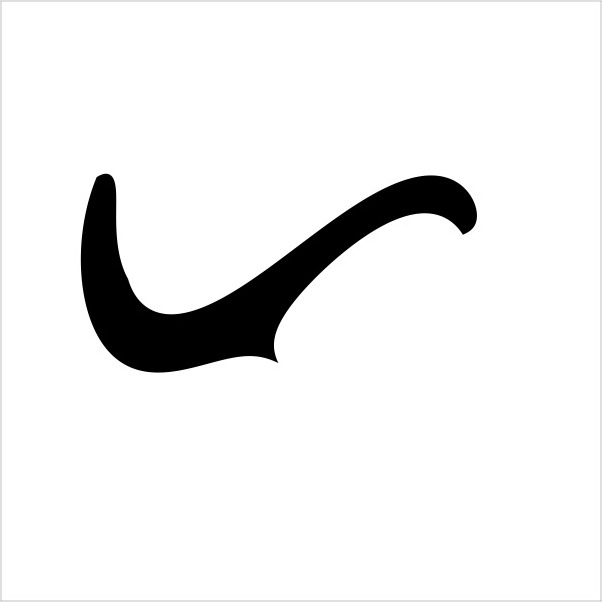}%
\includegraphics[width=\textwidth/6]{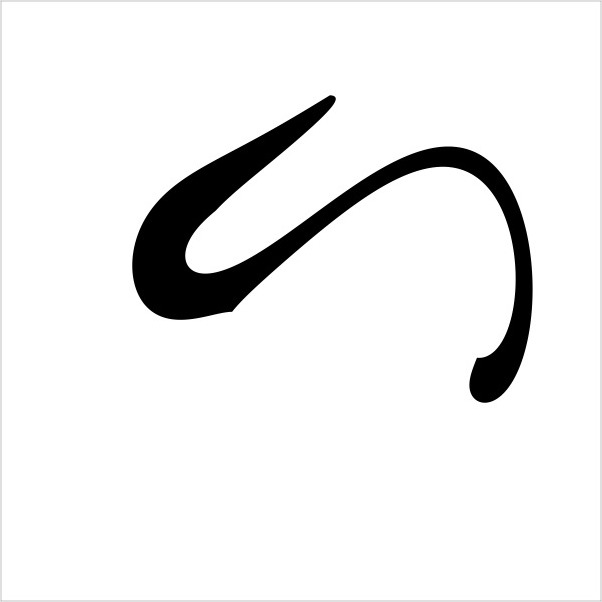}%

\vspace{-1pt}
\includegraphics[width=\textwidth/6]{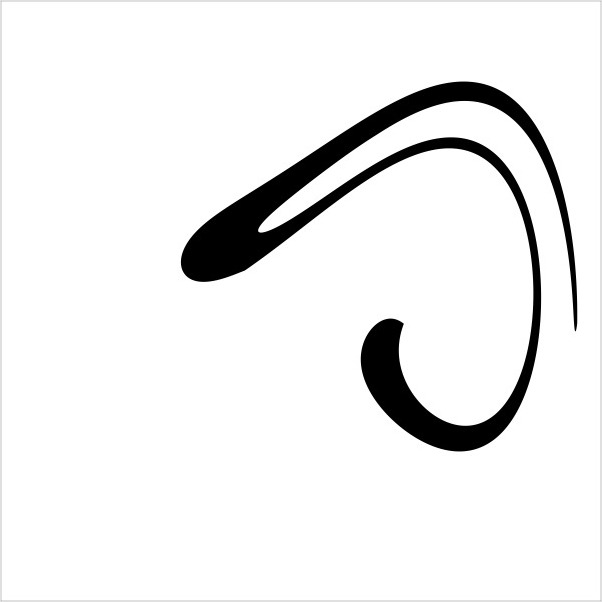}%
\includegraphics[width=\textwidth/6]{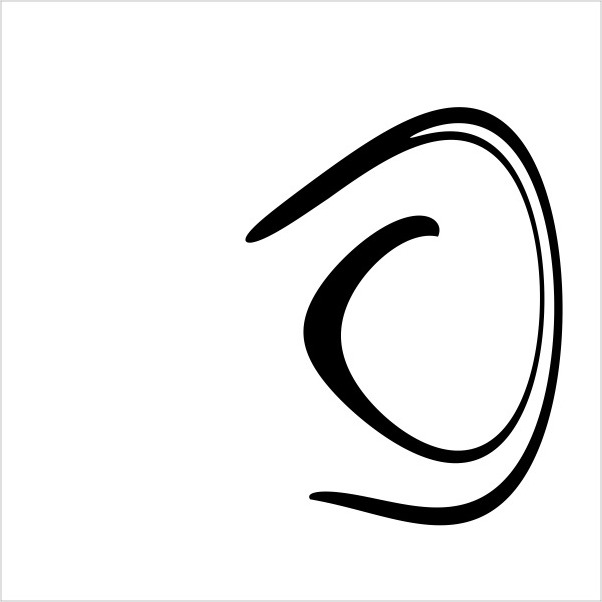}%
\includegraphics[width=\textwidth/6]{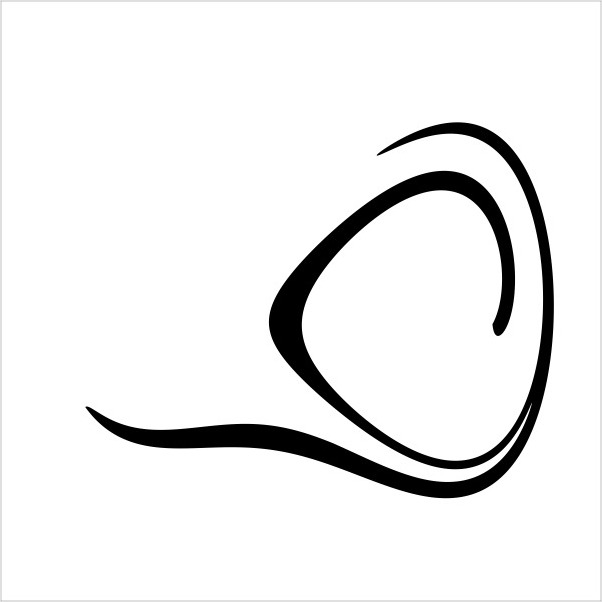}%
\includegraphics[width=\textwidth/6]{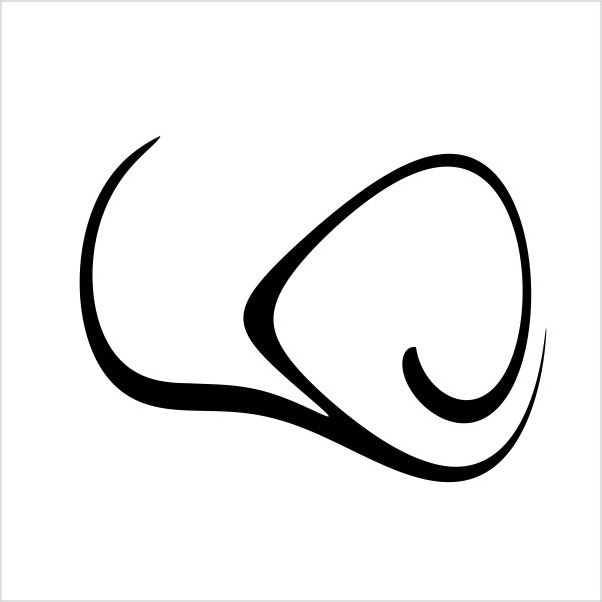}%
\includegraphics[width=\textwidth/6]{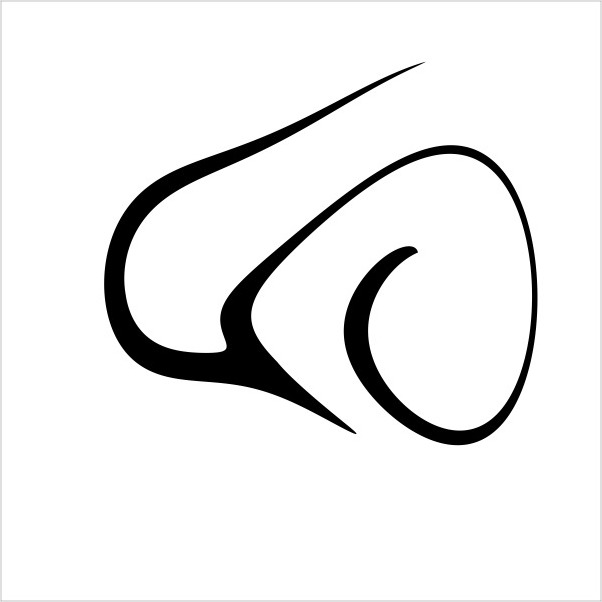}%
\includegraphics[width=\textwidth/6]{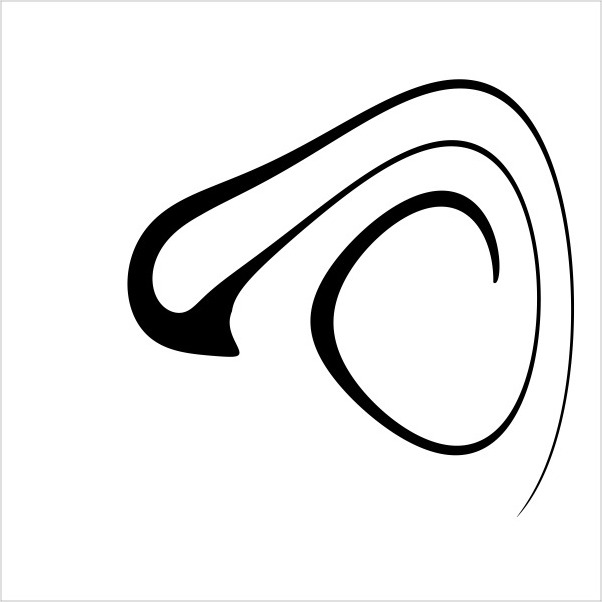}%

\vspace{-1pt}
\includegraphics[width=\textwidth/6]{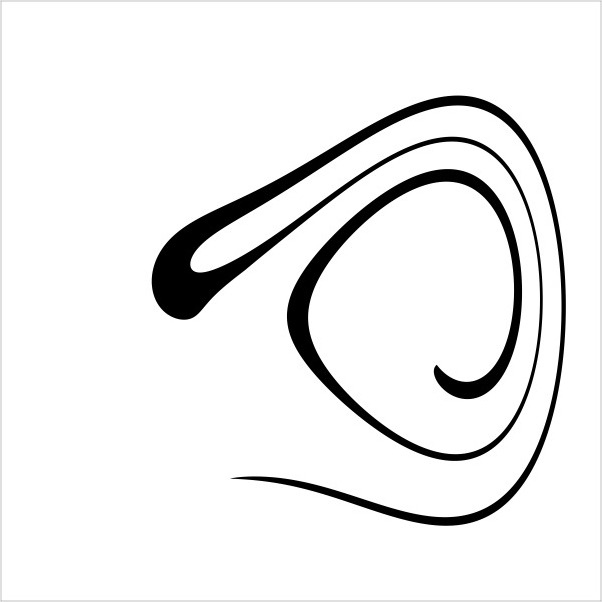}%
\includegraphics[width=\textwidth/6]{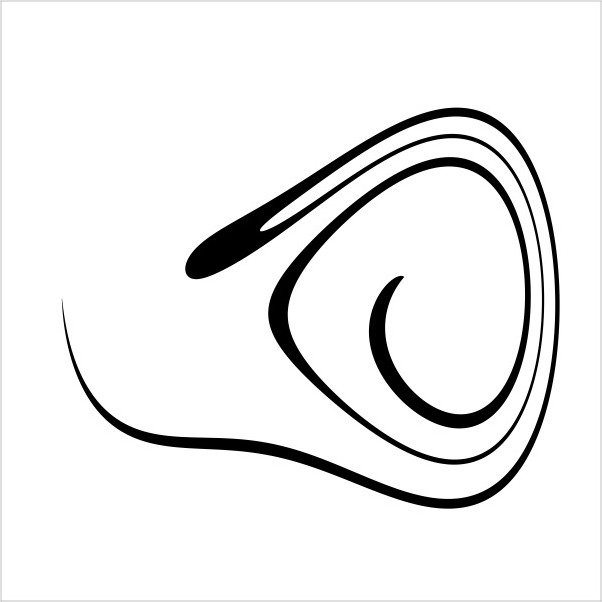}%
\includegraphics[width=\textwidth/6]{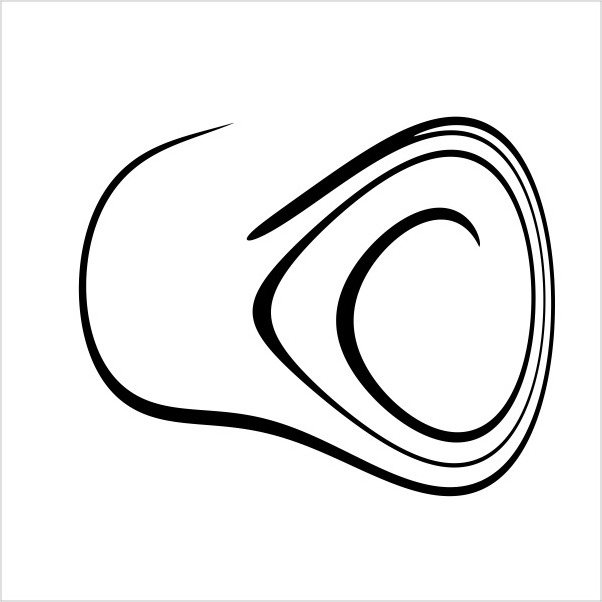}%
\includegraphics[width=\textwidth/6]{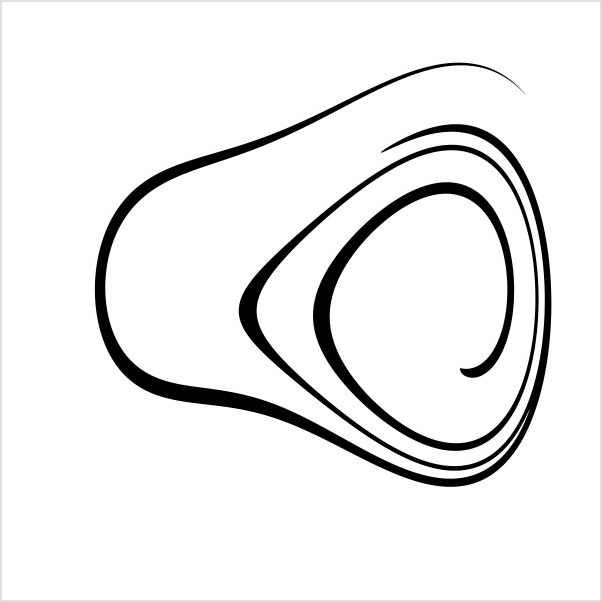}%
\includegraphics[width=\textwidth/6]{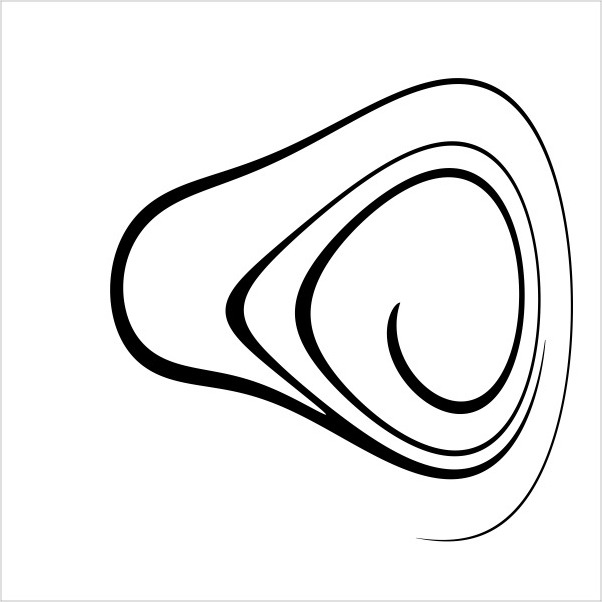}%
\includegraphics[width=\textwidth/6]{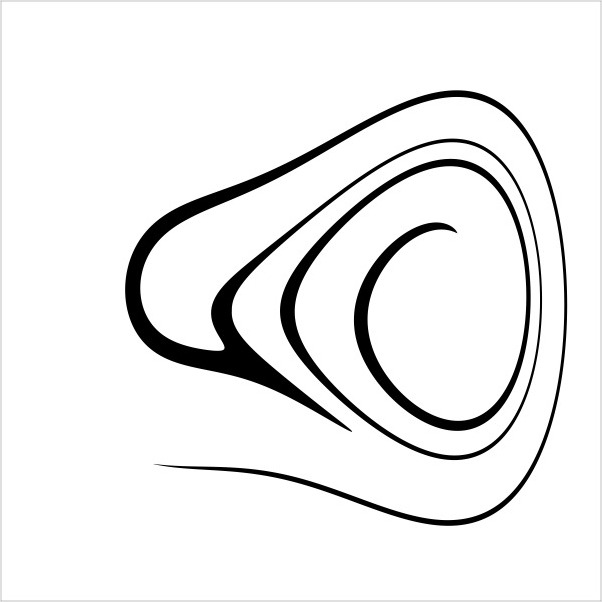}%

\vspace{-1pt}
\includegraphics[width=\textwidth/6]{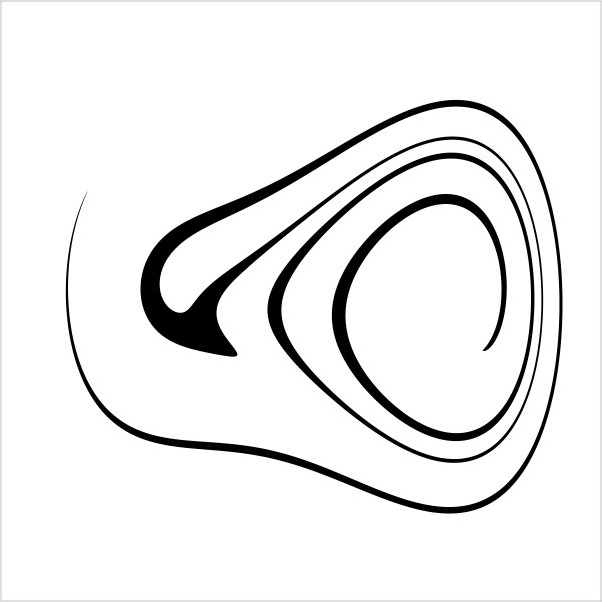}%
\includegraphics[width=\textwidth/6]{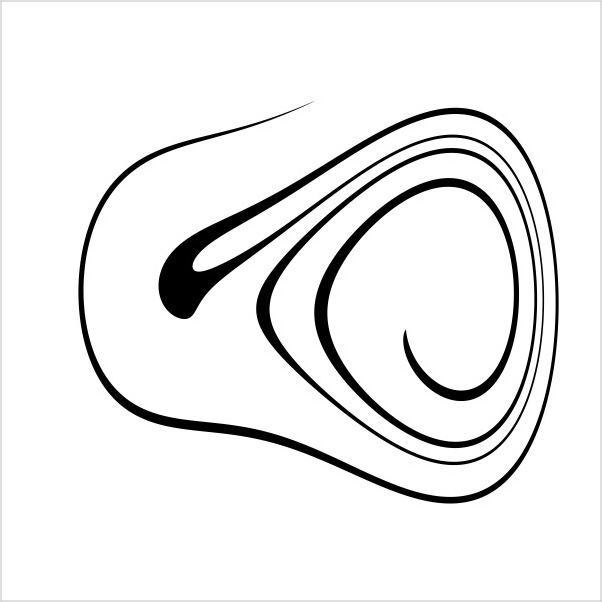}%
\includegraphics[width=\textwidth/6]{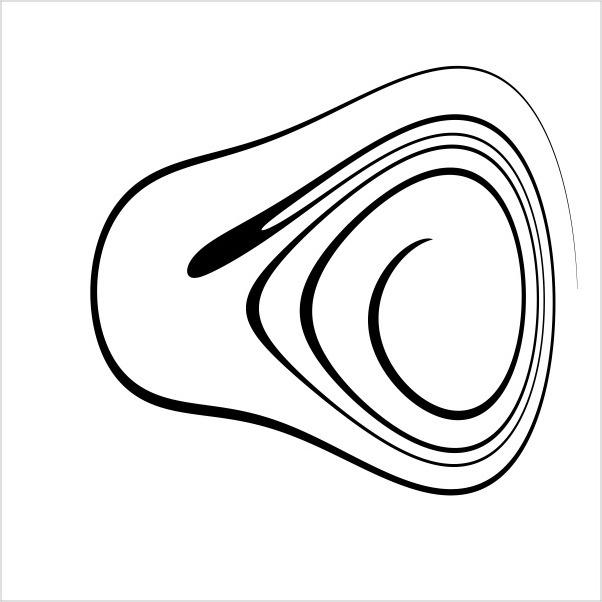}%
\includegraphics[width=\textwidth/6]{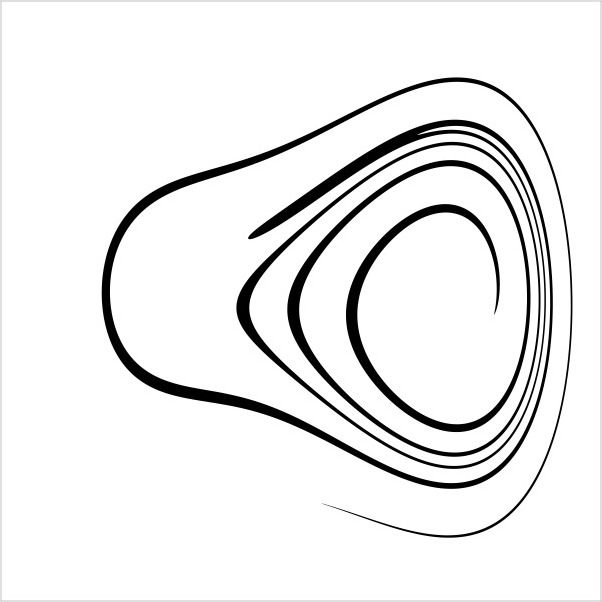}%
\includegraphics[width=\textwidth/6]{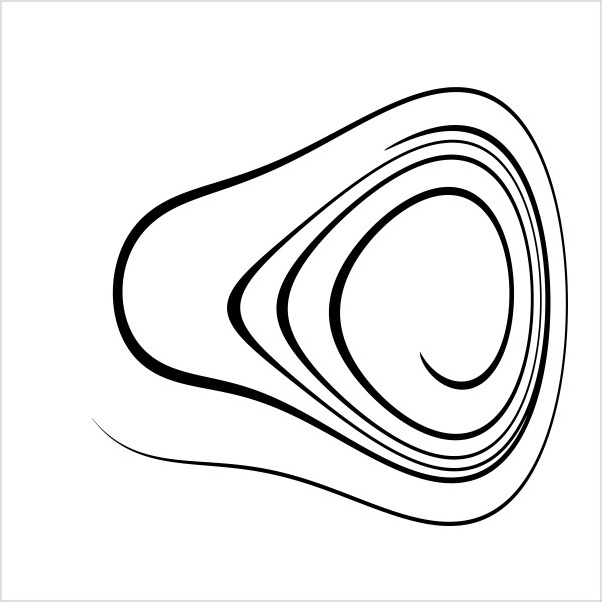}%
\includegraphics[width=\textwidth/6]{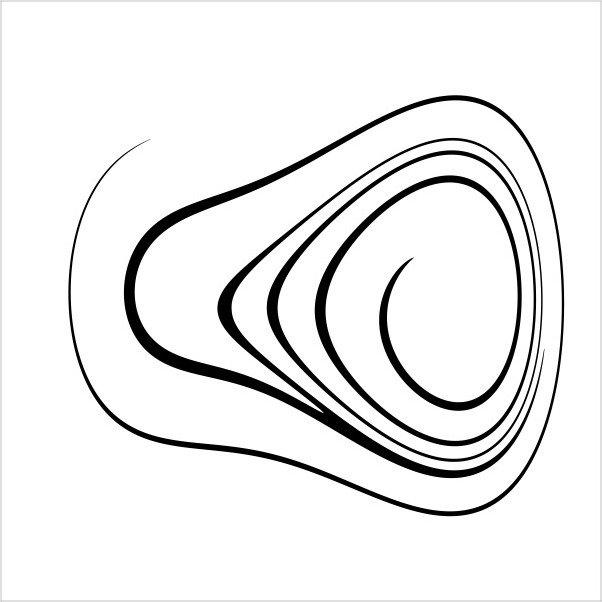}%

\vspace{-1pt}
\includegraphics[width=\textwidth/6]{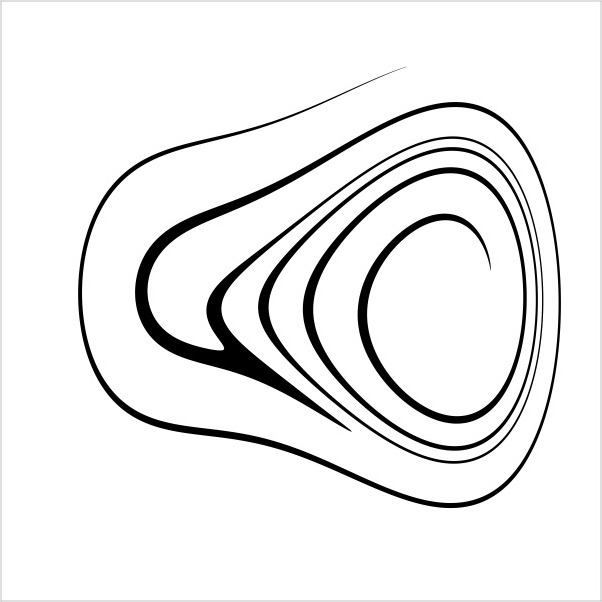}%
\includegraphics[width=\textwidth/6]{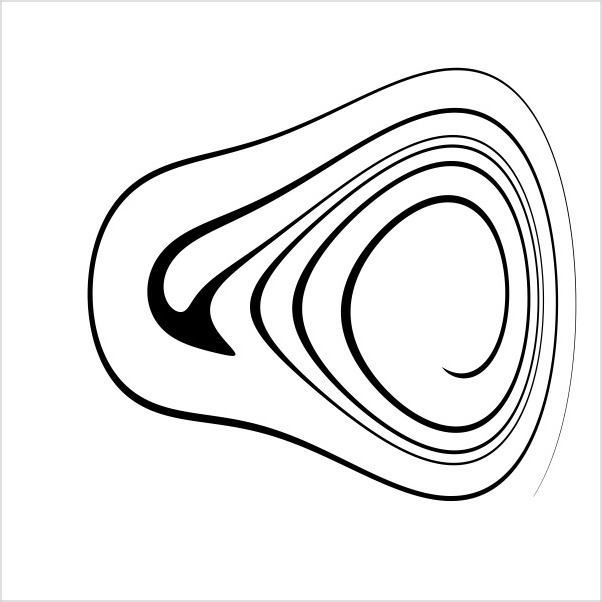}%
\includegraphics[width=\textwidth/6]{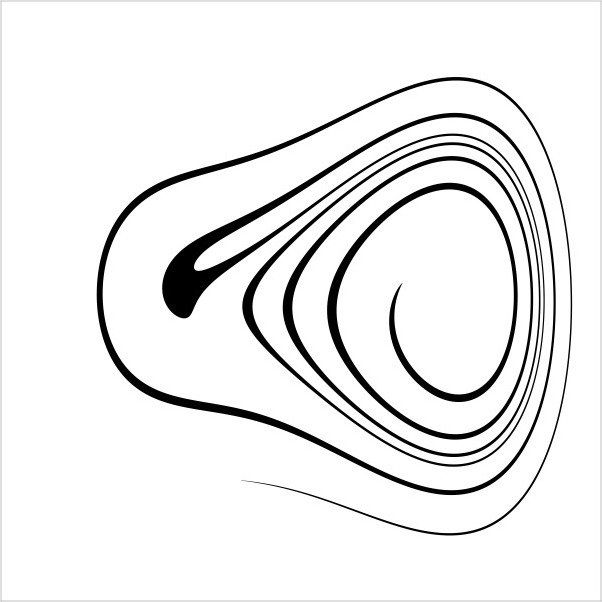}%
\includegraphics[width=\textwidth/6]{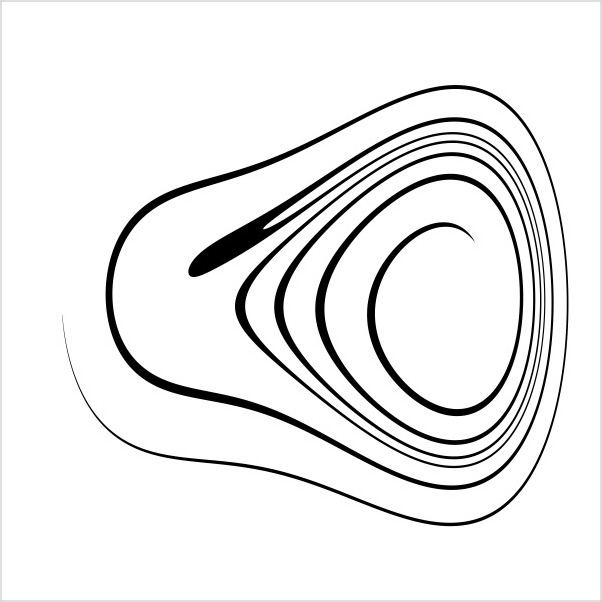}%
\includegraphics[width=\textwidth/6]{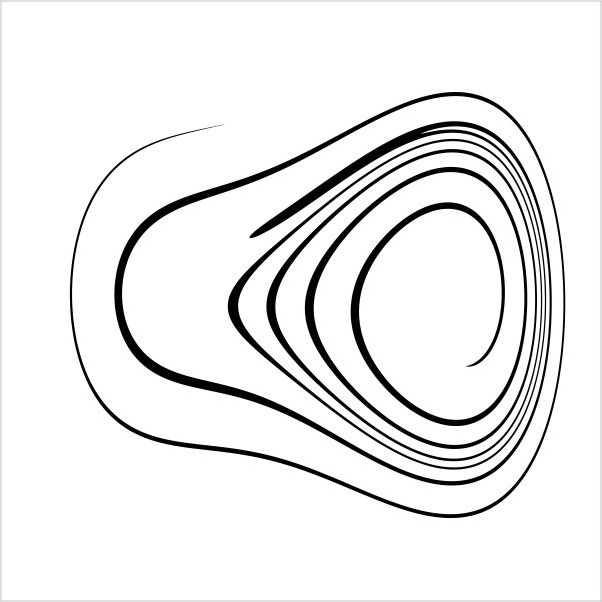}%
\includegraphics[width=\textwidth/6]{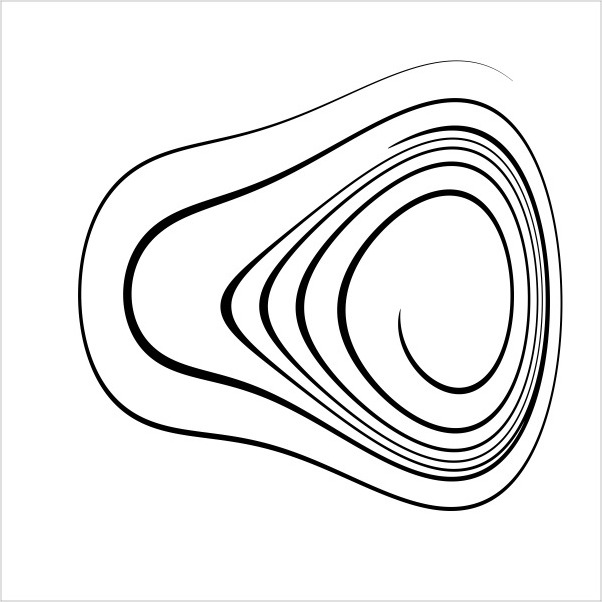}%

\vspace{-1pt}
\includegraphics[width=\textwidth/6]{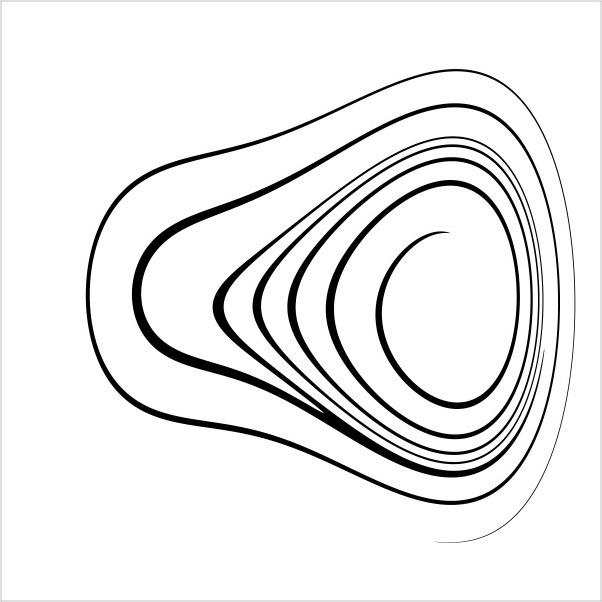}%
\includegraphics[width=\textwidth/6]{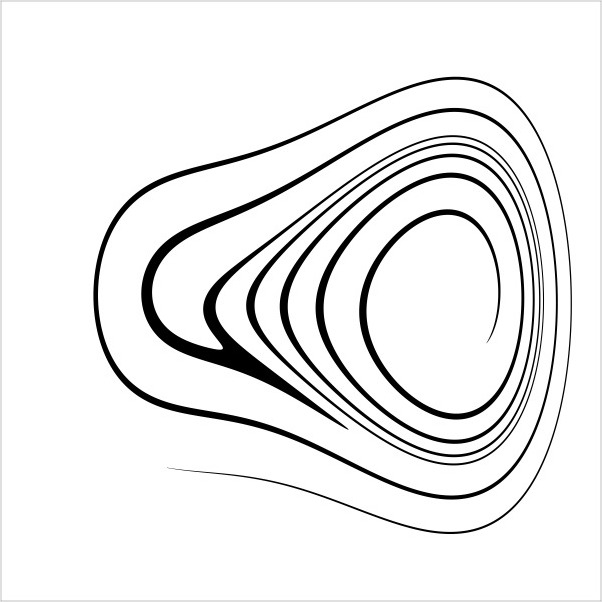}%
\includegraphics[width=\textwidth/6]{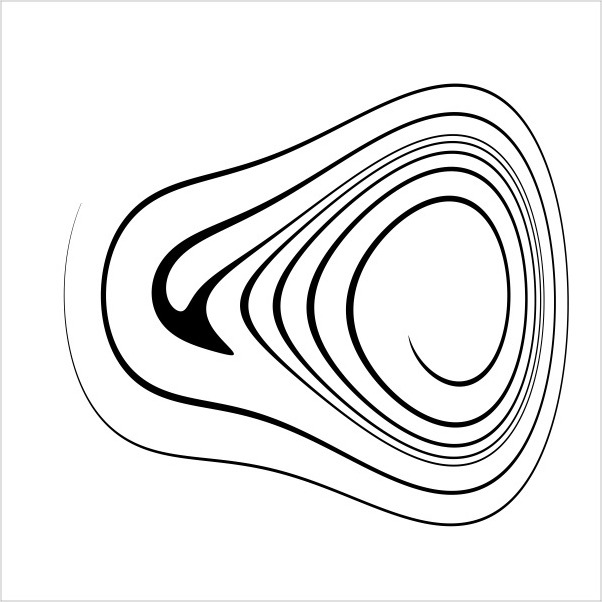}%
\includegraphics[width=\textwidth/6]{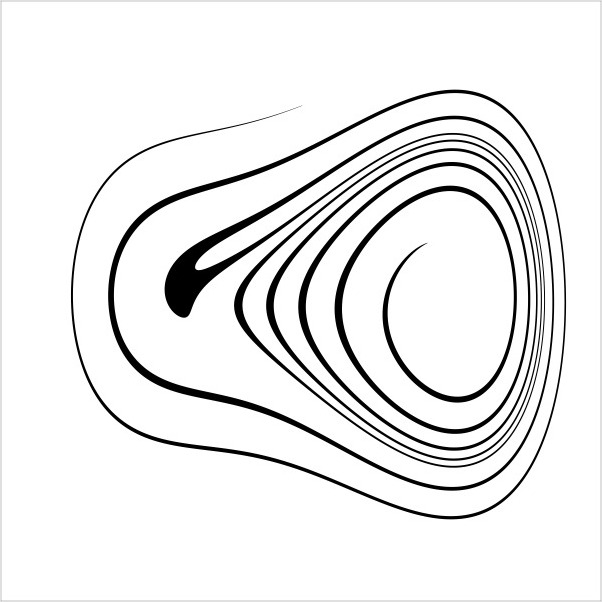}%
\includegraphics[width=\textwidth/6]{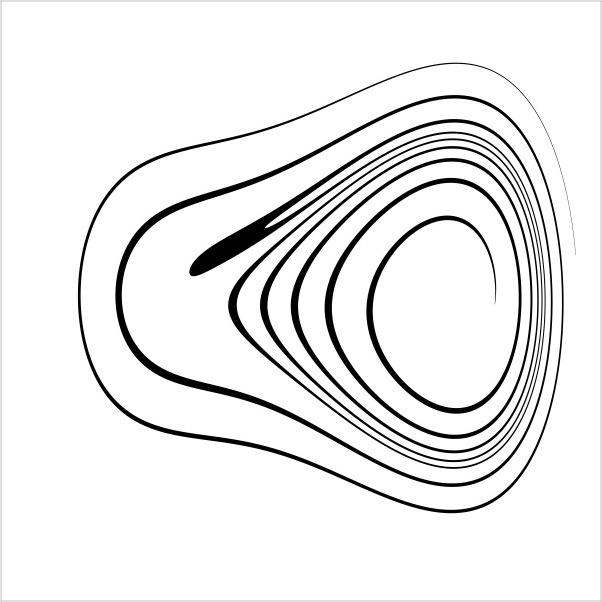}%
\includegraphics[width=\textwidth/6]{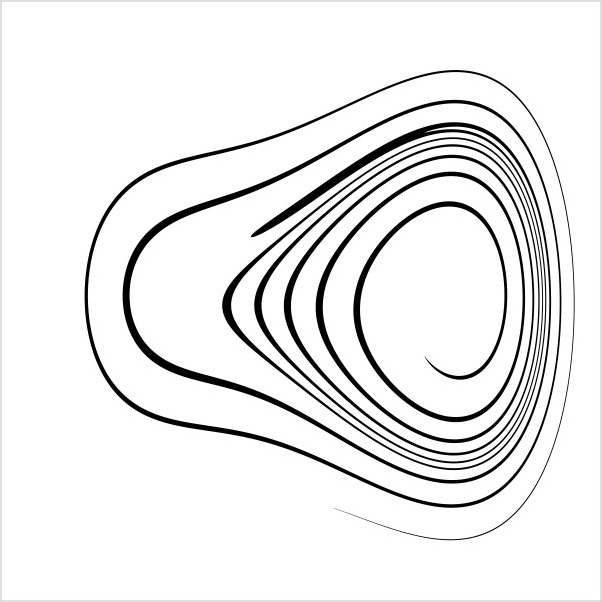}%

\vspace{-1pt}
\includegraphics[width=\textwidth/6]{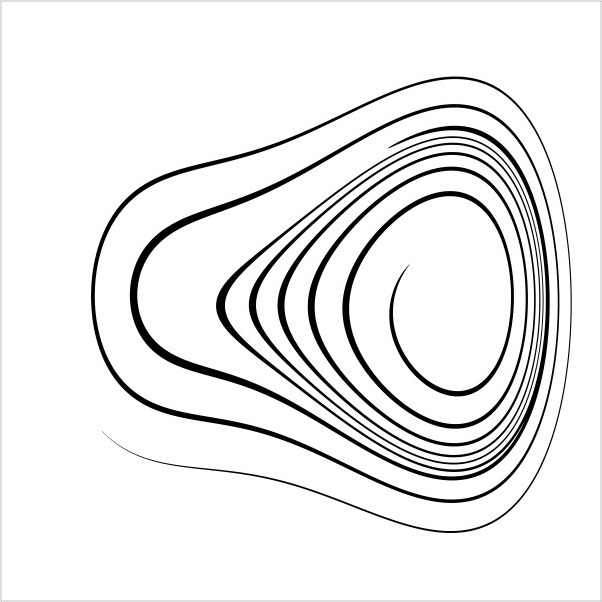}%
\includegraphics[width=\textwidth/6]{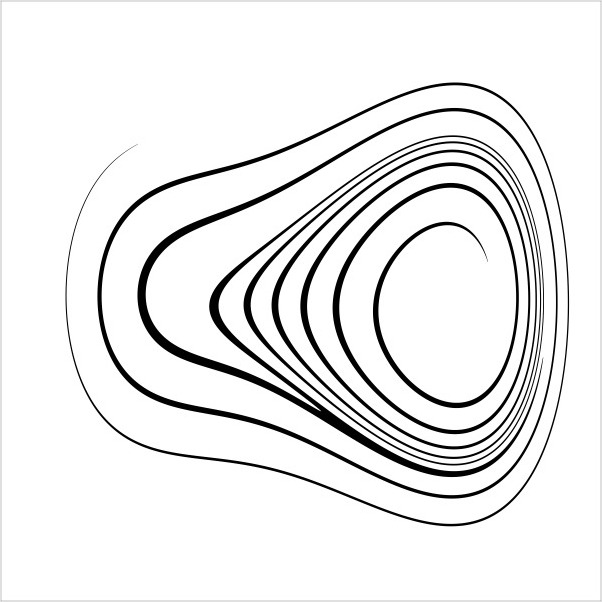}%
\includegraphics[width=\textwidth/6]{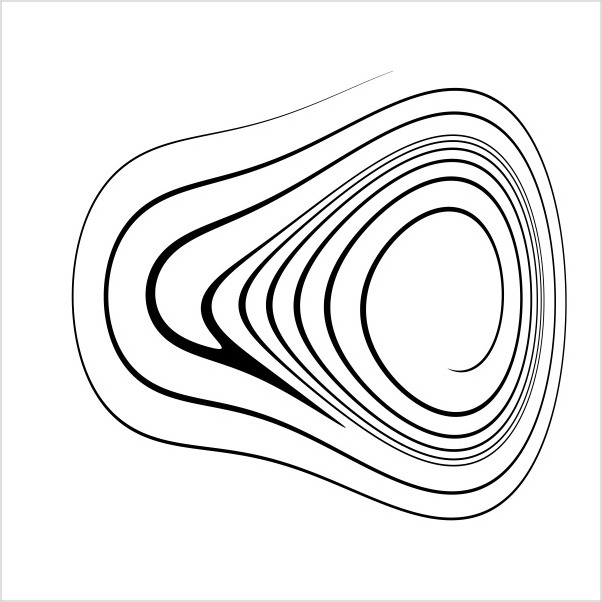}%
\includegraphics[width=\textwidth/6]{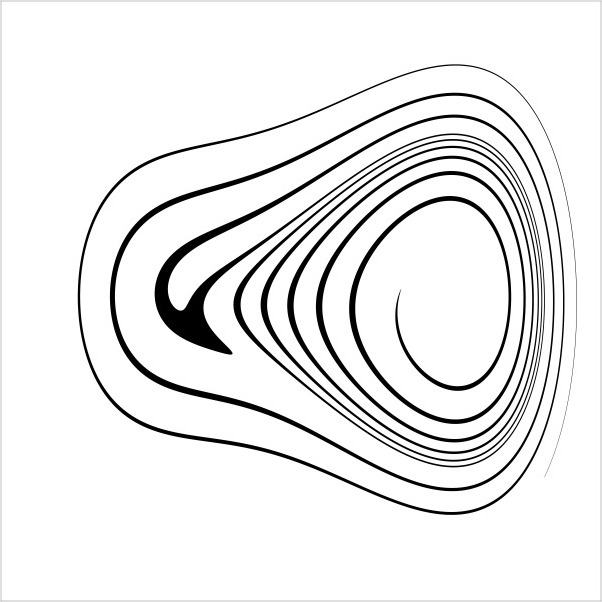}%
\includegraphics[width=\textwidth/6]{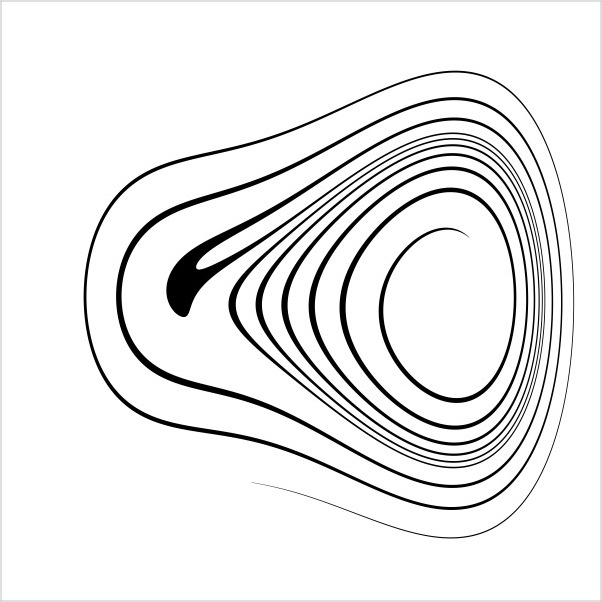}%
\includegraphics[width=\textwidth/6]{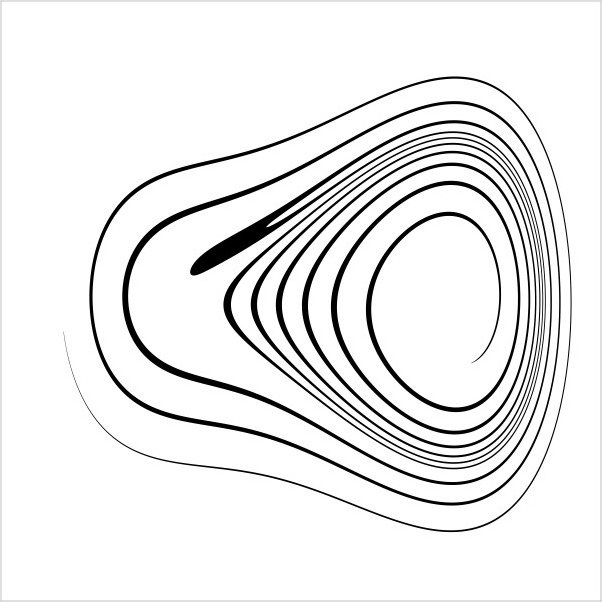}%

\vspace{-1pt}
\includegraphics[width=\textwidth/6]{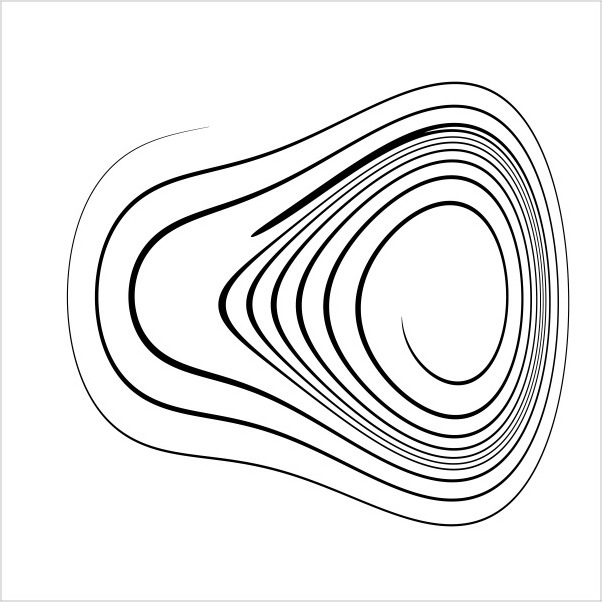}%
\includegraphics[width=\textwidth/6]{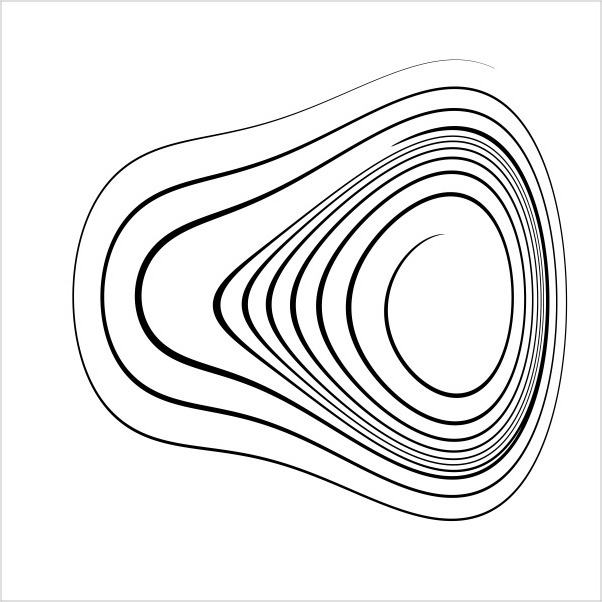}%
\includegraphics[width=\textwidth/6]{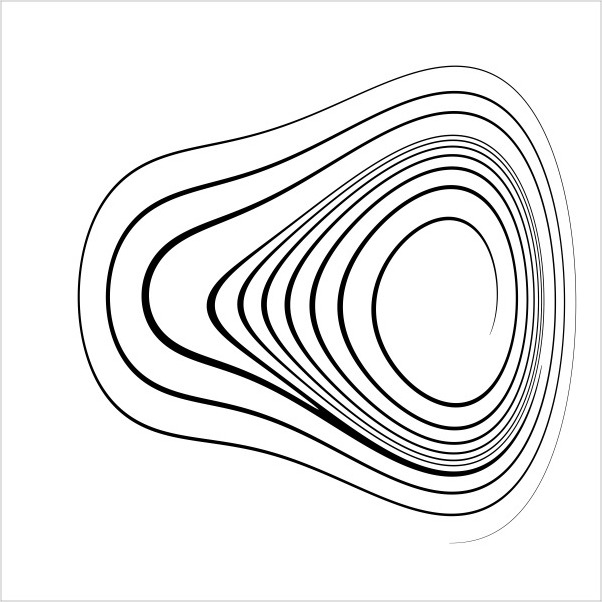}%
\includegraphics[width=\textwidth/6]{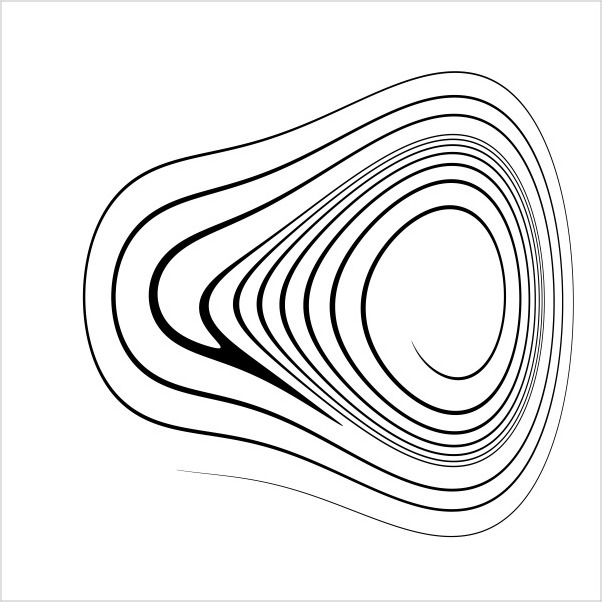}%
\includegraphics[width=\textwidth/6]{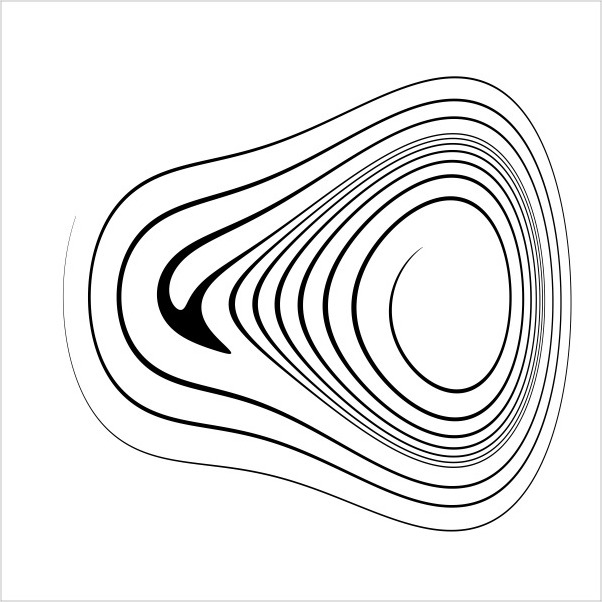}%
\includegraphics[width=\textwidth/6]{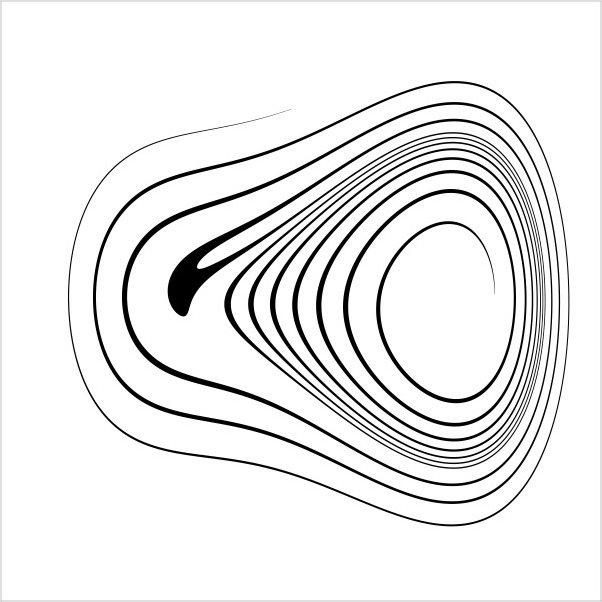}%

\vspace{-1pt}
\includegraphics[width=\textwidth/6]{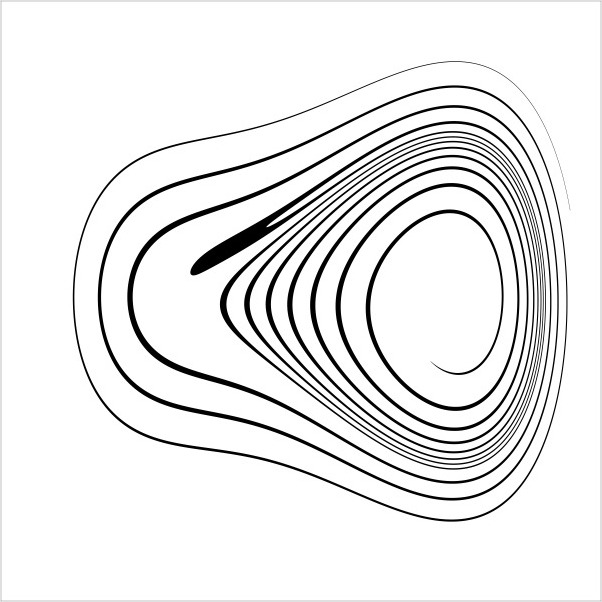}%
\includegraphics[width=\textwidth/6]{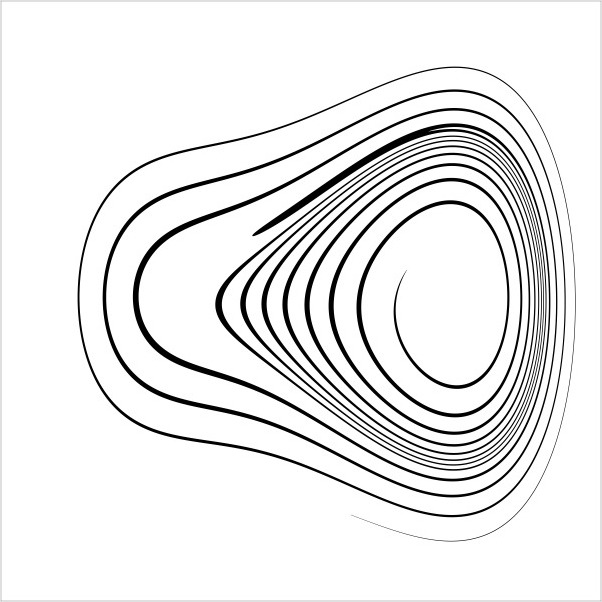}%
\includegraphics[width=\textwidth/6]{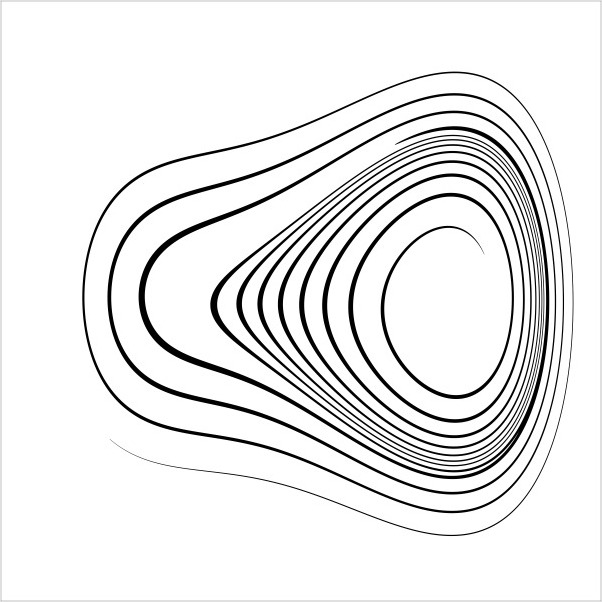}%
\includegraphics[width=\textwidth/6]{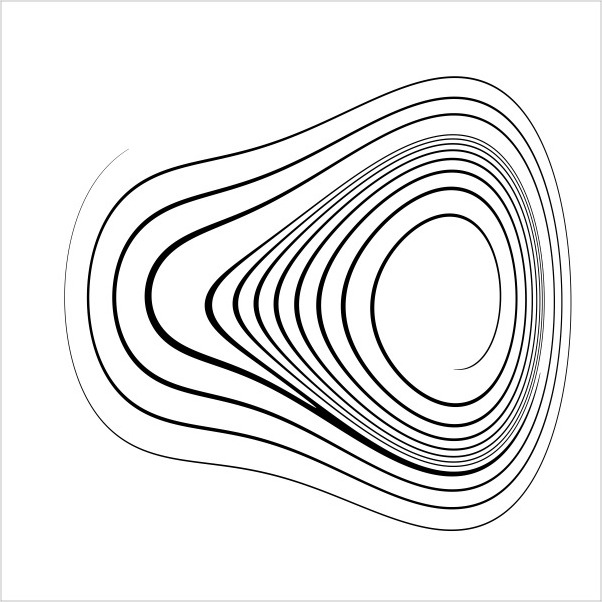}%
\includegraphics[width=\textwidth/6]{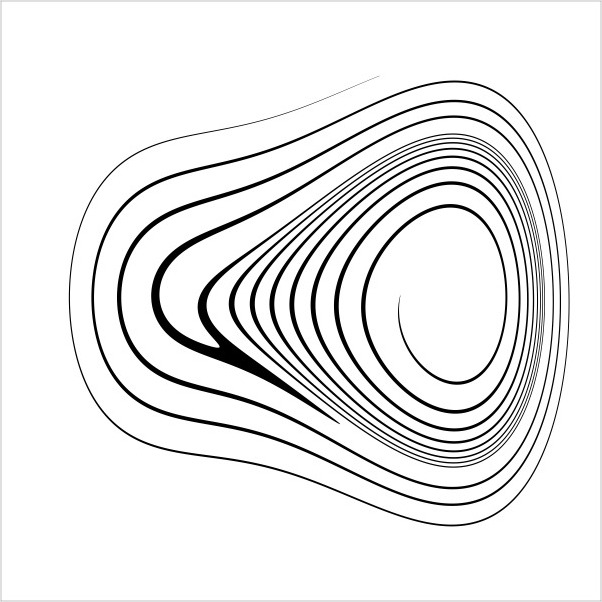}%
\includegraphics[width=\textwidth/6]{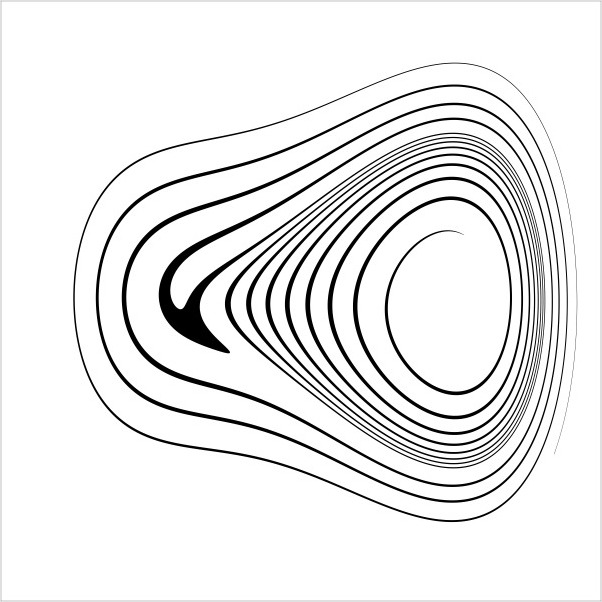}%
\phantomcaption
\end{figure}
%\pagebreak

\begin{figure}\ContinuedFloat
\includegraphics[width=\textwidth/6]{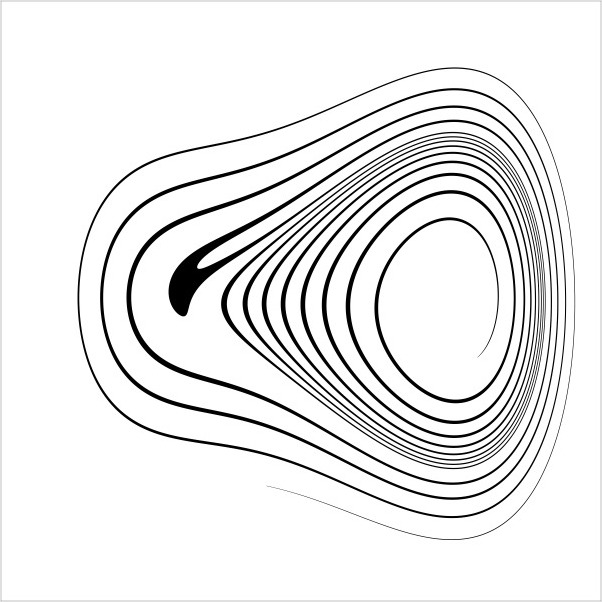}%
\includegraphics[width=\textwidth/6]{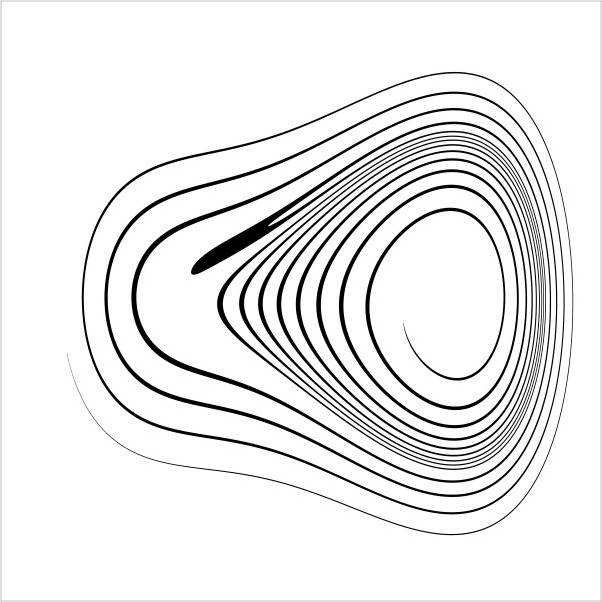}%
\includegraphics[width=\textwidth/6]{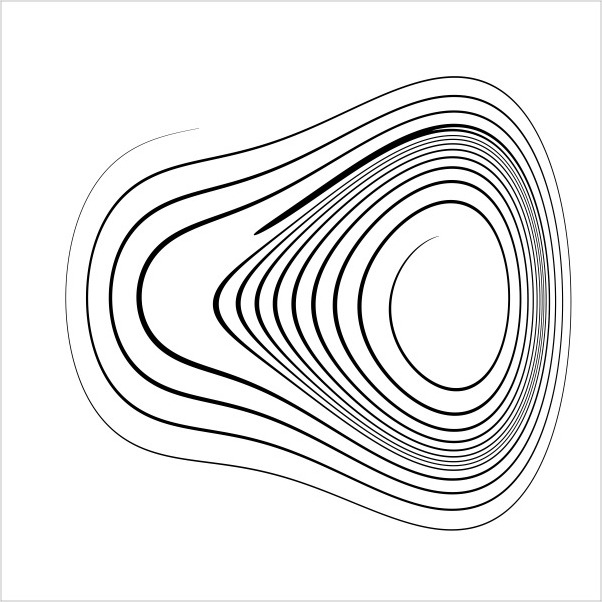}%
\includegraphics[width=\textwidth/6]{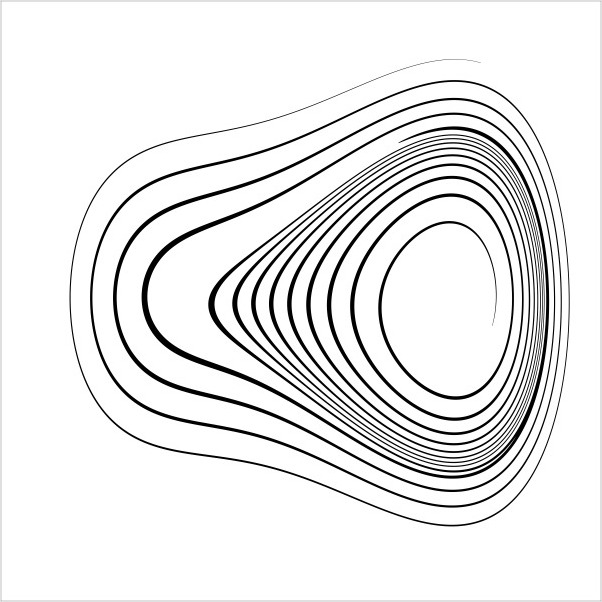}%
\includegraphics[width=\textwidth/6]{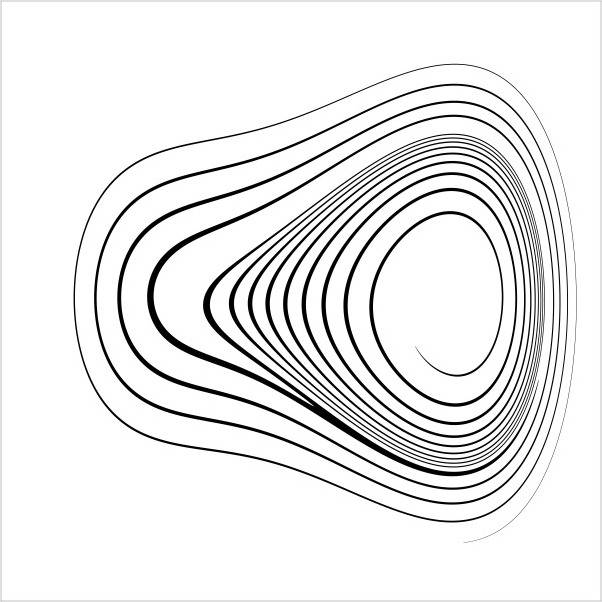}%
\includegraphics[width=\textwidth/6]{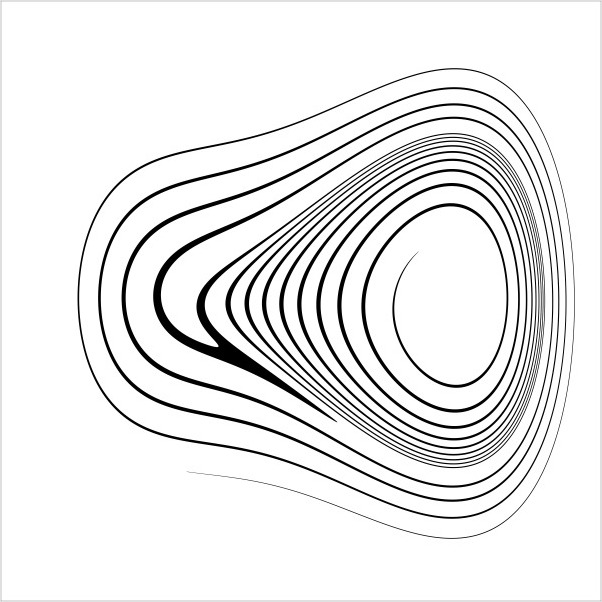}%

\vspace{-1pt}
\includegraphics[width=\textwidth/6]{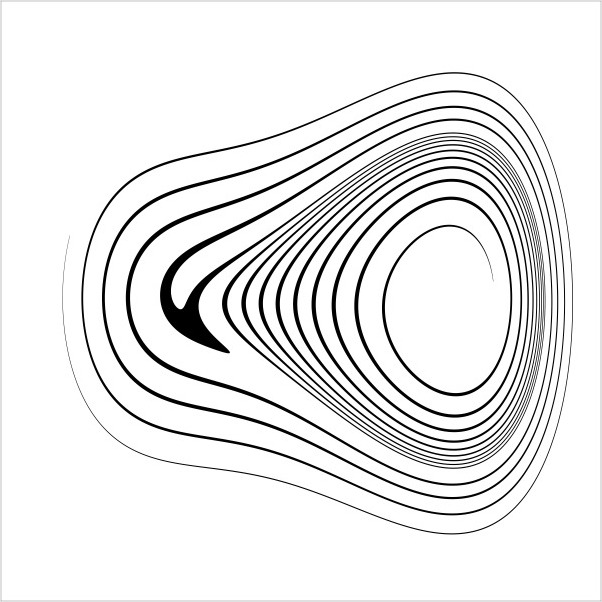}%
\includegraphics[width=\textwidth/6]{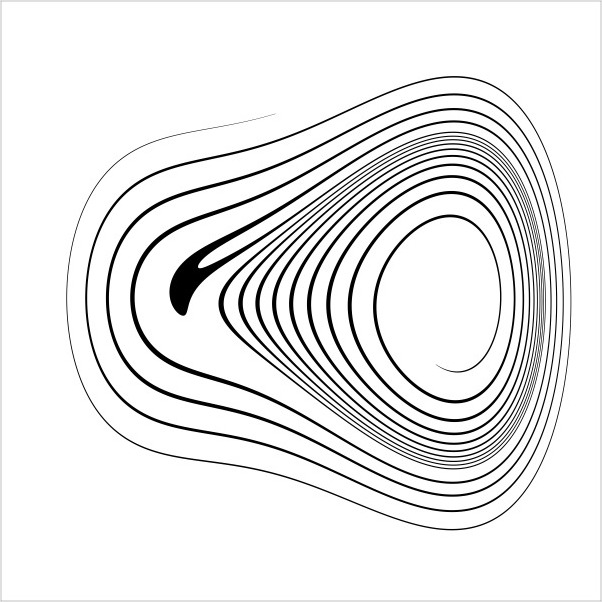}%
\includegraphics[width=\textwidth/6]{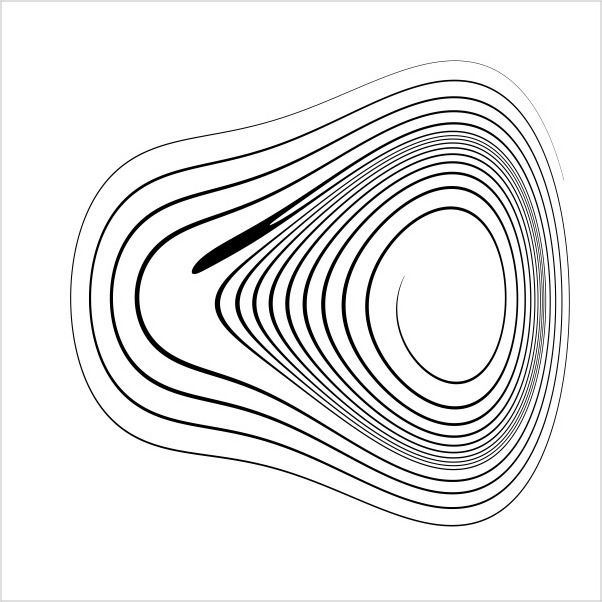}%
\includegraphics[width=\textwidth/6]{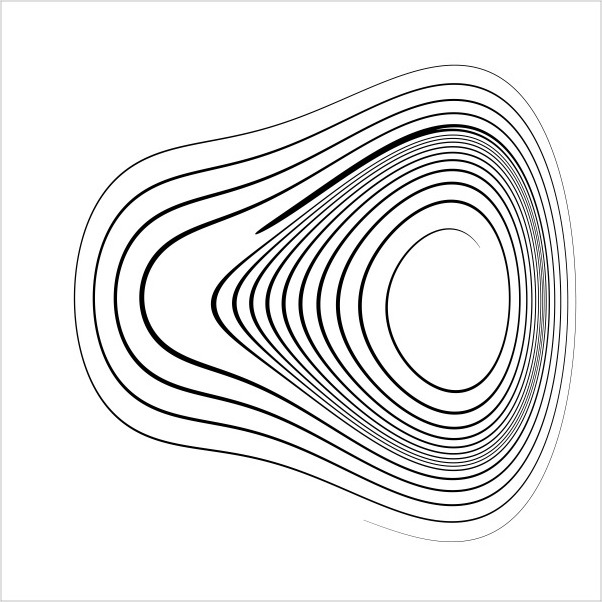}%
\includegraphics[width=\textwidth/6]{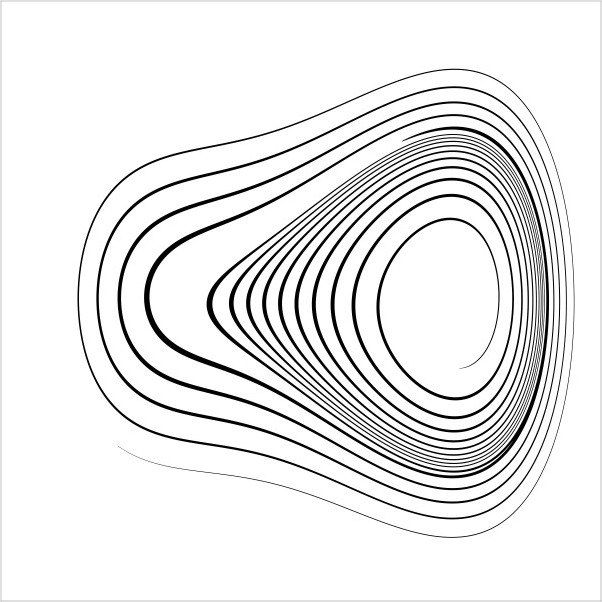}%
\includegraphics[width=\textwidth/6]{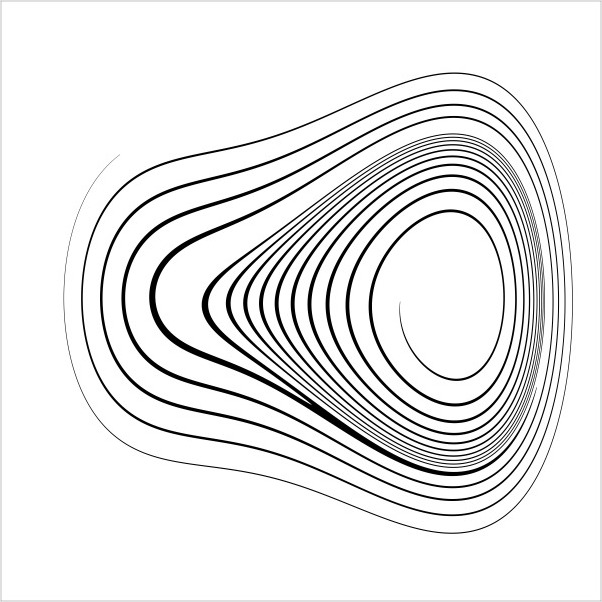}%

\vspace{-1pt}
\includegraphics[width=\textwidth/6]{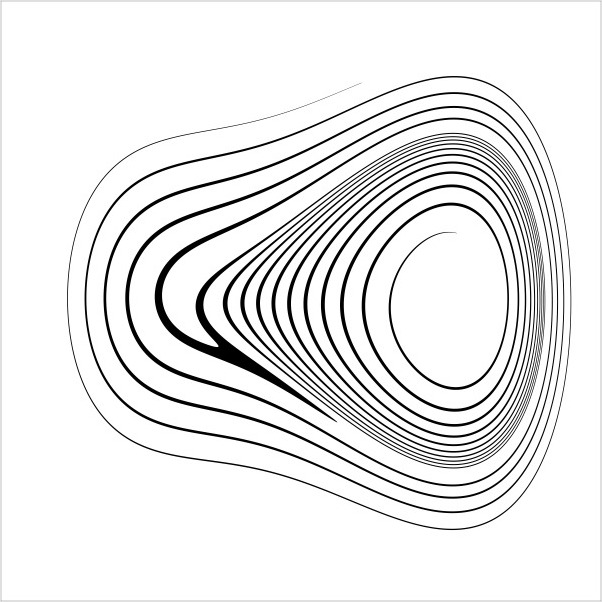}%
\includegraphics[width=\textwidth/6]{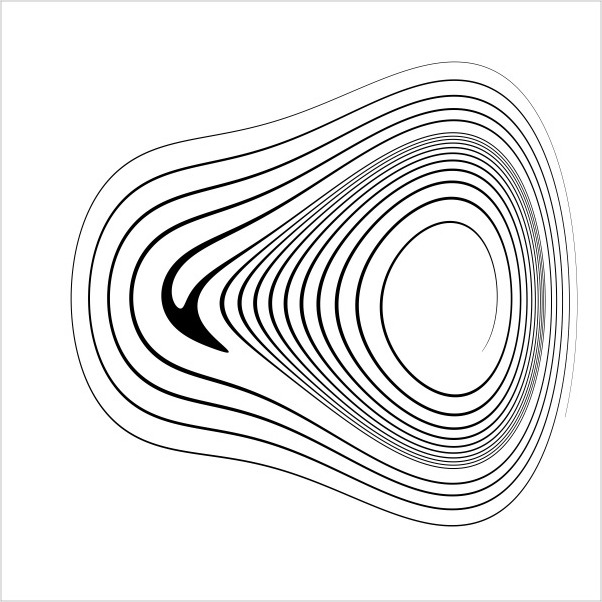}%
\includegraphics[width=\textwidth/6]{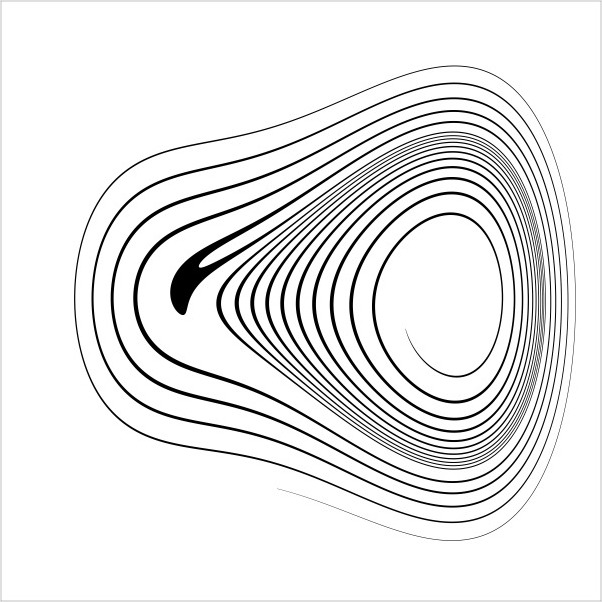}%
\includegraphics[width=\textwidth/6]{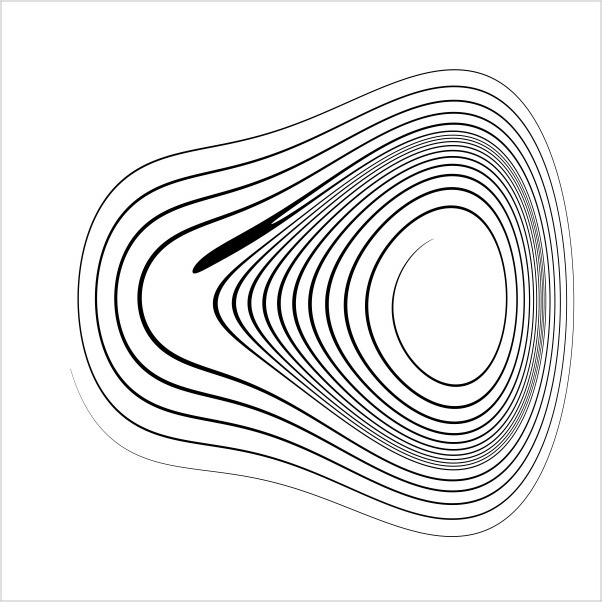}%
\includegraphics[width=\textwidth/6]{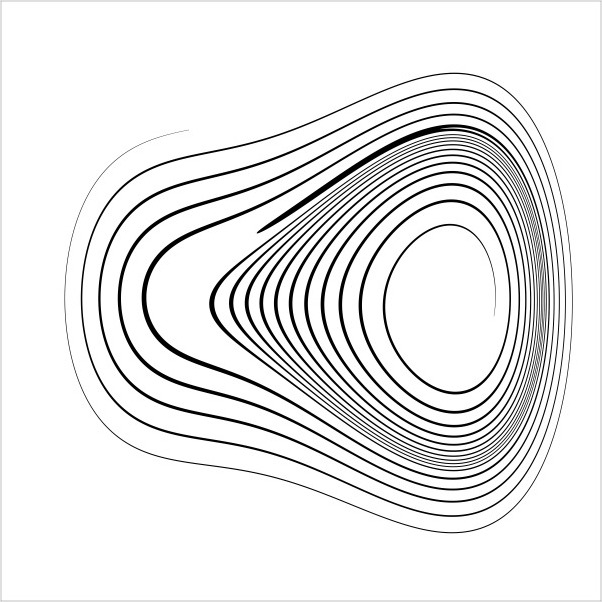}%
\includegraphics[width=\textwidth/6]{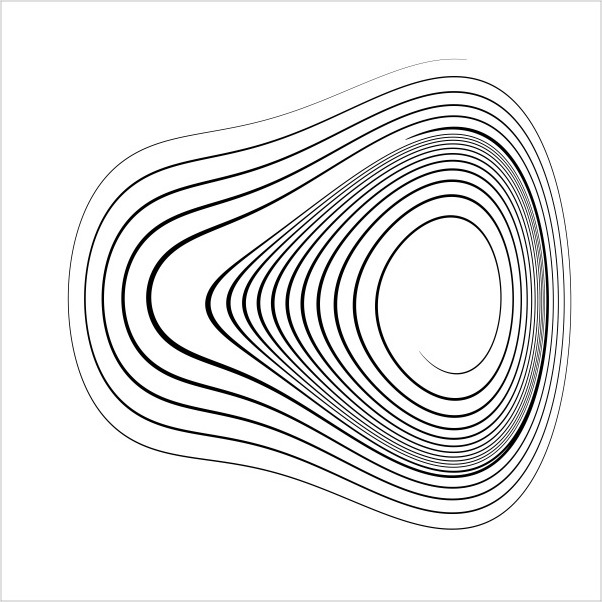}%

\vspace{-1pt}
\includegraphics[width=\textwidth/6]{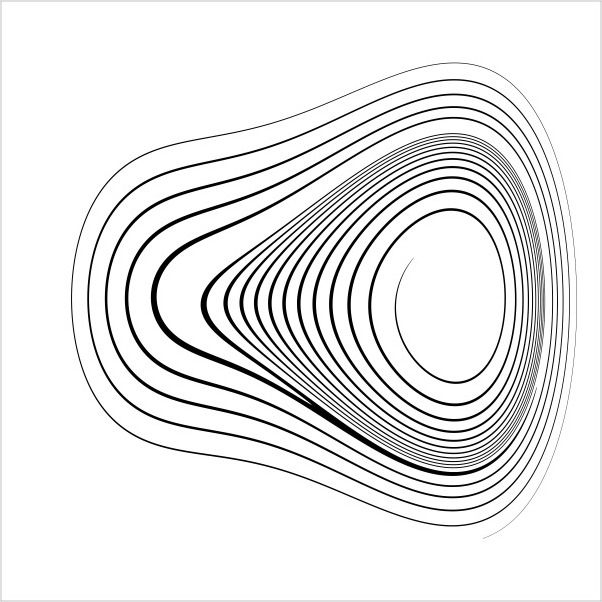}%
\includegraphics[width=\textwidth/6]{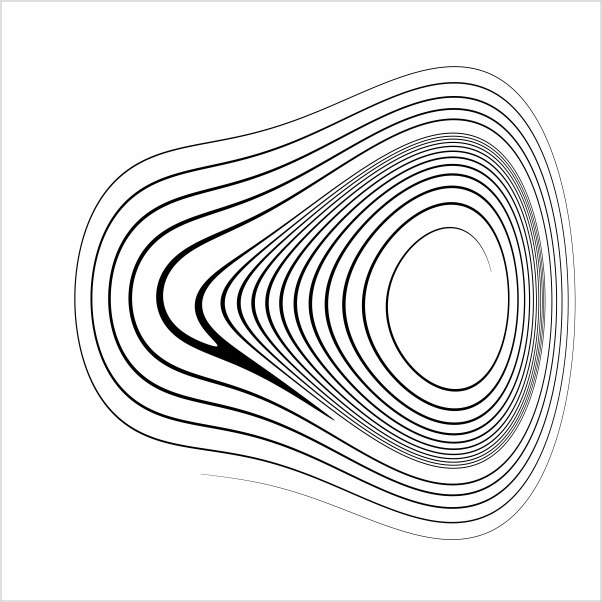}%
\includegraphics[width=\textwidth/6]{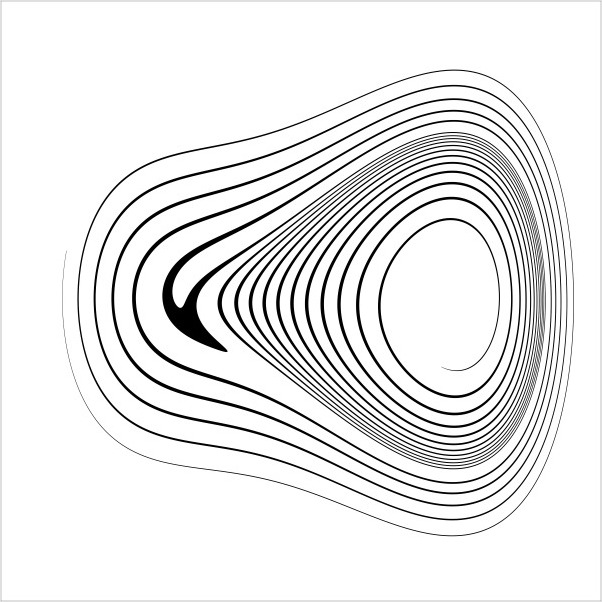}%
\includegraphics[width=\textwidth/6]{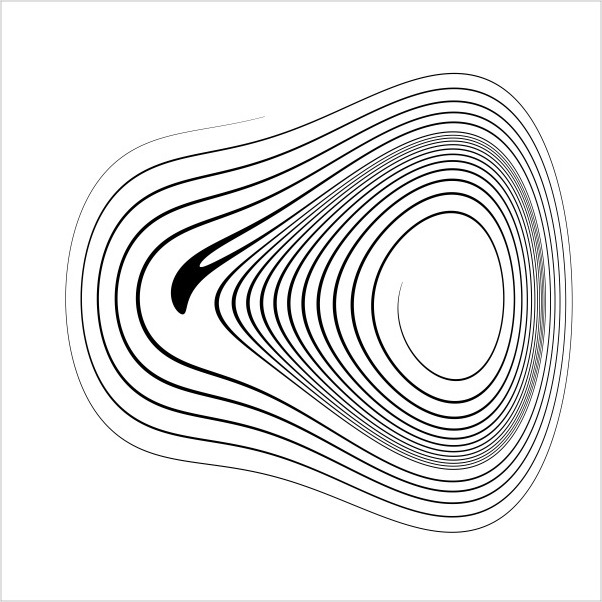}%
\includegraphics[width=\textwidth/6]{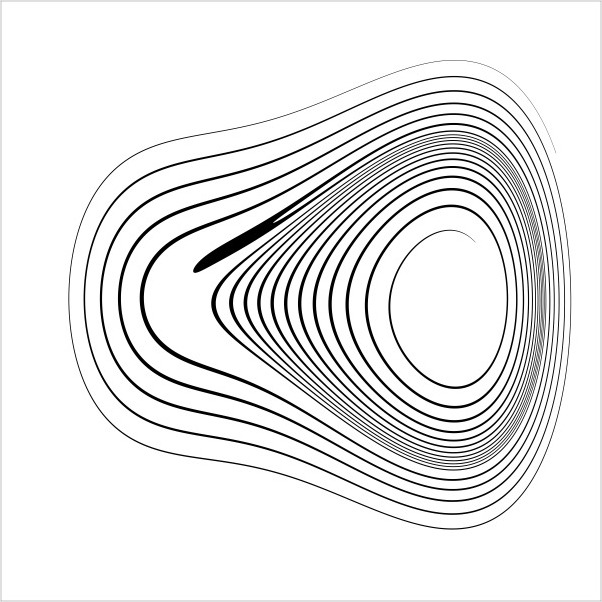}%
\includegraphics[width=\textwidth/6]{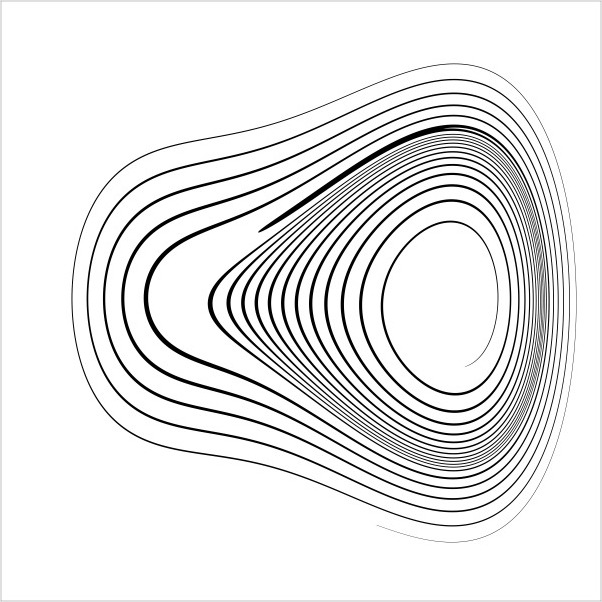}%
\caption[Classical relaxation as governed by Liouville's theorem (entropy conservation)]{Relaxation of classical nonequilibrium on the phase-space of a particle in the one-dimensional $V=q^4-q^3-q^2-q$. The probability density $\rho(q,p,t)$ of finding the particle in a particular region of the phase-space is depicted by shade. The darker tones indicate larger $\rho(q,p,t)$, and lighter tones the reverse. The first frame displays the initial conditions, in which $\rho(q,p,t)$ begins confined to a rectangular region and is constant within that region. Subsequent frames display snapshots in the evolution of $\rho(q,p,t)$. In the dark region $\rho(q,p,t)=\text{constant}$, and elsewhere $\rho(q,p,t)=0$. Hence the integrand of the entropy, $S=-\int \rho\log\rho\,\mathrm{d}q\mathrm{d}p$, is also constant in the dark region and zero elsewhere. During the Hamiltonian evolution, the first statement of Liouville's theorem, equation \eqref{L1}, dictates that the dark region may not change volume. The second statement of Liouville's theorem, equation \eqref{L2} forbids it from lightening or darkening in tone. Consequently, the \emph{exact} entropy is conserved by the dynamics. Over time however, the incompressible Hamiltonian flow stretches and warps the region. With the result that it develops fine, filamentary structure. After a sufficient time, in an empirical setting, this fine structure inevitably becomes too fine to reasonably resolve. The failure to resolve the fine structure makes the distribution appear larger in volume and lighter in shade. (Try observing the final frame with spectacles removed, for instance.) To experiment, therefore, entropy will rise, but it will only be a \emph{de facto} rise due to the experiment's own limitations.}\label{classical_relaxation}
\end{figure}

\clearpage
%%%%%%%%%%%%%%%%%%%%%%%%%%%%%%%%%%%%%%%%%%%%%%%%%%%%%%%%%%%%%%%%%%%%%%%%%%%%%%%%%%%%%%%%%%%%%%%%%%%%%%%%%%%%%%%%%%%%%%%%%%%%%%%%%%%%%
\section{Continuous state-space iRelax theories}\label{2.3}
To continue the development of iRelax theories, now consider an $n$-dimensional continuous state-space, $\Omega=\mathbb{R}^n$ (\ref{Ing1}). 
The system under consideration could be a number of particles, the Fourier modes of a physical field, or something more exotic altogether.
Whatever the system actually is, its state is written simply $x$. Together, all the possible $x$'s comprise the state-space $\Omega$.
Similarly to the discrete case, \emph{information} regarding the system is expressed as a probability distribution over the state-space, $\rho(x)$. This may represent (incomplete) knowledge on the state of a single system, or the variation of states in an ensemble of systems.
The measure of information (\ref{Ing4}) is the entropy,
\begin{align}\label{differential_entropy}
S=-\int_\Omega \rho\log\rho\, \mathrm{d}^nx.
\end{align}
This expression is something of a flawed generalization of discrete entropy \eqref{discrete_entropy}. 
It is usually called `differential entropy', and was originally introduced (without derivation) by Claude Shannon in his famous 1948 paper, reference \cite{Shannon}.
It is flawed in the sense that it doesn't take into account the issues that arise when extending the principle-of-indifference to continuous spaces (Bertrand's paradox). 
This said, much may be achieved by suspending disbelief for the time being.
The development of iRelax theories will proceed more clearly by stalling these considerations a little.
Full discussion on this point is left until section \ref{2.4}.
Until then it may be useful to note that, in contrast to discrete distributions, $\{p_i\}$, continuous distributions $\rho(x)$ are dimensionful. 
As a probability density, $\rho(x)$ inherits inverse units of state-space coordinate $x$. 
Hence, differential entropy \eqref{differential_entropy} involves a logarithm of a dimensionful quantity.
Even a simple change of units will cause the differential entropy to change. 
A true measure of information should be independent of the variables used to describe the states, and so this is problematic.

Nevertheless, the physical meaning of differential entropy \eqref{differential_entropy} is clear.
It straightforwardly carries over the notions of discrete entropy \eqref{discrete_entropy} to continuous spaces.
It is a scalar measure of how certain or uncertain the state $x$ of a system is known, according to $\rho(x)$.
It is low when the system is known to reside within a small region of the state-space (when $\rho(x)$ is well localized). 
It is high when the system cannot be well located with any confidence (when $\rho(x)$ has a large spread).
Unlike discrete entropy \eqref{discrete_entropy}, however, differential entropy \eqref{differential_entropy} is unbounded above and below. If the system is perfectly located $\rho(x)\rightarrow \delta^{(n)}(x-x_\text{location})$, it limits as $S\rightarrow-\infty$.
In the limit that $\rho(x)$ becomes uniformly spread over the state-space, $\Omega=\mathbb{R}^n$, it limits as $S=+\infty$. 
If the system were known to reside within some region $\omega\subset\Omega$ with volume $V(\omega)$, then the differential entropy would be
\begin{align}
S=\log V(\omega),
\end{align}
which parallels Boltzmann's famous expression, $S=k_B\log W$.
\subsection{State Propagators for information conservation}\label{sec:state_propagators}
To help probe the consequences of information conservation (\ref{Ing5}), it is useful to introduce a continuous equivalent of transition matrix $T_{ij}$ and its defining relation, equation \eqref{discrete_transition}.
This is the state propagator $T(x,x_0,t)$, which is the probability density that a system originally in state $x_0$ (i.e. $\rho(x')=\delta^{(n)}(x'-x_0)$) transforms to state $x$ in time $t$. An arbitrary distribution $\rho(x)$ evolves as 
\begin{align}\label{continuous_transformation}
\rho(x,t)=\int_\Omega T(x,x',t)\rho(x',0)\mathrm{d}^nx',
\end{align}
which is the continuous equivalent of equation \eqref{discrete_transition}. In the theory of Markov processes $T(x,x_0,t)$ is called the Markov kernel or the stochastic kernel. Formally, it is the Green's function of the operator $L(-t)$, which is defined by the transformation $L(-t)\rho(x,t)=\rho(x,0)$. For the present, however, the word propagator shall be used as this is familiar from quantum mechanics and quantum field theory. 

\begin{figure}[h]
\includegraphics[width=\textwidth]{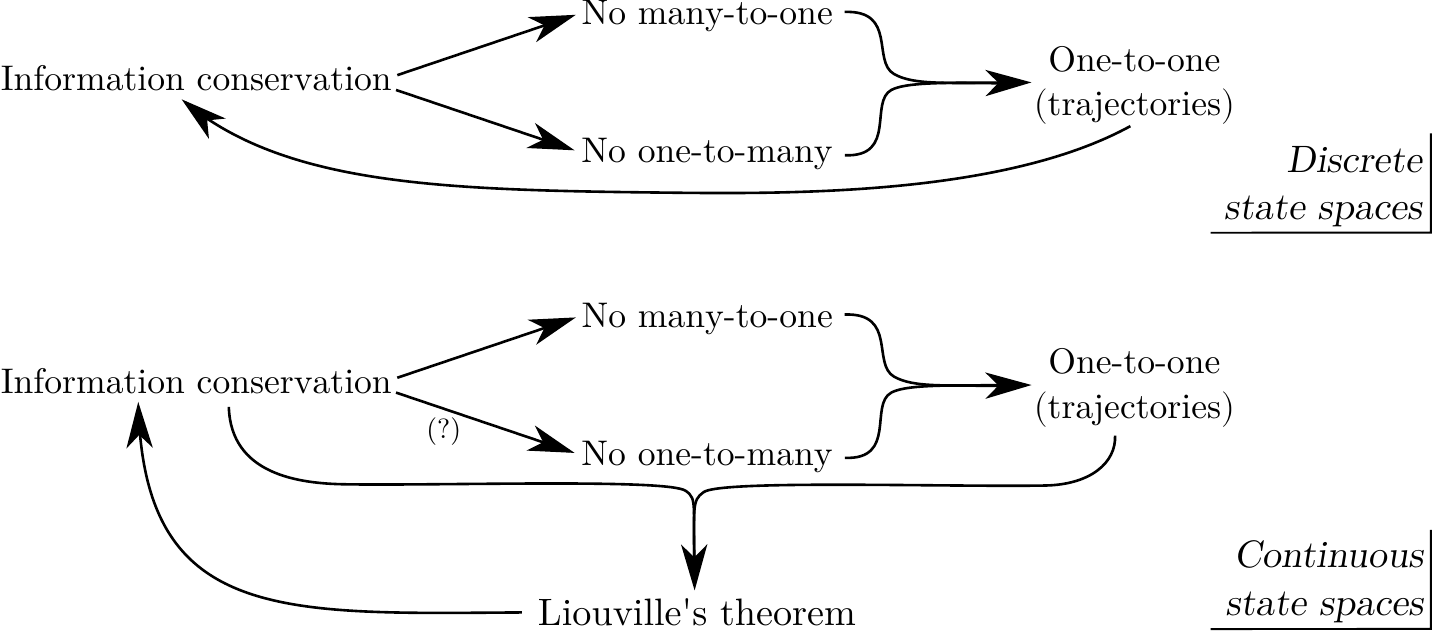}
\caption[Map of implications proven in sections \ref{2.3} and \ref{2.4}]{Contrasting mathematical implications found using the state propagator/transition matrix formalisms. The question mark indicates the missing link explained in sections \hyperref[sec:state_propagators]{2.3.2} and \hyperref[sec:modified_state_propagators]{2.4.2}.}\label{logic_map}
\end{figure}

In section \ref{2.2}, when treating discrete states, the one-to-one property trivially implied information conservation. In order to prove the equivalence of one-to-oneness and information conservation therefore, time was spent finding the reverse implication.
In contrast, for continuous state-spaces, one-to-oneness is not enough to ensure information conservation\footnote{For instance, let $\Omega=\mathbb{R}$. Consider the one-to-one law of evolution $x(t) = x_0/(1+t)$, beginning at $t=0$. Trajectories governed by this law will converge to $x=0$ as $t\rightarrow\infty$. The information-distribution evolves as $\rho(x)\rightarrow\rho(x(1+t))$. Hence, as time passes, the information-distribution will bunch-up near $x=0$, becoming ever-narrower, decreasing the entropy \eqref{differential_entropy}.}.
A more circuitous route must be taken. 
A map of the contrasting logical development is shown diagrammatically in figure \ref{logic_map}.
First, it will be shown that many-to-one laws are inconsistent with information conservation. 
Second, it shall be conjectured that one-to-many laws are also inconsistent with information conservation.
The lack of proof of this second implication (indicated with a question mark in figure \ref{logic_map}) is, admittedly, something of a loose end in the logical development of this chapter--a reflection that this chapter is still a work in progress. 
If the conjecture is accepted, then the implication is that information conservation requires systems to follow trajectories in the state-space, and forbids any probabilistic evolution.
This means the stated requirement for \ref{Ing2} to perfectly determine future evolution may be relaxed, as this is already implied by information conservation.
It is then shown that, by taking this one-to-one evolution together with information conservation, the classical Liouville's theorem (on the state-space $\Omega$ rather than phase-space) may be arrived at.
This then implies information conservation, hence showing that information conservation and Liouville's theorem are equivalent (if the no-one-to-many conjecture is accepted). 

\begin{theorem}
Information conservation $\implies$ No many-to-one on a $\Omega=\mathbb{R}^n$ state-space\label{theorem2}.
\end{theorem}
\begin{proof}
If entropy \eqref{differential_entropy} is conserved during the evolution \eqref{continuous_transformation}, then it follows that
\begin{align}\label{theorem2_misc1}
&\int_\Omega \rho(x',0)\log\rho(x',0)\mathrm{d}^nx'\nonumber\\
=& \int_\Omega\left[\int_\Omega T(x,x',0)\rho(x',0)\mathrm{d}^nx'\right]\log\left[\int_\Omega T(x,x',0)\rho(x',0)\mathrm{d}^nx'\right]\mathrm{d}^nx,
\end{align}
for all possible initial probability distributions $\rho(x',0)$, and for all times $t$.
Take $\rho(x',0)$ to be uniform on a finite sub-region $\omega\subset\Omega$. Then equation \eqref{theorem2_misc1} reduces to 
\begin{align}\label{theorem2_misc2}
0&=\int_\Omega\left[\int_\omega T(x,x',t)\mathrm{d}^nx'\right]\log\left[\int_\omega T(x,x',t)\mathrm{d}^nx'\right]\mathrm{d}^nx.
\end{align}
Now take $\omega$ to be the union of two small non-overlapping balls $\omega=B_{\epsilon_1}(x'_1)\cup B_{\epsilon_2}(x'_2)$, of radii $\epsilon_1$ and $\epsilon_2$, and centered on states $x'_1$ and $x'_2$. Take radius $\epsilon_1$ to be small enough so that $T(x,x',t)$ is approximately constant on $x'\in B_{\epsilon_1}(x'_1)$, and radius $\epsilon_2$ to be small enough so that $T(x,x',t)$ is approximately constant on $x'\in B_{\epsilon_2}(x'_2)$. Then,
\begin{align}\label{theorem2_misc3}
\int_\omega \mathrm{d}^nx' T(x,x',t) = \pi \epsilon_1^n T(x,x'_1,t) + \pi \epsilon_2^n T(x,x'_2,t).
\end{align}
Substituting \eqref{theorem2_misc3} into \eqref{theorem2_misc2}, then differentiating once with respect to $\epsilon_1$ and once with respect to $\epsilon_2$ gives
\begin{align}\label{theorem2_misc4}
0=\int_\Omega \mathrm{d}^nx\frac{n^2\epsilon_1^{n-1}\epsilon_2^{n-1}T(x,x'_1,t)T(x,x'_2,t)}{\epsilon_1^nT(x,x'_1,t)+ \epsilon_2^nT(x,x'_2,t)}
\end{align}
The integrand of equation \eqref{theorem2_misc4} is $\geq 0$ as $T(x,x',t)\geq 0$ for all $x$, $x'$, $t$. For equation \eqref{theorem2_misc4} to hold, therefore, it must be the case that
\begin{align}
T(x,x'_1,t)T(x,x'_2,t)=0
\end{align}
for all $x$, $t$, and $x'_1\neq x'_2$. Hence, no two distinct states, $x'_1$ and $x'_2$, may be mapped to the same final state $x$.
\end{proof}

\begin{conjecture}
Information conservation $\implies$ No one-to-many on a $\Omega=\mathbb{R}^n$ state-space\label{theorem2}.
\end{conjecture}

To carry on the logical development outlined in figure \ref{logic_map}, it should now be proven that a law of evolution which conserves information cannot be one-to-many. 
As indicated in the figure by a question mark, however, such a proof has not been forthcoming. 
The reason for this is the divergence to $-\infty$ of the differential entropy \eqref{differential_entropy} for distributions $\rho(x)$ that correspond to definite positions in the state-space.
For instance, suppose a system's state were known to be exactly $x_\text{A}$. Clearly this would correspond to the information distribution $\rho(x)=\delta^{(n)}(x-x_\text{A})$, which causes the entropy \eqref{differential_entropy} to diverge to $-\infty$.
This corresponds well with intuition--if a system's state is known exactly, there is no uncertainty in the state, and so the entropy should be the lowest it can be.
Suppose, however, that the system evolved in such a way to cause there to develop a 50:50 chance of finding the system in state $x_\text{A}$ or state $x_\text{B}$.
The new information distribution, $\rho(x)=1/2\left[\delta^{(n)}(x-x_\text{A})+\delta^{(n)}(x-x_\text{B})\right]$, would also correspond to a divergent $-\infty$ entropy. 
This is counter to intuition. 
As there is now some doubt over whether the system resides in state $x_\text{A}$ or state $x_\text{B}$, intuition says the uncertainty in the state (and therefore the entropy) should have risen. 
That differential entropy cannot tell the difference between two such $\rho(x)$'s seems simply to be a mathematical artifact caused by the delta functions. One might say the rise in entropy is `hidden in the infinity'.
If, for instance, the delta functions were replaced with tight Gaussians or bump functions, so that it was admitted in effect that it is impossible to pinpoint a state \emph{exactly}, then the entropy would not have diverged and it would also have risen when $\rho(x)$ split into two.  
This is the primary argument the author offers in support of the conjecture.
However it is hoped that this matter may be clarified in future.
For the reminder of this chapter, the conjecture shall be assumed to be true.

Together, the no many-to-one and no one-to-many conditions mean the law of evolution is one-to-one.
Given an initial state $x_0$ and a time $t$, information conserving laws specify a unique final state $x_f$.
In other words, the concept of trajectories has been recovered.
An individual trajectory may be written $x_f(t,x_0)$.
A probability distribution $\rho(x)$ representing perfect knowledge of such a trajectory may be written $\rho(x,t)=\delta^{(n)}\left(x-x_f(t,x_0)\right)$.
By substitution of this distribution into transformation \eqref{continuous_transformation}, it may be concluded that the information-preserving propagator is
\begin{align}\label{one-to-one-propagator}
T(x,x_0,t) = \delta^{(n)}(x-x_f(t,x_0)).
\end{align}

The standard statement of information conservation in classical mechanics is Liouville's theorem. 
The property of having one-to-oneness (trajectories) alone isn't enough to ensure Liouville's theorem is satisfied, however. 
Trajectories could bunch up, or spread apart causing $\rho(x)$ to narrow or widen in violation of Liouville's theorem.
Liouville's theorem may be recovered for an abstract law by imposing information conservation upon the trajectories in the following way.
\begin{theorem}
Information conservation $+$ one-to-one $\Rightarrow$ Liouville's Theorem.
\end{theorem}
\begin{proof}
Recall equation \eqref{theorem2_misc2}, the information conservation condition for a distribution $\rho(x_0,0)$ that is initial uniform on a sub-region $x_0\in\omega$ of the state-space $\Omega$,
\begin{align}
0&=\int_\Omega\left[\int_\omega T(x,x_0,t)\mathrm{d}^nx_0\right]\log\left[\int_\omega T(x,x_0,t)\mathrm{d}^nx_0\right]\mathrm{d}^nx.
\end{align}
In order to evaluate this expression once the one-to-one propagator \eqref{one-to-one-propagator} has been substituted in, it is necessary to evaluate the integrals 
\begin{align}
\int_\omega \mathrm{d}^nx_0\delta^{(n)}(x-x_f(t,x_0)).
\end{align}
However the argument of the delta function is linear in the final state-space coordinate $x_f$, and the integration is to be taken with respect to the initial state-space coordinate $x_0$.
As it has already been determined that there is a one-to-one mapping between these two, progress may be made through the substitution $x_0\rightarrow x_f(t,x_0)$.
The corresponding transformation of the integration measure is $\mathrm{d}^nx_0=|\mathrm{det}\mathbf{J}(x_f)|\mathrm{d}^nx_f$, where
\begin{align}
\mathbf{J}(x_f)=
\begin{pmatrix}
\frac{\partial x_0^1}{\partial x_f^1}&
\dots&
\frac{\partial x_0^n}{\partial x_f^1}\\
\vdots&\ddots&\vdots\\
\frac{\partial x_0^1}{\partial x_f^n}&
\dots&
\frac{\partial x_0^n}{\partial x_f^n}\\
\end{pmatrix}
\quad\quad\text{for}\quad\quad
\begin{matrix}
x_0=\left\{x_0^1, x_0^2,\dots,x_0^n\right\}\\
x_f=\left\{x_f^1, x_f^2,\dots,x_f^n\right\}
\end{matrix}.
\end{align}
This is the Jacobian matrix for the evolution (i.e.\ the transformation $x_0\rightarrow x_f(t,x)$). It is written with the argument $x_f$ so as to indicate the position in the state-space at which it is evaluated. It follows that
\begin{align}
\int_\omega d^nx_0\delta^{(n)}\left(x-x_f(t,x_0)\right)&=\int_{x_f(\omega)}\mathrm{d}^nx_f|\mathrm{det}\mathbf{J}(x_f)|\delta^{(n)}(x-x_f)\\
&=
\begin{cases}
|\mathrm{det}\mathbf{J}(x)|&\,\text{for}\, x\in x_f(\omega)\\
0&\text{Otherwise}
\end{cases},
\end{align}
where $x_f(\omega)$ is the projection of volume $\omega$ through the time interval time $t$, and $\mathbf{J}(x)$ is the Jacobian matrix as evaluated at $x_f=x$. By substitution of this into equation \eqref{one-to-one-propagator}, it is found that
\begin{align}
0=\int_{x=x_f(\omega)}\mathrm{d}^nx|\mathrm{det}\mathbf{J}(x)|\log |\mathrm{det}\mathbf{J}(x)|
=\int_{\omega}\mathrm{d}^nx_0\log |\mathrm{det}\mathbf{J}(x(x_0,t))|,
\end{align}
where the final equality is found as a result of making the substitution $x_f\rightarrow x_0$.
The only way for this to be equal to zero for arbitrary domains $\omega$ and arbitrary times $t$, is for 
\begin{align}\label{L1}
|\det{\mathbf{J}\left(x(x_0,t)\right)}|=1 
\end{align}
everywhere in $x_0\in\Omega$ and for all t.
Hence if an initial distribution $\rho(x,0)$ is confined to region $\omega$ of volume $\text{Vol}(\omega)$, then since 
\begin{align}
\text{Vol}(\omega)
=\int_\omega\mathrm{d}^nx
=\int_{x_f(\omega)}|\text{det}\mathbf{J}(x_f)|\mathrm{d}^nx_f
=\int_{x_f(\omega)}\mathrm{d}^nx_f
=\text{Vol}(x_f(\omega)),
\end{align}
the distribution will occupy a region $x_f(\omega)$ of identical volume after its evolution.
So information conserving laws of evolution conserve state-space volumes.
This result, along with equation \eqref{L1}, shall be referred to as the first statement of Liouville's theorem.
(For more detail on this point, see the section headed \emph{Suitable coordinates} below.)
\end{proof}
It may help to make the connection to the alternative, equivalent statement of Liouville's theorem.
In order to do so, consider a general initial information-distribution $\rho(x_0,0)$, which transforms as
\begin{align}
\rho(x,t)&=\int_\Omega \mathrm{d}^nx_0 T(x,x_0,t)\rho(x_0,0)\\
&=\int_\Omega \mathrm{d}^nx_0 \delta^{(n)}(x-x_f(t,x_0))\rho(x_0,0)\\
&=\int_{x_f(\Omega)} \mathrm{d}^nx_f |\mathrm{det}\mathbf{J}(x_f)|\delta^{(n)}(x-x_f)\rho(x_0(x_f,t),0)\\
&=|\mathrm{det}\mathbf{J}(x)|\rho(x_0(x,t),0).
\end{align}
Hence $\rho(x_f,t)=\rho(x_0,0)$, with the understanding that $x_0$ evolves to $x_f$. In other words, the information-distribution is constant along system trajectories,
\begin{align}\label{L2}
\frac{\mathrm{d}\rho}{\mathrm{d}t}=0.
\end{align}
(This is the material/convective derivative $\mathrm{d}/\mathrm{d}t=\partial/\partial t+\sum_i (\partial x_i/\partial t) (\partial/\partial x_i)$ along the trajectory.) Equation \eqref{L2} shall be referred to as the second statement of Liouville's theorem.

In order to complete the circuit of implications in figure \ref{logic_map}, one final implication must be sought.
Namely, Liouville's theorem $\implies$ information conservation. 
To proceed, note that the second statement of Liouville's theorem \eqref{L2} means that if trajectories are understood to evolve as $x_0\rightarrow x_f$ in time $t$, then $\rho(x_0,0)=\rho(x_f,t)$. So,
\begin{align}
S(0)&=-\int_\Omega \mathrm{d}^nx_0\, \rho(x_0,0)\log\rho(x_0,0)\\
&=-\int_{x_f(\Omega)} \mathrm{d}^nx_f\, |\mathrm{det}\mathbf{J}(x_f)|\rho(x_f,t)\log\rho(x_f,t)\\
&=S(t),
\end{align}
completing the round of implications in figure \ref{logic_map}.
(It is not necessary for the domain of the final $x_f$ integral, $x_f(\Omega)$ to equal the entire state-space $\Omega$. If the projection of $\Omega$ through the evolution, $x_f(\Omega)$, is indeed a subset of the whole state-space $x_f(\Omega)\subset\Omega$, then $\rho(x,t)=0$ in the complement $\left\{x\in\Omega|x\notin x_f(\Omega)\right\}$, and so the complement does not contribute to the entropy.

\subsection{Constraints on laws of evolution (deriving Hamilton's equations)}
In section \ref{2.2} information conservation constrained discrete laws of evolution to be simple permutations between states.
An analogous situation arises for continuous state-spaces.
Since $\rho(x)$ is a probability distribution, it is required to satisfy the continuity equation\footnote{Strictly speaking, in order for the distribution to satisfy a continuity equation, the system trajectories are required to be continuous in the state-space. Else systems could disappear at one point in the state-space and appear in another. This would happen if the particle were to strike a wall for instance. Upon rebounding the particle would experience a discontinuity in momentum, making the phase-space coordinate jump. It is assumed that such considerations could be accounted for when they arise without significantly affecting the formalism.},
\begin{align}\label{misc_eom}
0&=\frac{\partial\rho}{\partial t} + \nabla\cdot\left(\rho \dot{x}\right),
\end{align}
where $\nabla\cdot$ is the $n$-dimensional divergence operator and $\dot{x}$ is the time derivative of the state occupied by a system (which thereby acts as a placeholder for a possible law of evolution, \ref{Ing2}). 
The right hand side of equation \eqref{misc_eom} may be modified as
\begin{align}
&=\frac{\partial\rho}{\partial t} + \dot{x}\cdot\left(\nabla\rho\right) + \rho\left(\nabla\cdot\dot{x}\right)\\
&=\frac{d\rho}{dt} + \rho\left(\nabla\cdot\dot{x}\right).\label{modified_continuity}
\end{align}
The first term in this expression vanishes as a result of the second statement of Liouville's theorem, equation \eqref{L2}.
Hence, the law of evolution $\dot{x}$ must satisfy $\nabla\cdot\dot{x}=0$; in order to conserve information, laws of evolution must be divergence-free in the state-space. 
Of course, the property of being divergence free for all $x\in\Omega$ may be used in conjunction with the divergence theorem to conclude that
\begin{align}
0=\int_{\partial\omega} \dot{x}\cdot\widehat{\partial\omega}\,\, \mathrm{d}^{n-1}x\quad
\end{align}
for all $\omega\subset\Omega$, where $\partial\omega$ denotes the boundary of region $\omega$, and where $\widehat{\partial\omega}$ denotes a unit vector that is orthogonal to this boundary in the usual manner. 
In other words, there is no net flow of $\dot{x}$ in or out of any conceivable sub-region of the state-space $\Omega$.
The flow is said to be incompressible.

So iRelax theories that conserve differential entropy \eqref{differential_entropy} have laws of evolution that are akin to the flow of an incompressible fluid. 
Individual systems, in this analogy, would be depicted as being suspended in the fluid and guided around by it as it flows. 
The probability distribution $\rho(x)$ is a scalar defined upon the fluid, and so could correspond to a suspension of dye within the fluid. (What Josiah Gibbs called `coloring matter' in the context of incompressible Hamiltonian flow, \cite{Gibbsbook}.)
It should be qualified that the dye cannot diffuse however. 
The only way the dye may disperse throughout the fluid is by being `stirred' by the law of evolution (i.e. carried by the flow of the fluid).
As shall be described in the following section, this is key to the mechanism by which \ref{Ing5} and \ref{Ing6} may live in accord.\\
\\
The distribution $\rho(x)$ of maximum entropy (\ref{Ing3}) is always a special case for iRelax theories.
Later de Broglie-Bohm theory will be derived by proposing the Born distribution to be the maximum entropy distribution and then inferring the rest of the iRelax structure from this.
For the present case, the distribution of maximum entropy and the expression for differential entropy \eqref{differential_entropy} both reflect a decision to take the principle-of-indifference at its word. 
Differential entropy \eqref{differential_entropy} treats all regions of the state-space `equally' (more discussion on this point in section \ref{2.4}). 
The distribution $\rho(x)$ of maximum entropy is consequently uniform. 
If such a distribution of maximum entropy is unique, as is the case presently, then entropy conservation ensures it must be conserved by the law of evolution.
It becomes relevant, therefore, to ask which laws of evolution are consistent with the conservation of a uniform distribution.
The answer is of course those without sources or sinks of flow, with incompressible flow, with $\nabla \cdot \dot{x}=0$.
In this regard, the constraint placed upon the law of evolution by information conservation amounts to restricting it to conserve the distribution of maximum entropy. 

The incompressible flow that is a consequence of the conservation of differential entropy \eqref{differential_entropy}, parallels the incompressible Hamiltonian flow of classical mechanics. (Although textbooks usually state the incompressibility of Hamiltonian flow and hence ultimately information conservation to result from the form of Hamilton's equations, rather than the other way around.)
To elaborate further on this point, suppose $x=(q,p)$ and $(q,p)\in\Omega=\mathbb{R}^2$.
In $\mathbb{R}^2$, the curl of a vector field is a scalar field; a vector field $\mathbf{\bm{\xi}}$ may be decomposed into divergence and curl terms as $\mathbf{\bm{\xi}}=\nabla A + J\nabla B$, where $J=\left(\begin{smallmatrix}0&-1\\1&0\end{smallmatrix}\right)$, and where $A$ and $B$ are scalar fields. The divergence component is $\nabla A$ and the curl component is $J\nabla B$. 
As the law of evolution, $\dot{x}=(\dot{q},\dot{p})^\text{tr}$, is divergence free, it may be expressed $\dot{x}=J\nabla B$. 
In other words,
\begin{align}\label{possible_2d_eom}
\dot{q} = -\frac{\partial B}{\partial p},\quad\quad\dot{p}=\frac{\partial B}{\partial q}.
\end{align}
So an iRelax theory that is based upon a $(q,p)\in\mathbb{R}^2$ state-space, with differential entropy \ref{differential_entropy}, must have equations of motion \eqref{possible_2d_eom}, for some scalar function of the state-space $B(q,p)$.
Each permissible two-dimensional iRelax theory corresponds to a scalar field $B(q,p)$ upon the state-space.
Classical mechanics (Hamilton's equations) is recovered for the case where $B$ is the negative energy of the system, $B=-H$.

In two-dimensions, determining an information conserving law reduces to finding the correct scalar function on the two-dimensional state-space. 
In higher dimensions it is likely the situation becomes more complicated, meaning that other (symmetry?) arguments may need to be made to develop a theory fully.
For instance, if the state-space were $\mathbb{R}^3$, the standard Helmholtz decomposition could be used to express the law of evolution $\dot{x}=\nabla\times w$, where $w$ is a 3-vector field.
Hence, a 3-component field is required to determine a three-dimensional iRelax theory.
In more dimensions and on more complicated state-space manifolds, the relevant Hodge decomposition could be used.
Hodge decomposition extends the standard $\mathbb{R}^3$ Helmholtz decomposition to manifolds of arbitrary dimension using the language of differential geometry. Modulo boundary issues, Hodge decomposition allows a k-form to be decomposed into the exterior derivative $d$ of a $k-1$ form, and the codifferential $\delta$ of a $k+1$ form. The equation of motion is a vector field (1-form) upon the state-space of the system. By the Hodge decomposition, this may be expressed as the exterior derivative $d$ of a 0-form (1 component) plus the codifferential $\delta$ of a 2-form, which has a number of components equal to the binomial coefficient $\left(\begin{smallmatrix}n\\2\end{smallmatrix}\right)$, where $n$ is the dimensionality of the state-space. See references \cite{Nakahara,HHD} for further information.

\subsection{How to understand the tendency of entropy to rise}\label{2.3.3}
On the face of it, \ref{Ing5} and \ref{Ing6} appear to contradict each other.
On the one hand there is the information (entropy) conservation of \ref{Ing5} that is most usefully represented by Liouville's theorem in both its forms, equations \eqref{L1} and \eqref{L2}. This readily ensures that differential entropy \eqref{differential_entropy} is conserved. 
On the other hand, this goes against the usual notion of entropy, which famously tends to rise for isolated systems.
\ref{Ing6} is a statement stipulating that there should exist some mechanism by which entropy \emph{de facto} rises, even though the true entropy \eqref{differential_entropy} is forbidden from rising by information conservation, \ref{Ing5}.
This apparent incompatibility between \ref{Ing5} and \ref{Ing6} is at the center of the iRelax framework.

\begin{figure}
\subfloat[Evolution that is prohibited by information conservation (Liouville's theorem).]{\includegraphics[width=\textwidth]{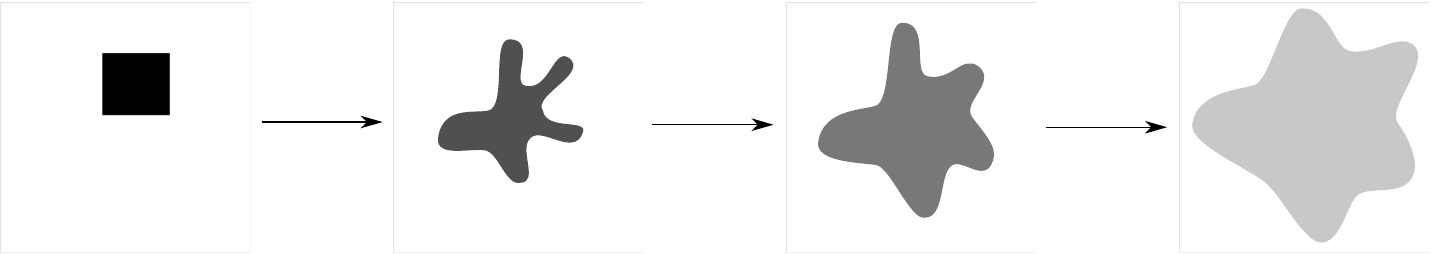}\label{prohibited}}

\subfloat[Evolution that is permitted by information conservation (Liouville's theorem).]{\includegraphics[width=\textwidth]{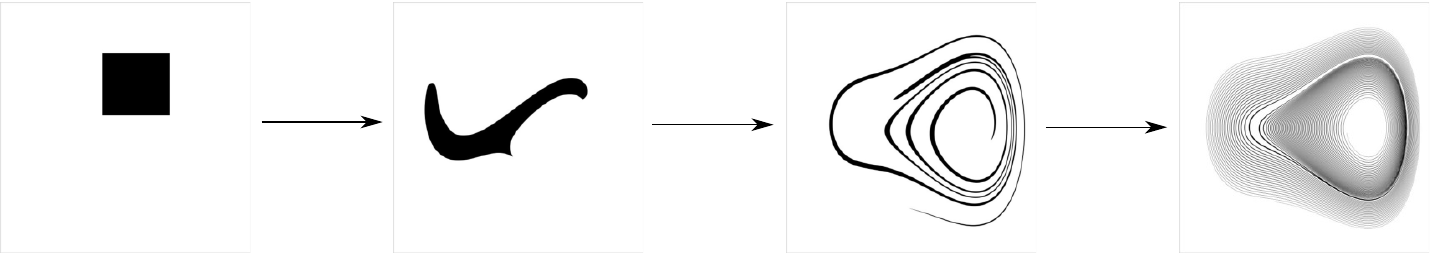}\label{permitted}}
\caption[Permitted and prohibited evolution according to Liouville's theorem]{Contrasting evolution of information-distribution $\rho(x)$ upon a 2 dimensional state-space $\Omega=\{x\}=\mathbb{R}^2$. Individual plots are colored according to the density of $\rho(x)$; darker regions indicate denser $\rho(x)$, while white regions indicate $\rho(x)=0$.}\label{permitted_vs_prohibited}
\end{figure}

The resolution of the conflict is best illustrated by example. Consider figure \ref{permitted_vs_prohibited}, which is divided into two sub-figures, \ref{prohibited} and \ref{permitted}, contrasting two possible manners in which an initial $\rho(x)$ could evolve upon a two dimensional state-space $\Omega=\mathbb{R}^2$.
In each of the individual square frames, the density of $\rho(x)$ upon the background state-space is represented by grayscale tone.
Darker regions indicate places where $\rho(x)$ is large.
Lighter regions indicate small $\rho(x)$.
White regions represent $\rho(x)=0$.
A low entropy $\rho(x)$ (denoting low ignorance of the state) is therefore represented by a small dark region. A high entropy $\rho(x)$ is represented by a wide and pale region.

The first statement of Liouville's theorem, equation \eqref{L1}, dictates that co-moving state-space volumes translate to state-space volumes of the same size if information is conserved. The second statement, equation \eqref{L2}, stipulates that along trajectories the density $\rho(x)$ remains constant. 
Clearly the evolution displayed in sub-figure \ref{prohibited} is prohibited by Liouville's theorem. The region in which $\rho\neq 0$ grows with time, in violation of \eqref{L1}, and as it does so its color lightens in violation of \eqref{L2}.
(If it grows it must lighten in color as the probability distribution $\rho(x)$ must remain normalized.)
Hence entropy does grow in figure \ref{prohibited}, as should be expected of an isolated system. But this is at the cost of violating information conservation.

In contrast, sub-figure \ref{permitted} displays an evolution that is permitted by Liouville's theorem.
The incompressible flow generated by the law means the region of non-zero $\rho(x)$ does not increase or decrease in size.
It doesn't lighten or darken either.
Rather, it becomes stretched and warped over time.
As it becomes ever more stretched and warped it develops structure that becomes ever finer and more stratified.   
Eventually the filamentary structure becomes too fine to observe clearly.
This results in a distribution that appears to have been spread over a wider region, even though strictly speaking it has not%
\footnote{In higher dimensional systems the process may happen in different ways. For instance in a three-dimensional state-space, one law of evolution may act to elongate a density $\rho(x)$ so that it eventually resembles a ball of cotton wool. Another law of evolution could instead flatten $\rho(x)$ in a single dimension so that it instead resembles a scrunched up piece of paper. In both instances, $\rho(x)$ would appear at first glance to occupy a larger volume, at a lower density, than is truly the case.}.
So the evolution proceeds in accordance with information (entropy) conservation.
But there arise practical (experimental) reasons why the fine structure becomes obscured, resulting in a distribution that appears to have been spread over a larger region, and ultimately an increase in entropy.
More detail on the classical system used to generate figure \ref{permitted} is given in figure \ref{classical_relaxation}.

Entropy-raising, fine-structure blurring practical limitations arise in numerous ways. 
Consider attempting to assign an empirical value to differential entropy \eqref{differential_entropy}.
First $\rho(x)$ must be `measured'. 
As any set of measurements performed on any individual system can yield at most the system's state, $x$, the probability distribution $\rho(x)$ must be reconstructed from repeated measurements over an ensemble. 
In the two-dimensional case displayed in figures \ref{classical_relaxation} and \ref{permitted}, the reconstruction of $\rho(x)$ could be carried out, for instance, by dividing the state-space into two-dimensional `bins' and creating a histogram from the number of ensemble members with states $x$ that fall into each bin.  
In order to reproduce the fine structure displayed in the final frame of figure \ref{permitted}, however, the bins would have to be exceedingly small. And in order to populate the bins sufficiently, the ensemble size would have to be correspondingly large. 
Clearly a tremendously large ensemble would be required to reproduce the fine structure displayed in the final frame of figure \ref{permitted} with any accuracy.
In a real empirical setting, bin size would likely be determined by the ensemble size available, and it may be expedient or necessary to choose bins that are too large to capture the fine structure. Such a choice will blur the fine structure in the reconstructed $\rho(x)$, resulting in an increase in entropy.
Even an infinite sample size could result in an entropy rise if the state measurements were performed with finite precision.
To get a better sense of this, the reader is invited to view figure \ref{permitted} with their spectacles removed.
(The infinite sample of state readings is then represented by the large number of photons entering the eye and the finite measurement precision is represented by the removal of the spectacles.)
Even if an infinite number of infinitely precise state measurements could be made, there may be other reasons why the fine structure in $\rho(x)$ is obscured. Systems may still be evolving whilst the measurements are being taken, for instance. There may be no way to ensure that the measurements are all taken at the same time. The measurements themselves may take a finite time to perform, so that systems evolve during the measurement process itself.
All these empirical limitations lead to a blurring of the fine structure in $\rho(x)$, resulting in a \emph{de facto} entropy rise, \ref{Ing6}.  

The argument above describes how the tendency of $\rho(x)$ to become stretched and warped leads to the rise of \emph{de facto} entropy.
It might be asked however, whether this stretching and warping is to be expected generally. 
Certainly it is possible to invent scenarios in which it does not take place.
The humble harmonic oscillator is a famous example. 
The paths taken by oscillators form concentric ellipses in phase-space. 
As the period of the oscillation is famously independent of amplitude, nearby oscillators remain in phase with each other throughout the evolution. 
Hence, in the case of an oscillator, $\rho(x)$ does not get stretched and warped. It merely rotates%
\footnote{This fact is put to use in section \hyperref[sec:classical_fields]{2.3.7} in a field theoretic context.}.
Entropy does not rise. It does not lower either. 
Still, any perturbation around the Hamiltonian of the oscillator would destroy this balance in favor of an entropy rise.
This particular scenario appears, therefore, to be a peculiar property of the oscillator (reflecting the rotational symmetry of the Hamiltonian in phase-space), rather than something to be concerned about.  
A much more commonly raised scenario involves inventing a $\rho(x)$ that will evolve from a complicated and structured distribution to a simple smooth distribution. 
This is easy to do since the time reversibility of the law of evolution ensures it may be extrapolated backwards. Thus creating a complicated filamentary distribution appear to evolve from a simply smooth distribution.
This argument is commonly made in the foundations of statistical physics.
The issue with it is that the possibility of this behavior, though non-zero, is so overwhelmingly unlikely that most scientists are minded to drop such concerns in the interest of pragmatism and utility.
Bearing in mind the above discussion on the incompressible flow of the law of evolution, the situation is directly analogous to stirring a glass of fruit squash and finding it separate into syrup and water\footnote{Or stirring chocolate milk and finding it separate into milk and chocolate powder.}. 
These so-called `conspiracies' of initial conditions \cite{AV91a} are therefore discounted for the usual reasons. See reference \cite{Daviesbook} for further discussion.

\subsection{Proof of entropy rise via coarse-graining}
In the view of the author, no further proof of entropy rise is necessary.
The justifications given above are conclusive. 
Nevertheless, the proof below has been important historically in the development of the subject area.
Through comparison, its presentation here will help to highlight some important details in section \ref{2.4}.
And consideration of the assumptions required by the proof will also help clarify some remarks on the arrow of time that will follow the proof. 
In particular, the proof does not require the particular details of a law of evolution, only the conservation of information, Liouville's theorem.
Before beginning, however, a word of warning.
The proof is based on so-called coarse-graining (discretization of the state-space).
In the past, this has lead some to falsely conclude that coarse-graining is an essential component for relaxation to take place. 
(Reference \cite{Penrosebook} even uses the apparent arbitrariness of how the coarse-graining is applied to argue against the robustness of the second law of thermodynamics.)
This is certainly not the case.
As discussed at length above, coarse-graining is but one method available to realize the experimental limitations that cause the \emph{de facto} rise in entropy.
In fact any kind of blurring of $\rho(x)$ would cause an equivalent entropy rise. 
In principle it should be possible, for instance, to model the blurring with some sort of Gaussian smearing and perform the proof using that in place of the coarse-graining. 
In this regard, the coarse-graining should be viewed as a convenient mathematical tool, and not something of vital importance\footnote{In the opinion of the author, it is regrettable that the term `coarse-graining' persists. In modern parlance, discretization, pixelation, and rasterization all seem more appropriate. The use of the word pixelated to describe a coarse grained image is widespread for instance.}.

\begin{proof}
The following is based upon the proof given in reference \cite{Daviesbook}, which in turn is based upon the treatment by P. and T. Ehrenfest in reference \cite{Ehrenfests} (English translation).
Experimental limitations are modeled via coarse-graining, or in other words, the discretization of the state-space.
Divide the state-space into cells, the volume of which, $\delta \omega$, reflects the precision of the experimental readings.
Suppose that experiment only has access to the average of $\rho(x)$ within each of the cells,
\begin{align}
\bar\rho(x):=\int_{\delta \omega}\rho(x)\,\mathrm{d}^nx.
\end{align}
Hence, experiment calculates entropy as,
\begin{align}
\bar S:=-\int_\Omega \mathrm{d}^nx \bar\rho(x)\log \bar\rho(x),
\end{align}
which is often called the coarse-grained entropy.
The argument relies on one crucial assumption, the postulate of equal \emph{a priori} probabilities. This asserts that $\bar\rho(x)=\rho(x)$ for the initial density $\rho_0(x)$. Of course, this will never be completely true. Implications of this postulate are interpreted below. For now, it will be assumed that the initial distribution is smooth enough so that for a sensible choice of coarse-graining the postulate is at least close to being satisfied. If it is satisfied, then the initial coarse-grained entropy $\bar S_0$ will equal the initial exact entropy $S_0$, which is preserved in time as a result of the information conservation and so will equal the final exact entropy $S$. Hence, the effect of the postulate is to make the correspondence between the rise in entropy and the effect of coarse-graining upon the final density, 
\begin{align}
\underbrace{\bar S-\bar S_0}_{\text{Rise in CG entropy}}&= \underbrace{\bar S-S}_{\text{Effect of coarse-graining}}.
\end{align}
Then, since $\log\bar\rho(x)$ is constant over each cell it follows that $\int \mathrm{d}^nx\bar\rho(x)\log\bar\rho(x)=\int \mathrm{d}^nx\rho(x)\log\bar\rho(x)$ and thus the change in entropy may be written
\begin{align}
\bar S-\bar S_0&=-\int_\Omega \mathrm{d}^nx\,\left\{\bar\rho(x)\log\bar\rho(x)-\rho(x)\log\rho(x)\right\}\\
&=\int_\Omega \mathrm{d}^nx\,\rho(x)\log\frac{\rho(x)}{\bar\rho(x)}.
\end{align}
The integrand here, $\rho(x)\log\left[\rho(x)/\bar \rho(x)\right]$, is positive whenever the exact density is greater than its coarse-grained counterpart, and negative for the reverse. Since there is a factor of $\rho$, however, a greater weight is given to the regions of the state-space in which it is positive. This fact may be exploited by adding $\bar\rho-\rho$ to the integrand. Since $\bar\rho$ is equal to $\rho$ on average, such a term will contribute nothing to the final change in entropy. It will furthermore ensure that the integrand is now positive for all regions in which $\rho<\bar\rho$. Hence, the change in entropy may be written
\begin{align}
\Delta S &=\int_\Omega \mathrm{d}^nx\,\left\{\rho(x)\log\frac{\rho(x)}{\bar\rho(x)}+\bar\rho(x)-\rho(x)\right\}\geq 0,
\end{align}
where the inequality is found as $a\log(a/b)+b-a\geq0$ for all $a\geq 0$, $b\geq 0$. 

\end{proof}

\subsection{Remarks on the arrow of time}
Contained in the above discussion are some points relevant to that perennial source of interest, the emergence of the arrow of time.
There are many directions from which to approach the question of the arrow of time. Relativity has of course much to say regarding time, for instance.
As the abstract iRelax theories under analysis at present do not yet possess a concept of space, however, relativistic considerations are a little premature.  
Rather, iRelax theories exhibit a very thermodynamical kind of time-asymmetry. Exactly the kind found in classical mechanics.

iRelax theories are time-symmetric theories. Not only does the one-to-oneness of the laws of evolution ensure they are in principle perfectly predictive both forwards and backwards in time. The information conservation also ensures that there is no explicit time asymmetry hard-coded. 
Nevertheless iRelax theories do exhibit spontaneous thermodynamic relaxation to equilibrium.
As has been covered at length, the reason for this is the tendency of the incompressible flow to stretch and warp simple smooth $\rho(x)$'s into complicated structures with fine filaments. This filamentary structure is then `blurred' by the inevitable shortcomings of experiment, obscuring the fine structure, making $\rho(x)$ appear to have been smeared out, resulting in an entropy rise. 
The process of stretching and warping is shown in detail in figure \ref{classical_relaxation} for a simple two-dimensional state-space.
A key point to note is that, had this simulation been evolved backwards in time from the initial uniform $\rho(x)$, a similarly filamentary structure would also have formed.
In iRelax theories the formation of structure, and thus thermodynamic relaxation, occurs \emph{regardless of the direction time flows}.
In this respect, the notion of `future' appears to have a close association to complex structure and the notion of `past', simple structure. The complex filamentary `future' distributions evolve from simple `past' distributions regardless of the direction time flows.  

This has lead some to question the validity of the postulate of equal \emph{a priori} probabilities. 
If there is an intrinsic tendency to form structure, it could be asked, is it reasonable to assume there to be no structure in any `initial' distribution?
To help answer this question, suppose there were an ensemble of systems distributed with fine structure, and that this structure were too fine to be perceptible to `low-resolution' experiment. 
Recall that $\rho(x)$ does not just have an ensemble interpretation, however.
In the Bayesian sense it also represents ignorance and knowledge, and knowledge of how an ensemble is distributed may be incomplete.
In cases of incomplete knowledge of $\rho(x)$, the principle of maximum entropy prescribes the assignment of the $\rho(x)$ with the maximum entropy that is consistent with the information available.
This may be a smooth distribution, despite the presence of fine structure in the actual ensemble distribution of systems.
The above proof goes through fine with the caveat that if $\rho(x,0)$ reflects both the variation of an ensemble and also incomplete knowledge of said variation, then the subsequent later time higher-entropy $\rho(x,t)$ also reflects both these factors.
A subsequent `high-resolution' measurement of $\rho(x,t)$ would of course yield a different ensemble distribution to that predicted.
This is because it is directly observing the ensemble distribution, not the prediction made of it from a position of ignorance. 
This is not surprising at all.
For the purposes of comparison, suppose an individual system were assigned a probability distribution $\rho(x,0)$ based on a poor-resolution measurement of its state. This would evolve to $\rho(x,t)$. A high-resolution measurement of its later state, however, would return its true state, not the $\rho(x,t)$ calculated. How could it? It is a single coordinate, not a distribution.

Nevertheless it is disturbing to some that entropy could appear to decrease were initially imperceptible fine structure to become visible at later times. 
Consider however, that the initially too-fine-to-see structure is nonetheless visible at some smaller lengthscale.
If the coarse-graining employed in the above proof were instead applied on this smaller lengthscale, an entropy rise would still be concluded. And this would be attributed to the structure evolving to become even finer than it was initially. 
It follows there is a hierarchy. 
Smooth distributions become structured.
Structured distributions become more finely structured.
Finely structured distribution become still more finely structured.
And so on. And so on.

It may be useful to draw a comparison between the assumption of a smooth initial $\rho(x)$ and the so-called coastline paradox, first introduced by Lewis Fry Richardson in reference \cite{R61}, and later popularized by Benoit Mandelbrot \cite{M67}.
Famously the question `How long is the coastline of Britain?' was answered `It depends how long your ruler is!'. 
As the coastline is jagged, an estimate of its length very much depends upon how much detail one is willing to take into account. 
Rulers about a mile long may capture the larger structure sufficiently, but miss bays smaller than a mile wide. 
Rulers about 100 meters long may capture these bays well, but smooth over the jagged, rocky parts of the coastline. 
Smaller rulers may capture the rocky parts well but fail to capture even smaller structure. And so on. 
So the shorter the ruler, the more detail taken into account, the longer the estimate of the coastline. 
And in principle this need not have a limit except for the discovery of some smallest lengthscale, which may in practice be out of reach in any case.
In order to make an estimate that may be of some utility (perhaps by comparison with other countries), a pragmatic choice of ruler length must be made. 
A similarly pragmatic choice must be made with the lengthscale of coarse-graining used to calculate entropy from an ensemble. 
In any real ensemble, there may exist many different lengthscales upon which structure is present.
The finer the coarse-graining used to calculate entropy, the more this structure will be resolved, and the lower the resulting estimate of the entropy\footnote{A parallel may also be drawn with astronomical observations. More sensitive telescopes may view more distant objects, from which the light will have departed earlier, and thus in a sense view farther backwards in time. The finer the structure in $\rho(x)$, the farther back in time it will have formed. Hence, finer scale coarse-graining that is able to resolve such structure is in effect able to measure a lower entropy from an earlier time.}. 
It may be that the finest structure requires a lengthscale of coarse-graining that is out of reach experimentally. 
Nevertheless, a pragmatic choice of lengthscale will still lead to an estimate that may be of some utility. 
Just as expediency inspired a choice of ruler length to measure coastline, the limitations of the experiment at hand dictate the lengthscale of coarse-graining used to measure entropy. 
As long as the initial distribution $\rho(x,0)$ appears suitably smooth on the lengthscale chosen, the above proof follows, and so entropy at that coarse grained lengthscale may be expected to rise.

\subsection{Suitable coordinates}\label{sec:suitable}
The form of differential entropy, equation \eqref{differential_entropy}, betrays the fact that the distribution of maximum entropy is expected to be uniform with respect to coordinates, $x$. 
It would not be form-invariant under a general coordinate transform, $x\rightarrow x'$. 
Hence the coordinates $x$ are in a sense special to differential entropy. 
This should not be surprising however.
Even in classical mechanics there are special sets of coordinates called `canonical'. 
Use of these canonical coordinates guarantees form invariance of the law of evolution (Hamilton's equations).
Coordinate transforms that translate between sets of canonical coordinates, and so preserve the form of Hamilton's equations, are also given special status as canonical transformations.
Interestingly, as a result of preserving Hamilton's equations, canonical transforms also leave state-space volumes invariant \cite{Goldstein}.
This means they ensure the form-invariance of differential entropy and Liouville's theorem, essentially leaving the informational structure of the iRelax theory presently under analysis untouched.
Hence the coordinates of an iRelax theory may be changed with the aid of a canonical transform (at least if the state-space has an even number of dimensions).
In the context of the iRelax framework, however, entropy and Liouville's theorem are of foremost concern while Hamilton's equations are not, as they only apply in the context of classical mechanics.
So canonical transforms, as a class, may be too restrictive to capture all possible transforms that leave entropy \eqref{differential_entropy} and so the rest of the iRelax structure form-invariant.

It becomes relevant to ask after the iRelax framework equivalent of the canonical transformation. For the purposes of this section the word `suitable' shall be used as the iRelax equivalent of `canonical' so that, for instance, suitable transformations shall be said to translate between sets of suitable coordinates. 
As suitable transformations are required to leave differential entropy \eqref{differential_entropy} form-invariant, the distribution of maximum entropy (the equilibrium distribution $\rho_\text{eq}(x)$) must be uniform with respect to any set of suitable coordinates.
In order to ensure this is the case, the Jacobian matrix $\mathbf{J}(x\rightarrow x')$ of a suitable transformation $x\rightarrow x'$ must satisfy $\text{det}\mathbf{J}(x\rightarrow x')=\pm1$ at every point in the state-space.
To illustrate the implications this has for the sets of suitable coordinates, consider the following back-of-the-envelope calculation. Suppose a small vector points from position $x$ to position $x+\mathrm{d}x$.
In new coordinates $x'(x)$, the same small vector points from $x'(x)$ to $x'(x+\mathrm{d}x)$, so that in the new coordinates the vector is $\mathrm{d}x'=x'(x+\mathrm{d}x)-x'(x)$.
Taylor expanding $x'(x+\mathrm{d}x)$ and keeping only linear terms in $\mathrm{d}x$ gives
\begin{align}
\mathrm{d}x'=
\begin{pmatrix}
\frac{\partial x'_1}{\partial x_1}\mathrm{d}x_1+\frac{\partial x'_1}{\partial x_2}\mathrm{d}x_2+\dots+\frac{\partial x'_1}{\partial x_n}\mathrm{d}x_n\\
\vdots\\
\frac{\partial x'_n}{\partial x_1}\mathrm{d}x_1+\frac{\partial x'_n}{\partial x_2}\mathrm{d}x_2+\dots+\frac{\partial x'_n}{\partial x_n}\mathrm{d}x_n
\end{pmatrix}
= \mathbf{J}(x\rightarrow x')\mathrm{d}x.
\end{align}
So the action of the Jacobian matrix may be understood to be to translate small vectors at each point in state-space into the new coordinates. 
This may be used to view the effect of the coordinate transformation on state-space volumes in the following way. 
Combine $n$ such small vectors, $v^1$, $v^2$, \dots , $v^n$, into matrix
\begin{align}
P= 
\begin{pmatrix}
v^1_1&\dots&v^n_1\\
\vdots&\ddots&\vdots\\
v^1_n&\dots&v^n_n
\end{pmatrix}.
\end{align}
If the vectors $v^1$, $v^2$, \dots , $v^n$ are considered to define the edges of a small $n$-parallelepiped, the volume of the parallelepiped is given by the absolute value of the determinant $\text{Vol}(P)=|\mathrm{det}P|$.
As the coordinate transform acts on $P$ by matrix multiplication, $P\rightarrow P'=\mathbf{J}(x\rightarrow x')P$, the volume transforms as
\begin{align}
\mathrm{Vol}(P')&=|\mathrm{det}P'|=|\mathrm{det}\mathbf{J}(x\rightarrow x')P|\nonumber\\
&=|\mathrm{det}\mathbf{J}(x\rightarrow x')||\mathrm{det}P|=|\mathrm{det}P|=\text{Vol}(P).
\end{align}
Hence any coordinate transform that leaves these volumes invariant is considered a suitable transformation.
Group theoretically, Jacobian matrices of suitable transformations (i.e.\ that satisfy $|\mathrm{det}\mathbf{J}(x\rightarrow x')|=1$) belong to $SL^\pm(n,\mathbb{R})$, the special linear group of $n$-dimensional matrices with $\pm1$ determinant. If attention is restricted to handedness-preserving $+1$ determinant Jacobians, the symmetry group is instead the usual special linear group $SL(n,\mathbb{R})$. In two-dimensions this is isomorphic to the symplectic group $SL(2,\mathbb{R})=Sp(2,\mathbb{R})$, the symmetry group of canonical transformations, so there is no difference between canonical transforms and suitable transforms in two-dimensional state-spaces. 
In higher dimensions, $Sp(2n,\mathbb{R})$ is a subgroup of $SL(2n,\mathbb{R})$, reflecting the fact that canonical transformations must satisfy a more restrictive condition than suitable transformations.
The special linear group $SL(n,\mathbb{R})$ may be thought to depend on $n^2-1$ parameters, whilst the symplectic group $Sp(2,\mathbb{R})$ depends on $n(n+1)/2$ (and isn't defined for odd $n$).
As $1=\text{det}\mathbf{J}=\text{det}\left(e^{\epsilon \mathbf{X}}\right)=\left[\text{det}\left(e^\mathbf{X}\right)\right]^\epsilon=\left[e^{\text{tr}\mathbf{X}}\right]^\epsilon$, elements $\mathbf{X}$ of the generating Lie algebra $sl(n,\mathbb{R})$ are traceless matrices.

%%%%%%%%%%%%%%%%%%%%%%%%%%%%%%%%%%%%%%%%%%%%%%%%%%%%%%%%%%%%%%%%%%%%%%%%%%%%%%%%%%%%%%%%%%%%%%%%%%%%%%%%%%%%%%%%%%%%%%%%%%%%%%%%%%%%
\subsection{Classical fields (spanning the state-space with a complete set of solutions)}\label{sec:classical_fields}
As later chapters will focus on quantum fields, it is useful to consider information conservation in the context of classical fields.
Many of the above considerations also apply to field theories with the understanding that the individual states, $x(t)$, become fields, $\phi(\mathbf{q},t)$, and that everything else should be altered accordingly.
There are details that arise in the field theoretic context that are worth remarking upon, however.
In particular, free fields commonly Fourier transform to sets of decoupled harmonic oscillators, and this allows a property of harmonic oscillators to be exploited.
Namely the famous independence of period upon amplitude already highlighted in section \ref{2.3.3}.
A finely-tuned property of the harmonic oscillator that, due to the trivial evolution of $x(t)$, prevents relaxation of $\rho(x,t)$.
The exploit parallels the interaction picture of quantum theory, in that it extracts the `trivial' free behavior of the free field, in principle leaving behind the effect of the `more interesting' interactions.

The harmonic oscillator Hamiltonian, $H=(p^2+\omega^2 q^2)/2$, possesses the $U(1)$ symmetry $(\omega q, p)^\text{tr}\rightarrow R(\theta)(\omega q, p)^\text{tr}$, where $R(\theta)$ is the standard rotation matrix $R(\theta)=\left( \begin{smallmatrix}\cos\theta & -\sin\theta\\ \sin\theta & \cos\theta\end{smallmatrix}\right)$.
This results in a particularly simple evolution of $\rho(x,t)$.
By spanning the state-space with elliptical polar coordinates $r=\sqrt{\omega q^2+p^2/\omega}$ and $\theta=\arctan(p/\omega q)$, the equations of motion may be expressed simply $\dot{r}=0$, $\dot\theta=-\omega$.
Systems merely traverse paths of constant $r$ in a clockwise fashion at constant angular velocity $\dot{\theta}$.
By adopting the abbreviated coordinates $x=(q',p')^\text{tr}=(\sqrt{\omega} q, p/\sqrt{\omega})^\text{tr}$%
\footnote{The coordinates $x=(q',p')=(\omega q, p)^\text{tr}$ would also result in the trivial rotating evolution \eqref{rotating_evolution}. However, as this would amount to a stretch of coordinates in a single direction, the Jacobian determinant of the transformation does not equal unity, and hence it is not a suitable transformation in the above section \hyperref[sec:suitable]{2.3.6}. It does not leave entropy form invariant.}%
, the evolution of a system may be summarized by
\begin{align}\label{rotating_evolution}
x(t)= R(-\omega t) x(0).
\end{align}
With time, the corresponding probability distribution, $\rho(x,t)$, is simply rotated around the phase-space,
\begin{align}
\rho(x,t)=\rho(R(\omega t)x,0).
\end{align}
This particularly simple rotational evolution may be extracted into the definition of the coordinates by assigning coordinates $x'$ that co-move with the rotation, $x'=R(-\omega t)x$.
Then the evolution of $\rho(x')$ is frozen $\rho(x',t)=\rho(x',0)$.
And as a result, the entropy is trivially conserved, $\mathrm{d}S/\mathrm{d}t=0$.
If there were perturbations to the Hamiltonian, however, this would create perturbations in $\rho$, resulting in the same relaxation described earlier this chapter. 
And these perturbations would be much easier to view in the absence of the trivial free evolution.
A similar logic is used in the process of quantization of classical field theories, where it is common to employ a complete set of classical solutions \cite{Mandl_Shaw,Peskin_Schroeder}. 
The following example is intended to illustrate how this may be used to adopt useful bases for the description of field configurations.

Consider the complex Klein-Gordon field $\phi$ on a background spacetime $q=(t,\mathbf{q})$, where $\mathbf{q}\in \mathbb{R}^d$.
The corresponding Lagrangian is
\begin{align}
L=\int \dee^dq\left[\partial_\nu\phi\partial^\nu\phi^*-\mu^2\phi\phi^*\right].
\end{align}
In $\mathbf{q}$-space, the state of the field is described by coordinates $x=(\phi,\pi)$, and so the entropy would involve a functional integration $S=-\int\mathcal{D}\phi\mathcal{D}\pi\rho(\dots)\log\rho(\dots)$.
In order to avoid such a functional integration, assume the field to reside in a box, so that it may be expanded in a countable set of Fourier modes, 
\begin{align}\label{base_fourier_transform}
\phi(\mathbf{q},t)=\sum_\mathbf{k}\frac{e^{i\mathbf{k.q}}}{\sqrt{V}}\phi_\mathbf{k}(t),
\end{align}
where $V$ is the volume of the box. In terms of the Fourier mode components $\phi_\mathbf{k}(t)$, the Lagrangian may be written,
\begin{align}
L=\sum_\mathbf{k}\left[\dot{\phi}_\mathbf{k}(t)\dot{\phi}_\mathbf{k}^*(t) -(\mathbf{k}^2+\mu^2)\phi_\mathbf{k}(t)\phi_\mathbf{k}^*(t)\right].
\end{align}
Upon identification of canonical momenta $\pi_\mathbf{k}(t)=\partial L/\partial\dot{\phi}_\mathbf{k}=\dot{\phi}_\mathbf{k}^*(t)$ and their complex conjugates $\pi^*_\mathbf{k}(t)=\dot{\phi}_\mathbf{k}(t)$, this corresponds to the Hamiltonian
\begin{align}
H=\sum_\mathbf{k}\left[\pi_\mathbf{k}(t)\pi_\mathbf{k}^*(t) -\omega_\mathbf{k}^2\phi_\mathbf{k}(t)\phi_\mathbf{k}^*(t)\right],
\end{align}
where $\omega_\mathbf{k}=\sqrt{\mathbf{k}^2+\mu^2}$.
Hamilton's equations imply that each field mode $\phi_\mathbf{k}(t)$ satisfies the complex harmonic oscillator equation $\ddot{\phi}_\mathbf{k}=-\omega_\mathbf{k}\phi_\mathbf{k}$. This has the general solution
\begin{align}\label{complex_HO_solution}
\phi_\mathbf{k}(t) = C_\mathbf{k}e^{i\omega_\mathbf{k}t}+D_\mathbf{k}e^{-i\omega_\mathbf{k}t},
\end{align}
where $C_\mathbf{k},D_\mathbf{k}\in\mathbb{C}$. 
Note the dimensionality of the solution. Solution $\phi_\mathbf{k}(t)$ depends upon four real parameters corresponding to complex constants $C_\mathbf{k}$ and $D_\mathbf{k}$. This parallels the four real parameters required to specify complex canonical coordinates $\phi_\mathbf{k}$ and $\pi_\mathbf{k}$.
The complete set of solutions to $\phi(\mathbf{x},t)$ follows by substituting \eqref{complex_HO_solution} into \eqref{base_fourier_transform},
\begin{align}\label{expansion_by_mode}
\phi(\mathbf{q},t)=\sum_\mathbf{k}\frac{e^{i\mathbf{k.x}}}{\sqrt{V}}\left(C_\mathbf{k}e^{i\omega_\mathbf{k}t}+D_\mathbf{k}e^{-i\omega_\mathbf{k}t}\right).
\end{align}
By making the substitutions $C_\mathbf{k}=b^*_\mathbf{-k}/\sqrt{2\omega_\mathbf{k}}$ and $D_\mathbf{k}=a_\mathbf{k}/\sqrt{2\omega_\mathbf{k}}$, expansion \eqref{expansion_by_mode} may be brought into the usual form 
\begin{align}\label{canonical_field_expansion}
\phi(\mathbf{q},t) = \sum_\mathbf{k}\frac{1}{\sqrt{2\omega_\mathbf{k}V}}\left(a_\mathbf{k}e^{-i k_\nu q^\nu}+b^*_\mathbf{k}e^{i k_\nu q^\nu}\right)
\end{align}
as for instance may be found in reference \cite{Mandl_Shaw}. 
The state-space may be spanned either by canonical coordinates $\phi_\mathbf{k}$ and $\pi_\mathbf{k}$, or by $a_\mathbf{k}$ and $b_\mathbf{k}$, the parameters used to span the complete set of solutions. 
Either
\begin{align}
x=(\phi_{\mathbf{k}_1},\pi_{\mathbf{k}_1},\phi_{\mathbf{k}_2},\pi_{\mathbf{k}_2},...)\quad\text{or     }\quad
x=(a_{\mathbf{k}_1},b_{\mathbf{k}_1},a_{\mathbf{k}_2},b_{\mathbf{k}_2},...).
\end{align}
This seems to suggest two alternative ways of calculating the entropy. Either
\begin{align}
S&=-\int\rho(\dots)\log\rho(\dots)\prod_\mathbf{k}\dee(\Real\phi_\mathbf{k})\dee(\Imag\phi_\mathbf{k})\dee(\Real\pi_\mathbf{k})\dee(\Imag\pi_\mathbf{k}),
\end{align}
or
\begin{align}\label{entropy_from_interaction_basis}
S&=-\int\rho(\dots)\log\rho(\dots)\prod_\mathbf{k}\dee(\Real a_\mathbf{k})\dee(\Imag a_\mathbf{k})\dee(\Real b_\mathbf{-k})\dee(\Imag b_\mathbf{-k}).
\end{align}
In order for equation to \eqref{entropy_from_interaction_basis} to be valid, however, the coordinate transform away from the canonical coordinates must be suitable (in the sense of the above section). It's Jacobian determinant must equal one.
As for each wave vector $\mathbf{k}$ the canonical coordinates $(\phi_\mathbf{k},\pi_\mathbf{k})$ transform only to $(a_\mathbf{k},b^*_{-\mathbf{k}})$, the Jacobian matrix for the transformation may be written as a direct sum, in block diagonal form
\begin{align}
\mathbf{J}=\mathbf{J}_{\mathbf{k}_1}\oplus\mathbf{J}_{\mathbf{k}_2}\oplus ... =
\begin{pmatrix}
\mathbf{J}_{\mathbf{k}_1}&0&\dots \\
0 & \mathbf{J}_{\mathbf{k}_2}\\
\vdots&&\ddots
\end{pmatrix},
\end{align}
where 
\begin{align}
\mathbf{J}_\mathbf{k}=\frac{\partial \left(\Real\phi_{\mathbf{k}},\Imag\phi_{\mathbf{k}},\Real\pi_{\mathbf{k}},\Imag\pi_{\mathbf{k}}\right)}{\partial \left(\Real a_{\mathbf{k}},\Imag a_{\mathbf{k}},\Real b_{-\mathbf{k}},\Imag b_{-\mathbf{k}}\right)}.
\end{align}
As a result its determinant is the product of the individual determinants of each $\mathbf{J}_\mathbf{k}$,
\begin{align}
|\det\mathbf{J}|=\prod_\mathbf{k}|\det \mathbf{J}_\mathbf{k}|.
\end{align}
With only a little tedium, it is possible to show that $|\det \mathbf{J}_\mathbf{k}|=1$ for all $\mathbf{k}$, and so $|\det\mathbf{J}|=1$, meaning the transform does qualify as suitable.
Hence the state-space of an individual field mode $\phi_\mathbf{k}$ may be described in terms of coordinates $a_\mathbf{k}$ and $b^*_\mathbf{-k}$. And the entire field state may be described by a complete set of these (for all $\mathbf{k}$).
In these coordinates the trivial free-field behavior is hidden, so that for any individual system $d a_\mathbf{k}(t)/d t=d b_\mathbf{k}(t)/d t=0$ for all $\mathbf{k}$.
Hence, the distribution $\rho(a_{\mathbf{k}_1},b_{\mathbf{k}_2},...)$ is frozen, $d \rho/dt=0$.
And so clearly the entropy is conserved, $d S/dt=0$.

Classical Klein-Gordon scalar field theory is classically mechanical, and so it is not surprising that it conserves information.
Other partial differential equations in physics, particularly those that are meant to be effective descriptions of underlying phenomena, do not conserve information in the same way, however.
Take the heat equation,
\begin{align}\label{heat_equation}
\left(\partial_t-\nabla^2\right)\phi(q,t)=0,
\end{align}
for instance. For the purposes of argument, expand an initial distribution $\phi(q,0)$ in the tautological basis $\{\delta^{(n)}(q-q')\}$,
\begin{align}
\phi(q,0)=\int\mathrm{d}^3q'\phi(q',0)\delta^{(n)}(q-q').
\end{align}
The basis is referred to as tautological as the components of $\phi(q,0)$ are $\phi(q',0)$ in this basis.
According to heat equation \eqref{heat_equation}, an initial point distribution evolves as
\begin{align}
\delta^{(n)}(q-q')\quad\longrightarrow\quad K_\text{H}(q,q',t)=\frac{1}{(4\pi t)^{n/2}}e^{-|q-q'|^2/4t}.
\end{align}
It spreads out with time in accordance with the familiar notion of heat dispersal/diffusion. 
As the heat equation is linear, a general solution may be expanded in terms of such individual solutions, 
\begin{align}
\phi(q,t)=\int\mathrm{d}^3q'\phi(q',0)K_\text{H}(q,q',t).
\end{align}
Just like the above Klein-Gordon case, in this basis-of-solutions the components are frozen (in this case equal to $\phi(q',0)$).
The key difference to the entropy-conserving Klein-Gordon case is that, though the $K_\text{H}(q,q',0)$ initially form a good basis for the function-state-space, they do not remain as such. 
As functions $K_\text{H}(q,q',t)$ spread with time, they cannot capture quickly varying fields.
In contrast, the Fourier modes employed in the above treatment of the Klein-Gordon field \emph{do} remain a good basis at all times.

Of course, much of this rings true with common intuition on the heat equation.
The heat equation is not intended to describe fundamental behavior, but merely the effective flow created by an underlying many-body condensed matter physics.
For this reason it is not expected to conserve information (entropy).
And of course as $t\rightarrow\infty$ the heat equation is expected to approach an equilibrium configuration corresponding to initial and boundary conditions. 
A steady-state.
In the present conditions, in an unbounded space, states are expected to approach zero, $\lim_{t\rightarrow\infty}\phi=0$, regardless of the initial conditions. 
As all states approach a single state, a rise in entropy follows.
The difference is that this rise in entropy is not \emph{de facto} like that of a iRelax theory, but an exact, true rise in entropy, in violation of information conservation.

\section{Generalized continuous state-space iRelax theories} \label{2.4} 
The iRelax theories introduced in section \ref{2.3} were developed around use of differential entropy, equation \eqref{differential_entropy}, as the information measure, \ref{Ing4}. 
Differential entropy draws criticism on a number of grounds, however. 
The expression for differential entropy, equation \eqref{differential_entropy}, is dimensionally inconsistent, for instance, as it involves the logarithm of $\rho(x)$ which, as a density, possesses inverse units to that of $x$.
Furthermore, in contrast to the discrete entropy which is independent of how states are labeled and hence invariant under relabelings $i\rightarrow j(i)$, differential entropy would generally be expected to change under a coordinate transform $x\rightarrow y(x)$.
Even changing the units of $x$ would affect differential entropy.
(The class of coordinate transforms that do keep differential entropy \eqref{differential_entropy} form-invariant are discussed above in the section titled \hyperref[sec:suitable]{2.3.6}.)
Nevertheless, on the face of it, differential entropy does appear sympathetic to an important guiding principle, the principle-of-indifference. 
Discrete entropy \eqref{discrete_entropy} treats all discrete states $i$ equally, or perhaps `indifferently'. Is it not unreasonable to expect this indifference of states to carry over to continuous entropy?
For the following reason the principle-of-indifference is not so easily applied to continuous systems, however.
Equation \eqref{differential_entropy} does appear to treat all states $x$ equally, and as desired this results in a distribution of maximum entropy that is uniform.
But this uniformity may only be superficial, as a distribution that is uniform with respect to coordinates $x$ is not necessarily uniform with respect to another set of other coordinates $y(x)$\footnote{Suppose an ensemble of spheres varied such that their volumes were uniformly distributed. Their radii would not be uniformly distributed. The ambiguity in the concept of uniformity is the reason for the complications that arise in extending entropy to continuous spaces. This fact is famously exploited in Bertrand's paradox \cite{Bertrand}.}.
The implicit assumption made in using differential entropy \eqref{differential_entropy} is that coordinates $x$ are special. Or, using the terminology from above, suitable.
Of course this may not be the case.
Other critiques of expression \eqref{differential_entropy} point to the fact that it does not appear to be `derivable' from discrete entropy \eqref{discrete_entropy}.
Rather, it simply seems to have been written down.
Arguably the most damning criticism of differential entropy, however, is that given by Edwin Jaynes \cite{Jaynesbook}. Jaynes points out that Shannon's theorem establishing discrete entropy as the correct information measure (theorem 2 in \cite{Shannon}) does not follow through for differential entropy.  
For all of these reasons, the present section is devoted to exploring the consequences of correcting these seeming inadequacies in differential entropy. 

To proceed, a leaf will be taken from Edwin Jaynes. 
For this reason, the corrected version of differential entropy shall be referred to as Jaynes entropy.
The Jaynes entropy shall be used to help generalize the notion of entropy and thus the so-far developed notion of iRelax theories. 
The following theorem and proof is a repackaging of the passage beginning on page 375 of reference \cite{Jaynesbook}.
And as pointed out by Jaynes, the proof ``can be made as rigorous as we please, but at considerable sacrifice of clarity''.
\begin{theorem}
The correct generalization of discrete entropy \eqref{discrete_entropy} to a system with the continuous state-space $\Omega=\mathbb{R}^n$, is the Jaynes entropy
\begin{align}\label{Jaynes_entropy}
S=-\int_\Omega\mathrm{d}^nx\rho(x)\log\frac{\rho(x)}{m(x)},
\end{align}
where $m(x)$ is the normalized density-of-states at state $x$.
\end{theorem}
\begin{proof}
Without loss of generality, consider a discrete system in which each state $i$ may be associated with a one-dimensional continuous position $x_i$.
Suppose there were $n$ such states in all $\Omega$.
The proportion of states in a region $\omega\subset\Omega$ is given by 
\begin{align}
\frac{1}{n}\sum_{x_i\in\omega}1.
\end{align}
As discussed above, the number of states per unit state-space should not be considered uniform unless coordinates $x$ are `special' in the sense discussed above. 
A uniform density-of-states should not be assumed in general. 
To take account of this, in the limit of $n\rightarrow \infty$, the proportion of states in the region is given by
\begin{align}\label{discrete_to_continuous_measure}
\frac{1}{n}\sum_{x_i\in\omega}1 \rightarrow \int_\omega\mathrm{d}x\, m(x),
\end{align}
where $m(x)$ is the density-of-states with respect to coordinate $x$.
Jaynes called the function $m(x)$ the invariant measure.  
It is strictly positive, of course, and normalized in the sense that $\int_\Omega m(x) \mathrm{d}x=1$.
For the present purposes it is also taken to be smooth and suitably well-behaved. 
In this context $m(x)$ has a clear interpretation, and so rather than adopting Jaynes' terminology, the $m(x)$ distribution shall be referred to accordingly, as the \emph{density-of-states}. The iRelax theories discussed in section \ref{2.3} are recovered in the case of a uniform $m(x)$, modulo issues of normalizing a uniform $m(x)$ over $\Omega = \mathbb{R}^n$.
At any point $x_i$, the discrete analog of the probability density $\rho(x)$ is given by
\begin{align}\label{limit_1}
\rho_i= \frac{p_i}{x_{i+1}-x_i}.
\end{align}
This tends to $\rho(x)|_{x=x_i}$ as $n\rightarrow\infty$.
Similarly, the density-of-states at the same point $x_i$ is
\begin{align}\label{limit_2}
m_i= \frac{1/n}{x_{i+1}-x_i},
\end{align}
which tends to $m(x)|_{x=x_i}$ as $n\rightarrow\infty$.
Equations \eqref{limit_1} and \eqref{limit_2} suggest that prior to taking the limit, the $p_i$ in discrete entropy \eqref{discrete_entropy} should be substituted as $p_i=\rho_i/(nm_i)$.
Hence discrete entropy \eqref{discrete_entropy} may be expressed
\begin{align}
S=-\sum_{i}^{n}p_i\log p_i=-\sum_{i}^{n}\frac{\rho_i}{nm_i}\log\left(\frac{\rho_i}{nm_i}\right).
\end{align}
Upon taking the limit $n\rightarrow\infty$ (noting equation \eqref{discrete_to_continuous_measure} above), this would give
\begin{align}
S&=-\int \mathrm{d}x\rho(x)\log\left(\frac{\rho(x)}{m(x)}\right)+\lim_{n\rightarrow\infty}\log n.
\end{align}
However the $\log n$ term clearly diverges in the limit. To keep the final result finite, it should be subtracted before the limit is taken.
For this reason, the Jaynes entropy is defined 
\begin{align}\label{remove_divergent_term}
S:= \lim_{n\rightarrow\infty}\left(S_\text{discrete} - \log n\right).
\end{align}
Hence, the final expression for Jaynes entropy is given by equation \eqref{Jaynes_entropy}.
\end{proof}

\subsection{State propagators for information conservation (the modified Liouville theorem)}\label{sec:modified_state_propagators}
Now that the measure of information (\ref{Ing4}) has changed from the differential entropy \eqref{differential_entropy} to the Jaynes entropy \eqref{Jaynes_entropy}, there are a number of corresponding changes that must be made to the other ingredients outlined in section \ref{2.1}.
To find the consequences of imposing conservation of the Jaynes entropy, the same approach as section \ref{2.3} will be followed. 
State propagators (Green's functions), $T(x,x',t)$, defined by the operation
\begin{align}\label{propagator_defining_operation}
\rho(x,t)=\int_\Omega T(x,x',t)\rho(x',0)\mathrm{d}^nx',
\end{align}
will be used to demonstrate the circuit of implications shown in figure \ref{logic_map}. (Recall that $T(x,x',t)$ is the probability density that a system initially in state $x'$ moves to state $x$ in time $t$.)

\begin{theorem}
Conservation of Jaynes information \eqref{Jaynes_entropy} $\implies$ No many-to-one on a $\Omega=\mathbb{R}^n$ state-space\label{modified_no_many_to_one}.
\end{theorem}
\begin{proof}
For Jaynes entropy \eqref{Jaynes_entropy} to be conserved under evolution \eqref{propagator_defining_operation}, propagator $T(x,x',t)$ must satisfy
\begin{align}\label{modified_theorem_misc_1}
\int_\Omega\mathrm{d}^nx\,\rho(x,0)\log\frac{\rho(x,0)}{m(x)}
=\int_\Omega\mathrm{d}^nx\,\left[\int_\Omega\mathrm{d}^nx'\,T(x,x',t)\rho(x',0)\right]\nonumber\\ \times\log\frac{\int_\Omega\mathrm{d}^nx'\,T(x,x',t)\rho(x',0)}{m(x)}
\end{align}
for all initial distributions $\rho(x',0)$ and all time intervals $t$. Take $\rho(x',0)$ equal to the density-of-states equivalent of a uniform distributed region, $\rho(x,0) = cm(x)$ on $\omega\subset\Omega$, and equal to zero otherwise, where $c$ is normalized such that $\int_\omega cm(x)\mathrm{d}^nx=1$. Then equation \eqref{modified_theorem_misc_1} becomes
\begin{align}\label{modified_theorem_misc_2}
\resizebox{\textwidth}{!}{$
0=\int_\Omega\mathrm{d}^nx\,\left[\int_\omega\mathrm{d}^nx'\,cm(x')T(x,x',t)\right]\log\left[\frac{1}{m(x)}\int_\omega\mathrm{d}^nx'\,m(x')T(x,x',t)\right].$}
\end{align}
Now take region $\omega$ to be the union of two small non-overlapping balls, $\omega=B_{\epsilon_1}(x'_1)\cup B_{\epsilon_2}(x'_2)$, centered on points $x'_1$ and $x'_2$, of radii $\epsilon_1$ and $\epsilon_2$. Take $\epsilon_1$ and $\epsilon_2$ to be small enough so that for given $x$ and $t$, $T(x,x',t)$ and $m(x')$ are approximately constant on each of the balls. Then,
\begin{align}\label{modified_theorem_misc_3}
\resizebox{\textwidth}{!}{$
\int_\omega\mathrm{d}^nx'\,cm(x')T(x,x',t)=\pi\epsilon_1^ncm(x'_1)T(x,x'_1,t)+\pi\epsilon_2^ncm(x'_2)T(x,x'_2,t).$}
\end{align}
By substituting \eqref{modified_theorem_misc_3} into \eqref{modified_theorem_misc_2} and then taking derivatives with respect to $\epsilon_1$ and $\epsilon_2$, the information conservation condition may be found to reduce to
\begin{align}
0=\int_\Omega\mathrm{d}^nx\,m(x)\frac{\epsilon_1^{n-1}\epsilon_2^{n-1}m(x'_1)m(x'_2)T(x,x'_1,t)T(x,x'_2,t)}{\epsilon_1^nm(x'_1)T(x,x'_1,t)+\epsilon_2^nm(x'_2)T(x,x'_2,t)}.
\end{align}
For this final equality to be satisfied, the product $T(x,x'_1,t)T(x,x'_2,t)$ must vanish for all $x$, $t$ and $x'_1\neq x'_2$. Hence no two distinct states may evolve to the same state.
\end{proof}

\begin{conjecture}
Conservation of Jaynes information \eqref{Jaynes_entropy} $\implies$ No one-to-many on a $\Omega=\mathbb{R}^n$ state-space\label{Jaynes_conjecture}.
\end{conjecture}
This conjecture is identical to the no one-to-many conjecture of section \ref{2.3}, and as such the reader is referred to that section for arguments in its support. If the conjecture is accepted (as was necessary to do to proceed in section \ref{2.3}), then together the no many-to-one and no one-to-many requirements mean that a law of evolution which conserves Jaynes entropy \eqref{Jaynes_entropy} is one-to-one. In other words, conservation of the Jaynes information also implies the existence of trajectories. Hence, if a system trajectory that evolves from state $x_0$ to state $x_f$ in time $t$ is denoted $x_f(x_0,t)$, then a `perfect knowledge' information distribution may be written $\rho(x,t)=\delta^{(n)}(x-x_f(x_0,t))$.
By substitution of this into defining propagator relation \eqref{propagator_defining_operation}, it follows that the state propagators must take the form 
\begin{align}\label{one-to-one_propagators}
T(x,x_0,t)=\delta^{(n)}(x-x_f(x_0,t)).
\end{align}
Thus, the propagators appear unaffected by the new form of entropy. This is not the case for Liouville's theorem, however, as may be demonstrated as follows.

\begin{theorem}
Conservation of Jaynes entropy \eqref{Jaynes_entropy} $+$ one-to-one $\Rightarrow$ modified Liouville's Theorem.
\end{theorem}
\begin{proof}
Recall that in order for Jaynes entropy \eqref{Jaynes_entropy} to be conserved, propagator $T(x,x',t)$ must satisfy equation \eqref{modified_theorem_misc_2} for any subset of the state-space, $\omega\subset\Omega$.
As the state propagators take the form \eqref{one-to-one_propagators}, in order to evaluate the RHS of \eqref{modified_theorem_misc_2}, integrals of the form $\int_\omega \mathrm{d}^nx'\delta^{(n)}(x-x_f(x',t))m(x')$ must be performed.
As the law translates initial states $x'$ to final states $x_f$ in a one-to-one fashion, the integration may be performed with respect to final positions, $\int_\omega \mathrm{d}^nx'\rightarrow \int_{x_f(\omega,t)}\mathrm{d}^nx_f|\text{det}\mathbf{J}(x_f)|$.
So,
\begin{align}
&\int_\omega \mathrm{d}^nx'\,\delta^{(n)}(x-x_f(x',t))m(x')\\
=&\int_{x_f(\omega,t)} \mathrm{d}^nx_f \,|\text{det}\mathbf{J}(x_f)|\delta^{(n)}(x-x_f)m(x'(x_f,-t))\\
=&
\begin{cases}
|\text{det}\mathbf{J}(x)|m(x'(x,-t))&\text{for}\, x\in x_f(\omega,t)\\
0 & \text{otherwise}
\end{cases}.
\end{align}
This may be then substituted into equation \eqref{modified_theorem_misc_2} to yield
\begin{align}
0=&\resizebox{0.95\textwidth}{!}{$
\int_\Omega\mathrm{d}^nx\,\left[\int_\omega\mathrm{d}^nx'\,cm(x')T(x,x',t)\right]\log\left[\frac{1}{m(x)}\int_\omega\mathrm{d}^nx'\,m(x')T(x,x',t)\right]$}\\
=&\resizebox{0.95\textwidth}{!}{$
\int_{x_f(\omega,t)}\mathrm{d}^nx\,|\text{det}\mathbf{J}(x)|m(x'(x,-t))\log\left[\frac{1}{m(x)}|\text{det}\mathbf{J}(x)|m(x'(x,-t))\right]$}\\
=&\int_{\omega}\mathrm{d}^nx'\,m(x')\log\left[|\text{det}\mathbf{J}(x)|\frac{m(x')}{m(x(x',t))}\right].
\end{align}
For this expression to vanish for every possible state-space region $\omega\subset\Omega$, it must be the case that 
\begin{align}\label{ML1}
|\text{det}\mathbf{J}(x)|=\frac{m(x)}{m(x')},
\end{align}
where it is understood that state $x'$ evolves to state $x$, ie.\ $x=x_f(x',t)$.
So Liouville's theorem in it's original form, equation \eqref{L1}, no longer applies; state-space volumes are no longer conserved, at least if the volume is measured with respect to state-space coordinates $x$. If state-space volumes are measured in number (or proportion) of states, however, the volumes are preserved.
To see this, note that the proportion of all states in region $\omega$ is given by $\#(\omega)=\int_\omega m(x)\mathrm{d}^nx$. 
And from equation \eqref{ML1}, 
\begin{align}
\#(\omega)
=\int_\omega m(x)\mathrm{d}^nx
=\int_{x_f(\omega)} m(x)\frac{m(x_f)}{m(x)}\mathrm{d}^nx_f
=\#(x_f(\omega)).
\end{align}
Hence the number of states is preserved in comoving volumes. This observation, along with equation \eqref{ML1} shall be referred to as the first statement of the \emph{modified} Liouville theorem.
\end{proof}
To find the second statement of the modified Liouville theorem, first substitute the one-to-one propagator \eqref{one-to-one_propagators} into the propagator defining relation \eqref{propagator_defining_operation} to find $\rho(x_f,t)=|\text{det}\mathbf{J}|\rho(x_0,0)$, where it is understood that state $x_0$ evolves to state $x_f$. Then use the first statement of the modified Liouville theorem \eqref{ML1} to show that
\begin{align}\label{ML2_misc}
\frac{\rho(x_f,t)}{m(x_f)}=\frac{\rho(x_0,0)}{m(x_0)}.
\end{align}
Hence the ratio of $\rho$ to $m$ stays constant along trajectories. Or in analogy with \eqref{L2},
\begin{align}\label{ML2}
\frac{d}{dt}\left(\frac{\rho}{m}\right)=0,
\end{align}
where $\mathrm{d}/\mathrm{d}t$ is the material/convective derivative, $\mathrm{d}/\mathrm{d}t=\partial/\partial t+\sum_i (\partial x_i/\partial t) (\partial/\partial x_i)$, the derivative along the trajectory.
Equation \eqref{ML2} shall be referred to as the second statement of the modified Liouville theorem.

The final step in the circuit of implications of figure \ref{logic_map} is to show that the modified Liouville theorem $\implies$ conservation of the Jaynes entropy \eqref{Jaynes_entropy}.
To do this, substitute \eqref{ML2_misc} and \eqref{ML1} into \eqref{Jaynes_entropy} to find
\begin{align}
S(0)&=-\int_\Omega \mathrm{d}^nx_0\rho(x_0,0)\log\frac{\rho(x_0,0)}{m(x_0)}\\
&=-\int_{x_f(\Omega)}\left(\frac{m(x_f)}{m(x_0)}\mathrm{d}^nx_f\right)\left(\frac{m(x_0)}{m(x_f)}\rho(x_f,t)\right)\log\frac{\rho(x_f,t)}{m(x_f)}\\
&=S(t).
\end{align}
Note that, just as was the case in section \ref{2.3}, it does not matter that the integration range is $x_f(\Omega)$ and not $\Omega$. By definition $\rho(x,t)=0$ in the compliment $\left\{x\in\Omega|x\notin x_f(\Omega)\right\}$ and so these regions do not contribute to the entropy.

\subsection{Consequences of using the Jaynes entropy}
At present in this section ingredients \ref{Ing1}, \ref{Ing4}, and \ref{Ing5} (the state-space, the measure of information and the statement(s) of information conservation) are accounted for.
The distribution of maximum entropy \ref{Ing3} follows naturally from the Jaynes entropy \ref{Ing4}. 
To see why, recall that the Jaynes entropy introduces the idea of a normalized density-of-states upon the state-space, $m(x)$. 
And that the principle-of-indifference states that the distribution of maximum entropy corresponds to all states being equally likely.
Clearly the distribution of maximum entropy must equal the density-of-states,
\begin{align}\label{max_ent}
\rho_\text{eq}(x)= m(x).
\end{align}
As this distribution is uniquely high in entropy, the conservation of entropy requires that it be conserved by the dynamics.
By \ref{Ing6}, Valentini's relaxation theorem below, entropy naturally tends to rise, and so all other distributions $\rho(x)\neq m(x)$ will tend towards $\rho_\text{eq}(x)$. 
It is therefore natural to refer to it as the \emph{equilibrium distribution}. 
Likewise, all distributions for which $\rho(x)\neq m(x)$ are referred to as \emph{nonequilibrium distributions}.
And the process by which the equilibrium distribution arises from nonequilibrium distributions is called \emph{relaxation}.

This also solves an issue which is present in the differential entropy formulation of section \ref{2.3}.
It was noted previously that the tendency of both types of entropy to diverge leads to some pathological behavior for low ignorance distributions $\rho(x)$.
Neither type of entropy can tell the difference, for instance, between a perfectly determined state $\rho(x)=\delta^{(n)}(x-a)$ and one which is only half determined, $\rho(x)=1/2[\delta^{(n)}(x-a)+\delta^{(n)}(x-b)]$. In both cases $S\rightarrow-\infty$%
\footnote{This shouldn't be too disheartening. After all this issue also appears in standard classical mechanics.}.
It was \emph{not} noted, however, that differential entropy \eqref{differential_entropy} also exhibits pathological behavior at the other end of the spectrum, for high ignorance distributions.
Take a one-dimensional state-space $\Omega =\mathbb{R}$ for example. 
Differential entropy cannot tell the difference between a $\rho(x)$ that is uniform on $\mathbb{R}$, and a $\rho(x)$ that is uniform on $\mathbb{R}^+$ (and zero on $\mathbb{R}^-$), despite these distributions intuitively possessing different informational contents. 
In both cases differential entropy is divergent, $S\rightarrow+\infty$.
This issue does not emerge when Jaynes entropy is used. 
There is only one distribution of maximum entropy, $\rho(x)=m(x)$.
And this is unambiguously referred to by the Jaynes entropy as $S=0$ (as may be readily checked by substituting \eqref{max_ent} into \eqref{Jaynes_entropy}).

For uniformly distributed discrete probabilities $p_i=1/N$ $(N\leq n)$, discrete entropy \eqref{discrete_entropy} produces $S=\log N$ in analogy with Boltzmann's famous expression.
For the equivalent continuous case, $\rho(x)=1/V(\omega)$, differential entropy \eqref{differential_entropy} produces $S=\log V(\omega)$.
In the present case the equivalent of a uniform distribution on $\omega\subset\Omega$ is $\rho(x)=m(x)/\#(\omega)$, where $\#(\omega)$ is the proportion of total states in $\omega$. 
Substitution of this into Jaynes entropy \eqref{Jaynes_entropy} produces 
\begin{align}
S=\log\left(\#(\omega)\right).
\end{align}

When formulating iRelax framework using differential entropy, the second statement of Liouville's theorem \eqref{L2}, and thereby information conservation, was shown to imply an incompressible flow of trajectories, $\nabla.\dot{x}=0$. In this case, the same procedure yields
\begin{align}
-\nabla.\dot{x}=\frac{1}{m}\frac{\mathrm{d}m}{\mathrm{d}t}.
\end{align}
So the density of states experienced by a system along its trajectory changes in proportion to the \emph{convergence} of it's trajectory with those that surround it.
As there is an implicit $\dot{x}$ in the $\mathrm{d}/\mathrm{d}t$ derivative however, this is not the most useful manner in which to express the condition that information conservation places upon the law of evolution $\dot{x}$.
Instead note that the distribution of maximum entropy $\rho_\text{eq}(x,t)=m(x)$ is conserved by the dynamics. 
It also satisfies a continuity equation $\partial m/\partial t+\nabla\cdot(m\dot{x})=0$, and so the condition placed upon the law of evolution may be written
\begin{align}\label{2.4_condition}
\nabla\cdot\left(m\dot{x}\right)=0.
\end{align}
Hence, it is now the \emph{equilibrium probability current} $j=m\dot{x}$ rather than the law of evolution that is incompressible.
Nevertheless, this still represents one real condition upon the law of evolution $\dot{x}$, and so constrains the law of evolution in a similar manner to the $\nabla\cdot\dot{x}=0$ condition found using differential entropy.
In a practical situation it may be the case that the law of evolution \ref{Ing2} is not known, but the equilibrium state \ref{Ing3} is. (Think thermodynamics prior to kinetic theory.)
As the equilibrium distribution is $\rho_\text{eq}(x)=m(x)$, equation \eqref{2.4_condition} provides a means to arrive at a consistent law of evolution $\dot{x}$.
Of course, since \eqref{2.4_condition} is but one real condition, it does not fully determine the $n$ component law of evolution $\dot{x}$.
The set of consistent laws are related by the addition of an incompressible current $j_\text{inc}$ to the $m\dot{x}$ in \eqref{2.4_condition} so that
\begin{align}
\dot{x}'=\dot{x}+\frac{j_\text{inc}}{m(x)}.
\end{align}
In two-dimensions $x=(x_1,x_2)\in\Omega=\mathbb{R}^2$, for instance, the two component law may be decomposed as $\dot{x}=\nabla A + J\nabla B$, where $J=\left(\begin{smallmatrix}0&-1\\1&0\end{smallmatrix}\right)$, and where $A$ and $B$ are scalar fields.
Then the condition \eqref{2.4_condition} reduces to $\nabla\cdot(m\nabla A)=0.$
Alternatively, the current could be decomposed in the same manner, $j=\nabla A + J\nabla B$.
Then the condition \eqref{2.4_condition} means that $\nabla^2A=0$ in all $\Omega$, and hence $A(x)$ is constant everywhere in $\Omega$.
It follows that the law of evolution is
\begin{align}\label{eom_for_x}
\dot{x}_1=-\frac{1}{m(x_1,x_2)}\frac{\partial B}{\partial x_2},\quad\quad\dot{x}_2=\frac{1}{m(x_1,x_2)}\frac{\partial B}{\partial x_1}.
\end{align}
And Hamiltonian mechanics is recovered in the case that $B=-H$ and the density of states $m(x_1,x_2)$ is uniform.

\subsection{Proof of rise in Jaynes entropy (Valentini's relaxation theorem)}
The same arguments presented in support of the \emph{de facto} rise of differential entropy (\ref{Ing6}) also apply to Jaynes entropy. 
The incompressible flow is still present (and so the chocolate powder will still disperse into the milk as expected). The only difference is that the flow should be understood to be incompressible with respect to the actual state-space, rather than simply coordinates $x$.
More information on this point in the next section.
Nevertheless, the following proof has been important in the development of the discipline, and it is useful to show that the coarse-graining (discretization, pixelation, rasterization) arguments carry over. 
The proof is due to Valentini \cite{AV91a}, and was originally presented in the context of de Broglie-Bohm theory.

\begin{proof}
The proof follows most concisely with the introduction of the ratio $f:=\rho/m$ and its coarse-grained counterpart $\tilde f=\bar \rho/\bar m$. In analogy with section \ref{2.3}, define the coarse-grained (discretized) entropy,
\begin{align}
\bar{S}=-\int \mathrm{d}x\, \bar\rho\log\tilde f.
\end{align}
The exact entropy is of course conserved by the information preserving dynamics, so that $S_0$ is equal to $S$, the exact entropy at some later time.
The assumption of the postulate of \emph{a priori} probabilities establishes the equality of the initial coarse-grained entropy $\bar S_0$ and the initial exact entropy $S_0$.
The effect of the postulate is now, as before, to make the correspondence between the rise in entropy and the effect of coarse-graining upon the final densities,
\begin{align}
\Delta \bar S=\bar S -\bar S_0=\bar S-S=-\int \mathrm{d}x\,\left\{\bar\rho\log\tilde f-\rho\log f\right\}.
\end{align}
Then, since $\log\tilde f$ is constant over any cell, the first term in the integrand may be replaced with $\rho\log\tilde f$, so that
\begin{align}
\Delta \bar S=\int \mathrm{d}x\,m\left\{f\log(f/\tilde f)\right\}.
\end{align}
As $\tilde f$ is constant in each cell, $\int \mathrm{d}x\,m(\tilde f-f)=0$. This may be added to the RHS, giving
\begin{align}
\Delta S&=\int \mathrm{d}x\,m\left\{f\log(f/\tilde f)+\tilde f-f\right\}\\
&\geq 0,
\end{align}
where the final inequality is found as a result of $f\log(f/\tilde f)+\tilde f-f$ being positive unless $f=\tilde f$. 
\end{proof}

\subsection{Geometrical interpretation}
To use the Jaynes entropy \eqref{Jaynes_entropy} is to admit that the principle of indifference may not apply with respect to coordinates $x$. 
This does not mean that suitable coordinates $x'$ cannot be found for which the principle of indifference does apply, however. 
(At least locally.)
The process of finding such suitable coordinates might proceed as follows.
Suppose the equilibrium distribution $\rho_\text{eq}(x)$ of a system were known in coordinates $x$.
As $\rho_\text{eq}(x)=m(x)$, density-of-states $m(x)$ is also known in coordinates $x$.
If $m(x)$ is suitably well-behaved, it should be possible to find coordinates $x'$ such that $m(x)$ becomes flattened out locally.
This distribution must be appropriately normalized, so write $m'(x')=N$, where $N$ is a normalization constant.  
In some local region $\omega$, both sets of coordinates should be capable of counting the states,
\begin{align}\label{state_counting}
\#(\omega) = \int_\omega m(x)\mathrm{d}^nx=\int_\omega N \mathrm{d}^nx'=\int_\omega N\left|\text{det}\mathbf{J}\right|\mathrm{d}^nx,
\end{align}
and so it follows that
\begin{align}\label{geo_misc1}
\left|\text{det}\mathbf{J}\right|=\frac{m(x)}{N}.
\end{align}
The total probability of finding the system within region $\omega$ should be similarly independent of the coordinates,
\begin{align}\label{geo_misc2}
\int_\omega\rho(x)\mathrm{d}^nx=\int_\omega\rho'(x')\mathrm{d}^nx'=\int_\omega \rho'(x')\frac{m(x)}{N}\mathrm{d}^nx,
\end{align}
and so it follows that
\begin{align}
\frac{\rho(x)}{m(x)}=\frac{\rho'(x')}{N}.
\end{align}
Substitution of relations \eqref{geo_misc1} and \eqref{geo_misc2} into Jaynes entropy \eqref{Jaynes_entropy} produces
\begin{align}
S=-\int d^nx' \rho'(x')\log \frac{\rho'(x')}{N}.
\end{align}
So with respect to coordinates $x'$, the Jaynes entropy treats all positions in the state-space equally, manifestly demonstrating the principle of indifference% 
\footnote{The constant $N$ in the denominator contributes the constant factor $+\int d^nx' \rho'(x')\log N=\log N$ to the entropy. This becomes divergent in an unbounded state-space and is precisely the term removed in equation \eqref{remove_divergent_term} when deriving the Jaynes entropy.}.
To find an actual coordinate transform $x\rightarrow x'$ for which this is the case, a solution must be found to equation \eqref{geo_misc1}.
To provide a concrete example, further suppose a two-dimensional state-space $x=(x_1,x_2)\in\Omega=\mathbb{R}^2$ in which the density of states is only dependent upon coordinate $x_1$, $m(x)=m(x_1)$.
Then equation \eqref{geo_misc1} may be written
\begin{align}\label{geo_misc3}
\left|\frac{\partial x_1'}{\partial x_1}\frac{\partial x_2'}{\partial x_2} - \frac{\partial x_2'}{\partial x_1}\frac{\partial x_1'}{\partial x_2}\right|=\frac{m(x_1)}{N}
\end{align}
Since the density of states is purely a function of $x_1$, it is sensible to look for solutions of the form $x_1'=x_1'(x_1)$, and $x_2'=x_2$. The Jacobian matrix of such a transform would be
\begin{align}
\mathbf{J}=\begin{pmatrix}m(x_1)/N&0\\0&1\end{pmatrix},
\end{align}
which trivially satisfies equation \eqref{geo_misc3}. If the indefinite integral $\int m(x_1)\mathrm{d}x_1=M(x_1)$ could be found, the coordinate transforms would be $x_1'=M(x_1)/N + C$, $x_2'=x_2$.
And the law of evolution in $x'$ may be found by acting the Jacobian matrix on the law of evolution in $x$, 
\begin{align}
\begin{pmatrix}\dot{x}'_1\\\dot{x}'_2\end{pmatrix}
=\mathbf{J}\begin{pmatrix}\dot{x}_1\\\dot{x}_2\end{pmatrix}
=\begin{pmatrix}\frac{m(x_1)}{N}\dot{x}_1\\\dot{x}_2\end{pmatrix}.
\end{align}
Substituting in the law of evolution for $x$, equation \eqref{eom_for_x}, then gives
\begin{align}
\dot{x}_1'=-\frac{1}{N}\frac{\partial B}{\partial x_2'} ,\qquad\dot{x}_2'=\frac{1}{N}\frac{\partial B}{\partial x_1'},
\end{align}
and so locally, in the $x'$ coordinates, the form of the law of evolution is the same as for differential entropy. 
In the language of section \ref{2.3}, coordinates $x'$ may be considered locally suitable.

Of course, much of this discussion is suggestive of a geometrical interpretation to the state-space. 
The incompressibility condition, equation \eqref{2.4_condition}, may be written for instance as
\begin{align}
0=\frac{1}{m(x)}\nabla\cdot\left(m(x)\dot{x}\right) = \frac{1}{\sqrt{g}}\left(\sqrt{g}\dot{x}^a\right)_{,a}=\dot{x}\indices{^a_{;a}},
\end{align}
where $g$ is the determinant of metric $g_{ab}$. This of course indicates that $\sqrt{g}=m(x)$, and suggests the metric
\begin{align}\label{metric}
\mathrm{d}s^2&=m(x_1,x_2,\dots,x_n)\left(\mathrm{d}x_1^2+\mathrm{d}x_2^2+\dots,\mathrm{d}x_n^2\right).
\end{align}
Superficially, this looks like a FLRW metric with scale factor $a(t)$ replaced with $m(x)$, but there are some key differences.
Of course time doesn't feature in metric \eqref{metric}, which is therefore Riemannian rather than pseudo-Riemannian.
And $a(t)$ has time dependence, whereas $m(x)$ has state coordinate $x$ dependence. 
Unlike the scale factor, which is dimensionless, $m(x)$ has units of states per state-space volume. Consequently, the quantity $\mathrm{d}s^2$ has units of states per $[x_i]^{n-2}$, where $[x_i]$ are the as yet unspecified units of $x_i$.
Metric \eqref{metric} suggests that $n$-volume integrals should be performed with respect to measure $\mathrm{d}(\text{proper volume})=\sqrt{g}\mathrm{d}^nx=m(x)\mathrm{d}^nx=N\mathrm{d}^nx'$, so that state counting equation \eqref{state_counting} and Jaynes entropy \eqref{Jaynes_entropy} may be expressed in the more symmetrical form
\begin{align}
\#(\omega)=\int_\omega\left(m\mathrm{d}^nx\right) \frac{\rho}{m},\qquad
S=-\int_\Omega \left(m\mathrm{d}^nx\right)\frac{\rho}{m}\log\frac{\rho}{m}.
\end{align}
In dimensions $n\geq 2$ a nonconstant $m(x)$ does generally indicate curvature. For instance the Ricci scalar in $n$-dimensions may be found to be equal to
\begin{align}
R=-(n-1)\left[\frac{1}{m^2}\sum_i\partial_i^2m +\frac{1}{4}(n-6)\sum_i\frac{1}{m^3}\left(\partial_im\right)^2\right],
\end{align}
where $\partial_i$ denotes a partial derivative with respect to coordinate $x_i$. Recall however that if equation \eqref{geo_misc1} has a local solution, then the coordinate transform $x\rightarrow x'$ is able to make the metric locally flat, $\mathrm{d}s^2=N\left(\mathrm{d}^2x'_1+\dots\mathrm{d}^2x'_n\right)$. 
As this $x\rightarrow x'$ transform affects the metric as $g_{ab}\rightarrow g'_{ab} =(N/m(x))g_{ab}$, it is a conformal transformation. 
In other words, metric \eqref{metric} is locally conformally flat \cite{encyclopedia_of_mathematics}. 
In dimensions $\geq 4$, this has a consequence for the traceless part of the Riemann tensor, the Weyl tensor $C\indices{^a_b_c_d}$.
The Weyl tensor is invariant under conformal transformations $g_{ab}\rightarrow g'_{ab}=f^2(x)g_{ab}$.
Hence the components of the Weyl tensor calculated with metric \eqref{metric} are the same as those calculated with a locally flat metric, and so all vanish identically\footnote{A similar argument holds for the Cotton tensor in three-dimensions. Two-dimensional Riemannian manifolds are automatically conformally flat.}.

Realist quantum theories based upon similar geometrical considerations have been proposed from time to time. For more information, see reference \cite{R15}, which provides a nice overview of these proposals.

\begin{figure}
\includegraphics[width=\textwidth/4]{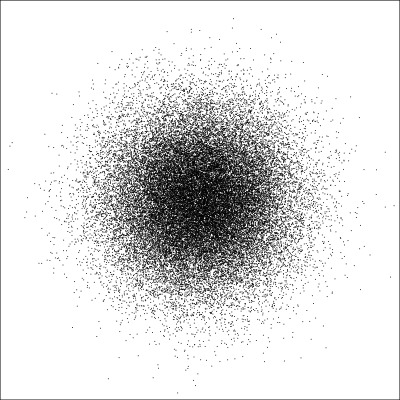}%
\includegraphics[width=\textwidth/4]{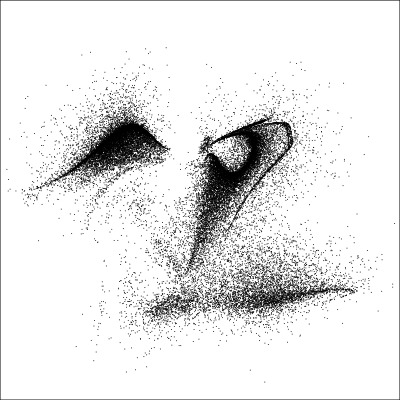}%
\includegraphics[width=\textwidth/4]{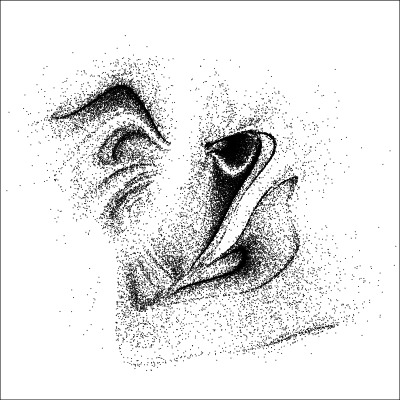}%
\includegraphics[width=\textwidth/4]{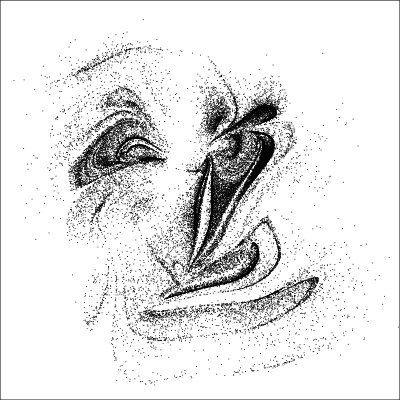}\vspace{-3.0pt}\\ \vspace{-3.0pt}%
\includegraphics[width=\textwidth/4]{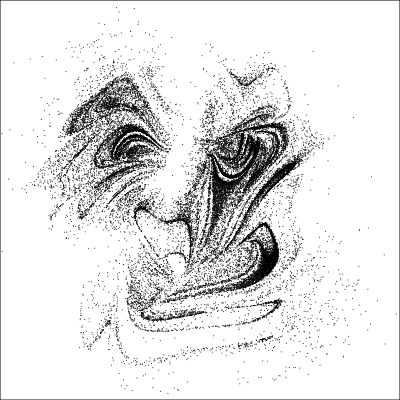}%
\includegraphics[width=\textwidth/4]{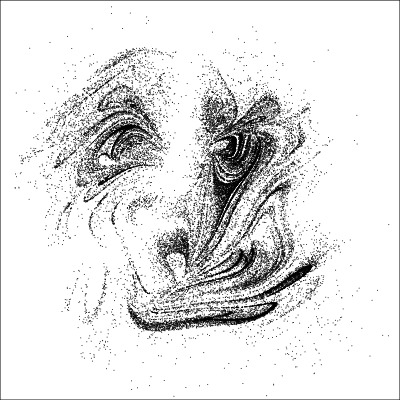}%
\includegraphics[width=\textwidth/4]{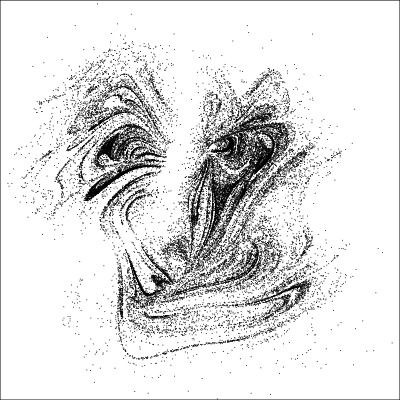}%
\includegraphics[width=\textwidth/4]{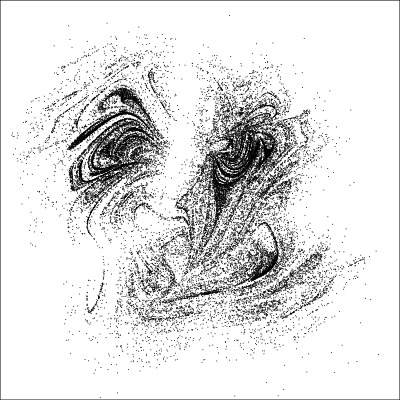}\\ \vspace{-3.0pt}%
\includegraphics[width=\textwidth/4]{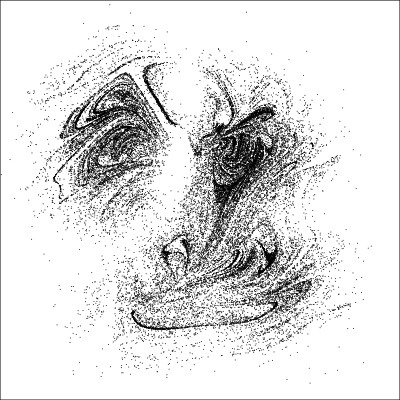}%
\includegraphics[width=\textwidth/4]{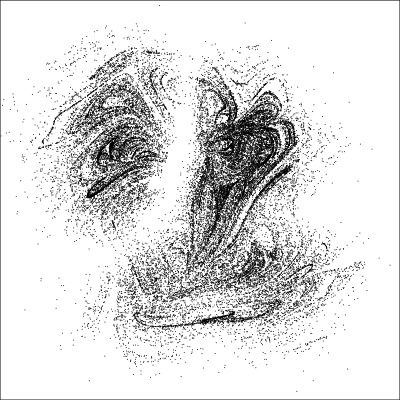}%
\includegraphics[width=\textwidth/4]{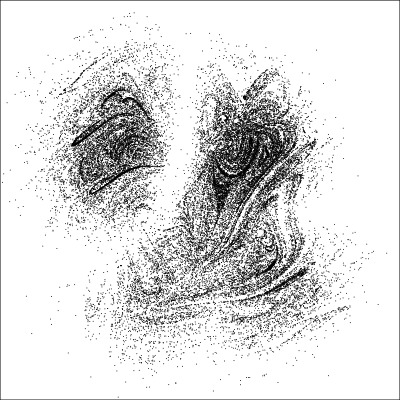}%
\includegraphics[width=\textwidth/4]{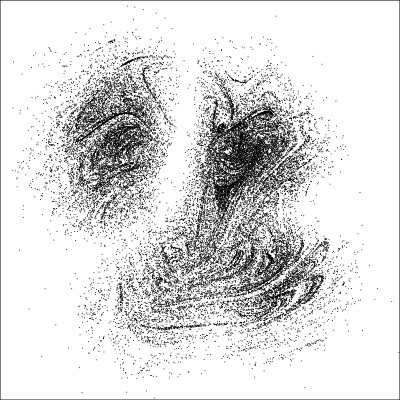}\\ \vspace{-3.0pt}%
\includegraphics[width=\textwidth/4]{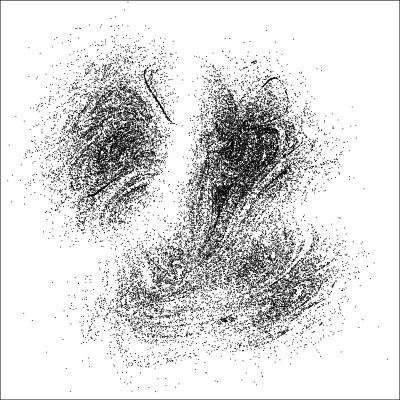}%
\includegraphics[width=\textwidth/4]{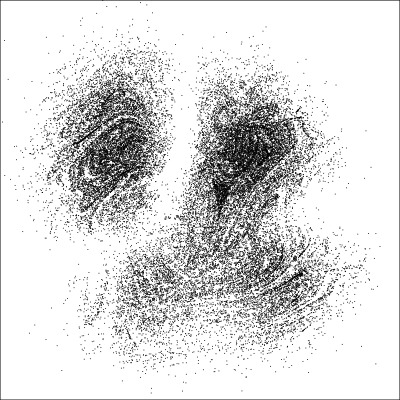}%
\includegraphics[width=\textwidth/4]{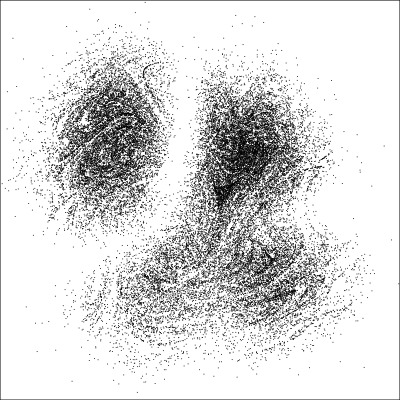}%
\includegraphics[width=\textwidth/4]{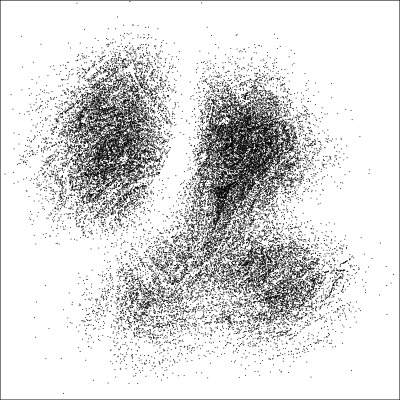}\\ \vspace{-3.0pt}%
\includegraphics[width=\textwidth/4]{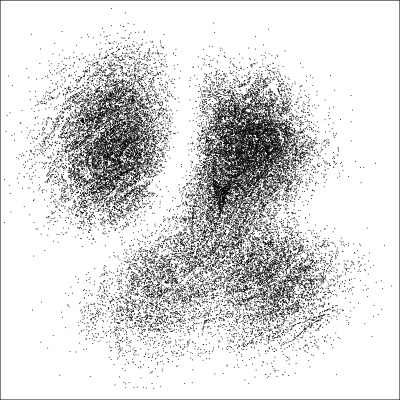}%
\includegraphics[width=\textwidth/4]{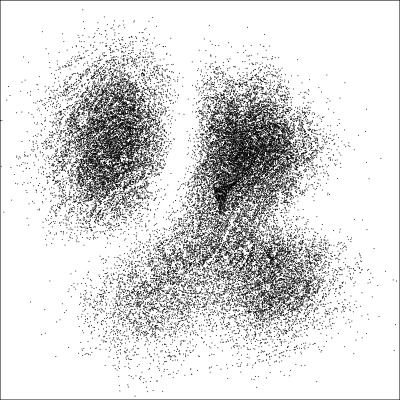}%
\includegraphics[width=\textwidth/4]{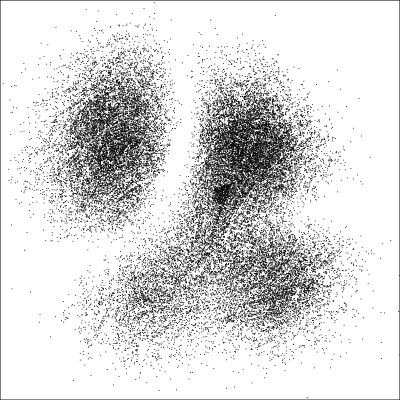}%
\includegraphics[width=\textwidth/4]{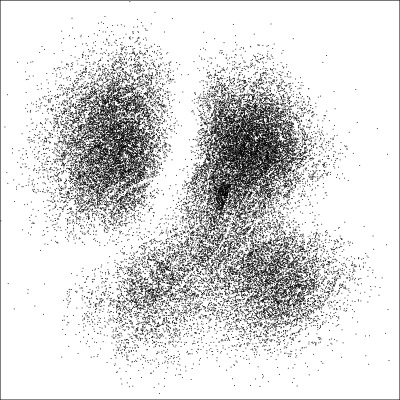}%
\caption[Example of quantum relaxation to quantum equilibrium]{Quantum relaxation for a two-dimensional harmonic oscillator with a superposition of nine low energy states. The first frame displays the initial $\rho(x,0)$, which is Gaussian. The quantum state is periodic. Subsequent frames display the progress of the quantum relaxation at the end of each period. 
Similarly to the classical relaxation in figure \ref{classical_relaxation}, quantum relaxation proceeds by the creation of structure, which becomes ever more fine as time progresses.
Eventually this fine structure becomes too fine to resolve, leading to a \emph{de facto} rise in entropy.
The equilibrium state $\rho_\text{eq}(x,t)=|\psi(x,t)|^2$ is obtained to a reasonable degree by the final frame, a mere 19 wave function periods into the evolution. }\label{quantum_relaxation}
\end{figure}
\begin{figure}
\includegraphics[width=\textwidth/4]{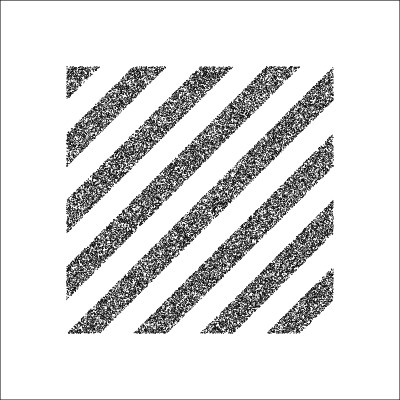}%
\includegraphics[width=\textwidth/4]{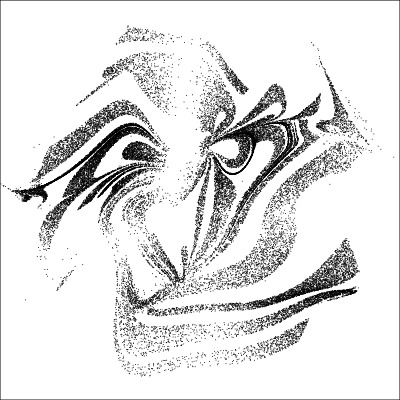}%
\includegraphics[width=\textwidth/4]{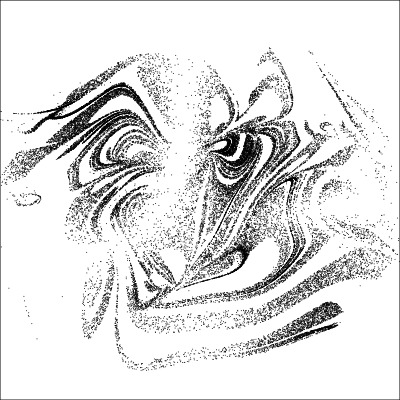}%
\includegraphics[width=\textwidth/4]{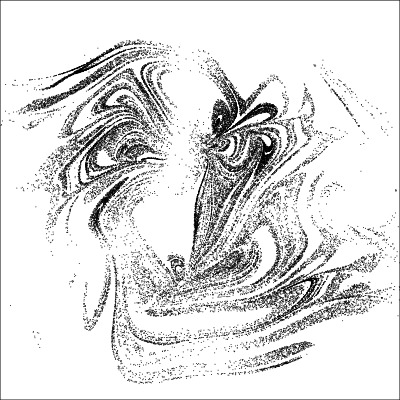}\vspace{-3.0pt}\\ \vspace{-3.0pt}%
\includegraphics[width=\textwidth/4]{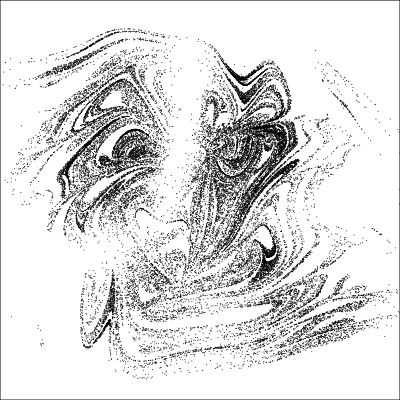}%
\includegraphics[width=\textwidth/4]{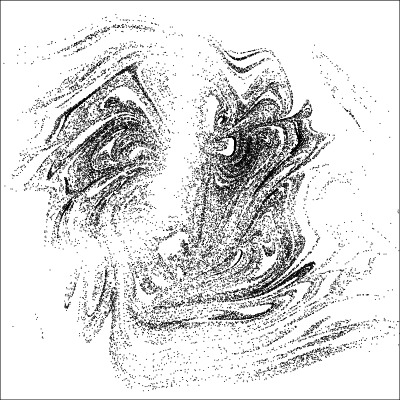}%
\includegraphics[width=\textwidth/4]{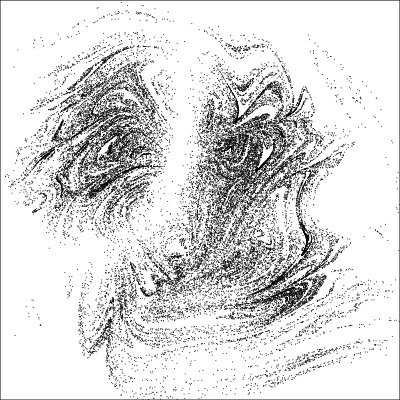}%
\includegraphics[width=\textwidth/4]{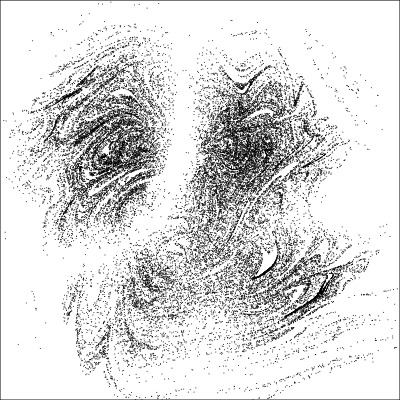}\\ \vspace{-3.0pt}%
\includegraphics[width=\textwidth/4]{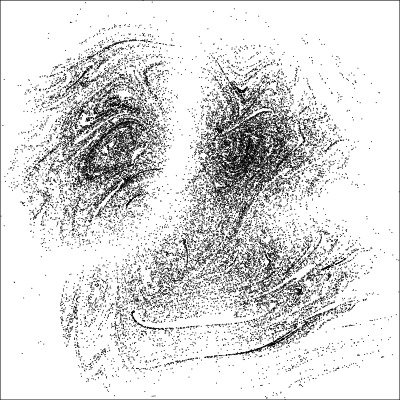}%
\includegraphics[width=\textwidth/4]{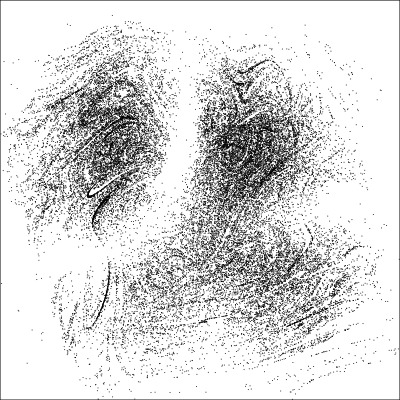}%
\includegraphics[width=\textwidth/4]{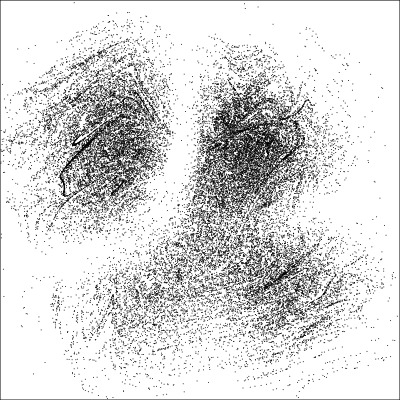}%
\includegraphics[width=\textwidth/4]{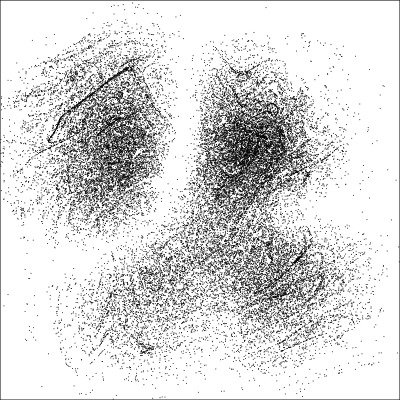}\\ \vspace{-3.0pt}%
\includegraphics[width=\textwidth/4]{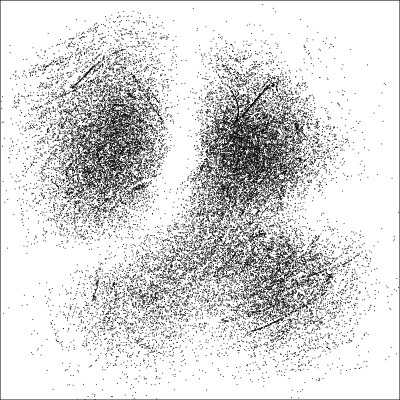}%
\includegraphics[width=\textwidth/4]{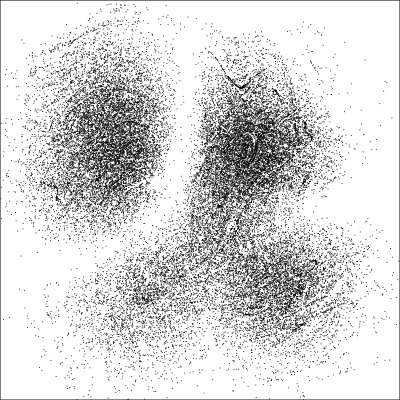}%
\includegraphics[width=\textwidth/4]{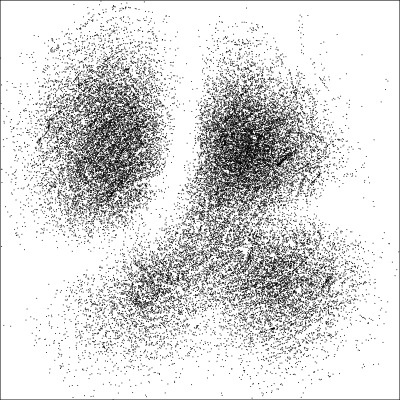}%
\includegraphics[width=\textwidth/4]{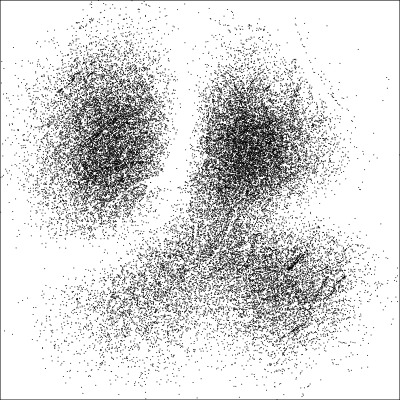}\\ \vspace{-3.0pt}%
\includegraphics[width=\textwidth/4]{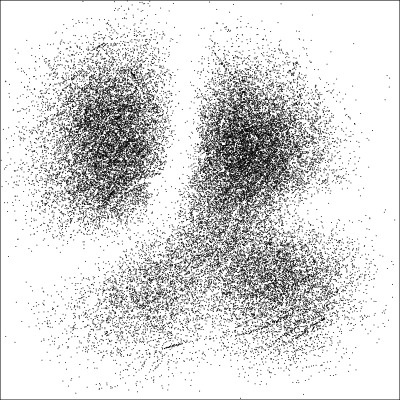}%
\includegraphics[width=\textwidth/4]{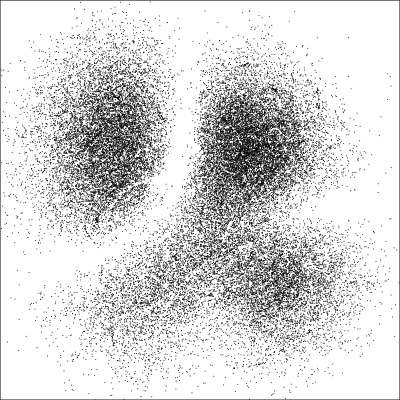}%
\includegraphics[width=\textwidth/4]{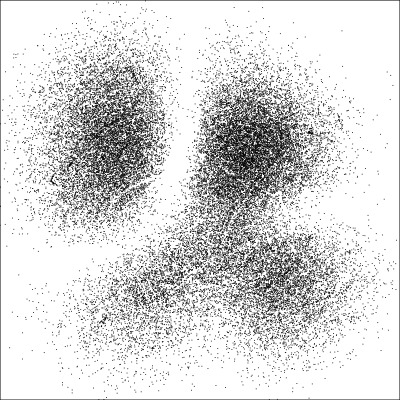}%
\includegraphics[width=\textwidth/4]{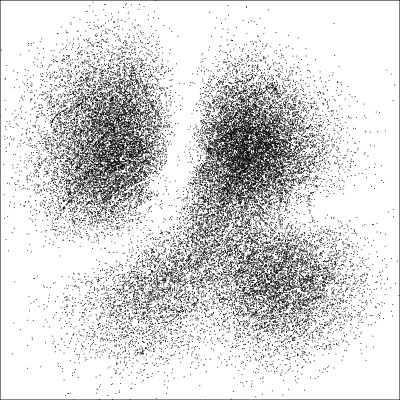}%
\caption[Another example of quantum relaxation to quantum equilibrium]{Relaxation of an identical system to figure \ref{quantum_relaxation} except with a different initial nonequilibrium, displayed in the top-left frame.}
\end{figure}
\newpage
\section{Quantum theory as an iRelax theory?}\label{2.5}
With the construction of the statistical underpinnings of this work now complete, attention returns to quantum theory.
The minimal iRelax framework outlines how a deterministic, information preserving theory may, through a \emph{de facto} loss of information, admit a probabilistic description with a concept of entropy rise.
Low entropy \emph{nonequilibrium} distributions $\rho_\text{noneq}(x)$ evolve towards the high entropy \emph{equilibrium} distribution $\rho_\text{eq}(x)=m(x)$.
The process is called \emph{relaxation} for the reason that it is precisely the thermodynamic relaxation that occurs in classical mechanics.
With this framework in place, it becomes natural to speculate on whether quantum theory, which after all is canonically probabilistic, will admit the iRelax framework\footnote{Although of course this work was inspired by de Broglie-Bohm theory in the first place, and so this was a foregone conclusion.}
The wonderful thing about the iRelax framework is that if an equilibrium distribution is known, perhaps from experiment, then the informational structure of the theory may be backwards inferred, and potential consistent laws of evolution deduced. 
To see how this works in practice, this method shall now be used to derive canonical and non-canonical variants of de Broglie-Bohm theory.

\subsection{The Valentini entropy (a derivation of de Broglie-Bohm theory)}
In 1991, in reference \cite{AV91a}, Valentini uncovered the informational structure of de Broglie Bohm theory. 
By using a known feature of de Broglie-Bohm theory, an observation by Bohm \cite{B53} that the ratio $f=\rho/|\psi|^2$ is conserved along trajectories, Valentini was able to argue that in de Broglie Bohm theory, $|\psi|^2$ should be regarded as a distribution of maximum entropy.
An equilibrium distribution to which all other nonequilibrium distributions tend to relax. 
The following passage describes how the logic works in reverse. 
The motivation being the question--Which realist, deterministic, time-reversible theories (iRelax theories) are consistent with our knowledge of quantum theory?

The quantum theory in question may remain abstract for now. Particle, field, or something more exotic altogether.
Suppose, however, that it admits a continuous complete basis $\{\ket{x}\}$ spanned by $n$ real parameters $x\in \mathbb{R}^n$ such that the resulting Born distribution $|\braket{x}{\psi}|^2=|\psi(x,t)|^2$ obeys a continuity equation, $\partial |\psi(x,t)|^2/\partial t + \nabla\cdot j(x,t)=0$. 
Regardless of the system and basis in question, if basis parameters $x$ are such that $|\psi(x,t)|^2$ obeys a continuity equation, then $|\psi(x,t)|^2$ may be regarded as a conserved distribution.
From the perspective of the IRPT formalism, the spirit of de Broglie-Bohm theory is to identify $|\psi(x,t)|^2$ as the equilibrium, maximum entropy distribution, $\rho_\text{eq}(x,t)=|\psi(x,t)|^2$. 
Hence the entry point of de Broglie-Bohm theory into the iRelax framework is at \ref{Ing3}. 
This distribution is of course defined upon the background state-space, $\Omega$. 
So it follows that the state-space (\ref{Ing1}) corresponding to this choice is $\Omega=\{x\}=\mathbb{R}^n$.
As by equation \eqref{max_ent}, density of states $m(x)$ is equal to $\rho_\text{eq}(x)$, it follows that $m(x)=\rho_\text{eq}(x,t)=|\psi(x,t)|^2$\footnote{This of course means that the density of states should be regarded as time-dependent, a feature that is critiqued in the next section.}.
Substitution of this into Jaynes entropy \eqref{Jaynes_entropy}, produces the entropy formula (\ref{Ing4}),
\begin{align}\label{Valentini_entropy}
S=-\int_\Omega \mathrm{d}^nx \rho(x,t)\log\frac{\rho(x,t)}{|\psi(x,t)|^2}.
\end{align}
This formula shall be referred to as the Valentini entropy, as it was first written down in reference \cite{AV91a}.
By the modified Liouville theorem, equation \eqref{ML2}, the corresponding statement of information/entropy conservation (\ref{Ing5}) is
\begin{align}\label{valentini_entropy_conservation}
\frac{d}{dt}\left(\frac{\rho}{|\psi|^2}\right)=0.
\end{align}
This formula, originally written down by Bohm \cite{B53}, was the starting point used by Valentini for his development of the statistical structure of de Broglie-Bohm theory \cite{AV91a}.
Finally, the \emph{de facto} rise in entropy \ref{Ing6}, follows from \ref{Ing4} and \ref{Ing5} as a result of Valentini's relaxation theorem above. 
The corresponding relaxation is commonly referred to as `quantum relaxation', and is shown in figure \ref{quantum_relaxation} for a two-dimensional system.

\subsection{Guidance equations}
The iRelax framework ensures that \ref{Ing1}, \ref{Ing4}, \ref{Ing5}, and \ref{Ing6} follow from \ref{Ing3}. 
The final ingredient is \ref{Ing2}, the law of evolution, which in de Broglie-Bohm Theory is commonly referred to as the guidance equation(s).

As the density of states $m(x,t)=|\psi(x,t)|^2$ is now time dependent, the situation differs a little from the iRelax theories in section \ref{2.4}.
Since $\partial m /\partial t=\partial |\psi|^2 /\partial t$ no longer vanishes, the condition placed on the law of evolution, equation \eqref{2.4_condition}, is modified to 
\begin{align}\label{debb_condition}
\nabla\cdot j(x,t)=-\frac{\partial |\psi(x,t)|^2}{\partial t}.
\end{align}
Hence the law of evolution $\dot{x}$ is constrained such that the equilibrium current $j=|\psi|^2\dot{x}$ has a \emph{convergence} $(-\nabla\cdot j_\text{eq})$ equal to the local rate of change of the density $|\psi|^2$. 
In order to derive a set of guidance equations, a current $j(x,t)$ must be sought such that it obeys continuity equation \eqref{debb_condition}.
The guidance equations $\dot{x}$ are then
\begin{align}
\dot{x}=\frac{j(x,t)}{|\psi(x,t)|^2}.
\end{align}
Of course, this is only one real condition upon the $n$-component law of evolution $\dot{x}$. Once one solution $\dot{x}$ is found to equation \eqref{debb_condition}, others $\dot{x}'$ may be found by adding an incompressible current to the equilibrium current, so that
\begin{align}
\dot{x}'=\dot{x}+\frac{j_\text{inc}}{|\psi|^2},
\end{align}
where $\nabla.j_\text{inc}(x,t)=0$.

The following two passages explain two different methods to arrive at de Broglie-Bohm guidance equations. The first is the traditional method, and takes advantage of the kinetic term found in many Hamiltonians. This method results in the canonical de Broglie-Bohm guidance equations that are favored by the majority of researchers in the field. 
Although guidance equations of this form are used exclusively in later chapters, in order to highlight that this is not the only option, a second derivation making use of Green's functions follows.
This results in an integral rather than a differential expression that bears a resemblance to classical electromagnetism.
Both sets of guidance equations are in their own sense unique, and may be argued to be in their own sense natural.

\subsubsection{Canonical guidance equations}
In any continuous basis $\ket{x}$ with components $\braket{x}{\psi}=\psi(x)$, the Schr\"{o}dinger equation may be used to write
\begin{align}\label{traditional_method_misc1}
\frac{\partial |\psi(x)|^2}{\partial t} - \frac{2}{\hbar}\text{Im}\left(\psi^* H \psi\right)=0.
\end{align}
The traditional method used to arrive at guidance equations involves taking advantage of the quadratic kinetic term found in many Hamiltonians to transform the quantity $-(2/\hbar)\text{Im}\left(\psi^* H \psi\right)$ in equation \eqref{traditional_method_misc1} into a total divergence. 
Once this is done, the current $j$ may be picked out, from which the guidance equations $\dot{x}=j/|\psi|^2$ follow.

To see how this works, note that in any quantum theory which has been canonically quantized \cite{Dirac}, canonical coordinates $q_i$ and $p_i$ carry over as operators $\hat{q}_i$ and $\hat{p}_i$.
These obey canonical commutation relations $\left[\hat{q}_i,\hat{p}_j\right]=i\hbar\delta_{ij}$%
\footnote{Extensions to De Broglie-Bohm to account for fermionic field theories in which anti-commutators are used may be found for instance in references \cite{CS07,CW11}.}.
If the Hamiltonian is of the form $\hat{H}=\sum_i \hat{p}_i^2/2m_i + V(\hat{q}_1,...,\hat{q}_n)$, then by adopting the standard Schr\"{o}dinger $\ket{q}$ representation\footnote{The term `Schr\"{o}dinger representation' is used to refer to the representation of the canonical operators by differential operators as for instance used by reference \cite{Hatfield}, and shouldn't be confused with the Schr\"{o}dinger picture. For instance, commutation relation $[q,p]=i\hbar$ is satisfied with $\hat{q}\rightarrow q$, $\hat{p}\rightarrow -i\hbar\partial/\partial q$ in the Schr\"{o}dinger coordinate or configuration representation, and by $\hat{q}\rightarrow i\hbar\partial/\partial p$, $\hat{p}\rightarrow p$ in the Schr\"{o}dinger momentum representation.}, $\hat{q}_i\rightarrow q_i$, $\hat{p}_i\rightarrow -i\hbar\partial/\partial q_i$, the term may be written
\begin{align}
- \frac{2}{\hbar}\text{Im}\left(\psi^* H \psi\right)
&=\sum_i\frac{\hbar}{m_i}\text{Im}\left(\psi^* \partial^2_{q_i} \psi\right)\nonumber \\
&=\sum_i\frac{\hbar}{m_i}\text{Im}\partial_{q_i}\left(\psi^* \partial_{q_i} \psi\right)\nonumber\\
&=\sum_i\partial_i\left[|\psi|^2\frac{\hbar}{m_i}\text{Im}\left(\frac{\partial_{q_i} \psi}{\psi}\right)\right],
\end{align}
which indicates guidance equations,
\begin{align}\label{basic_guidance_long}
\dot{q}_i =\frac{\hbar}{m_i}\text{Im}\left(\frac{\partial_{q_i} \psi}{\psi}\right).
\end{align}
This is most commonly expressed by writing $\psi$ in complex polar form $\psi=|\psi|e^{iS/\hbar}$, whereby equations \eqref{basic_guidance_long} become
\begin{align}\label{basic_guidance_short}
\dot{q}_i = \frac{\partial_{q_i} S}{m_i},
\end{align}
the canonical form of de Broglie-Bohm guidance equations.

When treating a field theory, much the same process is followed, except with the state-space $\Omega$ promoted to a function space, with states promoted to functions, and with functions of such states promoted to functionals.
Take a classical real scalar field $\phi(y,t)$, where $y$ is a 3-space coordinate.
Canonical quantization results in the replacement of $\phi(y,t)$ and its momentum conjugate $\pi(y,t)$ with operators $\hat{\phi}(y)$ and $\hat{\pi}(y)$ that satisfy commutation relations $[\hat{\phi}(y),\hat{\pi}(y')]=i\delta^{(3)}(y-y')$. (Working in the Schr\"{o}dinger picture with units where $\hbar=c=1$.)
The evolution of quantum state $\ket{\psi}$ is still determined by the Schr\"{o}dinger equation, except now the Hamiltonian is the operator equivalent of a functional, $H=\int \mathrm{d}^3y\frac{1}{2}\hat{\pi}(y)^2 + V(\hat{\phi})$.
To find the current $j$, represent the components of quantum state $\ket{\psi}$ in a field configuration basis $\ket{\phi}$, with functional $\braket{\phi}{\psi}=\psi[\phi]$. 
This is called the wave functional.
Take the corresponding functional Schr\"{o}dinger representation for the field operators $\hat{\phi}(y)=\phi(y)$, $\hat{\pi}(y) =\delta/\delta\phi(y)$.
Then, 
\begin{align}
- 2\text{Im}\left(\psi^* H \psi\right)
=\int \mathrm{d}^3y\text{Im}\left(\psi^* \frac{\delta^2\psi}{\delta\phi(y)^2}\right)
=\int \mathrm{d}^3y\text{Im}\frac{\delta}{\delta\phi(y)}\left(\psi^*\frac{\delta\psi}{\delta\phi(y)}\right)\\
=\int \mathrm{d}^3y\frac{\delta}{\delta\phi(y)}\left[|\psi|^2\text{Im}\left(\frac{1}{\psi}\frac{\delta\psi}{\delta\phi(y)}\right)\right].
\end{align}
The operator $\int \mathrm{d}^3y\frac{\delta}{\delta\phi(y)}$ is the functional equivalent of $\nabla\cdot$, and so this final expression may be thought of as a total divergence. Hence, the guidance equation may be read off as 
\begin{align}
\dot{\phi}(y)=\text{Im}\left(\frac{1}{\psi}\frac{\delta\psi}{\delta\phi(y)}\right),
\end{align}
which, using the polar form of the wave functional, $\psi[\phi]=|\psi[\phi]|\exp(iS[\phi])$, may be expressed
\begin{align}
\dot{\phi}(y)=\frac{\delta S}{\delta\phi(y)}
\end{align}
in analogy with equation \eqref{basic_guidance_short}.
This is one method by which quantum fields may be treated formally. 
For practical calculations involving quantum fields (see later chapters), it is generally speaking more convenient to work in Fourier space $\braket{\tilde{\phi}(k)}{\psi}=\psi[\tilde{\phi}(k)]$ with a box-normalized quantum field. 
The box normalization ensures that the number of modes $k$, and hence the dimensionality of the configuration space is countable.

\myparagraph{The character of the canonical law}
The canonical method produces guidance equations that stipulate system velocities proportional to the gradient of the complex phase of the wave function, $\dot{x}\sim\nabla S(x,t)$.
This expression appears particularly compelling when considering a single non-relativistic point particle in 3-space.
Such a particle would be guided around 3-space by the gradient of the 3-space complex phase $S$ of the wave function $\psi(x,t)$.
In this sense, such a theory would be a true `pilot-wave' theory.
Care should be taken with this `wave in 3-space' idea however.
If there are $m$ such particles, $S$ is not defined upon 3-space, but upon configuration space $\mathbb{R}^{3m}$.
And in the scalar field theory described above, the complex phase is the functional $S[\phi]$, and so this resembles the wave-in-space notion even less.

As $\dot{x}\sim \nabla S(x,t)$, it may seem at first glance that the resulting dynamics may be characterized as irrotational.
This is only partially true however. 
As $S(x,t)$ is the complex phase of the wave function, it maps the configuration space as $S:\Omega\rightarrow [0,2\pi)$, not $\Omega\rightarrow\mathbb{R}$, and this has important consequences for the dynamics.
Note for instance that $S(x,t)$ is ill-defined on the nodes of the wave function, at $\psi=0$.
That the canonical guidance equations are not defined on nodes might appear concerning at first, but as $\psi=0$ constitutes two real conditions, in practical situations the nodes are $(n-2)$-dimensional hypersurfaces on the $n$-dimensional configuration space, and hence measure zero.
In any usual circumstance, there is a zero chance of finding a system \emph{on} a node.
Nevertheless, nodes do play an important role in the dynamics.
Consider for instance applying the classical Kelvin-Stokes theorem to velocity field $\dot{x}$, in two or three-dimensions, to a two-dimensional subregion $\omega$ and its boundary $\partial\omega$,
\begin{align}\label{Kelvin-Stokes}
\oint_{\partial\omega}\dot{x}\cdot \mathrm{d}x=\iint_\omega \left(\nabla\times\dot{x}\right)\cdot\mathrm{d}^2x.
\end{align}
If the region $\omega$ contains no nodes, then $\dot{x}=\nabla S$ everywhere in $\omega$, and so $\nabla\times\dot{x}=0$, causing the RHS of equation \eqref{Kelvin-Stokes} to vanish.
If on the other hand, there is a node in $\omega$, then $\dot{x}$ does not equal $\nabla S$ on the node, and so the RHS of equation \eqref{Kelvin-Stokes} need not vanish.
By writing 
\begin{align}
\oint_{\partial\omega}\dot{x}\cdot\mathrm{d}x=\oint_{\partial\omega}\nabla S\cdot\mathrm{d}x=\oint_{\partial\omega}\mathrm{d}S=\Delta S,
\end{align}
the LHS of \eqref{Kelvin-Stokes} may be seen to equal $\Delta S$, the change in complex phase $S$ around the circuit defined by $\partial\omega$.
Of course, in order to ensure $S$ is single valued, it can only change by some integer multiple of $2\pi$ around a closed circuit.
Hence the nodes may be said to have `quantized vorticities' of $\pm2\pi m$ for $m\in \mathbb{Z}^+$, in the sense that $\Delta S=\oint_{\partial\omega}\dot{x}\cdot\mathrm{d}x$ takes only these values when evaluated on a path $\partial\omega$ that winds around a single node%
\footnote{In higher dimensions the same result follows by regarding the `velocity field', $\dot{x}$, as the 1-form found by taking the exterior derivative $d$ of scalar $S$, $\dot{x}=dS$. This is defined everywhere except on nodes. The generalized Stokes' theorem may then be applied to $\dot{x}$, 
\begin{align}
\int_{\partial \omega}\dot{x}=\int_\omega d\dot{x}.
\end{align}
If there is no node within $\omega$, then as the exterior derivative is nilpotent, $d\dot{x}=d^2S=0$ and so the RHS vanishes.
If on the other hand there is a node in $\omega$, then it need not vanish but, again, it may only change as $\pm2\pi m$ because $S$ is single valued.}.
The case of $m=0$ is ruled out as the $m$ relates to the order of the zero created by the node.
See chapter \ref{3} for more on this point.
Close to nodes, the primary component of $\dot{x}$ is that which circulates around the node $\dot{x}_\text{circ}$. 
By taking $\omega$ to be a small disk of radius $r$ around a node%
\footnote{In n-dimensions, as the nodes are n-2 dimensional hypersurfaces, there are always 2 spare dimensions in which to orient a circular $\partial\omega$ so that it winds around a node.}.
Then, by approximating the integral as $|\Delta S|\sim 2\pi|\dot{x}_\text{circ}|r$, and noting that $|\Delta S|$ is restrained to a fixed value, it may be concluded that $|\dot{x}_\text{circ}|\sim 1/r$. 
The circular component of the velocity $\dot{x}$ varies inversely with the distance from the node, and hence diverges as the node is approached.
This vortex that surrounds each node is a generator of chaos, and is often described as a driving factor in quantum relaxation.
So the canonical guidance equations may be characterized as irrotational except for nodes, which are each surrounded by a vortex of quantized strength.
A systematic, in-depth description of the properties of the nodes and how they relate to quantum relaxation is given in chapter \ref{3}.

\subsubsection{Non-canonical guidance equations from Green's functions}
The need to find a continuity equation in some basis $\ket{x}$ suggests the need find some vector $j$ such that
\begin{align}\label{alternative_method_misc1}
\nabla \cdot j(x) = -\frac{\partial |\psi|^2}{\partial t},
\end{align}
where $\nabla \cdot$ is the $n$-dimensional divergence. 
This is bears some resemblance to Gauss' law in classical electromagnetism, and may be solved with a similar procedure.
To find a solution, find a vector Green's function $G(x,x')$ such that 
\begin{align}\label{green_equation}
\nabla \cdot G(x,x') =\delta^{(n)}(x-x').
\end{align}
In $n$-dimensional hyperspherical coordinates the divergence of distribution $f=(f_r,f_{\theta_1},...,f_{\theta_{n-1}})$ is
\begin{align}
\nabla \cdot f= \frac{1}{r^{n-1}}\frac{\partial}{\partial r}\left(r^{n-1} f_r\right)+\text{angular terms}.
\end{align}
As may be readily checked, equation \eqref{green_equation} is satisfied by Green's function 
\begin{align}\label{alternative_method_green_function}
G(x,x')=-\frac{1}{A(n)}\frac{\widehat{\Delta x}}{|\Delta x|^{n-1}},
\end{align}
where $\Delta x$ is a vector pointing from state $x$ to state $x'$ in $\Omega$ (not the Hilbert space), $\widehat{\Delta x}$ is the corresponding unit vector, and $A(n)$ is the surface area of the unit $n$-sphere, $A(n)=2\pi^{n/2}/\Gamma(n/2)$. (In 3-space this is the familiar factor of $4\pi$.)
Green's function \eqref{alternative_method_green_function} suggests equation \eqref{alternative_method_misc1} is solved by current
\begin{align}\label{Green_current}
j(x,t)=\frac{1}{A(n)} \int_\Omega\mathrm{d}^nx' \frac{\widehat{\Delta x}}{|\Delta x|^{n-1}} \frac{\partial |\psi(x')|^2}{\partial t}.
\end{align}
The corresponding guidance equation is
\begin{align}\label{Green_velocity}
\dot{x}=\frac{1}{A(n)|\psi(x)|^2} \int_\Omega\mathrm{d}^nx' \frac{\widehat{\Delta x}}{|\Delta x|^{n-1}} \frac{\partial |\psi(x')|^2}{\partial t}.
\end{align}

So, in contrast to the differential expression $\dot{x}\sim\nabla S$ obtained through the canonical method, the outcome of using this method is an integral expression.
As this is an integral expression, it is more difficult than the canonical equations to model computationally.
Possibly it is for this reason that a study focusing on guidance equation \eqref{Green_velocity} is yet to make it into the literature.
Practical considerations aside, however, expression \eqref{Green_current} does possess some intriguing properties.
Recall that canonical expression $\dot{x}\sim\nabla S$ may be considered unique in the sense that it is irrotational (excepting nodes).
In its own sense, the equilibrium current \eqref{Green_current} may be considered unique and minimal. 
Of course, continuity equation \eqref{alternative_method_misc1} determines the divergence component of the current, but it allows for arbitrary `curl' parts.
Current \eqref{Green_current} is unique and minimal in the sense that it doesn't feature any such `curl' part, and so may be considered `curl-free' or irrotational.
As may be checked readily in 2 and 3-dimensions, 
\begin{align}
\nabla\times\frac{\widehat{\Delta x}}{|\Delta x|^{n-1}}=0.
\end{align}
In higher dimensions the terms curl-free and irrotational are used in the sense that the current may be expressed as the gradient of a scalar. 
In this case as a result of the fact that
\begin{align}
\frac{\widehat{\Delta x}}{|\Delta x|^{n-1}}=\nabla\left[\frac{1}{n-2}\frac{1}{|\Delta x|^{n-2}}\right].
\end{align}
Equilibrium current \eqref{Green_current} and guidance equation \eqref{Green_velocity} are intuitively minimal.
The quantity $\widehat{\Delta x}/|\Delta x|^{n-1}$ is nothing more than an $n$-dimensional inverse square.
If the density of states $|\psi(x')|^2$ nearby to a state $x$ increases, systems in state $x$ will be drawn towards $x'$ with a velocity that varies inversely to the distance of separation. 
In this regard, this alternative formulation could be considered to involve \emph{straight} lines-of-influence.
This might be considered appealing for the purposes of conserved quantities for instance. 
Another feature that is lacking from the canonical formulation is a momentum representation.
Canonical equations \eqref{basic_guidance_short} rely on the presence of a quadratic kinetic term in the Hamiltonian, and could only in principle permit a momentum representation in the presence of a quadratic potential term.
The alternative guidance equation \eqref{Green_velocity} does not rely on any particular form of the Hamiltonian and so permits a momentum representation to be used.

\subsection{An imperfect fit? The need for a unified realist theory?}
The bulk of this chapter has been devoted to exposing the consequences of adopting the two minimal tenets of realism\footnote{The term realism is used in its broadest possible sense, to mean the existence of a state-space.}+ information conservation.  
The resulting theories are automatically endowed with a concept of thermodynamic relaxation, the notion of equilibrium and nonequilibrium distributions, and the concept of thermodynamic relaxation.
Although this whole approach was initially motivated by extraordinary success of quantum relaxation and quantum nonequilibrium in De Broglie-Bohm theory, it is useful to reflect on how well de Broglie-Bohm theory fits into the formalism developed.
To illustrate, consider the first three iRelax ingredients.

\paragraph{\ref{Ing1}}
The state-space of classical mechanics is of course phase-space, which may be loosely characterized as spanning all possible canonical coordinates $q$ and $p$. The state-space of canonical De Broglie-Bohm theory, on the other hand, is \emph{configuration space}, and only spans the $q$'s, not the $p$'s.
So by quantizing a classical theory, the dimensionality of the state-space (the number of parameters required to determine a system's evolution), is apparently halved.
Of course this is not without its consequences.
Recall that a coordinate in state-space is supposed to `constitute sufficient information to entirely determine to future state of the system'. 
In de Broglie-Bohm theory, this is not the case, as is made evident by the form of the guidance equations, \ref{Ing2}.
\paragraph{\ref{Ing2}}
Although it has not specifically been mentioned yet, guidance equations \eqref{basic_guidance_short} and \eqref{Green_velocity}, \emph{actually do not} by themselves determine the future evolution of a quantum system.
In both cases, some knowledge of the wave function $\psi(x,t)$ is required to be entered on the RHS%
\footnote{Curiously, the two different methods require distinct (non-overlapping) pieces of information on $\psi(x,t)$. The Canonical guidance equations \eqref{basic_guidance_short} require knowledge of the complex phase $S(x,t)$ of the wave function, whereas the non-canonical require knowledge only of its norm $|\psi(x,t)|^2$.}.
Hence, if de Broglie-Bohm theory is viewed as an iRelax theory, it must be one that is necessarily supplemented with extra knowledge on the quantum state.
\paragraph{\ref{Ing3}}
The derivation of de Broglie-Bohm theory above began with the premise of taking the equilibrium distribution $\rho_\text{eq}(x,t)$ to be the Born distribution $|\psi(x,t)|^2$.
The obvious question of course is why this distribution is time-dependent.
In a self-contained theory with no external factors, there is no reason for this to be so.
Indeed, by the considerations of section \ref{2.4}, $\rho_\text{eq}$ is equal to density of states $m$, and so the density of states must also be time-dependent $m=m(x,t)$.
The need for outside factors to determine the geometry of the state-space is surely in conflict with the principle of indifference.

So although de Broglie-Bohm theory does possess an iRelax framework of sorts, it does not appear to be complete.
The de Broglie-Bohm equivalent of the iRelax state is the configuration, but this does not supply the requisite information to determine a system's evolution. In addition to the configuration, the quantum state $\psi$ must also be known.  
Suppose for a moment the state-space in question to be $\Omega=\mathbb{R}^n$. Then true iRelax theories map states to other states as
\begin{alignat}{3}
\dot{x}:\mathbb{R}^n&\longrightarrow\mathbb{R}^n &\qquad\qquad
x:\mathbb{R}^n\times\mathbb{R}&\longrightarrow\mathbb{R}^n\\
x&\longmapsto \dot{x}(x) &\qquad\qquad
(x_0,t)&\longmapsto x(x_0,t)\nonumber,
\end{alignat}
whereas de Broglie-Bohm maps states as
\begin{alignat}{3}
\dot{x}:\mathbb{R}^n\times\mathcal{H}&\longrightarrow\mathbb{R}^n &\qquad\qquad
x:\mathbb{R}^n\times\mathbb{R}\times\mathcal{H}&\longrightarrow\mathbb{R}^n\\
(x,\psi)&\longmapsto \dot{x}(x,\psi) &\qquad\qquad
(x_0,t,\psi)&\longmapsto x(x_0,t,\psi)\nonumber
\end{alignat}
If nature does admit a complete underlying quantum iRelax framework, then clearly it is one in which configuration space $\Sigma$ does not exhaust the true state-space $\Omega$.
The quantum state exists within its own Hilbert space, $\mathcal{H}$, though. 
And so it may be informative to entertain the idea of adjoining the two spaces $\Omega\overset{?}{\sim}\Sigma\otimes\mathcal{H}$, and searching for a theory with a time-invariant density $m(x)$, that is consistent with conventional quantum theory given knowledge of $\ket{\psi}$.
Conventional quantum theory could then be viewed a effective theory obtained by integrating over the configuration space degrees of freedom. (A marginal theory if you will.)
De Broglie-Bohm would be viewed as the effective theory found by presupposing perfect knowledge of the Hilbert space degrees of freedom. (A cross-section of the underlying theory.)
Any further speculation is unwarranted at the present time. 
However this is clearly an avenue for future work.
\section{Outlook}\label{2.6}
As has hopefully been made clear, if nature truly does possess an iRelax framework, whether fully implemented, or with a partial de Broglie-Bohm style implementation, then its details are still very much an open question.
With all the open possibilities (the correct basis, the `curl' components of the velocity, etc.) it feels unwise to put too much faith into any single formulation.
Despite these misgivings over the exact details, the startling fact remains that the imposition of the really-quite-minimal principle of information conservation upon a background abstract state-space results in theories with a concept of thermodynamic relaxation.
And de Broglie-Bohm theory is the prototypical example of how this might be applied to quantum theory, so that quantum probabilities arise in the same way as classical ignorance/uncertainty probabilities. 

There appear to be two obvious paths forward.
For the sake of argument, call these the route of the idealist and the route of the pragmatist.
The idealist would attempt to find ways to pin down the as-yet undetermined factors in the implementation. They might search for an application of the iRelax framework that unifies quantum theory with the ideas brought forward by de Broglie-Bohm.
The pragmatist, on the other hand, would note that all conceived theories to-date predict violations of standard quantum predictions for nonequilibrium distributions $\rho_\text{noneq}$.
And so surely the pursuit of these quantum violations resulting from quantum nonequilibrium is a clear way forward. 
Such a discovery, if made, would rule out all other major interpretations of quantum mechanics, most of which do not deviate from conventional quantum mechanics in their predictions. 
(More on this in chapters \ref{4} and \ref{5}.)
So the real task is to discover evidence for, or possibly even a source of, quantum nonequilibrium, as this would give \emph{experimental} weight to the whole affair.
In order to make experimental predictions, however, the pragmatist faces the dual challenges of an uncertain theory and as yet unknown quantum nonequilibrium distributions.

The idealist might berate the pragmatist for attempting to make predictions based upon a theory with meaningful unknown factors. 
The pragmatist might scorn the idealist for pursuing a theory with no experimental proof. 
To the mind of the author, however, both approaches have significant merit.
For the prospect is showing that quantum probabilities fit into a classical framework. Thus helping to de-mystify quantum theory in general, and providing a wealth of other quantum-violating predictions.
The purpose of this chapter has been to attempt to clearly pose the problem of the idealist, so that it may provide a foundation for future idealists to build upon.
For the remainder of this thesis, the role of the pragmatist shall be adopted.

\newpage
\begin{figure}
\begin{center}
\resizebox{0.9\textwidth}{!}{%
\begin{turn}{270}
\begin{tabular}{|p{70pt}|p{76pt}|p{75pt}|p{70pt}|p{75pt}|p{70pt}|p{85pt}|}
\hline
&
\mbox{\ref{Ing1}:} State Space $\Omega=\{x\}$&
\mbox{\ref{Ing2}:} Law of evolution $\dot{x}$&
\mbox{\ref{Ing3}:} Equilibrium distribution $\rho_\text{eq}(x)$&
\mbox{\ref{Ing4}:} \mbox{Information} measure \mbox{(entropy)} $S$&
\mbox{\ref{Ing5}:} Information Conservation statement 
&Issues?\\
\hline Discrete states (section \ref{2.2})&%%%%%%%%%%%%%%%%%%%%%%%%%%%%%%%%%%%%%%%%%%%%%%%%%%%%%%%%%%%%%%%%%%%%%%%%%%%%
\mbox{$\Omega=\{A,B,C,...\}$}&
Iterative \mbox{permutations}&
\mbox{$p_A=p_B=\dots$}&
Gibbs entropy $-\sum_{i}p_i\log p_i$&
One-to-one map between states&
No issues\\
\hline Flat Continuous state-space (section \ref{2.3})&%%%%%%%%%%%%%%%%%%%%%%%%%%%%%%%%%%%%%%%%%%%%%%%%%%%%%%%%%%%%%%%%%%%
$\Omega=\{x\}=\mathbb{R}^n$&
Divergence free, $\nabla\cdot\dot{x}=0$&
Uniform $\rho(x)$ w.r.t.\ $x$, \mbox{$\rho(x)=\text{constant}$}&
Differential entropy $-\int \mathrm{d}^nx\rho\log\rho$&
Liouville's theorem \mbox{$\frac{d\rho(x,t)}{dt}=0$}&
Issues with diff.\ entropy  (cf. section \ref{2.4})\\ 
\hline Classical \mbox{mechanics} (section \ref{2.3})&%%%%%%%%%%%%%%%%%%%%%%%%%%%%%%%%%%%%%%%%%%%%%%%%%%%%%%%%%%%%%%%%%%%%%%%
Phase Space, $\Omega=\{x\}=\{q,p\}$&
Hamilton's equations&
Uniform $\rho(q,p)$ (w.r.t. canonical coords.)&
Differential entropy $-\int\mathrm{d}q\mathrm{d}p\,\rho\log\rho$&
Liouville's theorem $\frac{d\rho(q,p,t)}{dt}=0$&
Issues with diff.\ entropy  (cf. section \ref{2.4})\\ 
\hline Generalized continuous state-space (section \ref{2.5})&%%%%%%%%%%%%%%%%%%%%%%%%%%%%%%%%%%%%%%%%%%%%%%%%%%%%%%%%%%%%%%%%%%
$\Omega = \{x\}$ (at least locally)&
Divergence free w.r.t.\ density of states $m(x)$, \mbox{$\nabla\cdot[m(x)\dot{x}]=0$}&
$\rho(x,t)=m(x)$&
Jaynes entropy $-\int\mathrm{d}^nx\,\rho\log\frac{\rho}{m}$&
Generalized Liouville theorem $\frac{d}{dt}\left(\frac{\rho}{m}\right)=0$&
No issues (excepting the aforementioned missing link)\\
\hline De Broglie-Bohm type (section \ref{2.4})&%%%%%%%%%%%%%%%%%%%%%%%%%%%%%%%%%%%%%%%%%%%%%%%%%%%%%%%%%%%%%%%%%%%%%%%%%%%%%%%%%%%%%%%%%%
Configuration space $\{x\}=\Sigma$ &
$\dot{x}=j/|\psi|^2$. Canonically \mbox{$\dot{x}\sim\partial_{x_i}S/m_i$}, where \mbox{$\psi=|\psi|^{iS/\hbar}$}.&
$\rho(x,t)=|\psi(x,t)|^2$&
Valentini entropy $-\int\mathrm{d}^nx\,\rho\log\frac{\rho}{|\psi|^2}$&
Generalized Liouville theorem $\frac{d}{dt}\big(\frac{\rho}{|\psi|^2}\big)=0$ due to $m=|\psi|^2$&
State \mbox{incompletely} specified by $x$. Need to know $\psi$. Partial iRelax structure.\\
\hline Unified \mbox{quantum theory?} (section \ref{2.5})&%%%%%%%%%%%%%%%%%%%%%%%%%%%%%%%%%%%%%%%%%%%%%%%%%%%%%%%%%%%%%%%%%%%%%%%%%%%%%%%%%%%%%%%%%%
$\Omega=\{x\}=\Sigma\otimes\mathcal{H}$?&
Time independent $\dot{x}=\dot{x}(x)$. Subject to \mbox{$\nabla\cdot[m(x)\dot{x}]=0$}&
Time \mbox{independent} $\rho_\text{eq}=\rho_\text{eq}(x)$&
$-\int\mathrm{d}^nx \rho\log\frac{\rho}{\rho_\text{eq}}$&
$\frac{d}{dt}\big(\frac{\rho}{\rho_\text{eq}}\big)=0$.&
Not yet formulated. Unclear whether possible.\\
\hline%%%%%%%%%%%%%%%%%%%%%%%%%%%%%%%%%%%%%%%%%%%%%%%%%%%%%%%%%%%%%%%%%%%%%%%%%%%%%%%%%%%%%%%%%%%%%%%%%%%%%%%%%%%%%%%%%%%
\end{tabular}
\end{turn}
}
\end{center}\vspace{-5mm}
\caption{A summary of the iRelax theories considered in chapter \ref{2}.}\label{IPRT_summary}
\end{figure}
%%%%%%%%%%%%%%%%%%%%%%%%%%%%%%%%%%%%%%%%%%%%%%%%%%%%%%%%%%%%%%%%%%%%%%%%%%%%%%%%%%%%%%%%%%%%%%%%%%%%%%%%%%%%%%%%%%%%%

\chapter*{PREFACE TO CHAPTER 3}\addcontentsline{toc}{chapter}{PREFACE TO CHAPTER 3}
\vspace{-2mm}Chapter \ref{2} is the theoretical background and justification for relaxation, both quantum and classical.
It attests the tendency of nonequilibrium distributions to approach equilibrium.
However it does not say anything about whether equilibrium will actually be reached.
In any individual case, there may arise practical barriers that prevent full relaxation to equilibrium.
For classical systems, conservation laws are such barriers.
Consider, for instance, the simple two-dimensional example depicted in figure \ref{classical_relaxation}.
In that example, the Hamiltonian is time-independent and so energy is conserved by the individual systems.
But the ensemble distribution corresponding to equilibrium and maximum differential entropy \eqref{differential_entropy}, is uniform on the phase space, and hence is unbounded in energy.
Energy conservation restricts trajectories to closed circuits in the phase space defined by $E(q,p)=$ constant.
The result is that the relaxation reaches the roughly stationary distribution depicted in the final frame of figure \ref{classical_relaxation}, but goes no further towards full relaxation to a uniform distribution.
This distribution is certainly closer to equilibrium than the initial conditions.
To the eye it appears more spread out. 
It is higher in entropy than initially.
But it is not full relaxation.

In de Broglie-Bohm theory, conservation laws exist on the level of the quantum state, and not the individual trajectories.
So there is no directly analogous barrier to quantum relaxation.
Nevertheless, the field is still very much in its infancy, and the possibility of there existing non-trivial barriers to relaxation still is very much open.
The discovery of such barriers would certainly be of importance to the field as they would lend support to the prospect of quantum nonequilibrium persisting and being detectable today. 

Since the advent of the widespread availability of personal computers in the 1990s, there have been a variety of computer-aided investigations into the properties of de Broglie-Bohm trajectories.  
These may be roughly divided into three main categories.
First there were investigations into the chaotic nature of the trajectories \cite{deP95,FS95,PV95,DM96,SC96,C00,MPD00,SF08}.
Second, through a series of papers, \cite{F97,WS99,FF03,WP05,WPB06,EKC07,WPB06,CE08,EKC09,CMS16,TCE16}, there came a gradual but general acceptance that the chaos was generated (or at least driven) by the nodes of the wave function.
Third, for the most part inspired by the first demonstration of quantum relaxation in 2005 \cite{VW05}, there came a series of papers discussing the properties of the relaxation, \cite{EC06,B10,CS10,SC12,TRV12,CDE12,ACV14,KV16,ECT17}.

Of this final category, references \cite{ACV14} and \cite{KV16} have most directly considered the possibility of relaxation barriers.
Reference \cite{ACV14} found that for sufficiently simple states, there can remain a residual nonequilibrium that is unable to decay away.
Reference \cite{KV16} found that for system that are perturbatively close to the ground (vacuum) state, trajectories could be confined to regions in the configuration space.
To obtain their results, both of these studies evolved de Broglie trajectories numerically through periods of time that push the boundaries of what is computationally feasible at present. 
Chapter \ref{3} also contributes to the effort to find barriers, but does so through a different approach in order to avoid the computational bottleneck.

\chapter{EXTREME QUANTUM NONEQUILIBRIUM, NODES, VORTICITY, DRIFT, AND RELAXATION RETARDING STATES}\label{3}
\vspace{-15mm}
\begin{center}
\textit{Nicolas G Underwood\hyperlink{1st_address_chap_3}{$^\dagger$}\hyperlink{2nd_address_chap_3}{$^\ddagger$}}
\end{center}
\begin{center}
\textit{Adapted from J. Phys. A: Math. Theor. 51 055301 (2018)} \cite{NU18}
\end{center}
\begin{center}
\hypertarget{1st_address_chap_3}{$^\dagger$}Perimeter Institute for Theoretical Physics, 31 Caroline Street North, Waterloo, Ontario, N2L 2Y5, Canada\\
\hypertarget{2nd_address_chap_3}{$^\ddagger$}Kinard Laboratory, Clemson University, Clemson, 29634, SC, United States of America
\end{center}

\section*{Abstract}
Consideration is given to the behaviour of de Broglie trajectories that are separated from the bulk of the Born distribution with a view to describing the quantum relaxation properties of more `extreme' forms of quantum nonequilibrium. For the 2-dimensional isotropic harmonic oscillator, through the construction of what is termed the `drift field', a description is given of a general mechanism that causes the relaxation of `extreme' quantum nonequilibrium. Quantum states are found which do not feature this mechanism, so that relaxation may be severely delayed or possibly may not take place at all. A method by which these states may be identified, classified and calculated is given in terms of the properties of the nodes of the state. Properties of the nodes that enable this classification are described for the first time.

\section{Introduction}\label{3.1}
De Broglie-Bohm theory \cite{deB28,B52a,B52b,Holl93} is the archetypal member of a class \cite{CS10,W07,DG97} of interpretations of quantum mechanics that feature a mechanism by which quantum probabilities arise dynamically \cite{CS10,AV91a,AV92,VW05,TRV12,ACV14,SC12} from standard ignorance-type probabilities. 
This mechanism, called `quantum relaxation', occurs spontaneously through a process that is directly analogous to the classical relaxation (from statistical nonequilibrium to statistical equilibrium) of a simple isolated mechanical system as described by the second thermodynamical law.
Demonstration of the validity of quantum relaxation \cite{VW05,EC06,CS10,SC12,TRV12,ACV14} has meant that the need to postulate agreement with quantum probabilities outright may be dispensed with. It has also prompted the consideration of `quantum nonequilibrium'--that is, violations of standard quantum probabilities, which in principle could be observed experimentally\footnote{This conclusion is not entirely without its critics. For an alternative viewpoint see references \cite{DGZ92,DT09}. For counterarguments to this viewpoint see references \cite{AV96,AV01}.}.
In other words, these theories allow for arbitrary nonequilibrium probabilities, only reproducing the predictions of standard quantum theory in their state of maximum entropy.
It has been conjectured that the Universe could have begun in a state of quantum nonequilibrium \cite{AV91a,AV92,AV91b,AV96,AV01,UV15} which has subsequently mostly degraded, or that exotic gravitational effects may even generate nonequilibrium \cite{AV10,AV07,AV04b}.
If such nonequilibrium distributions were indeed proven to exist, this would not only demonstrate the need to re-assess the current quantum formalism, but could also generate new phenomena \cite{AV91b,AV08,AV02,PV06}, potentially opening up a large field of investigation.
In recent years some authors have focused their attention upon the prospect of measuring such violations of quantum theory. This work is becoming well developed in some areas, for instance with regard to measurable effects upon the cosmic microwave background \cite{AV10,CV13,CV15,CV16}. 
In other areas, for instance regarding nonequilibrium signatures in the spectra of relic particles \cite{UV15,UV16}, the literature is still in the early stages of development.

Other authors have instead focused upon the de Broglie trajectories and what this might tell us about the process of quantum relaxation.
In this regard, some statements may be made with certainty.
De Broglie trajectories tend to be chaotic in character \cite{F97,WP05,EC06,EKC07,TCE16,WS99}, and this is generally \cite{F97,WS99,WPB06}, but not always \cite{SC12}, attributed to the nodes of the wave function.
This chaos, in turn, is generally seen as the driving factor in quantum relaxation \cite{WPB06,TRV12,ECT17}.
Certainly, for relatively small systems with a small superposition and relatively minor deviations from equilibrium, relaxation takes place remarkably quickly \cite{VW05,TRV12,CS10}. The speed of the relaxation appears to scale with the complexity of the superposition \cite{TRV12}, whilst not occurring at all for non-degenerate energy states. 
Of course, the particles accessible to common experimentation have had a long and turbulent astrophysical past, and hence will have had ample time to relax. Consequently, it is not unreasonable to expect a high degree of conformance to standard quantum probabilities in everyday experiments.
This of course lends credence to the notion of quantum nonequilibrium despite the apparent lack of experimental evidence to date.
That said, there are still many open questions regarding the generality of quantum relaxation, and for sufficiently minimally interacting systems it may be possible that there are some windows of opportunity for quantum nonequilibrium to persist. The intention of this work is to address two such questions. These may be phrased as follows.\\
\\
\emph{If a system is very out-of-equilibrium, how does this affect the relaxation time?}\\
To date, every study of quantum relaxation has been concerned only with nonequilibrium distributions that begin within the bulk of the standard quantum probability distribution (variously referred to as $|\psi|^2$, the Born distribution or quantum equilibrium). There is however no \emph{a priori} reason why the initial conditions shouldn't specify a distribution that is far from equilibrium. (We refer to nonequilibrium distributions that are appreciably separated from the bulk of the Born distribution as `extreme' quantum nonequilibrium.) In the early development of the theory the main priority was to prove the validity of the process of relaxation, and so it was natural to choose initial distributions that had the possibility of relaxing quickly enough to be computationally tractable. When considering the possibility of quantum nonequilibrium surviving to this day however, this choice may appear as something of a selection bias. Certainly, the chaotic nature of the trajectories can make longer term simulations impractical, and so it is not entirely surprising that such situations have yet to make it into the literature. For the system discussed here (the 2 dimensional isotropic oscillator), state of the art calculations (for instance \cite{ACV14,KV16}) achieve evolution timescales of approximately $10^2$ to $10^4$ wave function periods and manage this only with considerable computational expense. A new numerical methodology described in section \ref{sec:drift} (construction of the `drift field') allows this computational bottleneck to be avoided, providing the means to describe the behaviour of systems which relax on timescales that may be longer by many orders of magnitude. \\
\\
\\
\emph{Is it possible that for some systems, relaxation does not take place at all?}\\
As has already been mentioned, relaxation does not occur for non-degenerate energy eigenstates. This is widely known and follows trivially as a result of the velocity field vanishing for such states. It is not yet clear, however, whether relaxation may be frozen or impeded more commonly and for more general states\footnote{Of course, since de Broglie-Bohm is time-reversible it is always possible to contrive initial conditions such that distributions appear to fall out of equilibrium. (A so-called `conspiracy' of initial conditions leading to `unlikely' entropy-decreasing behaviour \cite{AV91a}.) Indeed, since the cause of quantum relaxation is in essence the same as the cause of classical statistical relaxation (see chapter \ref{2}), this is also true in classical mechanics. The possibility of such behaviour is excluded from the discussion for the usual reasons (see for instance reference \cite{Daviesbook}).}. Some authors have recently made some inroads regarding this question. One recent study \cite{KV16} found that for states that were perturbatively close to the ground state, trajectories may be confined to individual subregions of the state space, a seeming barrier to total relaxation. Another found that for some sufficiently simple systems there could remain a residual nonequilibrium that is unable to decay \cite{ACV14}. In this paper, a number of points are added to this discussion. Firstly, a general mechanism is identified that causes distributions that are separated from the bulk of the Born distribution to migrate into the bulk, a necessary precursor to relaxation for such systems. Also identified are states for which this mechanism is conspicuously absent. It is argued that for such states, the relaxation of systems separated from the bulk of $|\psi|^2$ will be at least severely retarded and possibly may not take place at all. It is shown that quantum states may be cleanly categorised according to the vorticity of their velocity field and that this vorticity may be calculated from the state parameters. With supporting numerical evidence, it is conjectured that states belonging to one of these categories always feature the relaxation mechanism, and also that for states belonging to another of the categories the mechanism is always absent. This provides a method by which states that feature either efficient or retarded (possibly absent) relaxation may be calculated. 

The article is structured as follows. Section \ref{sec:nodes} provides a description of the quantum system under investigation and provides an explanation of the role of nodes in quantum relaxation. A number of previously unknown properties of nodes are described which are valid under two general assumptions. These properties allow the categorisation of quantum states by their total vorticity. In section \ref{sec:drift} the `drift field' is introduced, its construction is outlined, and categories of drift field are described. An account is given of the mechanism by which `extreme' quantum nonequilibrium relaxes and it is argued that states that do not feature this mechanism will in the least exhibit severely retarded relaxation. A systematic study of 400 randomised quantum states is described which allows a link to be made between vorticity categories of section \ref{sec:nodes} and the categories of drift fields. These links are then stated as formal conjectures and supported with data from a further 6000 quantum states.
Finally, in section \ref{sec:conclusion} the main results of the investigation are summarised and implications for the prospect of quantum nonequilibrium surviving to this day are commented upon, with a view to future research directions.

\section{The system, its nodes and the vorticity theorem}\label{sec:nodes}
In studies of quantum relaxation, the two-dimensional isotropic harmonic oscillator is often chosen as the subject of investigation.
This is partly as it may be shown to be mathematically equivalent to a single uncoupled mode of a real scalar field \cite{AV07,AV08,ACV14,CV15,UV15,KV16}, granting it physical significance in studies concerning cosmological inflation scenarios or relaxation in high energy phenomena (relevant to potential avenues for experimental discovery of quantum nonequilibrium).
From a more practical perspective, two-dimensional systems also lend themselves well to plotting, allowing a more illustrative description to be made.
There is yet to be a systematic study of relaxation in many dimensions, and such a study would certainly add to our current understanding of quantum relaxation. That said, one primary intention of this work is to introduce a method which is tailored to provide an effective description of more `extreme' forms of nonequilibrium, and for this purpose it is useful to have other works with which to draw comparison. 
For these reasons, the two-dimensional isotropic harmonic oscillator is the subject of this investigation.
It has been useful to depart from convention, however, by working primarily with an `angular' basis of states that is built, not from the energy states of two one-dimensional oscillators (what will be referred to as a Cartesian basis), but from states that are simultaneous eigenstates of total energy and angular momentum. This is for a number of reasons. Principally, many of the results described in this section have been found as a result of using this basis. Also, as described in section \ref{sec:drift}, the long-term drift of trajectories decomposes well into radial and angular components, the size of which may differ by orders of magnitude. The angular basis is naturally described in terms of polar coordinates, and so its use helps to reduce numerical errors. Finally, the angular basis is arguably more natural as it takes advantage of the symmetry of the system. There are many sets of Cartesian bases but only one angular basis. Adoption of a Cartesian basis has lead to the study of states that could be considered `unusual' or `finely-tuned'. For instance, in \cite{ACV14} the authors used the notation $M=4$ and $M=25$ to denote the `first 4' and the `first 25' energy states. By $M=4$ it was meant that the state had $\psi_{00}$, $\psi_{10}$, $\psi_{01}$ and $\psi_{11}$ components (using the notation below). It is the case, however, that $\psi_{11}$ has the same energy as $\psi_{20}$ and $\psi_{02}$. Furthermore, it may easily be shown that under rotations energy states mix amongst other states of equal energy. A rotation of such a state would therefore supplement the superposition with $\psi_{20}$ and $\psi_{02}$ parts, turning what was referred to as $M=4$ into a superposition of 6 states. Similarly, rotated axes transform what was referred to as $M=25$ into the lowest 45 energy states. It is debatable whether states that have been selected in this manner should or should not be considered finely-tuned. In the course of the present study, however, it was found that such states exhibit quite unusual, indeed pathological behaviour. It is the intention here to keep the discussion general and hence to avoid discussion of such states, which will be the subject of another paper. In order to exclude such `fine-tuning', the following criterion is held throughout this work. If any state may be considered to have been selected using criteria that restricts the Hilbert space in such a manner that the resulting subspace is measure-zero, then this shall be considered fine-tuning and worthy to be disregarded from discussion. For example, superpositions with components that are exactly the same magnitude or that differ in phase by exactly $\pi/2$ are considered to be finely-tuned. Of course, a state in a Hilbert space that is completely unrestricted in this manner would require an infinite number of parameters to specify and make the problem intractable. Hence, it is necessary to make one exception to the rule. This exception is chosen to be an upper limit to the allowable energy in a superposition that, following the notation of \cite{ACV14} will be specified by the symbol $M$. With the stated considerations, however, $M$ should be understood to take only the values 1, 3, 6, 10, 15 etc. Note that $M$ may be unambiguously referred to in this manner as for instance, accounting for degeneracy, a combination of the first 15 energy states remains a combination of the first 15 energy states regardless of the basis used (Cartesian, rotated Cartesian, or angular). 

It is useful to develop the Cartesian description in parallel with the angular description.
The subject of investigation is an isotropic two-dimensional harmonic oscillator of mass $m$ and frequency $\omega$.
The standard Cartesian coordinates $(q_x,q_y)$, radial coordinate $r$, and time $t$, are replaced with dimensionless counterparts $(Q_x,Q_y)$, $\eta$, and $T$ respectively. These are related by
\begin{align}
Q_x=\sqrt{\frac{m\omega}{\hbar}}q_x,\quad
Q_y=\sqrt{\frac{m\omega}{\hbar}}q_y,\quad
\eta=\sqrt{\frac{m\omega}{\hbar}}r,\quad
T=\omega t.
\end{align}
The symbol $\varphi$ is used to denote the anticlockwise angle from the $Q_x$ axis. The partial derivative with respect to dimensionless time $T$ is denoted by the placement of hollow dot $\overset{\circ}{}$ above the subject of the derivative. 
In these coordinates the Schr\"{o}dinger equation for the oscillator is
\begin{align}
\frac{1}{2}\left(\partial_{Q_x}^2+\partial_{Q_y}^2+Q_x^2+Q_y^2\right)\psi&=i\overset{\circ}{\psi},\\
\frac{1}{2}\left(\partial_\eta^2+\eta^{-1}\partial_\eta+\eta^{-2}\partial_\varphi^2+\eta^2\right)\psi&=i\overset{\circ}{\psi}.
\end{align}
From these equations, continuity equations
\begin{align}
0&=\overset{\circ}{|\psi|^2}+\partial_{Q_x}\left[|\psi|^2\text{Im}(\partial_{Q_x}\psi/\psi)\right]+\partial_{Q_y}\left[|\psi|^2\text{Im}(\partial_{Q_y}\psi/\psi)\right],\\
0&=\overset{\circ}{|\psi|^2}+\eta^{-1}\left[\partial_\eta\left(\eta|\psi|^2\text{Im}(\partial_\eta\psi/\psi)\right)+\partial_\varphi\left(|\psi|^2\eta^{-1}\text{Im}(\partial_\varphi\psi/\psi)\right)\right],
\end{align}
may be found in the usual manner, from which follow the de Broglie guidance equations
\begin{align}\label{xy_guidance}
\overset{\circ}{Q_x}=\text{Im}(\partial_{Q_x}\psi/\psi),&\quad\overset{\circ}{Q_y}=\text{Im}(\partial_{Q_y}\psi/\psi),\\
\overset{\circ}{\eta}=\text{Im}(\partial_\eta\psi/\psi),&\quad\overset{\circ}{\varphi}=\eta^{-2}\text{Im}(\partial_\varphi\psi/\psi).\label{angular_guidance}
\end{align}
The expansion in terms of Cartesian basis states $\psi_{n_x n_y}$ is expressed
\begin{align}\label{xy_state}
\psi=\sum_{n_n n_y}D_{n_x n_y}e^{-i(n_x+n_y)T}\psi_{n_x n_y}(Q_x,Q_y),
\end{align}
where the basis states are
\begin{align}
\psi_{n_x n_y}=\frac{1}{\sqrt{n_x!n_y!}}\left(a_x^\dagger\right)^{n_x}\left(a_y^\dagger\right)^{n_y}\psi_{00}=\frac{H_{n_x}(Q_x)H_{n_y}(Q_y)}{\sqrt{2^{n_x}2^{n_y}n_x!n_y!}}\psi_{00},\label{xy_basis}
\end{align}
$\psi_{00}$ is the ground state, and $a_x^\dagger$ and $a_y^\dagger$ denote the raising operators in their respective $Q_x$ and $Q_y$ directions. It will sometimes also be useful to express the complex coefficients $D_{n_x n_y}$ in terms of the real $d_{n_x n_y}$ and $\theta_{n_x n_y}$ which are related as $D_{n_x n_y}=d_{n_x n_y}\exp(i\theta_{n_x n_y})$. 

The angular basis is composed of simultaneous eigenstates of the total energy and angular momentum\footnote{Much of the notation adopted follows the conventions of \cite{CT}. In particular, it may be useful to note that the subscripts of symbols $n_d$ and $n_g$ derive from the French for right and left.}. The basis states $\chi_{n_d n_g}$ may be constructed in an analogous manner to equation \eqref{xy_basis}, as
\begin{align}\label{angular_basis}
\chi_{n_d n_g}=\frac{1}{\sqrt{n_d!n_g!}}(a_d^\dagger)^{n_d}(a_g^\dagger)^{n_g}\chi_{00}=e^{i(n_d-n_g)\varphi}f_{n_d n_g}(\eta)\chi_{00},
\end{align}
where $a_d^\dagger$ and $a_g^\dagger$ are right and left raising operators related to the usual Cartesian raising operators by
\begin{align}
a_d^\dagger=\frac{1}{\sqrt{2}}\left(a_x^\dagger + ia_y^\dagger\right),\quad a_g^\dagger=\frac{1}{\sqrt{2}}\left(a_x^\dagger-ia_y^\dagger\right).
\end{align}
The $f_{n_d n_g}(\eta)$ are polynomials of order $n_g+n_d$ in radial coordinate $\eta$, the first 15 of which are
\begin{align}\label{eq:f}
&f_{00}=1,\quad f_{10}=f_{01}=\eta,\quad f_{20}=f_{02}=\frac{1}{\sqrt{2}}\eta^2, \quad f_{11}=\eta^2-1,\nonumber\\
&f_{30}=f_{03}=\frac{1}{\sqrt{6}}\eta^3,\quad f_{21}=f_{12}=\frac{1}{\sqrt{2}}\left(\eta^3-2\eta\right),\\
&\resizebox{\textwidth}{!}{$
f_{40}=f_{04}=\frac{1}{\sqrt{24}}\eta^4,\quad\!\! f_{31}=f_{13}=\frac{1}{\sqrt{6}}\left(\eta^4-3\eta^2\right),\quad\!\! f_{22}\frac{1}{\sqrt{4}}\left(\eta^4-4\eta^2+2\right).$}\nonumber
\end{align}
The energy eigenvalue of state $\chi_{n_d n_g}$ is proportional to $n_d + n_g$ and its angular momentum eigenvalue is proportional to $n_d -n_g$. The state expansion in the angular basis states is denoted
\begin{align}\label{angular_expansion}
\psi=\sum_{n_d n_g}C_{n_d n_g}e^{-i(n_d+n_g)T}\chi_{n_d n_g},
\end{align}
and it is occasionally useful to express the complex coefficients in the polar form $C_{n_d n_g}=c_{n_d n_g}\exp(i\phi_{n_d n_g})$.
Components for up to $M=15$ states may be translated between bases using
{\footnotesize
\begin{align}\label{eq:coef_transforms}
&D_{00}=C_{00},\,
D_{10}=\sqrt{1/2}C_{10}+\sqrt{1/2}C_{01},\,
D_{01}=i\sqrt{1/2}C_{10}-i\sqrt{1/2}C_{01},\nonumber\\
&D_{20}=\sqrt{1/4}C_{20}+\sqrt{2/4}C_{11}+\sqrt{1/4}C_{02},\,
D_{11}=i\sqrt{1/2}C_{20}-i\sqrt{1/2}C_{02},\nonumber\\
&D_{02}=-\sqrt{1/4}C_{20}+\sqrt{2/4}C_{11}-\sqrt{1/4}C_{02},\nonumber\\
&D_{30}=\sqrt{2/16}C_{30}+\sqrt{6/16}C_{21}+\sqrt{6/16}C_{12}+\sqrt{2/16}C_{03},\nonumber\\
&D_{21}=i\sqrt{6/16}C_{30}+i\sqrt{2/16}C_{21}-i\sqrt{2/16}C_{12}-i\sqrt{6/16}C_{03},\nonumber\\
&D_{12}=-\sqrt{6/16}C_{30}+\sqrt{2/16}C_{21}+\sqrt{2/16}C_{12}-\sqrt{6/16}C_{03},\\
&D_{30}=-i\sqrt{2/16}C_{30}+i\sqrt{6/16}C_{21}-i\sqrt{6/16}C_{12}+i\sqrt{2/16}C_{03},\nonumber\\
&D_{40}=\sqrt{1/16}C_{40}+\sqrt{4/16}C_{31}+\sqrt{6/16}C_{22}+\sqrt{4/16}C_{40}+\sqrt{1/16}C_{04},\nonumber\\
&D_{31}=i\sqrt{4/16}C_{40}+i\sqrt{4/16}C_{31}-i\sqrt{4/16}C_{13}-i\sqrt{4/16}C_{04},\nonumber\\
&D_{22}=-\sqrt{6/16}C_{40}+\sqrt{4/16}C_{22}-\sqrt{6/16}C_{04},\nonumber\\
&D_{13}=-i\sqrt{4/16}C_{40}+i\sqrt{4/16}C_{31}-i\sqrt{4/16}C_{13}+i\sqrt{4/16}C_{04},\nonumber\\
&D_{40}=\sqrt{1/16}C_{40}-\sqrt{4/16}C_{31}+\sqrt{6/16}C_{22}-\sqrt{4/16}C_{40}+\sqrt{1/16}C_{04}.\nonumber
\end{align}}

The nodes of the wave function are often cited as the primary source of chaos in de Broglie trajectories \cite{F97,WS99,EKC07,ECT17}.
This chaos is in turn thought to be one the primary driving factors in quantum relaxation \cite{WPB06,TRV12,ECT17}.
As is described in the next section, the nodes (which are mostly to be found amongst the bulk of the Born distribution) also play an important role in the long-term relaxation of systems that may be far away and exhibit very regular behaviour. Indeed, global properties of the `drift field' described in section \ref{sec:drift} may be inferred from the properties of the nodes. It is therefore useful to know some of these properties. 
For the sake of clarity, these properties are listed, giving a short justification of each. For the sake of brevity and simplicity the system is taken to be the two-dimensional isotropic oscillator with the dual assumptions of no fine-tuning and limited energy expressed earlier, although most if not all of these properties may be easily generalised. (Properties (i) through (vi) are widely known or have been previously noted and are included to aid the reader. To our knowledge, properties (vii) through (xii) are new.) \\
\\
\textbf{(i) Nodes are points.} This follows trivially as $\psi=0$ places two real conditions on the two-dimensional space.\\
\textbf{(ii) The velocity field has zero vorticity away from nodes.} By writing the wave function in the complex exponential form $\psi=|\psi|\exp(iS)$, guidance equations \eqref{xy_guidance} may be seen to define a velocity field, $v=(\overset{\circ}{Q_x},\overset{\circ}{Q_y})=\nabla S$, that is the (two-dimensional) gradient of complex phase $S$. Consequently the vorticity (as the curl of the velocity $\nabla\times v$) must vanish everywhere except on nodes, where $S$ is ill-defined.\\
\textbf{(iii) The vorticity of nodes is `quantised'.} This observation was originally made by Dirac \cite{Dirac31} and is well known to the field \cite{Holl93}. It is also well known in other fields like chemical physics \cite{HCP74,HGB74,W06}. By Stokes' theorem, the line integral of the velocity field around a closed curve $\partial\Sigma$, that defines the boundary of some region $\Sigma$ which does not contain a node, must vanish; $\oint_{\partial\Sigma}v.\mathrm{d}l=\int_{\Sigma}\nabla\times v \mathrm{d}\Sigma=0$.
If on the other hand, the region $\Sigma$ is chosen such that it contains a node, this need not be the case.
The single-valuedness of $\psi$, however, assures that along any closed path, the phase $S$ can only change by some integer number of $2\pi$ from its initial value, i.e.
\begin{align}\label{eq:v_around_node}
\oint_{\partial \Sigma}v.\mathrm{d}l=\oint_{\partial \Sigma}\mathrm{d}S=2\pi n,
\end{align}
for integer $n$.
In this manner, it is said that the vorticity is `quantised'. Note that $\nabla\times v$ is actually ill-defined on the node and that the vorticity $\mathcal{V}$ of a node should instead be understood to refer to $\oint_{\partial \Sigma}v.\mathrm{d}l$ as evaluated around a closed path enclosing only the node in question.\\
\textbf{(iv) Nodes generate chaos.} It is possible to go into much detail on this point. For the intricacies regarding this process, articles \cite{EKC07,ECT17} and references therein are recommended. It is possible, however, to illustrate this point in the following simple way. 
 By taking the region $\Sigma$ in equation \eqref{eq:v_around_node} to be a ball $B_{\rcurs}$ of radius $\rcurs$, centred on some node, it may be concluded that the component of the velocity field $v$ around the node varies in proportion to $1/\rcurs$.
In other words, a trajectory that approaches close to a node will tend to circle the node with a speed that is strongly dependent upon how close the trajectory manages to get. Two trajectories that are initially near to each other and come close to a node may only differ by a small amount in their approach, but due to the $1/\rcurs$ dependence, this small difference may cause a significant difference in how the trajectories circle the node. The trajectories will likely be scattered in completely different directions.
This `butterfly effect' may lead to highly erratic trajectories and is commonly thought to be one of the primary causes of chaos in de Broglie trajectories.\\
\textbf{(v) Nodes pair create/annihilate.} This is explained for the three dimensional case in \cite{H77}. For our purposes, the condition of a node ($\psi=0$) may be regarded as two separate real conditions upon the 2-dimensional space. Each of these conditions defines a (plane algebraic) curve that evolves with time with nodes appearing where the curves intersect. Two curves that are not initially intersecting may begin to do so, creating a node. In an open system however, this must take place either at infinity (as for instance must be the case for two straight lines), or if the curves begin to intersect at a finite coordinate, they must necessarily intersect twice. This creates the appearance of the pair-creation of nodes. Similarly, two curves may cease to intersect, creating the appearance of annihilation. (Consider for instance an ellipse whose path crosses a static straight line. Upon initial contact, two nodes are created. These then annihilate when the ellipse completes its passage through the line.)\\
\textbf{(vi) Vorticity is locally conserved.} This was noted by \cite{WPB06} although to our knowledge there is yet to be an explanation of why this is the case. By local conservation it is meant that pair creation and annihilation of nodes may only take place between two nodes of opposite vorticity. Consider for instance the line integral $\oint_{\partial\Sigma}v.\mathrm{d}l=\oint_{\partial\Sigma}\mathrm{d}S$ around a region $\Sigma$ the instant before a pair of nodes is created. It is expected that $S$ is smooth in both space and time everywhere but on the node, and hence immediately after the pair creation there cannot be a jump to one of the quantised values of $2\pi n$ allowed by equation \eqref{eq:v_around_node}. Hence, it must be the case that the nodes created are equal and opposite in their vorticity. It follows that a local conservation law must hold, but as nodes may appear from or disappear to infinity in a finite time, this does not guarantee global conservation.\\
\textbf{(vii) Nodes have $\pm 2\pi$ vorticity.}
This stronger condition than property (iii) may be arrived at with the no fine-tuning assumption.
Suppose there exists a node in $\psi$ at coordinate $Q_{x_0},Q_{y_0}$. Since $\psi$ is analytic (in the real analysis sense), at any moment in time and in some small region around the node, it may be Taylor expanded
\begin{align}
\psi(Q_x,Q_y)=&a_x(Q_x-Q_{x_0})+a_y(Q_y-Q_{y_0})\\&+\mathcal{O}\left[(Q_x-Q_{x_0})^2,(Q_y-Q_{y_0})^2\right]\nonumber\\
=&a_x\epsilon\cos\theta +a_y\epsilon\sin\theta+\mathcal{O}(\epsilon^2),
\end{align}
where $a_x$ and $a_y$ are complex constants that depend upon the state parameters, the coordinate of the node and the time. Vanishing $a_x$ or $a_y$ would either be instantaneous or require fine tuning of the state parameters. It is assumed therefore that $a_x$ and $a_y$ are non-zero. The $\epsilon$ and $\theta$ are polar coordinates centred upon the node. The vorticity of the node is calculated by integrating the change in phase around the edge of a small ball $B_{\epsilon}(Q_{x_0},Q_{y_0})$ centred on the node,
\begin{align}\label{eq:that_one}
\mathcal{V}=\oint_{\partial B_{\epsilon}(Q_{x_0},Q_{y_0})} \mathrm{d}S
=\oint_{\partial B_{\epsilon}(Q_{x_0},Q_{y_0})}\text{Im}\partial_\theta\log(a_x\epsilon\cos\theta+a_y\epsilon\sin\theta)\mathrm{d}\theta.
\end{align}
The radius of the ball $\epsilon$ is taken to be small enough firstly to ignore $\mathcal{O}(\epsilon^2)$ terms, but also to exclude other nodes from the interior of the ball. The factor $\epsilon$ then drops out of \eqref{eq:that_one} once the imaginary part is taken. By making the substitution $z=a_x\cos\theta+a_y\sin\theta$ the vorticity may be written
\begin{align}
\mathcal{V}=\text{Im}\oint_{\gamma(\theta)}\frac{\mathrm{d}z}{z}=\pm 2\pi,
\end{align}
where the final equality is found as the contour $\gamma(\theta)=a_x\cos\theta+a_y\sin\theta$ winds around $z=0$ once. The sign of the vorticity is determined by whether the path $\gamma$ is clockwise or anticlockwise, which may be found to be equal to the sign of $\sin(\text{Arg}(a_x)-\text{Arg}(a_y))$. Note that states which are finely-tuned in the manner described may not satisfy this property. For instance, the exact angular momentum state $\chi_{3\,0}$ has a pole of vorticity $6\pi$ at the origin.\\
\textbf{(viii) Vorticity is globally conserved.}
The total vorticity of a state may be written
\begin{align}\label{A1}
\mathcal{V}_\text{tot}=\lim_{\eta\rightarrow\infty}\oint_{\partial B_{\eta}(0)}\mathrm{d}S=\lim_{\eta\rightarrow\infty}\int_{0}^{2\pi}\frac{\partial S}{\partial \varphi}\,\mathrm{d}\varphi=\lim_{\eta\rightarrow\infty}\text{Im}\int_{0}^{2\pi}\frac{\partial_\varphi\psi}{\psi}\,\mathrm{d}\varphi.
\end{align}
By substituting \eqref{angular_basis} into \eqref{angular_expansion}, an arbitrary state may be written
\begin{align}\label{A2}
\psi&= \sum_{n_d n_g}C_{n_d n_g}\exp\left[-i(n_d+n_g)T+i(n_d-n_g)\varphi\right]f_{n_dn_g}(\eta)\chi_{00}(\eta),
\end{align}
and hence the integrand in \eqref{A1} becomes
\begin{align}\label{eq:dpsi}
&\resizebox{\textwidth}{!}{$
\frac{\partial_\varphi\psi}{\psi}=\frac{\sum_{n_d n_g}i(n_d-n_g)C_{n_d n_g}\exp\left[-i(n_d+n_g)T+i(n_d-n_g)\varphi\right]f_{n_dn_g}(\eta)}{\sum_{n_d n_g}C_{n_d n_g}\exp\left[-i(n_d+n_g)T+i(n_d-n_g)\varphi\right]f_{n_dn_g}(\eta)}.$}
\end{align}
The polynomials $f_{n_d n_g}(\eta)$ are of order $n_d+n_g$ and so in the limit $\eta\rightarrow\infty$, only the highest order terms in the numerator and denominator expansions will contribute. Since these terms have identical time-phase factors $\exp[-i(n_d+n_g)T]$, these cancel leaving the integrand, and hence the total vorticity, time independent.\\
\textbf{(ix) States of total vorticity $\mathcal{V}_\text{tot}$ have a minimum of $\mathcal{V}_\text{tot}/2\pi$ nodes.} This follows from properties (vii) and (viii). Of course, pair creation allows there to be more than this number.\\
\textbf{(x) The `vorticity theorem' provides a simple method by which the total vorticity of a state may be calculated.}
The vorticity theorem may be stated as follows. \\
\emph{Let $\psi$ be a state of the 2-dimensional isotropic oscillator with an expansion \eqref{angular_expansion} that is bounded in energy by $n_d+n_g=m$. Let $f(z)$ be the complex Laurent polynomial,
\begin{align}\label{laurent_polynomial}
f(z)=\sum_{n_d=0}^{m}\frac{C_{n_d\, (m-n_d)}}{\sqrt{n_d!(m-n_d)!}}z^{2n_d-m},
\end{align}
formed from the coefficients corresponding to the states of highest energy in the expansion. If no zeros or poles of $f(z)$ lie on the complex unit circle $\partial B_1(0)$, then the total vorticity $\mathcal{V}_\text{tot}$ of the state $\psi$ is
\begin{align}\label{eq:tot_vort}
\mathcal{V}_\text{tot}&=2\pi \left[\#\text{ of zeros of }f(z)\text{ in }B_{1}(0)-\#\text{ of poles of }f(z)\text{ in }B_{1}(0)\right],
\end{align}
the difference between the number of zeros and poles of $f(z)$ in the unit disk $B_{1}(0)$ multiplied by $2\pi$, where it is understood that the counting should be according to multiplicity.}\\
This may be proven by inserting \eqref{eq:dpsi} into \eqref{A1}, assuming the energy upper bound $n_d+n_g=m$, and taking the limit $\lim_{\eta\rightarrow\infty}$. This results with the expression
\begin{align}
\resizebox{\textwidth}{!}{$
\mathcal{V}_\text{tot}=\nonumber\text{Im}\int_0^{2\pi}\frac{\sum_{n_d=0}^m i(2n_d-m)C_{n_d (m-n_d)}[n_d!(m-n_d)!]^{-\frac{1}{2}}\exp\left[i(2n_d-m)\varphi\right]}{\sum_{n_d=0}^mC_{n_d (m-n_d)}[n_d!(m-n_d)!]^{-\frac{1}{2}}\exp\left[i(2n_d-m)\varphi\right]}\mathrm{d}\varphi,$}\label{A4}
\end{align}
where the factors of $[n_d!(m-n_d)!]^{-\frac{1}{2}}$ come from the leading factors of $f_{n_d n_g}(\eta)$ (see equation \eqref{eq:f}). 
By writing $z=e^{i\varphi}$, this may be expressed as a complex contour integral around the unit circle $\partial B_{1}(0)$,
\begin{align}
&\mathcal{V}_\text{tot}=
\resizebox{0.9\textwidth}{!}{$
\oint_{\partial B_{1}(0)}\frac{\sum_{n_d=0}^m i(2n_d-m)C_{n_d (m-n_d)}[n_d!(m-n_d)!]^{-\frac{1}{2}}z^{2n_d-m}}{\sum_{n_d=0}^mC_{n_d (m-n_d)}[n_d!(m-n_d)!]^{-\frac{1}{2}}z^{2n_d-m}}(-iz^{-1})\mathrm{d}z$}\\
&=\text{Im}\oint_{\partial B_{1}(0)}\frac{f'(z)}{f(z)}\mathrm{d}z,
\end{align}
where $f(z)$ is defined as in \eqref{laurent_polynomial}. The Argument Principle of complex analysis (see for instance p.119 of \cite{Conway}) may then be used to arrive at the final result \eqref{eq:tot_vort}.\\
\textbf{(xi) A state that is bounded in energy by $n_d+n_g=m$ may only take the following $m+1$ possible total vorticities,}
\begin{align}\label{eq:poss_vort}
\mathcal{V}_\text{tot}=-2\pi m,\; -2\pi m+4\pi,\;\dots,\;2\pi m-4\pi,\;2\pi m.
\end{align}
This follows from taking a factor of $z^{-m}$ out of \eqref{laurent_polynomial}. The $z^{-m}$ contributes an $m$th order pole at the origin and the remaining polynomial has the property that if $z$ is a zero then so also is $-z$ prompting a rephrasing of the vorticity theorem as follows.\\
\emph{Let $\psi$ be a state of the 2-dimensional isotropic oscillator with an expansion \eqref{angular_expansion} that has an upper limit in energy, $n_d+n_g=m$. Let $g(z)$ be the complex polynomial,
\begin{align}\label{polynomial}
g(z)=\sum_{n_d=0}^{m}\frac{C_{n_d\, (m-n_d)}}{\sqrt{n_d!(m-n_d)!}}z^{n_d}.
\end{align}
formed from the coefficients corresponding to the states of highest energy in the expansion. If there are no zeros on the unit circle $\partial B_{1}(0)$, the total vorticity $\mathcal{V}_\text{tot}$ of the state $\psi$ is
\begin{align}
\mathcal{V}_\text{tot}&=2\pi \left[2\times\#\text{ of zeros of }g(z)\text{ in }B_{1}(0)-m\right],
\end{align}
where again it is understood that the counting should be according to multiplicity.}\\
An $m$th order complex polynomial has $m$ roots, each of which may or may not be inside the unit ball. Accounting for every possible case, it may be concluded that a state of maximum energy $m$ may have the $m+1$ possible total vorticities expressed in equation \eqref{eq:poss_vort}.\\
\textbf{(xii) A state that is bounded in energy by $n_d+n_g=m$ has a maximum of $m^2$ nodes.}
Since $\psi_{00}$ is everywhere non-zero, the condition of a node may be written
\begin{align}
\sum_{n_n n_y}D_{n_x n_y}e^{-i(n_x+n_y)T}\frac{H_{n_x}(Q_x)H_{n_y}(Q_y)}{\sqrt{2^{n_x}2^{n_y}n_x!n_y!}}=0,
\end{align}
the real and imaginary parts of which define plane algebraic curves of degree $m$. B\'{e}zout's theorem (originally to be found in Newton's Principia) states that the maximum number of times two such curves may intersect is given by the product of their degrees, in this case $m^2$.

\section{The drift field and its structure}\label{sec:drift}
The chaos generated in trajectories close to the nodes should be understood to be in contrast to the regular behaviour of trajectories that are far from the nodes \cite{EKC07,ECT17}. Substitution of state \eqref{angular_expansion} and bases \eqref{angular_basis} into \eqref{angular_guidance} gives the guidance equations
\begin{align}
&\overset{\circ}{\eta}=\text{Im}\left\{\frac{\sum_{n_d n_g}C_{n_d n_g}\exp\left[-i(n_d+n_g)T+i(n_d-n_g)\varphi\right]\partial_\eta f_{n_d n_g}}{\sum_{n_d n_g}C_{n_d n_g}\exp\left[-i(n_d+n_g)T+i(n_d-n_g)\varphi\right] f_{n_d n_g}}\right\},\\
&
\resizebox{\textwidth}{!}{$
\overset{\circ}{\varphi}= \frac{1}{\eta^2}\text{Im}\left\{\frac{\sum_{n_d n_g}i(n_d-n_g)C_{n_d n_g}\exp\left[-i(n_d+n_g)T+i(n_d-n_g)\varphi\right]f_{n_d n_g}}{\sum_{n_d n_g}C_{n_d n_g}\exp\left[-i(n_d+n_g)T+i(n_d-n_g)\varphi\right] f_{n_d n_g}}\right\}$}.
\end{align}
Accordingly, for large $\eta$ the physical velocity will scale as $\sim \eta^{-1}$ and will vary as $\sim \eta^{-2}$. Away from the bulk of the Born distribution and away from the nodes, one therefore generally expects to find a small and smooth velocity field.
This smoothness may be exploited as follows.
\begin{enumerate}
\item
 For a grid of initial positions, evolve trajectories with high precision through a single wave function period using a standard numerical method. (For this investigation a 5th order Runge-Kutta algorithm with Cash-Karp parameters was used.)
\item Record the final displacement vector from the original position for each grid point.
\item Plot the grid of displacement vectors as a vector field. This is done in figures \ref{fig1}, \ref{fig2} and \ref{fig3}. In these figures it was found to be useful to separate the angular and radial components of the field. This helps to distinguish the radial component from the generally large angular component.
\end{enumerate}
 The resulting `drift field' may be regarded as a time-independent velocity field of sorts, and used to track long-term evolution without the need for large computational resources as would usually be the case.
Of course, the drift field is merely intended to be indicative of long-term, slow drift and will certainly not be a good approximation in regions near to nodes where the motion is chaotic or otherwise quickly varying. Nevertheless one might suspect that, in regions where the field varies slowly with respect to the grid size upon which it is calculated, the drift field may capture the long-term evolution of the system well. 
The drift field has in practice proven very useful in classifying the global properties of long term evolution. 

As the production of such plots is relatively computationally inexpensive, it is possible to compute many such plots in a reasonable time. 
To study the typical behaviour of systems far out from the bulk of $|\psi|^2$ therefore, 400 drift field plots were calculated with 100 plots each for $M=3,6,10$ and 15 states. For the sake of comparison, superposed upon these plots were the numerically calculated trajectories of the corresponding nodes. The state parameters were randomised as follows. Random numbers in the range $[0,1]$ were assigned to all the $C_{n_d n_g}$ present in the state. These were then normalised with the factor $(\sum_{n_d n_g} C_{n_d n_g}^2)^{-1/2}$, before assigning a random complex phase. In all cases, the radial component of the drift field was notably smaller than the angular component, and in practice this can make the radial behaviour difficult to read from the plot. It was therefore useful to plot the radial component of the drift separately to the angular component. It was found that the angular and radial components of the drift fields displayed distinct global structures that could be readily categorized and which mirrored properties of the nodes of the state concerned. \\
\\
\textbf{Drift field structure and types}\\
In all the drift fields plotted, the drift was predominantly angular with only a small radial component. The structure of this dominant angular drift may be used to divide the drift fields into types 0, 1 and 2. Type-0 fields display flows that are entirely clockwise or entirely anti-clockwise. An example of a type-0 field is shown in figure \ref{fig1}. Type-1 fields, as shown in figure \ref{fig2}, feature one central attractive axis and one central repulsive axis which may or may not be perpendicular. Naturally, the prevailing flow is away from the repulsive axis and towards the attractive axis. Finally, type-2 fields are similar to type-1 except with additional axes. Specifically, type-2 fields feature two attractive axes and two repulsive axes. An example of a type-2 field is shown in figure \ref{fig3}. Considering the relative likelihood of each type of drift field for each of the states studied, it is likely that further types may appear for superpositions with components with larger cut-off energies ($m$ values) than studied here.
\begin{figure}
%\begin{minipage}[c]{0.44\textwidth}
\includegraphics[width=0.5\textwidth]{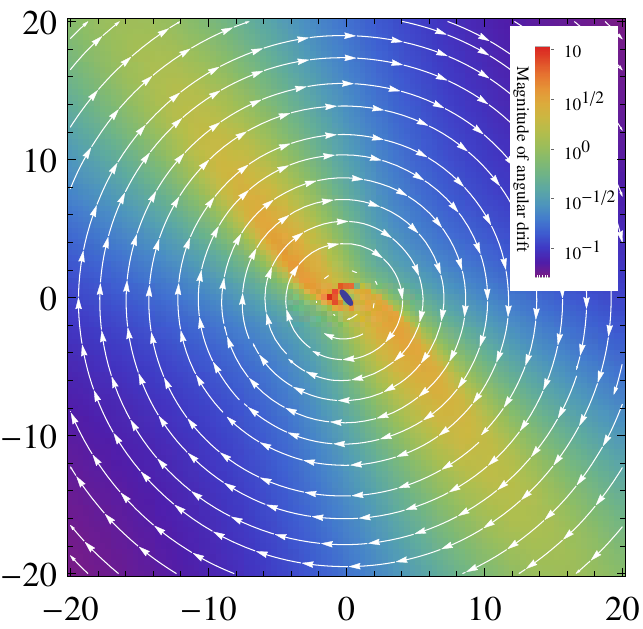}%
\includegraphics[width=0.5\textwidth]{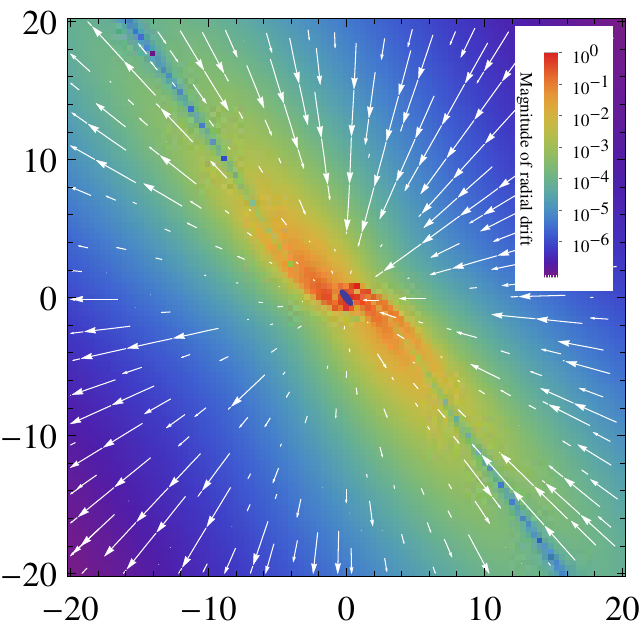}
%\end{minipage}\hfill
%\begin{minipage}[c]{0.56\textwidth}
\caption[A type-0 drift field (relaxation retarding)]{The angular and the radial components of a type-0 `drift field'. The drift field is calculated by computing the displacement of a grid of individual trajectories after one period of the quantum state. Every individual coloured square depicted in the plots represents a data point on the $100\times100$ grid used. Such a drift field may be regarded as a time-independent velocity field of sorts, and may be used to track the long-term evolution of systems in regions where the velocity field has slow spatial variation. In this respect drift fields may be a useful tool for tracking the evolution of systems that are far from the bulk of the Born distribution (which may be considered to occupy a region in the centre of these plots with an approximate radius of 4). The arrows should be understood only to represent drift direction (the length of the arrows is meaningless). Instead, the wide variety in drift magnitude is represented by colour. The small elliptical orbit of the single node is indicated in navy blue. This is a type-0 drift field. In such fields, the predominant angular component is monotonic and the smaller radial component appears to produce equal inwards and outwards flow. A trajectory in such a drift field will circle the Born distribution with a mildly oscillating radius. In contrast to the type-1 and type-2 drift fields shown in figures \ref{fig2} and \ref{fig3}, such type-0 fields do not display a clear mechanism for a trajectory to approach the central region, a necessary precursor to quantum relaxation. Instead it could be the case that the trajectory does not approach the centre at all. If it eventually does reach the centre, it will certainly be significantly retarded with respect to type-1 and type-2 drift fields (examples in figures \ref{fig2} and \ref{fig3}). Hence, for states with type-0 drift fields, quantum relaxation of extreme quantum nonequilibrium may be frozen or in the least severely impeded. } \label{fig1}
%\end{minipage}
\end{figure}

\begin{figure}
%\begin{minipage}[c]{0.60\textwidth}
\includegraphics[width=0.5\textwidth]{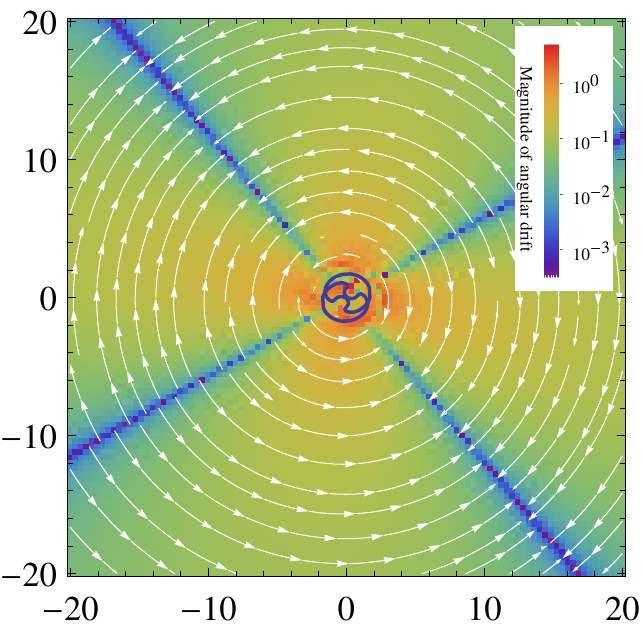}%
\includegraphics[width=0.5\textwidth]{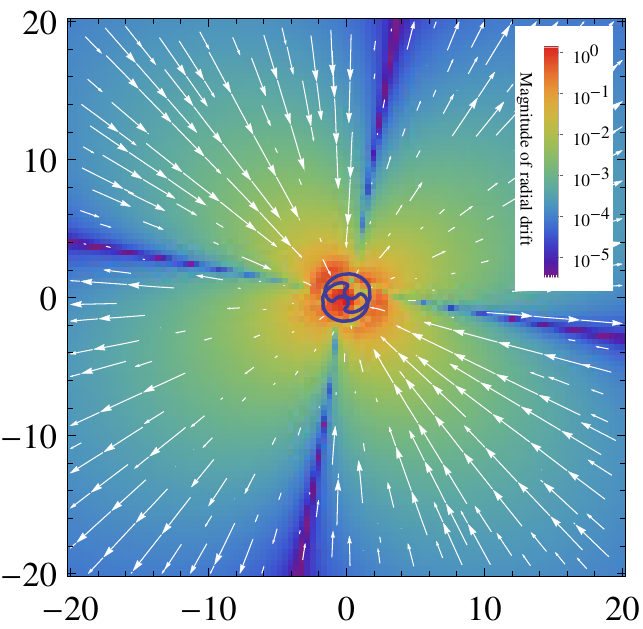}
%\end{minipage}\hfill
%\begin{minipage}[c]{0.40\textwidth}
\caption[A type-1 drift field (non-relaxation retarding)]{An example of the angular and radial components of a type-1 drift field with the nodal paths indicated in navy blue. In type-1 drift fields the dominant angular flow features two axes that may or may not be perpendicular. One of these axes is repulsive and the other attractive. The radial component of the field may be variously divided into sections (slices) that are inwards towards the bulk of the Born distribution and outwards away from it. It always appears to be the case, however, that the axes that are angularly repulsive reside within the sections that are radially outwards, whilst the axes that are angularly attractive reside within the sections that are radially inwards. Such a structure provides a clear mechanism for the drift of systems into the bulk of the Born distribution. Firstly the predominant angular drift will draw any system towards the attractive axis, which for this quantum state runs from top-left to bottom-right. The system will become trapped on this axis, allowing the inwards radial drift to draw it into the central region. In this sense, type-1 drift fields possess a mechanism which enables relaxation even for distributions that are highly separated from quantum equilibrium.}\label{fig2}
%\end{minipage}
\end{figure}

\begin{figure}
\begin{minipage}[c]{0.65\textwidth}
\includegraphics[width=\textwidth]{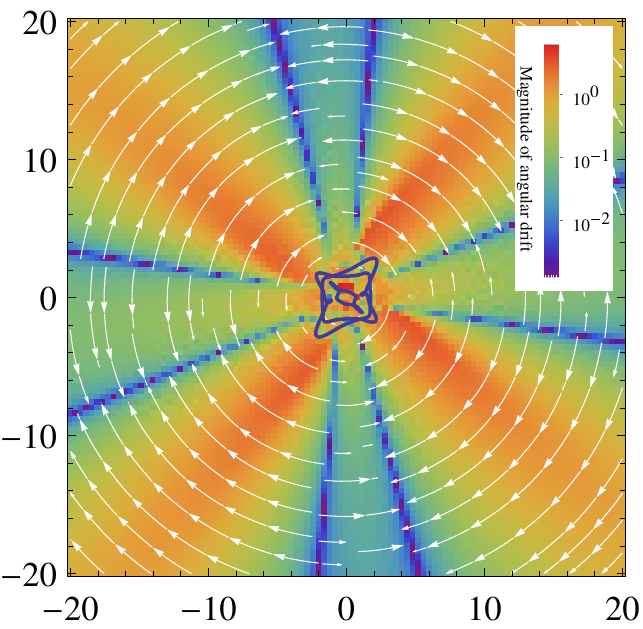}\\
\includegraphics[width=\textwidth]{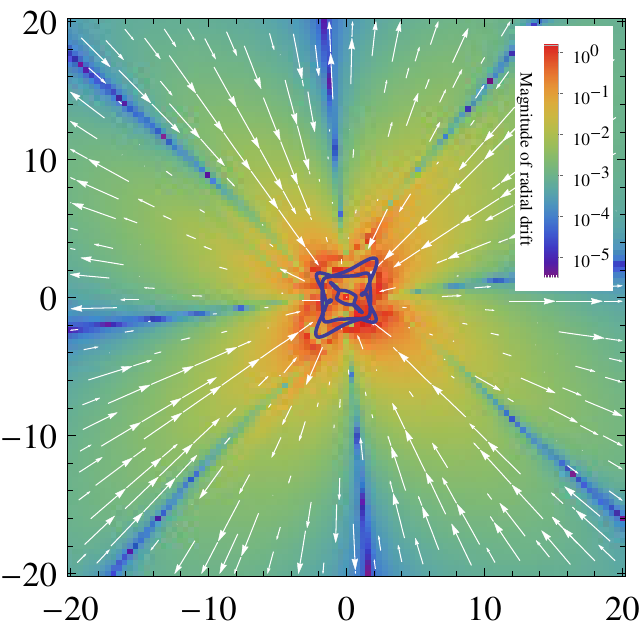}
\end{minipage}\hfill
\begin{minipage}[c]{0.35\textwidth}
\caption[A type-2 drift field (non-relaxation retarding)]{An example of the angular and radial components of a type-2 drift field with the nodal paths indicated in navy blue. Such a field has a similar structure to a type-1 field except with two axes that are angularly attractive and two that are angularly repulsive. In the states studied, it is always the case that angularly repulsive axes align with regions with radially outwards drift, whilst angularly attractive axes align with regions with drift that is radially inwards. This provides a clear mechanism by which trajectories may be drawn into the central region (the Born distribution) regardless of where they are to be found initially. As such, for states with a type-2 drift field it is expected that even extreme forms of quantum nonequilibrium will relax efficiently.}\label{fig3}
\end{minipage}
\end{figure}
In contrast to the angular drift, the radial component of the drift field appears always to be (equally) divided into regions that are radially inwards and regions that are radially outwards. No state tested featured a radial field that was entirely radially inwards or entirely radially outwards. The simplest way in which the space may be divided is by a single dividing axis through the origin as shown in figure \ref{fig1}. In this case it is evident that there is as much of the space with outwards drift as there is with inwards drift. Often it is the case that two or more of these axes divide the space (see figures \ref{fig2} and \ref{fig3}), so that the space is apportioned into wedges that alternate between flow that is radially inwards and flow that is radially outwards. In the case of two such axes (as for instance is shown in figure \ref{fig2}), it appears that these axes are always perpendicular, so that again half of the space exhibits inwards flow and the other half outwards flow. More complicated, bulb-like structures may sometimes be seen in the radial drift, especially in larger superpositions. To the eye however, it always appears the case that the space is \emph{equally} divided into inwards and outwards regions, although a mathematical proof of whether this is indeed the case remains to be seen.\\
\\
\textbf{Mechanism for the relaxation of extreme quantum nonequilibrium}\\
Any nonequilibrium distribution that is initially located away from the central region (where the Born distribution is located) cannot be considered relaxed until the vast majority of systems have at least relocated into the central region. 
As discussed in the introduction, it is expected that once this happens relaxation will proceed efficiently.  
It therefore becomes relevant to consider the time it takes an individual trajectory to migrate into the central region (or indeed whether this migration takes place at all).
To this end, type-1 and type-2 drift fields have a clear, identifiable mechanism by which this migration takes place.
Type-0 drift fields do not feature this mechanism. In every type-1 and type-2 field calculated, the angularly repulsive axes appear to coincide with portions of the field with inwards radial drift whilst the angularly attractive axes coincide with regions that are radially inwards. This structure is clearly displayed in figures \ref{fig2} and \ref{fig3}. In this regard, any system that is initially separated from the bulk of the Born distribution will be swept by the dominant angular drift towards one of the angularly attractive axes. As these axes always coincide with a region of inwards radial drift, the trajectory will then be dragged into the central region where the bulk of the Born distribution is found.

In contrast to this, type-0 drift fields produce trajectories that perpetually orbit around the central region. As a trajectory does so it will sample both regions (wedges) with inwards drift and regions with outwards drift. As these appear (at least to the eye) equal in size it is tempting to conclude that, for such trajectories, outwards drift will balance inwards drift. If this were the case then trajectories would not be drawn into the central region and relaxation would not take place. Certainly it is the case that there is a balancing effect between inwards and outwards flow and so for type-0 fields relaxation from outside the central region will be at least significantly retarded if not stopped altogether. Whether or not there exists some delicate imbalance between inwards and outwards flow that eventually produces relaxation remains to be seen. An answer to this question could have important implications for studies considering conjectured quantum nonequilibrium in the early Universe, and so this is clearly a direction for further work.
\begin{figure}
\includegraphics[width=0.5\linewidth]{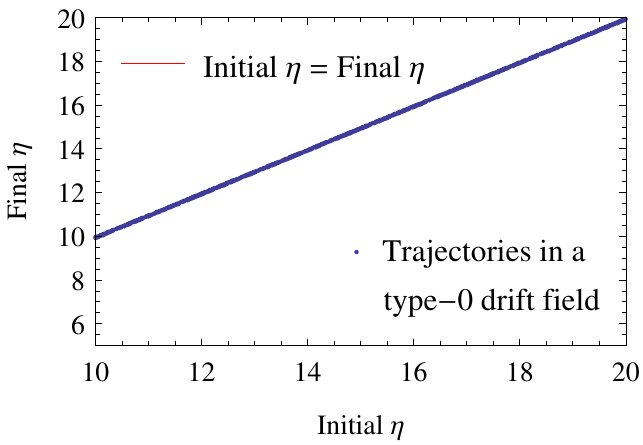}
\includegraphics[width=0.5\linewidth]{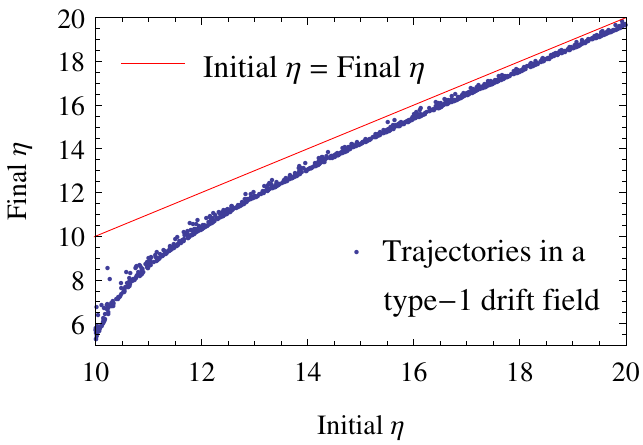}
\caption[Comparison of relaxation after 1000 periods by the relaxation-retarding and non-relaxation-retarding drift fields shown in figures \ref{fig1} and \ref{fig2}.]{A comparison of radial displacement of trajectories after 1000 periods produced by the type-0 and type-1 drift fields shown in figures \ref{fig1} and \ref{fig2}. The type-1 drift field causes a clear net drift inwards that appears to be absent in the type-0 field. Under the influence of a type-1 drift field, an `extreme' nonequilibrium distribution (that is concentrated away from the central region--small $\eta$) will be transported to the bulk of the Born distribution where it is presumed that it will relax efficiently to quantum nonequilibrium. In contrast, the balancing effect on radial drift for type-0 fields (discussed in figure \ref{fig1}) results in no such clear radial displacement. These frames were produced by placing 1000 trajectories randomly in the interval $10<\eta<20$ and numerically evolving for 1000 wave function periods with a standard Runge-Kutta algorithm. }\label{fig4}
\end{figure}

Although a detailed study is left for future work, the effect of the relaxation mechanism (the resulting radial transport of trajectories) is illustrated in figure \ref{fig4}. To produce this figure, 1000 trajectories were chosen randomly in the radial interval $10<\eta<20$, and numerically evolved for 1000 wave function periods with a standard Runge-Kutta algorithm. In the first frame of the figure, the wave function parameters that produced the type-0 drift field displayed in figure \ref{fig1} were used. In the second frame the wave function parameters that produced the type-1 drift field of figure \ref{fig2} were used. In accordance with our statements regarding the mechanism, the type-1 drift field causes a clear inward drift of the trajectories. An `extreme' quantum nonequilibrium ensemble (that is concentrated away from the central region) will therefore be transported towards the Born distribution where, as discussed, it is presumed to relax efficiently. In contrast, the type-0 drift field produces no such radial drift and, after 1000 wave function periods, the trajectories still appear close to their original radial coordinate $\eta$.  \\
\\
\hspace{-\parindent}\textbf{Specific properties of states studied}\\
\textbf{M=3}\\
An $M=3$ state has only $\chi_{00}$, $\chi_{10}$ and $\chi_{01}$ components. (Or equivalently $\psi_{00}$, $\psi_{10}$ and $\psi_{01}$ components.) Since the cut-off energy is $m=1$, by property (xi) the total vorticity of the state may only be $\mathcal{V}_\text{tot}=\pm2\pi$. Hence, there must always exist a single node with $\mathcal{V}=\mathcal{V}_\text{tot}$. Property (xii) ensures there is a maximum of one node and so no pair creation can take place. (The two plane algebraic curves mentioned at the end of section \ref{sec:nodes} are in this case straight lines and so intersect only once.) The sign of the vorticity may be determined with the vorticity theorem. The relevant Laurent polynomial $C_{01}/z+C_{10}z$ has a simple pole at $z=0$ and two simple zeros at $z=\pm\sqrt{-C_{01}/C_{10}}$. Hence, if $|C_{10}|>|C_{01}|$ the vorticity is positive, else it is negative. Clearly the two possibilities appear with equal probability when wave function parameters are chosen at random. For the 100 states generated, the drift field was type-0 in every case. Hence, the relaxation of extreme nonequilibrium may be retarded for every $M=3$ state. The angular flow is in the direction indicated by $\mathcal{V}_\text{tot}$. Note however that due to the way the $\chi_{n_d n_g}$ are defined, $d$ should be understood to refer to positive, anticlockwise vorticity rather than the common notion of rotating to the right implied by the letter $d$.

For the simple $M=3$ case, further properties of the node may be found as follows. By taking the modulus square of the state \eqref{xy_state}, the path of the node may be written $AQ_x^2+BQ_xQ_y+CQ_y^2+DQ_x+EQ_y+F=0$ (which is the general Cartesian form of a conic-section) with the coefficients
\begin{align}
\resizebox{\textwidth}{!}{$
A=d_{10}^2,\enspace B=(D_{10}D^*_{01}+D^*_{10}D_{01}),\enspace C=d_{01}^2,\enspace D=E=0,\enspace F=-\frac12 d_{00}^2.$}
\end{align}
The condition that this path is elliptical is $B^2-4AC<0$, which is the case unless $\theta_{10}$ and $\theta_{01}$ differ by an exact integer number of $\pi$, the limiting case in which the ellipse becomes a straight line. 
That the ellipse is centred upon the origin is implied by the vanishing $D$ and $E$ coefficients.
The nodal trajectory is circular if $B=0$ and $A=C$, which is the case when $\theta_{01}$ and $\theta_{10}$ differ by a half-integer number of $\pi$ and $d_{10}=d_{01}$, or equivalently $D_{10}=\pm i D_{01}$. By translating the angular basis, it may be shown that the ellipse has a semi-minor axis of $c_{00}/(c_{10}+c_{01})$ and a semi-major axis of $c_{00}/|c_{10}-c_{01}|$. (A circular trajectory is found if either $c_{10}$ or $c_{01}$ vanish.) As the closest the node approaches the origin is $c_{00}/(c_{10}+c_{01})$, for a ground state with only perturbative contributions from $\psi_{10}$ and $\psi_{01}$ (as studied by \cite{KV16}), the node stays far from the bulk of the Born distribution and interesting relaxation properties are to be expected. The orientation of the ellipse varies linearly with the difference between $\phi_{01}$ and $\phi_{10}$. The area of the ellipse is $\pi c_{00}^2/|c_{10}^2-c_{01}^2|$. 
The nodal trajectory as a function of time may be expressed
\begin{align}
\resizebox{\textwidth}{!}{$
Q_x(T)=\frac{1}{\sqrt{2}}\frac{d_{00}}{d_{10}}\frac{\sin(\theta_{01}-\theta_{00}-T)}{\sin(\theta_{10}-\theta_{01})},\,\,
Q_y(T)=\frac{1}{\sqrt{2}}\frac{d_{00}}{d_{01}}\frac{\sin(\theta_{10}-\theta_{00}-T)}{\sin(\theta_{01}-\theta_{10})}.$}
\end{align}
\textbf{M=6}\\
By property (xi) the total vorticity may only be $-4\pi$, $0$ or $4\pi$. As an $M=6$ state is bounded in energy by $m=2$, by property (xii) there is a maximum of 4 nodes at any time. In the case of $\mathcal{V}_\text{tot}=\pm 4\pi$ vorticity, there must always exist at least 2 nodes, each with vorticity $\mathcal{V}=\mathcal{V}_\text{tot}/2$, and one pair may create and annihilate in addition to these. All states generated with $\mathcal{V}_\text{tot}=\pm 4\pi$ produced type-0 drift fields. In the case of $\mathcal{V}_\text{tot}=0$, there may be zero nodes, or up to two pairs of opposite vorticity nodes. (The plane algebraic curves upon whose intersections the nodes reside are in this case conic sections which may intersect 0, 2 or 4 times.) The Vorticity 0 states produced type-1 drift fields. If the wave function parameters are chosen at random in the way described, the $\mathcal{V}_\text{tot}=0$ type-1 drift fields appear approximately 66\% of the time whilst the $\mathcal{V}_\text{tot}=\pm 4\pi$ type 0 appear approximately 17\% of the time each. (The relative frequency of states was determined by randomly selecting state parameters for 100,000 states and calculating $\mathcal{V}_\text{tot}$ using the vorticity theorem \eqref{polynomial}.) Hence, states that exhibit retarded relaxation are relatively less common than for $M=3$ states, appearing in 34\% of cases.\\
\\
\textbf{M=10}\\
By property (xi) the total vorticity may be $\pm 2\pi$ or $\pm6\pi$, with the former possibility always featuring at least one node and the latter at least 3 nodes. In both cases property (xii), allows pair creation to increase this number up to a total of 9 nodes. The $\mathcal{V}_\text{tot}=\pm6\pi$ are relatively rare, each appearing in approximately 3\% of randomly selected cases each. All of the 6 randomly generated states that had $\mathcal{V}_\text{tot}=\pm6\pi$ produced type-0 drift fields. States with $\mathcal{V}_\text{tot}=\pm2\pi$ each appear in approximately 47\% of cases. These may exhibit any one of the three types of drift field. Of the 94 fields generated by states with $|\mathcal{V}_\text{tot}|=2\pi$, type-0 drift fields (with retarded relaxation) appeared in 7 cases, whilst type-1 and type-2 fields appeared in 72 and 15 cases respectively.\\
\\
\textbf{M=15}\\
By property (xi) the total vorticity may be 0, $\pm4\pi$ or $\pm8\pi$. States with $\mathcal{V}_\text{tot}=0$ may feature no nodes, but states with $\mathcal{V}_\text{tot}=4\pi$ and $\mathcal{V}_\text{tot}=8\pi$ must always retain at least 2 and 4 nodes respectively. For all possible vorticities, property (xii) and pair creation allow up to a maximum of 25 nodes. As with the $M=10$ superpositions, the states of maximal total vorticity are relatively rare. For $M=15$ superpositions, maximal vorticities $\mathcal{V}_\text{tot}=\pm8\pi$ each appear in approximately 0.4\% of cases, whilst $\mathcal{V}_\text{tot}=\pm4\pi$ each appear in approximately 19\% of cases. The most commonly found vorticity is $\mathcal{V}_\text{tot}=0$, appearing in approximately 61\% of cases. The single maximal vorticity $\mathcal{V}_\text{tot}=\pm8\pi$ state that appeared in the 100 randomly selected states produced a type-0 retarded drift field. Of the 42 random trials that had $|\mathcal{V}_\text{tot}|=4\pi$, 15 had type-0 drift fields. Of the remaining 57 trials that were found to have $\mathcal{V}_\text{tot}=0$, not a single one produced a type-0 drift field.\\
\\
\textbf{Vorticity-drift conjectures}\\
Although proofs have not been forthcoming, we advance the following statements as conjectures and argue their validity primarily on the basis of lack of counterexample.\\
\\
\textbf{Conjecture 1 - A state with zero total vorticity cannot produce a type-0 drift field}\\
For sufficiently large $\eta$, the velocity field of a state with zero total vorticity cannot be uniformly clockwise or uniformly anticlockwise. 
To demonstrate this, consider that for a zero-vorticity state it is the case that $\int_{0}^{2\pi}\frac{\partial S}{\partial\varphi}\mathrm{d}\varphi=0$ for sufficiently large $\eta$. Hence, the integrand $\partial S/\partial\varphi$ (which is proportional to the angular component of the velocity field) must change sign or be trivially zero (in the presence of fine tuning). If, then, the drift field classifications were applied instead to velocity fields, it would be the case that zero total vorticity fields could not produce type-0 velocity fields. To support the validity of this statement when applied to drift fields, 1000 extra zero-vorticity states were generated for each of the $M=6$ and $M=15$ categories. (Property (xi) means $M=3$ and $M=10$ states cannot have zero-vorticity.) (States with maximal or zero total vorticity may be easily calculated by randomly generating many states and checking the vorticity with the vorticity theorem.) All 2000 cases were found to be type-1 or type-2, in support of the conjecture.\\
\\
\textbf{Conjecture 2 - All states of maximal vorticity produce type-0 drift fields}\\
Consistently it was found that all states of maximal vorticity produced type-0 relaxation retarding states. As these states are relatively rare for $M=10$ and $M=15$ categories, however, this conjecture is certainly in need of further substantiation. To provide this, 1000 extra maximal vorticity states were generated for each of the $M=3,6,10,15$ state categories.  All 4000 cases were found to be of type-0. We note that this conjecture provides a convenient method of generating these type-0 relaxation retarding states were one to wish to study their properties. One need simply to randomly generate state parameters and then check that their vorticity is maximal using the vorticity theorem.

\section{Conclusions}
\label{sec:conclusion}
In this paper, consideration has been given to the behaviour of de Broglie trajectories that are separated from the bulk of the Born distribution with a view to describing the quantum relaxation properties of more `extreme' forms of quantum nonequilibrium. The main results are as follows. 
A number of new properties of nodes have been described ((vii) to (xii) in section \ref{sec:nodes}) that are true under the assumption of a quantum state with bounded energy and in the absence of finely-tuned parameters. It is hoped that these results prove useful to the community and that they may be extended and generalised to suit problems other than that considered here. For the 2-dimensional isotropic oscillator (which has physical significance in studies regarding quantum nonequilibrium in the early Universe \cite{AV10,CV13,CV15,CV16,UV15,UV16}), these properties have been shown to have consequences for the quantum relaxation of systems that are separated from the bulk of the Born distribution (what has been referred to as `extreme' quantum nonequilibrium). The relaxation properties of these systems have been classified in terms of the structure of their `drift field', a new concept introduced here. States have been divided into 3 distinct classes. Type-1 and type-2 drift fields have been shown to feature a mechanism that gives rise to efficient relaxation of extreme quantum nonequilibrium. Type-0 fields lack such a relaxation mechanism, appearing instead to sample the drift field in a manner that suggests the prevention of quantum relaxation. Even if it is not the case that relaxation is entirely blocked in such states, it will at least be significantly retarded. Whether or not some delicate imbalance of flow may eventually cause relaxation, and the calculation of relaxation timescales if this is indeed the case, presents a clear avenue for further investigation. Another subject that will be the focus of a future paper concerns the consequences of fine-tuning and perturbations around finely-tuned states. (In the process of this investigation, such states were found to produce highly unusual, pathological behaviour that was judged to be extraneous to the intended focus.)

 For the states studied, it was found that all states of maximal total vorticity $\mathcal{V}_\text{tot}$ (according to their energy bound) produced type-0, relaxation-retarding drift fields. It was also found that no zero-vorticity states produced exhibit type-0 drift. In section \ref{sec:drift} it was formally conjectured that these two correspondences are true for all (non-finely-tuned) states. These conjectures are a central result of this work which, when used with the vorticity theorem (point (x) in section \ref{sec:nodes}), allow the generation of states that should exhibit efficient relaxation as well as those for which relaxation is retarded (or possibly stopped altogether). The relative abundances of these states given random parameters have been discussed. It is hoped that the identification of relaxation retarding states, and the methodology that has otherwise been formulated, may be useful to those investigating the intriguing possibility of discovering quantum nonequilibrium.

\chapter*{PREFACE TO CHAPTER 4}\addcontentsline{toc}{chapter}{PREFACE TO CHAPTER 4}
Experimental evidence of quantum nonequilibrium would clearly be a momentous finding for the foundations of quantum mechanics and physics in general.
The main route through which this possibility has thus far been explored is in the potential of nonequilibrium imprints on the power spectrum of the CMB.
The work done in this regard is summarized at the beginning of chapter \ref{4}, and the most recent work to date may be found in reference \cite{VPV19}.
Chapter \ref{4} introduces a second, more speculative possibility. 
Namely that quantum nonequilibrium could have been preserved for some species of relic particles. 

Two main scenarios are outlined that could conceivably leave relic particles with nonequilibrium statistics today. 
The first is through the decay of a nonequilibrium inflaton field.
(The case for the inflaton field to have existed in nonequilibrium has been previously developed due to its potential to leave an imprint on the CMB sky.)
But the decay of a nonequilibrium inflaton field would also transfer nonequilibrium to its decay products.
And as these decay products are thought to constitute the majority of the matter density in the universe, it seems conceivable that nonequilibrium could yet be stored in some subset of these particles.
As discussed, practical considerations appear to make this less likely. 
For instance, the particles would have had to be created at a time after their corresponding, decoupling time $t_\text{dec}$, else the interactions with other particles would be likely to cause rapid relaxation.
Moreover, in order to be subsequently detectable, such particles would need to decay or annihilate into detectable particles.
And if the decays/annihilations produced photons, say, then in order to avoid subsequent relaxation of the photons, this would need to take place after recombination.
Nevertheless, simple estimates suggest that this is a non-vanishing possibility, and an illustrative scenario involving the gravitino is outlined.
To complement this, a second, simpler scenario is outlined concerning relic nonequilibrium in the vacuum modes of conformally coupled fields.
Vacuum modes are simple enough for quantum nonequilibrium to be conserved trivially. 
And particle physics processes taking place on such a nonequilibrium background would have their statistics altered accordingly.
Particles created on such a background would be expected to pick up traces of the nonequilibrium in the vacuum.
The gravitino field is again adopted as an illustrative candidate for this second scenario.

In order to support the claims made, sections \ref{particle_decay} and \ref{IV} describe the results of some numerical work.
Through a toy model based upon the Jaynes-Cummings model of quantum optics, section \ref{particle_decay} explicitly demonstrates the transfer of nonequilibrium from one quantum field to another during particle decay and creation.
Section \ref{IV} then demonstrates how simple quantum-mechanical measurements performed on a nonequilibrium electromagnetic field may yield incorrect and anomalous spectra.  
This line of argument is picked up and significantly elaborated in chapter \ref{5}.

\chapter{QUANTUM FIELD THEORY OF RELIC NONEQUILIBRIUM SYSTEMS}\label{4}
\vspace{-15mm}
\begin{center}
\textit{Nicolas G. Underwood\hyperlink{address_chap_4}{$^\dagger$} and Antony Valentini\hyperlink{address_chap_4}{$^\dagger$}}
\begin{center}
\textit{Adapted from Phys. Rev. D 92, 063531 (2015)} \cite{UV15} 
\end{center}
\end{center}
\begin{center}
\hypertarget{address_chap_4}{$^\dagger$}Kinard Laboratory, Clemson University, Clemson, 29634, SC, United States of America
\end{center}
\subsection*{Abstract}
In terms of the de Broglie-Bohm pilot-wave formulation of quantum theory, we develop field-theoretical models of quantum nonequilibrium systems which could exist today as relics from the very early Universe. We consider relic excited states generated by inflaton decay, as well as relic vacuum modes, for particle species that decoupled close to the Planck temperature. Simple estimates suggest that, at least in principle, quantum nonequilibrium could survive to the present day for some relic systems. The main focus of this paper is to describe the behaviour of such systems in terms of field theory, with the aim of understanding how relic quantum nonequilibrium might manifest experimentally. We show by explicit calculation that simple perturbative couplings will transfer quantum nonequilibrium from one field to another (for example from the inflaton field to its decay products). We also show that fields in a state of quantum nonequilibrium will generate anomalous spectra for standard energy measurements. Possible connections to current astrophysical observations are briefly addressed.

\section{Introduction}
\label{I}
In the de Broglie-Bohm pilot-wave formulation of quantum theory
\cite{deB28,BV09,B52a,B52b,Holl93}, the Born probability rule has a dynamical
origin \cite{AV91a,AV92,AV01,VW05,EC06,TRV12,SC12} and ordinary quantum
physics is recovered as a special equilibrium case of a wider nonequilibrium
physics
\cite{AV91a,AV91b,AV92,AV96,AV01,AV02,AV07,AV08,AV09,AV10,AVPwtMw,PV06}. On
this view, we may understand the Born rule as arising from a relaxation
process that took place in the remote past. Quantum nonequilibrium -- that is,
violations of the Born rule -- may have existed in the very early Universe
before relaxation took place \cite{AV91a,AV91b,AV92,AV96}. Such effects could
leave observable traces today -- in the cosmic microwave background (CMB)
\cite{AV07,AV08,AV09,AV10,CV13,CV15} or in relic systems that decoupled at
very early times \cite{AV01,AV07,AV08}. The former possibility has been
developed in some detail and comparisons with data are beginning to be made
\cite{AV10,CV13,CV15,VPV19}. The latter possibility is the focus of this paper.

According to our current understanding, the observed temperature anisotropy in
the CMB was ultimately seeded by quantum fluctuations during an inflationary
era \cite{LL00,Muk05,W08,PU09}. Inflationary cosmology then provides us with
an empirical window onto quantum probabilities in the very early Universe. On
an expanding radiation-dominated background, relaxation in pilot-wave theory
can be suppressed at long (super-Hubble)\ wavelengths while proceeding
efficiently at short (sub-Hubble)\ wavelengths
\cite{AV07,AV08,AV10,CV13,CV15,AVbook}. Thus, in a cosmology with a
radiation-dominated pre-inflationary phase \cite{VF82,L82,S82,PK07,WN08}, one
may obtain a large-scale or long-wavelength power deficit in the CMB
\cite{AV07,AV08,AV10,CV13,CV15}. For an appropriate choice of cosmological
parameters, the expected deficit is consistent with the deficit found in data
from the \textit{Planck} satellite \cite{PlanckXV,CV13,CV15}. Whether the
observed deficit is in fact caused by quantum relaxation suppression during a
pre-inflationary era or by some other more conventional effect remains to be seen.

A pilot-wave or de Broglie-Bohm treatment of the early Bunch-Davies vacuum
shows that relaxation to quantum equilibrium does not take place at all during
inflation itself \cite{AV07,AV10}. Thus, if a residual nonequilibrium still
existed at the end of a pre-inflationary era, the inflaton field would carry
traces of that nonequilibrium forward to much later times.\textbf{ }Similarly,
should nonequilibrium be generated during the inflationary era by exotic
gravitational effects at the Planck scale \cite{AV10}, the resulting
departures from the Born rule will be preserved in the inflaton field and
carried forward into the future where they might have an observable effect.

As we shall discuss in this paper, as well as imprinting a power deficit onto
the CMB sky, a nonequilibrium inflaton field would also transfer
nonequilibrium to the particles that are created by inflaton decay. Since such
particles make up almost all of the matter present in our Universe, it seems
conceivable that today there could exist relic particles that are still in a
state of quantum nonequilibrium. We will also consider relic vacuum modes for
other fields (apart from the inflaton) as potential carriers of nonequilibrium
at late times.

These scenarios raise a number of immediate questions. First of all, even if
nonequilibrium relics were created in the early Universe, could the
nonequilibrium survive until late times and be detected today? As we shall
see, simple estimates suggest that (at least in principle) relaxation to
equilibrium could be avoided for some relic systems. A second question that
must be addressed is the demonstration, in pilot-wave field theory, that
perturbative interactions will in general transfer nonequilibrium from one
field to another. This will be shown for a simplified model of quantum field
theory involving just two energy levels for each field. Finally, one must ask what kind of
new phenomena might be generated by relic nonequilibrium systems in an
astrophysical or cosmological context. This opens up a potentially large
domain of investigation. General arguments have already shown that the
quantum-theoretical predictions for single-particle polarisation probabilities
(specifically Malus' law) would be broken for nonequilibrium systems
\cite{AV04a}. In this paper we focus on measurements of energy as a simplified
model of high-energy processes. It will be shown that conventional energy
measurements performed on nonequilibrium systems would generate anomalous
 spectra. We may take this as a broad indication of the kinds of
anomalies that would be seen in particle-physics processes taking place in the
presence of quantum nonequilibrium.

In this paper we are not concerned with the question of practical detection of
relic nonequilibrium. Rather, our intention is to make a case that detection
might be possible at least in principle, and to begin the development of
field-theoretical models of the behaviour of relic nonequilibrium matter.

Generally speaking, even a lowest-order calculation of perturbative processes
in quantum field theory will involve all of the field modes that are present
in the system. While such calculations are in principle possible in de
Broglie-Bohm theory, in practice it would involve integrating trajectories for
an unlimited number of field modes. In this paper, we make a beginning by
confining ourselves to simplified or truncated models of quantum field theory
involving only a small number of field modes. Our models are inspired by
approximations commonly used in quantum optics, where one is often interested
in the dynamics of a single (quantised)\ electromagnetic field mode inside a
cavity. Our main aim is to justify the assertions that underpin our scenarios.
In particular, we wish to show by explicit calculation of examples that
perturbative couplings will in general transfer nonequilibrium from one field
to another, and that nonequilibrium will affect the spectra for
basic particle-physics processes involving measurements of energy. We
emphasise that the calculations presented in this paper are only intended to
be broadly illustrative. The development of more realistic models is left for
future work.

In Section \ref{IIA} we summarise the background to our scenario, and in particular
the justification for why the inflaton field singles itself out as a natural
carrier of primordial quantum nonequilibrium. In Section \ref{IIB} we argue that
inflaton decay can generate particles in a state of quantum nonequilibrium
(induced by nonequilibrium inflaton perturbations and also by the other
nonequilibrium degrees of freedom that can exist in the vacuum), and that such
nonequilibrium could in principle survive to the present day for those decay
particles that were created at times later than the relevant decoupling time.
The gravitino provides a suggestive, or at least illustrative, candidate. In
Section \ref{IIC} we consider a somewhat simpler scenario involving relic
nonequilibrium field modes for the vacuum only. For simplicity we restrict
ourselves to conformally-coupled fields, as these will not be excited by the
spatial expansion. It is argued that super-Hubble vacuum modes that enter the
Hubble radius after the decoupling time for the corresponding particle species
will remain free of interactions and could potentially carry traces of
primordial nonequilibrium to the present day (for sufficiently long comoving
wavelengths). An illustrative example is provided by the massless gravitino.
In Section \ref{IID} we indicate how particle-physics processes would be affected by
a nonequilibrium vacuum.

These preliminary considerations provide motivation for the subsequent
detailed calculations. In Section \ref{particle_decay} we give an example of the perturbative
transfer of nonequilibrium from one field to another, a process that could
play a role in inflaton decay as well as in the decay of relic nonequilibrium
particles generally. In Section \ref{IV} we provide a field-theoretical model of
energy measurements, and we show by detailed calculation of various examples
that nonequilibrium will entail corrections to the energy spectra generated by
high-energy physics processes. Finally, in Section V we present our
conclusions. We briefly address the possible relevance of our scenarios to
current searches for dark matter. We also comment on some practical obstacles
to detecting relic nonequilibrium (even if it exists) and we emphasise the
gaps in our scenarios that need to be filled in future work.

\section{Relic nonequilibrium systems}

\label{II}
In this section we first summarise the background to our scenario and in
particular the role that quantum nonequilibrium might play in the very early
Universe. We then provide some simple arguments suggesting that primordial
violations of the Born rule might survive until much later epochs and perhaps
even to the present day \cite{AVbook}. These arguments motivate the detailed
analysis of nonequilibrium systems provided later in the paper.

\subsection{Nonequilibrium primordial perturbations}
\label{IIA}

In de Broglie-Bohm pilot-wave theory \cite{deB28,BV09,B52a,B52b,Holl93}, a
system has a configuration $q(t)$ whose velocity $\dot{q}\equiv dq/dt$ is
determined by the wave function $\psi(q,t)$. As usual, $\psi$ obeys the
Schr\"{o}dinger equation $i\partial\psi/\partial t=\hat{H}\psi$ (with
$\hbar=1$). For standard Hamiltonians $\dot{q}$ is proportional to the phase
gradient $\operatorname{Im}\left(  \partial_{q}\psi/\psi\right)  $. Quite
generally,
\begin{align}\label{the_first}
\frac{dq}{dt}=\frac{j}{|\psi|^{2}}
\end{align}
where $j=j\left[  \psi\right]  =j(q,t)$ is
the Schr\"{o}dinger current \cite{SV08}. The configuration-space `pilot wave'
$\psi$ guides the motion of an individual system and has no intrinsic
connection with probabilities. For an ensemble with the same wave function we
may consider an arbitrary distribution $\rho(q,t)$ of configurations $q(t)$.
By construction, $\rho(q,t)$ will satisfy the continuity equation
\begin{align}\label{ant_cont}
\frac{\partial\rho}{\partial t}+\partial_{q}\cdot\left(  \rho\dot{q}\right)  =0.
\end{align}
Because $\left\vert \psi\right\vert ^{2}$ obeys the same equation, an initial
`quantum equilibrium' distribution $\rho(q,t_{i})=\left\vert \psi
(q,t_{i})\right\vert ^{2}$ trivially evolves into a final quantum equilibrium
distribution $\rho(q,t)=\left\vert \psi(q,t)\right\vert ^{2}$. In equilibrium
we obtain the Born rule and the usual empirical predictions of quantum theory
\cite{B52a,B52b}. Whereas, for a nonequilibrium ensemble ($\rho(q,t)\neq
\left\vert \psi(q,t)\right\vert ^{2}$), the statistical predictions will
generally differ from those of quantum theory
\cite{AV91a,AV91b,AV92,AV96,AV01,AV02,AV07,AV08,AV09,AV10,AVPwtMw,PV06}.

If they existed, nonequilibrium distributions would generate new phenomena
that lie outside the domain of conventional quantum theory. This new physics
would allow nonlocal signalling \cite{AV91b} -- which is causally consistent
with an underlying preferred foliation of spacetime \cite{AV08a} -- and it
would also allow `subquantum' measurements that violate the uncertainty
principle and other standard quantum constraints \cite{AV02,PV06}.

The equilibrium state $\rho=\left\vert \psi\right\vert ^{2}$ arises from a
dynamical process of relaxation (roughly analogous to thermal relaxation).
This may be quantified by an $H$-function $H=\int dq\ \rho\ln(\rho/\left\vert
\psi\right\vert ^{2})$ \cite{AV91a,AV92,AV01}. Extensive numerical simulations
have shown that when $\psi$ is a superposition of energy eigenstates there is
rapid relaxation $\rho\rightarrow\left\vert \psi\right\vert ^{2}$ (on a
coarse-grained level) \cite{AV92,AV01,VW05,EC06,TRV12,SC12,ACV14}, with an
approximately exponential decay of the coarse-grained $H$-function with time
\cite{VW05,TRV12,ACV14}. In this way, the Born rule arises from a relaxation
process that presumably took place in the very early Universe
\cite{AV91a,AV91b,AV92,AV96}. While ordinary laboratory systems -- which have
a long and violent astrophysical history -- are expected to obey the
equilibrium Born rule to high accuracy, quantum nonequilibrium in the early
Universe can leave an imprint in the CMB \cite{AV07,AV10,CV13,CV15} and
perhaps even survive in relic particles that decoupled at sufficiently early
times \cite{AV01,AV07,AV08}. The latter possibility provides the subject
matter of this paper.

Much of the physics may be illustrated by the dynamics of a massless,
minimally-coupled and real scalar field $\phi$ evolving freely on an expanding
background with line element $d\tau^{2}=dt^{2}-a^{2}d\mathbf{x}^{2}$ (where
$a=a(t)$ is the scale factor and we take $c=1$). Beginning with the classical
Lagrangian density%
\begin{align}
\mathcal{L}=\frac{1}{2}\sqrt{-g}g^{\mu\nu}\partial_{\mu}\phi\partial_{\nu}%
\phi\ ,
\end{align}
where $g_{\mu\nu}$ is the background metric, and working with Fourier
components $\phi_{\mathbf{k}}=\frac{\sqrt{V}}{(2\pi)^{3/2}}\left(
q_{\mathbf{k}1}+iq_{\mathbf{k}2}\right)  $ -- where $V$ is a normalisation
volume and $q_{\mathbf{k}r}$ ($r=1,2$) are real variables -- the field
Hamiltonian becomes a sum $H=\sum_{\mathbf{k}r}H_{\mathbf{k}r}$ where%
\begin{align}
H_{\mathbf{k}r}=\frac{1}{2a^{3}}\pi_{\mathbf{k}r}^{2}+\frac{1}{2}%
ak^{2}q_{\mathbf{k}r}^{2}%
\end{align}
is formally the Hamiltonian of a harmonic oscillator with mass $m=a^{3}$ and
angular frequency $\omega=k/a$. Straightforward quantisation then yields the
Schr\"{o}dinger equation%
\begin{align}
i\frac{\partial\Psi}{\partial t}=\sum_{\mathbf{k}r}\left(  -\frac{1}{2a^{3}%
}\frac{\partial^{2}}{\partial q_{\mathbf{k}r}^{2}}+\frac{1}{2}ak^{2}%
q_{\mathbf{k}r}^{2}\right)  \Psi
\end{align}
for the wave functional $\Psi=\Psi\lbrack q_{\mathbf{k}r},t]$, from which one
may identify the de Broglie guidance equation%
\begin{align}
\frac{dq_{\mathbf{k}r}}{dt}=\frac{1}{a^{3}}\operatorname{Im}\frac{1}{\Psi
}\frac{\partial\Psi}{\partial q_{\mathbf{k}r}}%
\end{align}
for the evolving degrees of freedom $q_{\mathbf{k}r}$ \cite{AV07,AV08,AV10}.
(We have assumed a preferred foliation of spacetime with time function $t$. A
similar construction may be given in any globally-hyperbolic spacetime
\cite{AV04b,AV08a,AVbook}.)

An unentangled mode $\mathbf{k}$ has an independent dynamics with wave
function $\psi_{\mathbf{k}}(q_{\mathbf{k}1},q_{\mathbf{k}2},t)$. The equations
are the same as those for a nonrelativistic two-dimensional harmonic
oscillator with time-dependent mass $m=a^{3}$ and time-dependent angular
frequency $\omega=k/a$. Thus we may discuss relaxation for a single field mode
in terms of relaxation for such an oscillator \cite{AV07,AV08}. It has been
shown that the time evolution is mathematically equivalent to that of a
standard oscillator (with constant mass and constant angular frequency) but
with real time $t$ replaced by a `retarded time' $t_{\mathrm{ret}}(t)$ that
depends on the wavenumber $k$ \cite{CV13}. Thus, in effect, cosmological
relaxation for a single field mode may be discussed in terms of relaxation for
a standard oscillator.

Cosmological relaxation has been studied in detail for the case of a
radiation-dominated expansion, with $a\propto t^{1/2}$ \cite{CV13,CV15}. In
the short-wavelength or sub-Hubble limit, it is found that $t_{\mathrm{ret}%
}(t)\rightarrow t$ and so we obtain the time evolution of a field mode on
Minkowski spacetime, with rapid relaxation taking place for a superposition of
excited states. On the other hand, for long (super-Hubble) wavelengths it is
found that $t_{\mathrm{ret}}(t)<<t$ and so relaxation is
retarded.\footnote{Such retardation may also be described in terms of the mean
displacement of the trajectories \cite{AV08,AVbook}.} Thus, in a cosmology
with a radiation-dominated pre-inflationary era, at the onset of inflation we
may reasonably expect to find relic nonequilibrium at sufficiently large
wavelengths \cite{AV10,CV13,CV15}.

No further relaxation takes place during inflation itself. This has been shown
by calculating the de Broglie-Bohm trajectories of the inflaton field in the
Bunch-Davies vacuum \cite{AV07,AV10}. In terms of conformal time $\eta=-1/Ha$,
the wave functional is simply a product $\Psi\lbrack q_{\mathbf{k}r}%
,\eta]=\prod\limits_{\mathbf{k}r}\psi_{\mathbf{k}r}(q_{\mathbf{k}r},\eta)$ of
contracting Gaussian packets and the trajectories take the simple form
$q_{\mathbf{k}r}(\eta)=q_{\mathbf{k}r}(0)\sqrt{1+k^{2}\eta^{2}}$. The time
evolution of an arbitrary nonequilibrium distribution $\rho_{\mathbf{k}%
r}(q_{\mathbf{k}r},\eta)$ then amounts trivially to the same overall
contraction that occurs for the equilibrium distribution. It follows that the
width of the evolving nonequilibrium distribution remains in a constant ratio
with the width of the evolving equilibrium distribution. Thus the ratio%
\begin{align}
\xi(k)\equiv\frac{\left\langle |\phi_{\mathbf{k}}|^{2}\right\rangle
}{\left\langle |\phi_{\mathbf{k}}|^{2}\right\rangle _{\mathrm{QT}}}
\label{ksi}%
\end{align}
of the nonequilibrium variance $\left\langle |\phi_{\mathbf{k}}|^{2}%
\right\rangle $ to the quantum-theoretical variance $\left\langle
|\phi_{\mathbf{k}}|^{2}\right\rangle _{\mathrm{QT}}$ is preserved in time. Any
relic nonequilibrium ($\xi\neq1$) that exists at the beginning of inflation is
preserved during the inflationary era and is simply transferred to larger
lengthscales as physical wavelengths $\lambda_{\mathrm{phys}}=a\lambda
=a(2\pi/k)$ grow with time.

It follows that incomplete relaxation at long wavelengths during a
pre-inflationary era can change the spectrum of perturbations during inflation
and thus affect the primordial power spectrum for the curvature perturbations
that seed the temperature anisotropy in the CMB. An inflaton perturbation
$\phi_{\mathbf{k}}$ generates a curvature perturbation $\mathcal{R}%
_{\mathbf{k}}\propto\phi_{\mathbf{k}}$ (where $\phi_{\mathbf{k}}$ is evaluated
at a time a few e-folds after the mode exits the Hubble radius) \cite{LL00}.
This in turn generates the observed angular power spectrum%
\begin{align}
C_{l}=\frac{1}{2\pi^{2}}\int_{0}^{\infty}\frac{dk}{k}\ \mathcal{T}%
^{2}(k,l)\mathcal{P}_{\mathcal{R}}(k) \label{Cl2}%
\end{align}
for the CMB, where $\mathcal{T}(k,l)$ is the transfer function and%
\begin{align}
\mathcal{P}_{\mathcal{R}}(k)\equiv\frac{4\pi k^{3}}{V}\left\langle \left\vert
\mathcal{R}_{\mathbf{k}}\right\vert ^{2}\right\rangle \label{PPS}%
\end{align}
is the primordial power spectrum. From (\ref{ksi}) we have%
\begin{align}
\mathcal{P}_{\mathcal{R}}(k)=\mathcal{P}_{\mathcal{R}}^{\mathrm{QT}}(k)\xi(k)
\label{xi2}%
\end{align}
where $\mathcal{P}_{\mathcal{R}}^{\mathrm{QT}}(k)$ is the quantum-theoretical
or equilibrium power spectrum. Thus measurements of $C_{l}$ may be used to set
experimental limits on $\xi(k)$ \cite{AV10}.

The function $\xi(k)$ quantifies the degree of nonequilibrium as a function of
$k$. In a model with a pre-inflationary era, extensive numerical simulations
show that $\xi(k)$ is expected to take the form of an inverse-tangent -- with
$\xi<1$ for small $k$ and $\xi\simeq1$ at large $k$ \cite{CV15}. The extent to
which this prediction is supported by the data is currently under study
\cite{VPV19}.

Incomplete relaxation in the past is one means by which nonequilibrium could
exist in the inflationary era. Another possibility is that nonequilibrium is
\textit{generated} during the inflationary phase by exotic gravitational
effects at the Planck scale (ref.\ \cite{AV10}, section IVB). Trans-Planckian
modes -- that is, modes that originally had sub-Planckian physical wavelengths
-- may well contribute to the observable part of the inflationary spectrum
\cite{BM01,MB01}, in which case inflation provides an empirical window onto
physics at the Planck scale \cite{BM13}. It has been suggested that quantum
equilibrium might be gravitationally unstable \cite{AV04b,AV07}. In quantum
field theory the existence of an equilibrium state arguably requires a
background spacetime that is globally hyperbolic, in which case nonequilibrium
could be generated by the formation and evaporation of a black hole (a
proposal that is also motivated by the controversial question of information
loss) \cite{AV04b,AV07}. A heuristic picture of the formation and evaporation
of microscopic black holes then suggests that quantum nonequilibrium will be
generated at the Planck length $l_{\mathrm{P}}$. Such a process could be
modelled in terms of nonequilibrium field modes. Thus, a mode that begins with
a physical wavelength $\lambda_{\mathrm{phys}}<l_{\mathrm{P}}$ in the early
inflationary era may be assumed to be out of equilibrium upon exiting the
Planckian regime (that is, when $\lambda_{\mathrm{phys}}>l_{\mathrm{P}}$)
\cite{AV10}. If such processes exist, the inflaton field will carry quantum
nonequilibrium at \textit{short} wavelengths (below some comoving cutoff).

For our present purpose, the main conclusion to draw is that the inflaton
field may act as a carrier of primordial nonequilibrium -- whether it is relic
nonequilibrium from incomplete relaxation during a pre-inflationary era, or
nonequilibrium that was generated by Planck-scale effects during inflation
itself. This brings us to the question: in addition to leaving a macroscopic
imprint on the CMB, could primordial nonequilibrium survive all the way up to
the present and be found in microscopic relic systems today?

\subsection{Inflaton decay}
\label{IIB}

Post-inflation, the density of any relic particles (nonequilibrium or
otherwise) from a pre-inflationary era will be so diluted as to be completely
undetectable today. However, one may consider relic particles that were
created at the end of inflation by the decay of the inflaton field itself --
where in standard inflationary scenarios inflaton decay is in fact the source
of almost all the matter present in our Universe.

To discuss this, note that in pilot-wave theory it is standard to describe
bosonic fields in terms of evolving field configurations (as in our treatment
of the free scalar field in Section \ref{IIA}) whereas there are different
approaches for fermionic fields. Arguably the most straightforward pilot-wave
theory of fermions utilises a Dirac sea picture of particle trajectories
determined by a pilot wave that obeys the many-body Dirac equation
\cite{BH93,C03,CS07}. Alternatively, a formal field theory based on
anticommuting Grassmann fields may be written down \cite{AV92,AV96} but its
interpretation presents problems that remain to be addressed \cite{S10}. For
our purposes we will assume the Dirac sea model for fermions.

During the inflationary era the inflaton field $\varphi$ is approximately
homogeneous and may be written as%
\begin{align}
\varphi(\mathbf{x},t)=\phi_{0}(t)+\phi(\mathbf{x},t)\ ,
\end{align}
where $\phi_{0}(t)$ is a homogeneous part and $\phi(\mathbf{x},t)$ (or
$\phi_{\mathbf{k}}(t)$) is a small perturbation. As we have noted, during the
inflationary expansion perturbations $\phi_{\mathbf{k}}$ do not relax to
quantum equilibrium and in fact the exponential expansion of space transfers
any nonequilibrium that may exist from microscopic to macroscopic
lengthscales. The inflaton field is then a natural candidate for a carrier of
primordial quantum nonequilibrium (whatever its ultimate origin).

The process of `preheating' is driven by the homogeneous and essentially
classical part $\phi_{0}(t)$ (that is, by the $k=0$ mode of the inflaton
field) \cite{BTW06}. The inflaton is treated as a classical external field,
acting on other (quantum) fields which become excited by parametric resonance.
Because of the classicality of the relevant part of the inflaton field, this
process is unlikely to result in a transference of nonequilibrium from the
inflaton to the created particles. During `reheating', however, perturbative
decay of the inflaton can occur, and we expect that nonequilibrium in the
inflaton field will be at least to some extent transferred to its decay products.

Note that we follow the standard procedure of treating the large homogeneous
part $\phi_{0}(t)$ as a classical field and the small perturbation
$\phi(\mathbf{x},t)$ as a quantised field. This deserves some comment. In the
context of preheating, it has been argued that $\phi_{0}(t)$ arises from a
coherent state with a space-independent quantum expectation value \cite{TB90}. It
is also common to argue that the large amplitude and large occupation number
of the `zero mode' at the end of inflation justifies it being treated as a
classical field (see for example refs. \cite{BTW06} and \cite{ABCM10}).
Here we assume the standard formalism, albeit
rewritten in de Broglie-Bohm form. By construction, then, there is no
probability distribution for $\phi_{0}(t)$ (which has a classical `known'
value at all times). Whereas $\phi(\mathbf{x},t)$ has a probability
distribution, which in the standard theory is given by the Born rule and which
in de Broglie-Bohm theory can be more general. The probability distribution
for $\phi(\mathbf{x},t)$ is used to calculate the power spectrum emerging from
the inflationary vacuum. By allowing this distribution to be out of
equilibrium, new physical effects can occur in the CMB \cite{AV10}.
In contrast, because $\phi_{0}(t)$ is
treated as a classical background field with no probability distribution there
is no question of ascribing equilibrium or nonequilibrium to this part of the
field (at least not at the level of the effective description which we adopt here).\footnote{
Note that, in the standard formalism being assumed here, even at very long
wavelengths there remains a formal distinction between the large classical
homogeneous field $\phi_{0}(t)$ and modes of the small quantised field
$\phi(\mathbf{x},t)$.}

The perturbative decay of the inflaton occurs through local interactions. For
example, reheating can occur if the inflaton field $\varphi$ is coupled to a
bosonic field $\Phi$ and a fermionic field $\psi$ via an interaction
Hamiltonian density of the form%
\begin{align}
\mathcal{H}_{\mathrm{int}}=a\varphi\Phi^{2}+b\varphi\bar{\psi}\psi\ ,
\label{intn}%
\end{align}
where $a$, $b$ are constants (ref.\ \cite{PU09}, pp. 507--510). In actual
calculations, it is usual to consider only the dominant homogeneous part
$\phi_{0}$ of the field $\varphi=\phi_{0}+\phi$, and to ignore contributions
from the small perturbation $\phi$. Because the dominant homogeneous part
$\phi_{0}$ is treated essentially classically, inflaton decay bears some
resemblance to the process of pair creation by a strong classical electric
field. 

The decay particles will have physical wavelengths no greater than the
instantaneous Hubble radius, $\lambda_{\mathrm{phys}}\lesssim H^{-1}$, since
local processes cannot significantly excite super-Hubble modes (for which the
particle concept is in any case ill-defined). This standard argument -- that
super-Hubble modes are shielded from the
effects of local interactions -- is still valid in the de Broglie-Bohm
formulation since we are speaking of the time evolution of the wave functional
$\Psi$ (and of its mode decomposition) which still satisfies the usual
Schr\"{o}dinger equation. We have a nonlocal dynamical equation \eqref{the_first} for the
evolving configuration $q(t)$, but the Schr\"{o}dinger equation for $\Psi$
takes the usual form and therefore has the usual properties. Local Hamiltonian
terms in the Schr\"{o}dinger equation will be unable to excite super-Hubble
modes just as in standard quantum field theory.

How could quantum nonequilibrium exist in the decay products? There seem to be
two possible mechanisms.

First, note that the inflaton perturbation $\phi$ will also appear in the
interaction Hamiltonian (\ref{intn}). The dominant processes of particle
creation by the homogeneous part $\phi_{0}$ will necessarily be subject to
corrections from the perturbation $\phi$. If the perturbation is out of
equilibrium, the induced corrections will carry signatures of nonequilibrium
-- as will be illustrated by a simple model of field couplings in Section \ref{particle_decay}.

Second, as in any de Broglie-Bohm account of a quantum process, the final
probability distribution for the created particles will originate from the
initial probability distribution for the \textit{complete} hidden-variable
state.\footnote{In pilot-wave theory the outcome of a single quantum
measurement is determined by the complete initial configuration (together with
the initial wave function and total Hamiltonian). Over an ensemble, the
distribution of outcomes is then determined by the distribution of initial
configurations.} In this case the initial hidden-variable state will include
vacuum bosonic field configurations together with vacuum fermionic particle
configurations for the created species (assuming a Dirac-sea account of
fermions). Thus, if the relevant vacuum variables for the created species are
out of equilibrium at the beginning of inflaton decay, the created particles
will in general violate the Born rule. As we have discussed, inflaton
perturbations do not relax to equilibrium during the inflationary phase. One
may expect that the other degrees of freedom in the vacuum will show a
comparable behaviour -- in which case they could indeed be out of equilibrium
at the onset of inflaton decay, resulting in nonequilibrium for the decay products.

At least in principle, then, the particles created by inflaton decay could
show deviations from quantum equilibrium. However, subsequent relaxation can
be avoided only if the particles are created at a time after their
corresponding decoupling time $t_{\mathrm{dec}}$ (when the mean free time
$t_{\operatorname{col}}$ between collisions is larger than the expansion
timescale $t_{\exp}\equiv a/\dot{a}$) or equivalently at a temperature below
their decoupling temperature $T_{\mathrm{dec}}$. Otherwise the interactions
with other particles are likely to cause rapid relaxation.

A natural candidate to consider is the gravitino $\tilde{G}$, which arises in
supersymmetric theories of high-energy physics. In some models, gravitinos are
copiously produced by inflaton decay \cite{EHT06,KTY06,ETY07} and could make
up a significant component of dark matter \cite{T08}. (For recent reviews of
gravitinos as dark matter candidates see for example refs. \cite{EO13,C14}.)
Gravitinos are very weakly interacting and therefore in practice could not be
detected directly, but in many models they are unstable and decay into
particles that are more readily detectable. Again, as we shall see, in general
we expect any decay process to transfer quantum nonequilibrium from the
initial decaying field to the decay products. Thus, at least in principle, one
could search for deviations from the Born rule in (say) photons that are
generated by gravitino decay. However, the decay would have to take place
after the time $(t_{\mathrm{dec}})_{\gamma}$ of photon decoupling -- so that
the decay photons may in turn avoid relaxation.

It may then seem unlikely that primordial nonequilibrium could ever survive
until the present, since several stages may be required. But simple estimates
suggest that at least in principle the required constraints could be satisfied
for some models.

The unstable gravitino $\tilde{G}$ has been estimated to decouple at a
temperature $(T_{\mathrm{dec}})_{\tilde{G}}$ given by \cite{FY02}%
\begin{align}
k_{\mathrm{B}}(T_{\mathrm{dec}})_{\tilde{G}}&\equiv x_{\tilde{G}}%
(k_{\mathrm{B}}T_{\mathrm{P}})\\&\approx(1\ \mathrm{TeV})\left(  \frac{g_{\ast}%
}{230}\right)  ^{1/2}\left(  \frac{m_{\tilde{G}}}{10\ \mathrm{keV}}\right)
^{2}\left(  \frac{1\ \mathrm{TeV}}{m_{gl}}\right)  ^{2}\ ,
\end{align}
where $T_{\mathrm{P}}$ is the Planck temperature, $g_{\ast}$ is the number of
spin degrees of freedom (for the effectively massless particles) at the
temperature $(T_{\mathrm{dec}})_{\tilde{G}}$, $m_{gl}$ is the gluino mass, and
$m_{\tilde{G}}$ is the gravitino mass. For the purpose of illustration, if we
take $\left(  g_{\ast}/230\right)  ^{1/2}\sim1$ and $\left(  1\ \mathrm{TeV/}%
m_{gl}\right)  ^{2}\sim1$, then%
\begin{align}
x_{\tilde{G}}\approx\left(  \frac{m_{\tilde{G}}}{10^{3}\ \mathrm{GeV}}\right)
^{2}\ . \label{xG}%
\end{align}
If for example we take $m_{\tilde{G}}\approx100\ \mathrm{GeV}$, then
$x_{\tilde{G}}\approx10^{-2}$. Gravitinos produced by inflaton decay at
temperatures below $(T_{\mathrm{dec}})_{\tilde{G}}\equiv x_{\tilde{G}%
}T_{\mathrm{P}}$ could potentially avoid quantum relaxation. Any
nonequilibrium which they carry could then be transferred to their decay
products. If the gravitino is not the lightest supersymmetric particle, then
it will indeed be unstable. For large $m_{\tilde{G}}$ the total decay rate is
estimated to be \cite{NY06}%
\begin{align}
\Gamma_{\tilde{G}}=(193/48)(m_{\tilde{G}}^{3}/M_{\mathrm{P}}^{2})\ ,
\end{align}
where $M_{\mathrm{P}}\simeq1.2\times10^{19}\ \mathrm{GeV}$ is the Planck mass.
The time $(t_{\mathrm{decay}})_{\tilde{G}}$ at which the gravitino decays is
of order the lifetime $1/\Gamma_{\tilde{G}}$. Using the standard
temperature-time relation%
\begin{align}
t\sim(1\ \mathrm{s})\left(  \frac{1\ \mathrm{MeV}}{k_{\mathrm{B}}T}\right)
^{2}\ , \label{tT}%
\end{align}
the corresponding temperature is then%
\begin{align}
k_{\mathrm{B}}(T_{\mathrm{decay}})_{\tilde{G}}\sim(m_{\tilde{G}}%
/1\ \mathrm{GeV})^{3/2}\ \mathrm{eV}\ .
\end{align}
For example, again for the case $m_{\tilde{G}}\approx100\ \mathrm{GeV}$, the
relic gravitinos will decay when $k_{\mathrm{B}}(T_{\mathrm{decay}}%
)_{\tilde{G}}\sim1\ \mathrm{keV}$. This is prior to photon decoupling, so that
any (potentially nonequilibrium) photons produced by the decaying gravitinos
would interact strongly with matter and quickly relax to quantum equilibrium.
To obtain gravitino decay after photon decoupling, we would need
$k_{\mathrm{B}}(T_{\mathrm{decay}})_{\tilde{G}}\lesssim k_{\mathrm{B}%
}(T_{\mathrm{dec}})_{\gamma}\sim0.3\ \mathrm{eV}$, or $m_{\tilde{G}}%
\lesssim0.5\ \mathrm{GeV}$. For such small gravitino masses, decoupling occurs
at (roughly)%
\begin{align}
(T_{\mathrm{dec}})_{\tilde{G}}=x_{\tilde{G}}T_{\mathrm{P}}\approx\left(
m_{\tilde{G}}/10^{3}\ \mathrm{GeV}\right)  ^{2}T_{\mathrm{P}}\lesssim
10^{-7}T_{\mathrm{P}}\ .
\end{align}
In such a scenario, to have a hope of finding relic nonequilibrium in photons
from gravitino decay, we would need to restrict ourselves to those gravitinos
that were produced by inflaton decay at temperatures $\lesssim10^{-7}%
T_{\mathrm{P}}$.

Our considerations here are intended to be illustrative only. It may prove
more favourable to consider other gravitino decay products -- or to apply
similar reasoning to other relics from the Planck era besides the gravitino\footnote{Colin \cite{SC13} has developed the pilot-wave theory of (first-quantised) Majorana fermions and suggests that quantum nonequilibrium might survive at sub-Compton lengthscales for these systems.}.
And of course one could also consider photons that are generated by the
annihilation of relic particles as well as by their decay.

While definite conclusions must await the development of detailed and specific
models, in principle the required constraints do not seem insuperable. There
is, however, a further question we have yet to address: whether or not
relaxation will still occur even for decay particles that are decoupled.
Decoupling is necessary but not sufficient to avoid relaxation.
We may discuss this for decay particles whose physical wavelengths are sufficiently
sub-Hubble ($\lambda_{\mathrm{phys}}<<H^{-1}$) that the Minkowski limit
applies, since extensive numerical studies of relaxation have already been
carried out in this limit.
If the decay
particles are free but in quantum states that are superpositions of even
modest numbers of energy eigenstates, then rapid relaxation will occur (on
timescales comparable to those over which the wave function itself evolves)
\cite{AV92,AV01,VW05,EC06,TRV12,SC12,ACV14}. On the other hand, if the number
of energy states in the superposition is small then it is likely that
relaxation will not take place completely. It was shown in ref.\ \cite{ACV14}
that, if the relative phases in the initial superposition are chosen randomly,
then for small numbers of energy states it is likely that the trajectories
will not fully explore the configuration space, resulting in a small but
significant non-zero `residue' in the coarse-grained $H$-function --
corresponding to a small deviation from quantum equilibrium -- even in the
long-time limit. It appears that such behaviour can occur for larger numbers
of energy states as well, but will be increasingly rare the more energy states
are present in the superposition (see ref.\ \cite{ACV14} for a detailed
discussion). Decay particles will be generated with a range of effective
quantum states. For that fraction of particles whose wave functions have a
small number of superposed energy states, there is likely to be a small
residual nonequilibrium even in the long-time limit. Therefore, again, at
least in principle there seems to be no insuperable obstacle to primordial
nonequilibrium surviving to some (perhaps small) degree until the present day.

\subsection{Relic conformal vacua}
\label{IIC}

While inflaton decay will certainly create nonequilibrium particles from an
initially nonequilibrium vacuum, we have seen that there are practical
obstacles to such nonequilibrium surviving until the present day. The
obstacles do not seem insurmountable in principle, but whether a scenario of
the kind we have sketched will be realised in practice is at present unknown.
There is, however, an alternative and rather simpler scenario which appears to
be free of such obstacles. This involves considering relic nonequilibrium
field modes for the vacuum only. This has the advantage that vacuum wave
functions are so simple that no further relaxation can be generated -- any
relic nonequilibrium from earlier times will be frozen and preserved.

But how could primordial field modes remain unexcited in the post-inflationary
era? For a given field there are three mechanisms that can cause excitation:
(i) inflaton decay, (ii) interactions with other fields, and (iii) spatial
expansion. It is, however, possible to avoid each of these. Firstly, while a
field mode is in the super-Hubble regime it will in effect be shielded from
the effects of local physics and will not be subject to excitation from
perturbative interactions (with the inflaton or with other fields).\footnote{
Again, this standard argument is still valid in the de Broglie-Bohm
formulation since we are referring to the time evolution of the wave
functional only.}
 Secondly,
if during the post-inflationary radiation-dominated phase the field mode
enters the Hubble radius at a time $t_{\mathrm{enter}}$ that is later than the
decoupling time $t_{\mathrm{dec}}$ for the corresponding particle species, the
mode will remain free of interactions and continue to be unexcited (see figure
1). Finally, the effects of spatial expansion may be avoided altogether by
restricting our attention to fields that are conformally-coupled to the
spacetime metric. For example, for a massless scalar field $\phi$ with
Lagrangian density%
\begin{align}
\mathcal{L}=\frac{1}{2}\sqrt{-\,g}\left(  \,g_{\mu\nu}\partial^{\mu}%
\phi\partial^{\nu}\phi-\frac{1}{6}\,R\phi^{2}\right)
\end{align}
(where $R$ is the curvature scalar), the dynamics is invariant under a
conformal transformation $\,g_{\mu\nu}(x)\rightarrow\,\tilde{g}_{\mu\nu
}(x)=\Omega^{2}(x)\,g_{\mu\nu}(x)$, $\phi(x)\rightarrow\tilde{\phi}%
(x)=\Omega^{-1}(x)\phi(x)$, where $\Omega(x)$ is an arbitrary spacetime
function \cite{BD82,PT09}. Because a Friedmann--Lema\^{\i}tre spacetime is
conformally related to a section of Minkowski spacetime, the spatial expansion
will not create particles for a (free) conformally-coupled field. The natural
or conformal vacuum state is stable, just as in Minkowski spacetime
\cite{BD82,PT09}. Conformal invariance is however possible only for massless
fields, whether bosonic or fermionic. As examples of conformally-coupled
particle species, we may consider photons and (if they exist) massless
neutrinos and massless gravitinos.%

\begin{figure}
\centering
\footnotesize{\import{Chapter_4_figures/}{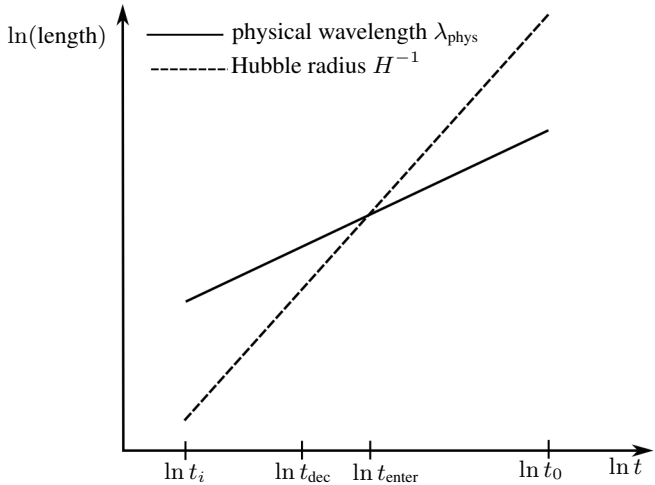}}
\caption[Lengthscales for a radiation-dominated expansion]{Lengthscales for a radiation-dominated expansion. The solid line
shows the time evolution of the physical wavelength $\lambda_{\mathrm{phys}%
}=a\lambda\propto t^{1/2}$. The dashed line shows the time evolution of the
Hubble radius $H^{-1}=2t$. The mode enters the Hubble radius after the
decoupling time $t_{\mathrm{dec}}$.}%
\end{figure}

Because ground-state wave functions and the associated de Broglie velocity
fields are so simple (indeed trivial), relic vacuum modes will not relax to
equilibrium and could therefore survive as carriers of nonequilibrium until
the present day. As we shall see, nonequilibrium vacuum modes would in
principle generate corrections to particle-physics processes.

At what lengthscale might relic nonequilibrium exist in the vacuum today? This
may be estimated by requiring that the modes enter the Hubble radius at times
$t_{\mathrm{enter}}>t_{\mathrm{dec}}$ (so as to avoid excitation and hence
likely relaxation). Thus we require that at the time $t_{\mathrm{dec}}$ the
vacuum modes have an instantaneous physical wavelength $\lambda_{\mathrm{phys}%
}^{\mathrm{vac}}(t_{\mathrm{dec}})$ that is super-Hubble,%
\begin{align}
\lambda_{\mathrm{phys}}^{\mathrm{vac}}(t_{\mathrm{dec}})\gtrsim
H_{\mathrm{dec}}^{-1}\ , \label{lam-con4}%
\end{align}
where $H_{\mathrm{dec}}^{-1}$ is the Hubble radius at time $t_{\mathrm{dec}}$
(as shown in figure 1). Now $\lambda_{\mathrm{phys}}^{\mathrm{vac}%
}(t_{\mathrm{dec}})=a_{\mathrm{dec}}\lambda^{\mathrm{vac}}$ (where
$a_{\mathrm{dec}}=T_{0}/T_{\mathrm{dec}}$ and $T_{0}\simeq2.7\ \mathrm{K}$)
and $H_{\mathrm{dec}}^{-1}=2t_{\mathrm{dec}}$ with $t_{\mathrm{dec}}$
expressed in terms of $T_{\mathrm{dec}}$ by the approximate formula
(\ref{tT}). The lower bound (\ref{lam-con4}) then becomes (inserting $c$)%
\begin{align}
\lambda^{\mathrm{vac}}\gtrsim2c(1\ \mathrm{s})\left(  \frac{1\ \mathrm{MeV}%
}{k_{\mathrm{B}}T_{\mathrm{dec}}}\right)  \left(  \frac{1\ \mathrm{MeV}%
}{k_{\mathrm{B}}T_{0}}\right)
\end{align}
or%
\begin{align}
\lambda^{\mathrm{vac}}\gtrsim(3\times10^{20}\ \mathrm{cm})\left(
\frac{1\ \mathrm{MeV}}{k_{\mathrm{B}}T_{\mathrm{dec}}}\right)  \,. \label{lb2}%
\end{align}
This is a lower bound on the comoving wavelength $\lambda^{\mathrm{vac}}$ at
which nonequilibrium could be found for conformally-coupled vacuum modes.

The lower bound (\ref{lb2}) becomes prohibitively large unless we focus on
fields that decouple around the Planck temperature or soon after. For photons
$k_{\mathrm{B}}(T_{\mathrm{dec}})_{\gamma}\sim0.3\ \mathrm{eV}$, and so for
the electromagnetic vacuum (\ref{lb2}) implies $\lambda_{\gamma}%
^{\mathrm{vac}}\gtrsim10^{27}\ \mathrm{cm}$. For massless, conformally-coupled
neutrinos (if such exist), $k_{\mathrm{B}}(T_{\mathrm{dec}})_{\nu}%
\sim1\ \mathrm{MeV}$ and $\lambda_{\nu}^{\mathrm{vac}}\gtrsim10^{20}%
\ \mathrm{cm}\simeq30\ \mathrm{pc}$ (or $\sim10^{2}$ light years). Relic
nonequilibrium for these vacua could plausibly exist today only at such huge
wavelengths and any induced effects would be far beyond any range of detection
in the foreseeable future.

We must therefore consider fields that decoupled close to the Planck
temperature. Gravitons are expected to be minimally-coupled and so would not
have a stable vacuum state under the spatial expansion. However, a massless
gravitino field should be conformally-coupled, in which case it would be a
candidate for our scenario. For massless gravitinos we have a lower bound%
\begin{align}
\lambda_{\tilde{G}}^{\mathrm{vac}}\gtrsim(10^{-2}\ \mathrm{cm})(1/x_{\tilde
{G}}) \label{lbGvac}%
\end{align}
(again writing $k_{\mathrm{B}}(T_{\mathrm{dec}})_{\tilde{G}}\equiv
x_{\tilde{G}}(k_{\mathrm{B}}T_{\mathrm{P}})\simeq x_{\tilde{G}}(10^{19}%
\ \mathrm{GeV})$ and with $x_{\tilde{G}}\lesssim1$). If, for example, we take
$x_{\tilde{G}}\approx10^{-2}$ then $\lambda_{\tilde{G}}^{\mathrm{vac}}%
\gtrsim1\ \mathrm{cm}$. According to this crude and illustrative estimate,
relic nonequilibrium for a massless gravitino vacuum today appears to be
possible for modes of wavelength $\gtrsim1\ \mathrm{cm}$.

\subsection{Particle physics in a nonequilibrium vacuum}
\label{IID}

If nonequilibrium vacuum modes do exist today, how could they manifest
experimentally? In principle they would induce nonequilibrium corrections to
particle creation from the vacuum (as already noted for inflaton decay) or to
other perturbative processes such as particle decay.

Consider for example a free scalar field $\Phi(\mathbf{x},t)$ that is massive
and charged. Let us again write the Fourier components as $\Phi_{\mathbf{k}%
}(t)=\left(  \sqrt{V}/(2\pi)^{3/2}\right)  \left(  Q_{\mathbf{k}%
1}(t)+iQ_{\mathbf{k}2}(t)\right)  $ with real $Q_{\mathbf{k}r}$ ($r=1,2$). In
Minkowski spacetime -- which is suitable for a description of local laboratory
physics -- the vacuum wave functional takes the form%
\begin{align}
\Psi_{0}[Q_{\mathbf{k}r},t]\propto%
%TCIMACRO{\dprod \limits_{\mathbf{k}r}}%
%BeginExpansion
{\displaystyle\prod\limits_{\mathbf{k}r}}
%EndExpansion
\exp\left(  {-}\frac{1}{2}{{\omega}}Q_{\mathbf{k}r}^{2}\right)  \exp\left(
{-i}\frac{1}{2}{{\omega t}}\right)  \label{PsiVac}%
\end{align}
where $\omega=(m^{2}+k^{2})^{1/2}$ and $m$ is the mass associated with the
field. (On expanding space the vacuum wave functional will reduce to this form
in the short-wavelength limit.)

Let us assume that the quantum state of the field is indeed the vacuum state
(\ref{PsiVac}). Assuming for simplicity that the (putative) long-wavelength
nonequilibrium modes are uncorrelated, we may then consider a hypothetical
nonequilibrium vacuum with a distribution of the form%
\begin{align}
P_{0}[Q_{\mathbf{k}r}]\propto%
%TCIMACRO{\dprod \limits_{\substack{\mathbf{k}r\\(k>k_{\mathrm{c}})}}}%
%BeginExpansion
{\displaystyle\prod\limits_{\substack{\mathbf{k}r\\(k>k_{\mathrm{c}})}}}
%EndExpansion
\exp\left(  {-{\omega}}Q_{\mathbf{k}r}^{2}\right)  .%
%TCIMACRO{\dprod \limits_{\substack{\mathbf{k}r\\(k<k_{\mathrm{c}})}}}%
%BeginExpansion
{\displaystyle\prod\limits_{\substack{\mathbf{k}r\\(k<k_{\mathrm{c}})}}}
%EndExpansion
\rho_{\mathbf{k}r}(Q_{\mathbf{k}r})\ , \label{VacNoneq}%
\end{align}
where $\rho_{\mathbf{k}r}(Q_{\mathbf{k}r})$ is a general nonequilibrium
distribution for the mode $\mathbf{k}r$ and the wavelength cutoff
$2\pi/k_{\mathrm{c}}$ is at least as large as the relevant lower bound
(\ref{lb2}) on $\lambda^{\mathrm{vac}}$. The short wavelength modes
($k>k_{\mathrm{c}}$) are in equilibrium while the long wavelength modes
($k<k_{\mathrm{c}}$) are out of equilibrium. (The vacuum distribution $P_{0}$
is time independent because the de Broglie velocity field generated by
(\ref{PsiVac}) vanishes, $\dot{Q}_{\mathbf{k}r}=0$, since the phase of the
wave functional depends on $t$ only.)

If the field $\Phi$ is now coupled to an external and classical
electromagnetic field $\mathbf{A}_{\mathrm{ext}}$, corresponding to a
replacement $\mathbf{\nabla}\Phi\rightarrow\mathbf{\nabla}\Phi+ie\mathbf{A}%
_{\mathrm{ext}}\Phi$ in the Hamiltonian, pairs of oppositely-charged bosons
will be created from the vacuum.\footnote{The de Broglie-Bohm theory of a
charged scalar field interacting with the electromagnetic field is discussed
in refs. \cite{AV92,AVbook}.} As in our discussion of inflaton decay, the
probability distribution for the created particles originates from the initial
probability distribution $P_{0}[Q_{\mathbf{k}r}]$ for the vacuum field $\Phi$.
(There are no other degrees of freedom varying over the ensemble, since the
given classical field $\mathbf{A}_{\mathrm{ext}}$ is the same across the
ensemble.) Clearly, if $P_{0}\neq\left\vert \Psi_{0}\right\vert ^{2}$ for
long-wavelength modes, the final probability distribution for the created
particles will necessarily carry traces of the initial nonequilibrium that was
present in the vacuum. We could for example consider an interaction
Hamiltonian $e^{2}A_{\mathrm{ext}}^{2}\Phi^{\ast}\Phi$ and calculate the final
particle distribution arising from a given initial nonequilibrium vacuum
distribution of the form (\ref{VacNoneq}).

Similarly, processes of particle decay will be affected by the nonequilibrium
vacuum. Consider, for example, the decay of a particle associated with a
(bosonic or fermionic) field $\psi$ that is coupled to $\Phi$ and to a third
field $\chi$. (For bosonic fields, the decay might be induced by an
interaction Hamiltonian of the form $a\chi\Phi^{2}\psi$ where $a$ is a
coupling constant.) An initial state $\left\vert \mathbf{p}\right\rangle
_{\psi}\otimes\left\vert 0\right\rangle _{\Phi}\otimes\left\vert
0\right\rangle _{\chi}$ -- where $\left\vert \mathbf{p}\right\rangle _{\psi}$
is a single-particle state of momentum $\mathbf{p}$ for the field $\psi$ and
$\left\vert 0\right\rangle _{\Phi}$, $\left\vert 0\right\rangle _{\chi}$ are
respective vacua for the fields $\Phi$ and $\chi$ -- may have a non-zero
amplitude to make a transition\footnote{
In a de Broglie-Bohm account, the apparent `collapse' of the quantum state as
indicated by equation \eqref{transition} is only an effective description. During a standard
quantum process -- such as a measurement, a scattering experiment, or general
transition between eigenstates -- an initial packet $\psi(q,0)$ on
configuration space evolves into a superposition $\psi(q,t)=\sum_{n}c_{n}%
\psi_{n}(q,t)$ of non-overlapping packets $\psi_{n}(q,t)$. The final
configuration $q(t)$ can occupy only one `branch' -- say $\psi_{i}(q,t)$,
corresponding to the $i$th `outcome'. The motion of $q(t)$ will subsequently
be affected by $\psi_{i}(q,t)$ alone, resulting in an effective `collapse' of
the wave function. The `empty' branches still exist but no longer affect the
trajectory $q(t)$. (See, for example, chapter 8 of ref. \cite{Holl93}.)}
\begin{align}\label{transition}
\left\vert \mathbf{p}\right\rangle _{\psi}\otimes\left\vert 0\right\rangle
_{\Phi}\otimes\left\vert 0\right\rangle _{\chi}\rightarrow\left\vert
0\right\rangle _{\psi}\otimes\left\vert \mathbf{k}_{1}\mathbf{k}%
_{2}\right\rangle _{\Phi}\otimes\left\vert \mathbf{p}^{\prime}\right\rangle
_{\chi}%
\end{align}
to a final state containing two excitations of the field $\Phi$ and one
excitation of $\chi$. The final probability distribution for the outgoing
particles will originate from the initial probability distribution for all the
relevant (hidden-variable) degrees of freedom -- which in this case consist of
the relevant vacuum variables for $\Phi$ and $\chi$ together with the
variables for the field $\psi$. (Again, if $\psi$ is fermionic the associated
hidden variables may consist of particle positions in the Dirac sea
\cite{BH93,C03,CS07}.) Because all these variables are coupled by the
interaction, an initial nonequilibrium distribution (\ref{VacNoneq}) for a
subset of them (that is, for the $Q_{\mathbf{k}r}$) will generally induce
corrections to the Born rule in the final joint distribution for the
collective variables and hence for the outgoing particles. Thus, for example,
for gravitinos decaying in a nonequilibrium vacuum we would expect the decay
photons to carry traces of nonequilibrium in the probability distributions for
their outgoing momenta and polarisations.

\section{Perturbative transfer of nonequilibrium}
\label{particle_decay}
We now turn to some simple but illustrative field-theoretical models of the behaviour of nonequilibrium systems. The first question that needs to be addressed is the perturbative transfer of nonequilibrium from one field to another. In this section we present a simple (bosonic) field-theoretical model that illustrates this process.

Suppose we have two Klein-Gordon fields $\phi_1$ and $\phi_2$, confined inside a box of volume $V$ with dimensions $l_x$, $l_y$, and $l_z$ such that the fields are necessarily zero valued on the boundaries of the box. In consideration of these boundary conditions, we expand and quantise the fields in a set of standing waves as $(i=1,2)$
\begin{align}
\phi_i(\mathbf{x})&=\sum_{\mathbf{k}}\frac{2^{3/2}q_{i\mathbf{k}}}{\sqrt{V}}\sin(k_x x)\sin(k_y y)\sin(k_z z),
\end{align}
with annihilation operators
\begin{align}
a_{i\mathbf{k}}&=\sqrt{\frac{\omega_{i\mathbf{k}}}{2}}\left(q_{i\mathbf{k}}+\frac{i}{\omega_{i\mathbf{k}}}p_{i\mathbf{k}}\right),
\end{align}
and a total Hamiltonian
\begin{align}
H_0&=\sum_{\mathbf{k}}\left(\omega_{1\mathbf{k}}a^\dagger_{1\mathbf{k}}a_{1\mathbf{k}}+\omega_{2\mathbf{k}}a^\dagger_{2\mathbf{k}}a_{2\mathbf{k}}\right).
\end{align}
We have dropped the zero point energy, and $k_x=n\pi/l_x$ and similarly for $y$ and $z$.  The two fields are coupled by the interaction Hamiltonian
\begin{align}
H_{\text{I}}&=g\int_V\mathrm{d}^3x\phi_1(\mathbf{x})\phi_2(\mathbf{x})\label{quad_int_ham}\\
&=\frac{g}{2}\sum_\mathbf{k}\frac{1}{\sqrt{\omega_{1\mathbf{k}}\omega_{2\mathbf{k}}}}(a_{1\mathbf{k}}+a^\dagger_{1\mathbf{k}})(a_{2\mathbf{k}}+a^\dagger_{2\mathbf{k}}),
\end{align}
where $g$ is a coupling constant.
If we suppose that at time $t=0$ the system is in the free (unperturbed) eigenstate $\ket{E_i}$, the first order perturbative amplitude to transition to the state $\ket{E_f}$ is 
\begin{align}
d_f^{(1)}(t)=\bra{f}H_{\text{I}}\ket{i}\frac{e^{-iE_ft}-e^{-iE_it}}{E_f-E_i}.\label{amp}
\end{align}
This will be damped for any $E_f$ significantly different from $E_i$. We may exploit this fact by further insisting that
\begin{itemize}
\item $l_x\gg l_y \gg l_z$, so that the lowest mode of each field is significantly lower than all others, and
\item the limit $m_2\rightarrow m_1$ is taken, so that the lowest modes of $\phi_1$ and $\phi_2$ have the same unperturbed energy.
\end{itemize}
These conditions ensure that the system state in which field $\phi_1$ has one particle occupying its lowest mode and the field $\phi_2$ is a vacuum has identical unperturbed energy to the system state in which the individual field states are reversed. We shall denote these states $\ket{1,0}$ and $\ket{0,1}$ respectively. Since these states have identical unperturbed energies, the first order perturbative amplitudes \eqref{amp} between the states is significantly amplified, whereas all others damped. This is the justification of the rotating wave approximation, familiar from quantum optics and cavity QED (see for instance ref.\ \cite{Knight}). Put simply, the states $\ket{1,0}$ and $\ket{0,1}$ are strongly coupled to each other and only very weakly coupled to any other state. 

We make the rotating wave approximation by removing all terms in the Hamiltonian that would effect an evolution to states other than $\ket{1,0}$ and $\ket{0,1}$. This allows us to employ the effective Hamiltonian,
\begin{align}
H_{\text{eff}}=\omega(a_1^\dagger a_1 + a_2^\dagger a_2)+\frac{g}{2\omega}(a_1a_2^\dagger + a_2a_1^\dagger).
\end{align}
We have suppressed the mode subscripts for simplicity. The approximate Schr\"{o}dinger equation $H_{\text{eff}}\ket{\psi}=i\partial_t\ket{\psi}$, along with the initial condition $\left.\ket{\psi}\right|_{t=0}=\ket{1,0}$, yields the solution 
\begin{align}
\ket{\psi}=e^{-i\omega t}\left(\cos\left(\frac{gt}{2\omega}\right)\ket{1,0}-i\sin\left(\frac{gt}{2\omega}\right)\ket{0,1}\right)\label{part_decay_state}.
\end{align}
The sine and cosine in Eq.\ \eqref{part_decay_state} describe an oscillatory decay process in which the first type of particle is seen to decay into the second type, which promptly decays back. This type of flip-flopping between one type of particle and the other is functionally equivalent to vacuum-field Rabi oscillations in the Jaynes-Cummings model \cite{Knight,jc} of quantum optics and cavity QED wherein an exchange of energy occurs between an atom and a cavity mode of the electromagnetic field, perpetually creating a photon, then destroying it only to create it once more.\footnote{From a field-theoretical viewpoint the quadratic interaction \eqref{quad_int_ham}
may seem too trivial an example since the interaction may be removed by a
linear transformation of the field variables. Such a transformation would not,
however, negate the physical meaning of the original system. Our aim is to
illustrate with a simple example how nonequilibrium may be passed from one
type of field to another. We expect a similar passing of quantum
nonequilibrium between fields to be caused by any reasonable interaction term.
Our example is based on a model -- widely used in quantum optics to study the
interaction between a two-level atom and a single mode of the quantised
electromagnetic field inside a cavity -- that is simple enough to be tractable
while at the same time providing a genuine field-theoretical account of energy
transfer to and from a quantised field.}

\begin{figure}
\centering
\footnotesize{\import{Chapter_4_figures/}{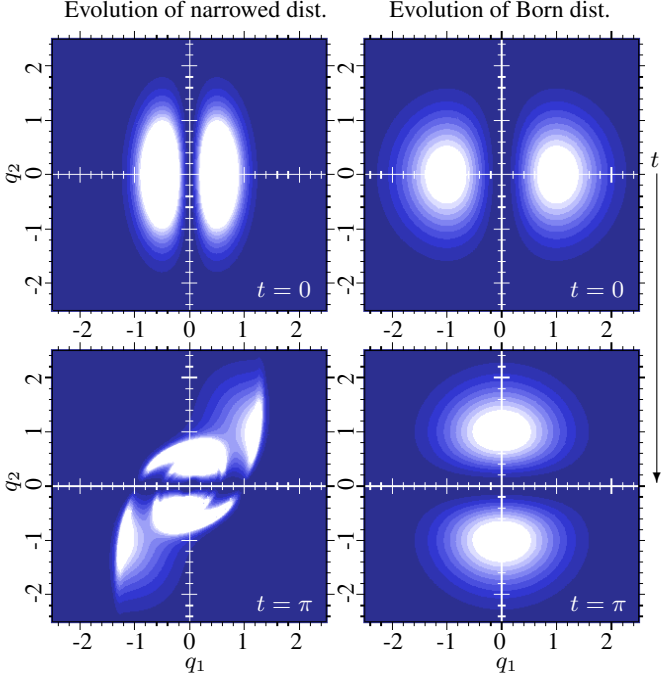}}
\caption[Quantum nonequilibrium in a particle decay process]{\label{part_decay_non-eq}
The evolution of quantum equilibrium and nonequilibrium through the particle decay process described by state \eqref{part_decay_state} and guidance equations \eqref{part_guidance_part}. Initially state \eqref{part_decay_state} is a product between an excited (one particle) state in $q_1$ and a ground (vacuum) state in $q_2$. This is shown in the top right graph. As time passes, $t=0\rightarrow\pi$, the excited state in $q_1$ decays into exactly the same excited state in $q_2$. At time $t=\pi$ the state \eqref{part_decay_state} exists in another product state, except this time with excited and ground states switched between fields. This is shown in the bottom right graph. The evolution of a quantum nonequilibrium distribution is shown in the left column. Before the interaction takes place, quantum nonequilibrium exists only in the one particle state of the first field; it has been narrowed with respect to the equilibrium distribution. As time passes, the first field generates nonequilibrium in the second. At $t=\pi$, by standard quantum mechanics, the decay process is complete and there exists another product state. In contrast, the introduction of quantum nonequilibrium has created a distribution at $t=\pi$ that is correlated between $q_1$ and $q_2$. The marginal distributions for the fields are shown in figure \ref{part_marginals}. (This figure takes $\omega=g=1$.)}
\end{figure}

\begin{figure}
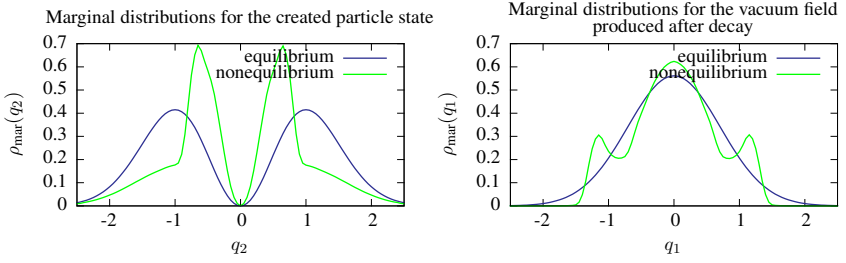

\resizebox{\linewidth}{!}{
\import{Chapter_4_figures/}{marginals_particle_created}
\import{Chapter_4_figures/}{marginals_particle_destroyed}}
\caption[Marginal distributions for the fields in figure \ref{part_decay_non-eq}]{\label{part_marginals}The ensemble marginal distributions of the particle and vacuum state created at $t=\pi$ in the decay process shown in figure \ref{part_decay_non-eq}. The top graph shows the marginal distribution for the excited field $\rho_{\text{mar}}(q_2)=\left.\int \mathrm{d}q_1\rho(q_1,q_2,t)\right|_{t=\pi}$. The bottom graph shows the nonequilibrium marginal distribution of the vacuum field $\rho_{\text{mar}}(q_1)=\left.\int \mathrm{d}q_2\rho(q_1,q_2,t)\right|_{t=\pi}$, obtained after the original particle state has decayed.}
\end{figure}

To develop a de Broglie-Bohm description of this particle decay process, one needs to specify the configuration of the system. For bosonic fields the canonical approach \cite{B52b} is to use the Schr\"{o}dinger representation with mode amplitudes as the configuration. In our case this is particularly simple; the state of any one system in an ensemble is described by the coordinates $q_1$ and $q_2$, proportional to the amplitudes of the lowest (standing) mode of each field. In this representation the Hamiltonian is
\begin{align}
H_{\text{eff}}=&-\frac{1}{2}\left(\partial_{q_1}^2+\partial_{q_2}^2\right)+\frac12\omega^2\left(q_1^2+q_2^2\right)\nonumber\\
&-\omega+\frac{g}{2}\left(q_1q_2-\frac{1}{\omega^2}\partial_{q_1}\partial_{q_2}\right).\label{coord_rwa_ham}
\end{align}
In the rotating wave approximation there are derivative terms in the interaction Hamiltonian. The de Broglie velocity fields associated with the Hamiltonian \eqref{coord_rwa_ham} may be derived in the standard way, and by using 
\begin{align}
&\psi^*\partial_{q_1}\partial_{q_2}\psi-\psi\partial_{q_1}\partial_{q_2}\psi^*\nonumber\\
&=i\partial_{q_1}\left(|\psi|^2\partial_{q_2}\text{Im}\ln\psi\right)+i\partial_{q_2}\left(|\psi|^2\partial_{q_1}\text{Im}\ln\psi\right)
\end{align}
(a special case of the general identity 2 of ref.\ \cite{SV08}). Writing $\psi=|\psi|e^{iS}$, the guidance equations may be expressed as
\begin{equation}\label{part_guidance}
\begin{aligned}
\dot{q_1}&=\partial_{q_1}S+\frac{g}{2\omega^2}\partial_{q_2}S,\\
\dot{q_2}&=\partial_{q_2}S+\frac{g}{2\omega^2}\partial_{q_1}S.
\end{aligned}
\end{equation}
For the particular state \eqref{part_decay_state}, these yield
\begin{equation}\label{part_guidance_part}
\begin{aligned}
\dot{q_1}&=\frac{\frac12\left(q_2-\frac{g}{2\omega^2}q_1\right)\sin\left(\frac{gt}{\omega}\right)}{q_1^2\cos^2\left(\frac{gt}{2\omega}\right)+ q_2^2\sin^2\left(\frac{gt}{2\omega}\right)},\\
\dot{q_2}&=\frac{\frac12\left(-q_1+\frac{g}{2\omega^2}q_2\right)\sin\left(\frac{gt}{\omega}\right)}{q_1^2\cos^2\left(\frac{gt}{2\omega}\right)+ q_2^2\sin^2\left(\frac{gt}{2\omega}\right)}.
\end{aligned}
\end{equation}
The configuration $q(t)=(q_1(t),q_2(t))$ and velocity $\dot{q}(t)=(\dot{q_1}(t),\dot{q_2}(t))$ of a particular member of an ensemble evolving along a trajectory described by Eqs.\ \eqref{part_guidance_part} has the properties  
$q(t)=q(t+2\pi\omega/g)$, $\dot{q}(t)=\dot{q}(t+2\pi\omega/g)$, $q(t)=q(-t)$, and $\dot{q}(t)=-\dot{q}(-t)$.
The trajectories $q(t)$ are periodic, and halfway through their period backtrack along their original path.% The sine in the numerators of Eqs.\ \eqref{part_guidance_part} serve to make this motion look roughly simple harmonic, as if each individual system was a pendulum, forced to move along a particular track in $(q_1,q_2)$. 

Given the velocity field \eqref{part_guidance_part}, we may integrate the continuity equation \eqref{ant_cont} to obtain the time evolution of an arbitrary distribution $\rho$. (Our numerical method is described in the appendix.)

In figure \ref{part_decay_non-eq} we compare the evolution of quantum nonequilibrium with that of equilibrium for the case $\omega=g=1$. The decay from an initial product state $\ket{1,0}$ to a final product state $\ket{0,1}$ is seen in the (product) equilibrium distributions on the right-hand side of figure \ref{part_decay_non-eq}. We illustrate the transfer of nonequilibrium in the left-hand side of figure \ref{part_decay_non-eq} for the case of an initial nonequilibrium that has simply been narrowed in $q_1$ (with respect to equilibrium). Hence only the first field is initially out of equilibrium. As time passes the distribution becomes correlated in $q_1$ and $q_2$. At $t=\pi$, when according to standard quantum mechanics we should find another product state (corresponding to $\ket{0,1}$), there exists a complicated overall nonequilibrium in $(q_1,q_2)$. The marginal distributions are shown in figure \ref{part_marginals}. 

The evolution of nonequilibrium depends strongly on the particular values of $\omega$ and $g$, although in general we see two important properties of this evolution. Firstly it is apparent from figures \ref{part_decay_non-eq} and \ref{part_marginals} that nonequilibrium in the marginal distribution of the original particle state (or excited field) will generate nonequilibrium in its decay product. Secondly, although the initial product state $\left.\ket{\psi}\right|_{t=0}=\ket{1,0}$ evolves into the product state $\ket{0,1}$ at $t=\pi\omega/g$, the nonequilibrium distribution is correlated between the two fields. Such correlation could not exist in standard quantum mechanics.

\section{Energy measurements and nonequilibrium spectra}
\label{IV}
In this section we focus on quantum-mechanical measurements of energy for elementary field-theoretical systems in nonequilibrium.
As we have discussed, in this paper we restrict ourselves to simple models that may be taken to illustrate some of the basic phenomena that could occur. 

The following analysis is presented for the electromagnetic field, partly because it provides a convenient  illustrative model and partly because (as explained in Section \ref{II}) we envisage the possibility of detecting decay photons produced by particles in nonequilibrium rather than the parent particles themselves. However, the analysis should apply equally well to other field theories. 
\subsection{Setup and effective wave function}\label{setup}
\begin{figure*}[t]
\resizebox{.98\linewidth}{!}{
\import{Chapter_4_figures/}{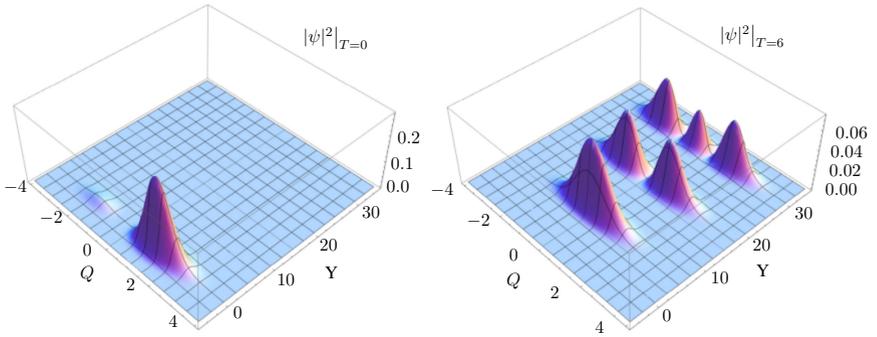}
}
\caption[Illustration of the energy measurement process]{\label{fig_1}Illustration of the energy measurement process, showing the evolution of the Born distribution into disjoint packets (for the case $c_0=c_1=c_2=1/\sqrt{3}$).
The variables $q$, $y$ and $t$ have been replaced by the rescaled variables $Q$, $Y$ and $T$ defined in section \ref{non_dim}.
The initial pointer wave function is chosen to be a Gaussian centred on Y=0.
Initially the components of the total wave function overlap and interfere. As time passes each component moves in the $Y$ direction with speed $2n+1$. After some time the components no longer overlap and the experimenter may unambiguously read off the energy eigenvalue from the position of the pointer ($Y$ coordinate).}
\end{figure*}
We work in the Coulomb gauge, $\nabla.\mathbf{A}(\mathbf{x},t)=0$, with the field expansion
\begin{align}
\mathbf{A}(\mathbf{x},t)=\sum_{\mathbf{k'}s'}\left[A_{\mathbf{k'}s'}(t)\mathbf{u}_{\mathbf{k'}s'}(\mathbf{x})+A_{\mathbf{k'}s'}^*(t)\mathbf{u}_{\mathbf{k'}s'}^*(\mathbf{x})\right],\label{u_expansion}
\end{align}
where the functions
\begin{align}
\mathbf{u}_{\mathbf{k'}s'}(\mathbf{x})=\frac{\bm{\varepsilon}_{\mathbf{k'}s'}}{\sqrt{2\varepsilon_0 V}}e^{i\mathbf{k'.x}}
\end{align}
and their complex conjugates define a basis for the function space. 
In expansion \eqref{u_expansion} and henceforth, summations over wave vectors are understood to extend over half the possible values of $\mathbf{k}'$. This is to avoid duplication of bases $\mathbf{u}^*_{\mathbf{k'}s'}$ with $\mathbf{u}_{\mathbf{-k'}s'}$. See for instance reference \cite{schiff}. The primes are included for later convenience. This expansion allows one to write the energy of the electromagnetic field as
\begin{align}
U&=\frac12 \int_V\mathrm{d}^3x\left(\varepsilon_0 \mathbf{E.E}+\frac{1}{\mu_0}\mathbf{B.B}\right)\\
&=\sum_{\mathbf{k'}s'}\frac{1}{2}\left(\dot{A}_{\mathbf{k'}s'}\dot{A}_{\mathbf{k'}s'}^*+\omega_\mathbf{k'}^2A_{\mathbf{k'}s'}A_{\mathbf{k'}s'}^*\right),\label{comp_HOs}
\end{align}
where $\omega_\mathbf{k'}=c|\mathbf{k'}|$.
Equation \ref{comp_HOs} defines a decoupled set of complex harmonic oscillators of unit mass. We shall prefer instead to work with real variables and so we decompose $A_{\mathbf{k}'s'}$ into its real and imaginary parts 
\begin{align}
A_{\mathbf{k'}s'}=q_{\mathbf{k'}s'1} +iq_{\mathbf{k'}s'2}.
\end{align}
One may then write the free field Hamiltonian as
\begin{align}
H_0=\sum_{\mathbf{k'}s'r'}H_{\mathbf{k'}s'r'}
\end{align}
with $r'=1,2$, where
\begin{align}
H_{\mathbf{k'}s'r'}=\frac{1}{2}\left(p_{\mathbf{k'}s'r'}^2+\omega_{\mathbf{k'}}^2q_{\mathbf{k'}s'r'}^2\right),
\end{align}
and where $p_{\mathbf{k'}s'r'}$ is the momentum conjugate of $q_{\mathbf{k'}s'r'}$.

Suppose we wish to perform a quantum energy measurement for a single mode of the field. We may follow the pilot-wave theory of ideal measurements described in ref.\ \cite{Holl93}. The system is coupled to an apparatus pointer with position variable $y$. The interaction Hamiltonian $H_{\text{I}}$ is taken to be of the form $g \hat{\omega}p_y$, where again $g$ is a coupling constant and $\hat{\omega}$ is the operator corresponding to the observable to be measured. In our case we have
\begin{align}\label{coupling}
H_{\text{I}}=g H_{\mathbf{k}sr}p_y,
\end{align}
where $p_y$ is the momentum conjugate to the pointer position $y$
and where $\mathbf{k}$, $s$ and $r$ refer specifically to the field mode that is being measured. Including the free Hamiltonian $H_{\text{app}}$ for the apparatus, the total Hamiltonian is
\begin{align}
H_\text{tot}=H_0+H_\text{app}+H_{\text{I}}.\label{hamiltonian}
\end{align}
We assume an initial product state 
\begin{align}
\psi(0)=\psi_{\mathbf{k}sr}(q_{\mathbf{k}sr},0)\phi(y,0)\chi(\mathcal{Q},0),
\end{align}
where $\psi_{\mathbf{k}sr}$ is the wave function for the mode in question, $\phi$ is the apparatus wave function and $\chi$ is a function of the rest of the field variables $\mathcal{Q}=\{q_{\mathbf{k'}s'r'}|(\mathbf{k'},s',r')\neq (\mathbf{k},s,r)\}$. The function $\chi$ is left unspecified as there is no need to make assumptions concerning the state of the rest of the field. Now, since $H_{\text{I}}$ and $H_\text{app}$ commute with all the terms in $H_0$ that include $\mathcal{Q}$, under time evolution the $\chi$ function remains unentangled with the rest of the system while the apparatus and the mode being measured become entangled. We may then write
\begin{align}
\psi(t)&=\Psi(q_{\mathbf{k}sr},y,t)\chi(\mathcal{Q},t),
\end{align}
where
\begin{equation}
\begin{aligned}
\Psi(q_{\mathbf{k}sr},y,t)&=\exp\left[-i(H_{\mathbf{k}sr}+H_\text{app}+gH_{\mathbf{k}sr}p_y)t\right]\\
&\times\psi_{\mathbf{k}sr}(q_{\mathbf{k}sr},0)\phi(y,0),\\
\chi(\mathcal{Q},t)&=\left[\prod_{(\mathbf{k'}s'r')\neq(\mathbf{k}sr)}\exp\left(-iH_{\mathbf{k'}s'r'}t\right)\right]\chi(\mathcal{Q},0).
\end{aligned}
\end{equation}
Since the system and apparatus remain unentangled with $\chi$, the dynamics remain completely separate. We may concern ourselves only with $\Psi(q_{\mathbf{k}sr},y,t)$ as an effective wave function. The velocity field in the ($q_{\mathbf{k}sr},y)$ plane depends on the position in that plane but is independent of the position in $\mathcal{Q}$. We may then omit the $\mathbf{k}sr$ labels in $q_{\mathbf{k}sr}$ and $H_{\mathbf{k}sr}$, and the $\mathbf{k}$ label in $\omega_\mathbf{k}$.

Let the measurement process begin at $t=0$ when the coupling is switched on. 
As usual in the description of an ideal von Neumann measurement (see for example ref.\ \cite{Holl93}), we take $g$ to be so large that the free parts of the Hamiltonian may be neglected during the measurement.
The system will then evolve according to the Schr\"{o}dinger equation
\begin{align}
\left(\partial_t+gH\partial_y\right)\Psi=0.\label{sch1}
\end{align}
Expanding $\Psi$ in a basis $\psi_n(q)$ of energy states for the field mode, we have the solution
\begin{align}
\Psi(q,y,t)=\sum_nc_n\phi(y-gE_nt)\psi_n(q).\label{key}
\end{align}
where we choose $\phi$ and $\psi_n$ to be real and $\sum_n|c_n|^2=1$.
Equation \eqref{key} describes the measurement process. If the initial system state is an energy eigenstate ($c_n=\delta_{mn}$ for some $m$), the pointer packet is translated with a speed proportional to the energy $E_m$ of the eigenstate. By observing the displacement of the pointer after a time $t$, an experimenter may infer the energy of the field mode. If instead the field mode is initially in a superposition of energy states, the different components of the superposition will be translated at different speeds until such a time when they no longer overlap and thus do not interfere. An example of this evolution into non-overlapping, non-interfering packets is shown in figure \ref{fig_1}. At this time an experimenter could unambiguously read off an energy eigenvalue. The weightings $|c_n|^2$ in the superposition could be determined by readings over an ensemble.

\subsection{Pointer packet and rescaling}\label{non_dim}
For simplicity we choose the initial pointer wave function $\phi$ in Eq.\ \eqref{key} to be a Gaussian centred on $y=0$, 
\begin{align}
\phi(y)&=\sigma^{-\frac12}(2\pi)^{-\frac14}e^{-y^2/4\sigma^2},
\end{align}
where $\sigma^2$ is the variance of $|\phi(y)|^2$.

It is convenient to introduce the rescaled parameters
\begin{align}\label{coords}
Q=\sqrt{\omega}q,\quad Y=\frac{y}{\sigma}, \quad T=\frac{g \omega t}{2\sigma}.
\end{align}
The evolution of the wave function is then determined by the Schr\"{o}dinger equation,
\begin{align}
\partial_T\Psi=\left(\partial_Q^2-Q^2\right)\partial_Y\Psi.\label{schro}
\end{align}
The general solution is
\begin{align}\label{wvfn}
&\Psi(Q,Y,T)=\sum_n\frac{c_n}{\sqrt{\pi 2^{n+1/2}n!}}\nonumber\\
&\times\exp\left[-\frac14\left(Y-(2n+1)T\right)^2\right]e^{-Q^2/2}H_n(Q),
\end{align}
where $H_n(Q)$ are Hermite polynomials.
(Equation \eqref{wvfn} differs from Eq.\ \eqref{key} by a factor $\sigma^{1/2}\omega^{-1/4}$, to normalise the wave function in the rescaled configuration space.)

\subsection{Continuity equation and guidance equations}
\label{sec:cont_eq}
From the Schr\"{o}dinger equation \eqref{schro}, it is simple to arrive at
\begin{align}
\partial_T|\Psi|^2&=\Psi^*\partial^2_Q\partial_Y\Psi+\Psi\partial^2_Q\partial_Y\Psi^*
-\partial_Y\left(Q^2|\Psi|^2\right).
\end{align}
From here we use the identity
\begin{align}
&\Psi^*\partial_Q^2\partial_Y\Psi+\Psi\partial_Q^2\partial_Y\Psi^*\nonumber\\
\equiv&\frac13\partial_Q\left(2\Psi\partial_Q\partial_Y\Psi^*-\partial_Y\Psi\partial_Q\Psi^*\right.\nonumber\\
&\left.-\partial_Q\Psi\partial_Y\Psi^*+2\Psi^*\partial_Q\partial_Y\Psi\right)\nonumber\\
&+\frac13\partial_Y\left(\Psi\partial^2_Q\Psi^*-\partial_Q\Psi\partial_Q\Psi^*+\Psi^*\partial_Q^2\Psi\right).
\end{align}
This, again, is a special case of the general identity 2 of \cite{SV08}. The continuity equation is found to be
\begin{align}
\partial_T|\Psi|^2+\partial_Q\text{Re}\left(\frac23\partial_Y\Psi\partial_Q\Psi^*-\frac43\Psi^*\partial_Q\partial_Y\Psi \right)&\nonumber\\
+\partial_Y\text{Re}\left(\frac13|\partial_Q\Psi|^2-\frac23\Psi^*\partial^2_Q\Psi+Q^2\right)&=0,
\end{align}
from which we may deduce the de Broglie guidance equations
\begin{align}
\partial_TQ&=\text{Re}\left(-\frac43\frac{\partial_Q\partial_Y\Psi}{\Psi}+\frac23\frac{\partial_Y\Psi\partial_Q\Psi^*}{|\Psi|^2}  \right),\label{guidance_a}\\
\partial_TY&=\text{Re}\left(-\frac{2}{3}\frac{\partial^2_Q\Psi}{\Psi}+\frac{1}{3}\frac{\partial_Q\Psi\partial_Q\Psi^*}{|\Psi|^2}\right)+Q^2.\label{guidance_b}
\end{align}
The factor $Q^2$ in Eq.\ \eqref{guidance_b} will turn out to have the most predictable effect on the evolution of quantum nonequilibrium in section \ref{results}. Any individual system in which $|Q|$ is abnormally large will, at least to begin with, have an abnormally large pointer velocity. The $Q^2$ term originates from the potential term in $H_{\mathbf{k}sr}=\frac{1}{2}p_{\mathbf{k}sr}^2+\frac12 \omega_{\mathbf{k}}^2q_{\mathbf{k}sr}^2$. 

\begin{figure}
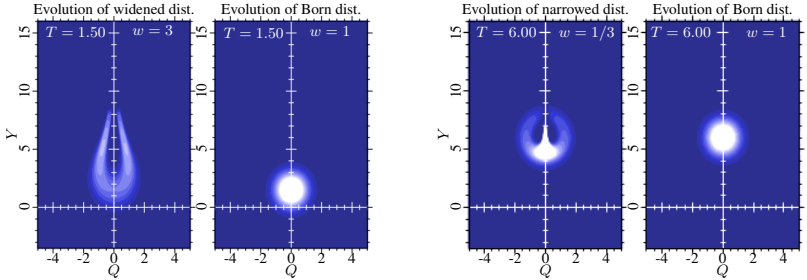

\footnotesize{\resizebox{\linewidth}{!}{
\import{Chapter_4_figures/}{gs_widened.eps_tex}
\import{Chapter_4_figures/}{gs_narrowed.eps_tex}
}}
\caption[Vacuum nonequilibria under an energy measurement process]{\label{ground_sims}The evolution of vacuum nonequilibria under an energy measurement process (as simulated by the code discussed in the appendix). On the left is a snapshot of the evolution of a widened initial $\rho$ with $w=3$, taken at $T=1.50$. The tails of $\rho$ evolve quickly to large $Y$ and small $Q$. These tails are evident in the marginal distribution for $Y$ shown in figure \ref{marginals}. On the right is the same simulation except narrowed by a factor $w=1/3$, and taken at $T=6.00$. In this case $\rho$ remains in what might loosely be deemed the support of $|\Psi|^2$, though displaying internal structure. The narrowed $\rho$ initially lags behind the Born distribution, before getting swept outwards and upwards, creating a double-bump in the pointer marginal distribution. Note that the equilibrium pointer distribution undergoes an upward displacement to indicate the zero-point energy of the vacuum mode.}
\end{figure}

In contrast with section \ref{particle_decay}, here we have chosen to retain the zero-point energy of the $\mathbf{k}sr$ mode. Since the pointer is coupled to the total energy of the $\mathbf{k}sr$ mode, this does affect the dynamics though only in a minor respect. Had we normal ordered Eq.\ \eqref{hamiltonian}, the pointer velocity Eq.\ \eqref{guidance_b} would have an extra additive term of $-1$. In the state-specific expressions of section \ref{state_specific}, normal ordering is equivalent to switching to a coordinate system moving in the $+Y$ direction at a (rescaled) velocity of $1$, the velocity of the vacuum component in Eq.\ \eqref{wvfn}. Equivalently, one may use the coordinate transformation $Y\rightarrow Y'=Y-T$, which we shall indeed do in section \ref{long_time}.

\subsection{Expressions for three examples}\label{state_specific}
\subsubsection{Vacuum}
If the field mode being measured is in its vacuum state ($c_n=\delta_{n0}$), the evolution of the total wave function \eqref{wvfn} and the associated velocity fields \eqref{guidance_a} and \eqref{guidance_b} are given by 
\begin{align}
\Psi&=2^{-\frac14}\pi^{-\frac12}\exp\left(-\frac12 Q^2 -\frac14(Y-T)^2 \right),\\
\partial_TQ&=\frac13 Q(T-Y),\label{gs_q_vel}\\
\partial_TY&=\frac{2}{3}(1+Q^2).\label{gs_y_vel}
\end{align}
\subsubsection{One particle state}
If instead the field mode being measured contains one particle or excitation ($c_n=\delta_{n1}$), the relevant expressions are
\begin{align}
\Psi&=2^{\frac14}\pi^{-\frac12}Q\exp\left(-\frac12 Q^2-\frac14 (Y-3T)^2\right),\label{sp_measurement}\\
\partial_TQ&=\frac13 (Y-3T)\left(\frac1Q -Q\right)\label{e1},\\
\partial_TY&=\frac13\frac{1}{Q^2}+\frac43 +\frac23 Q^2\label{e2}.
\end{align}
\subsubsection{Initial superposition of vacuum and one particle state}
For a superposition, the relative phases in the $c_n$'s will contribute to the dynamics. For a superposition of initial vacuum and one particle states, we take $c_0=e^{i\theta}/\sqrt{2}$ and $c_1=1/\sqrt{2}$. Our expressions then become
%\begin{widetext}
\begin{align}
\Psi&=\left(\frac{e^{i\theta}}{\sqrt{2}}+Qe^{T(Y-2T)}\right)\label{super1}
2^{-\frac14}\pi^{-\frac12}e^{-\frac14\left(Y-T\right)^2}
\exp\left(-\frac12 Q^2\right),\\
\partial_TQ&=
\resizebox{0.89\textwidth}{!}{$
\text{Re}\left(\frac{-\frac53T+\frac23Q^2T+\frac13Y}{\frac{e^{i\theta}}{\sqrt{2}}e^{T(2T-Y)}+Q}\right)+\frac{\frac23 QT}{\left|\frac{e^{i\theta}}{\sqrt{2}}e^{T(2T-Y)}+Q\right|^2}-\frac13(Y-T)Q\label{super2},$}\\
\partial_TY&=\text{Re}\left[\frac{\frac23Q}{\frac{e^{i\theta}}{\sqrt{2}}e^{T(2T-Y)}+Q}\right]+\frac{\frac13}{\left|\frac{e^{i\theta}}{\sqrt{2}}e^{T(2T-Y)}+Q\right|^2}+\frac23(Q^2+1)\label{super3}.
\end{align}
%\end{widetext}

\subsection{Results for nonequilibrium energy measurements}\label{results}
We will now consider outcomes of quantum energy measurements for nonequilibrium field modes. Like many features of quantum mechanics, the usual statistical energy conservation law emerges in equilibrium. But for nonequilibrium states there is no generally useful notion of energy conservation\footnote{The fundamental dynamical equation \eqref{the_first} is first-order in time and has no naturally conserved energy. When rewritten in second order form there appears a time-dependent `quantum potential' that acts as an effective external energy source \cite{Holl93}.}.

We may consider a parameterisation of nonequilibrium that simply varies the width of the Born distribution (as discussed in Section IIA for primordial perturbations). Our initial $\rho$ is written
\begin{align}
\rho(Q,Y,0)=\frac{1}{w}\left|\Psi(Q/w,Y,0)\right|^2,
\end{align}
where $w$ is a widening parameter equal to the initial standard deviation of $\rho$ relative to $|\Psi|^2$,
\begin{align}
w=\left.\frac{\sigma_\rho}{\sigma_{|\Psi|^2}}\right|_{t=0}.
\end{align}
(Comparing with eqn.\ \eqref{ksi}, we would have $w = \sqrt{\xi}$ for primordial perturbations.)

\subsubsection{Short-time measurement of vacuum modes}\label{short_time}
\begin{figure}
\centering
\import{Chapter_4_figures/}{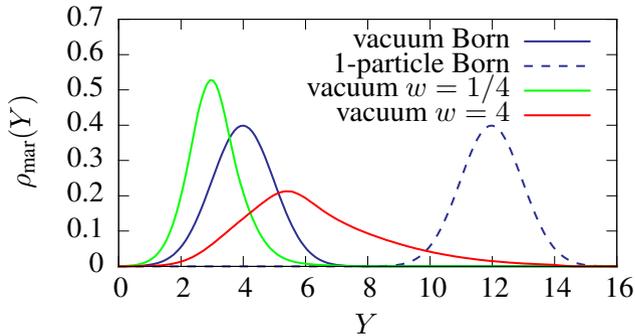}
\caption[Pointer distributions for nonequilibrium vacuum measurement]{\label{marginals}Marginal pointer distributions $\rho_{\text{mar}}(Y)$ under the energy measurement of vacuum mode nonequilibria at $T=4$. (For comparison we also show the Born pointer marginals for the vacuum and 1-particle cases.) For the widened vacuum mode ($w=4$), there is a significant probability of `detecting a particle' (that is, an excited state) in the vacuum mode. For this case there also exists a significant probability of finding the pointer around Y = 8 (which for all practical purposes would be impossible without nonequilibrium for any initial superposition). For the narrowed nonequilibrium ($w=1/4$), the pointer distribution lags behind the Born pointer distribution initially. As time progresses, $\rho$ will get swept outwards and upwards (cf.\ the right-hand side of figure \ref{ground_sims}), creating a double-bump in the pointer distribution. }
\vspace{-00mm}
\end{figure}

In figure \ref{ground_sims} we show the short-time behaviour of widened and narrowed nonequilibrium distributions $\rho$ under the energy measurement of a vacuum mode. As the $Q^2$ term in the $Y$ velocity \eqref{gs_y_vel} dominates for any $|Q|_{t=0}>1$, widened distributions show more initial movement of the pointer. The tails of widened distributions `flick' forwards and inwards, and then seem to linger. It is at this time that the pointer position could indicate the detection of an excited state (or ‘particle’) for the vacuum, or even occupy a position disallowed by standard quantum mechanics for any initial superposition of energy states (see figure \ref{marginals}). The closer the tails get to the $Q$-axis, the slower the pointer travels. Once inside $|Q|< 1/\sqrt{3}$, the tails move slower than the Born distribution (which eventually catches up). So although the widened distribution may produce the most dramatic deviations from standard quantum mechanics, the deviations are short-lived and any measuring device would need to make its measurement before the tails recede.

In contrast, the narrowed distribution shows less dramatic behaviour. It recedes slowly to the back of the Born distribution, and then some is swept out, up and around the Born distribution (see the right-hand side of figure \ref{ground_sims}). The pointer stays roughly where one would expect it to from standard quantum mechanics. 

If one were to perform an ensemble of similar preparations and measurements, recording the position of the pointer in each, one would find the marginal distribution $\rho_{\text{mar}}(Y)$. The marginal distributions for $w=1/4,1$ and $4$ are shown at $T=4$ in figure \ref{marginals}. Any deviation that this distribution shows from the marginal Born distribution would of course be indicative of quantum nonequilibrium.

\subsubsection{Long-time/large $g$ measurement of vacuum modes}\label{long_time}
\begin{figure}
\centering
\import{Chapter_4_figures/}{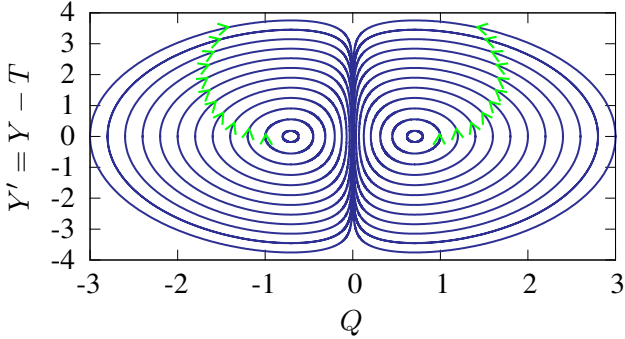}
\caption[Trajectories for vacuum energy measurement]{\label{groundstate_trajectories}Above, a selection of the trajectories for the measurement of a vacuum mode (with normal ordering). The velocity field is time independent, resulting in periodic orbits around $(\pm\sqrt{1/2},0)$. Numerical simulations show that the pointer marginals converge to stationary nonequilibrium distributions characteristic of the initial nonequilibrium state (see figure \ref{long_time_marginals}).}
\end{figure}
Let us discuss a second measurement regime, which may be thought of as valid for large $T$ and/or (since $T=g  \omega t/(2\sigma)$) large $g$. 

\begin{figure}
\centering
\resizebox{0.5\linewidth}{!}{
\import{Chapter_4_figures/}{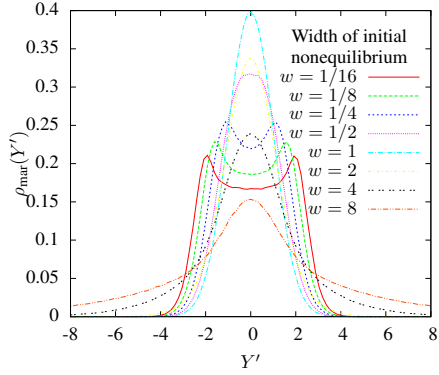}
}
\caption[Long-time marginals (incomplete relaxation)]{\label{long_time_marginals}Characteristic stationary pointer marginals $\rho_{\text{mar}}(Y')$ for energy measurement of nonequilibrium vacuum modes in the large $T$ or large $g$ approximation. In this regime, initial nonequilibrium in the field mode will produce a corresponding stationary nonequilibrium for the pointer. Field modes with larger spread produce pointer marginals with larger spread. Field modes with smaller spread form pointer marginals with central depressions.}
\end{figure}
To aid analysis, we shall continue as if we had normal ordered the Hamiltonian \eqref{hamiltonian}. This, as mentioned in section \ref{sec:cont_eq}, is equivalent to switching to the `reference frame' of the Born distribution with $Y\rightarrow Y'= Y-T$. Under normal ordering the wave function and guidance equations become
\begin{align}
\Psi&=2^{-\frac14}\pi^{-\frac12}\exp\left(-\frac12 Q^2 -\frac14Y'^2 \right),\\
\partial_TQ&=-\frac13 QY', \label{gs_q_vel_y'}\\
\partial_TY'&=\frac{2}{3}Q^2-\frac{1}{3}.\label{gs_y_vel_y'}
\end{align}
The guidance equations are now time-independent and conserve a stationary Born distribution.
The trajectories are periodic.
A selection of the trajectories produced by equations \eqref{gs_q_vel_y'} and \eqref{gs_y_vel_y'} are shown in figure \ref{groundstate_trajectories}.
The trajectories do not pass the line $Q=0$, and so we cannot find relaxation to the Born distribution for any initial $\rho$ asymmetric in $Q$.

Our numerical simulations indicate that any nonequilibrium in the vacuum mode will, in the large $T$ or large $g$ limit, produce a corresponding stationary nonequilibrium in the pointer distribution. Furthermore, from this pointer distribution, numerical simulations could deduce the initial nonequilibrium in the vacuum mode. Our simulations show that this limit will be reached at $T\sim 120$ for $1/8<w<8$. 

Eight such stationary pointer marginals are displayed in figure \ref{long_time_marginals}. These are found under the measurement of nonequilibrium vacuum modes described by width parameters ranging from $w=1/16$ to $8$. Nonequilibrium modes that are wider than equilibrium make the spread in the pointer position correspondingly wider. In contrast, for the measurement of nonequilibrium vacuum modes that are narrower than equilibrium, the pointer marginal forms a central depression whilst staying in the same region. Measurements of the pointer over an ensemble would be enough to deduce the character of the initial nonequilibrium for each case.

\subsubsection{Measurement of a single particle state}
\begin{figure}
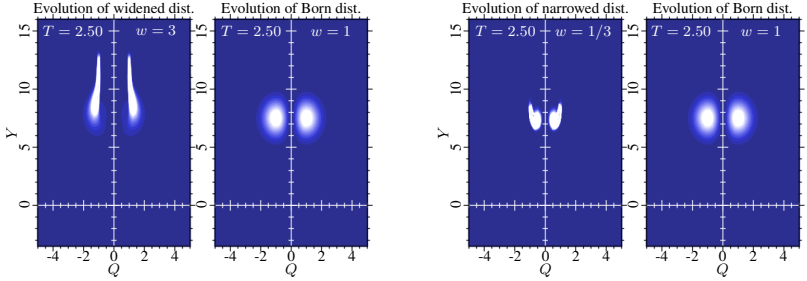

\footnotesize{\resizebox{\linewidth}{!}{
\import{Chapter_4_figures/}{excited_widened.eps_tex}
\import{Chapter_4_figures/}{excited_narrowed.eps_tex}
}}
\caption[One-particle nonequilibria under an energy measurement process]{\label{excited_sims}The evolution of nonequilibria under energy measurement of single-particle states. On the left, a widened ($w=3$) nonequilibrium distribution; on the right, a narrowed ($w=1/3$) nonequilibrium distribution. The Born distribution, shown for comparison in each case, moves at a rescaled speed of $dY/dT=3$ (although individual de Broglie trajectories have variable speeds). Pointer marginal distributions for this process are shown in figure \ref{excited_marginal}.}
\end{figure}
\begin{figure}
\centering
\hspace{0mm}
\import{Chapter_4_figures/}{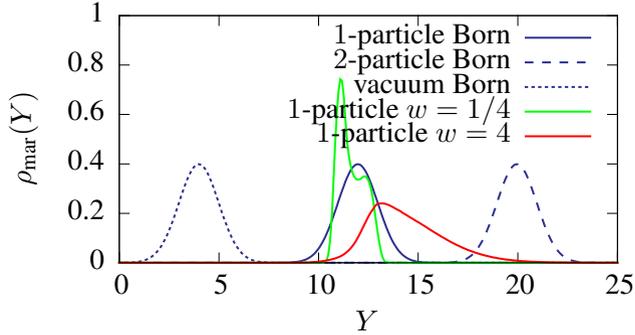}

\caption[Pointer distributions for nonequilibrium in one-particle measurement]{\label{excited_marginal}Marginal pointer distributions $\rho_{\text{mar}}(Y)$ under energy measurement of one particle state nonequilibria at time $T=4$. (For comparison we also show the Born pointer marginals for the vacuum, 1-particle and 2-particle cases.) The widened nonequilibrium ($w=4$) shows a significant probability of detecting two excitations (or `particles') instead of one, and again there is a significant probability of finding the pointer around $Y= 16$ (a position disallowed by standard quantum mechanics for any initial superposition). As in the case of the vacuum mode measurement, the narrowed nonequilibrium ($w=1/4$) will be distinguished only by its internal structure. The tendency to form a double-bump in the pointer distribution is also seen in this case.}
\end{figure}
\begin{figure}
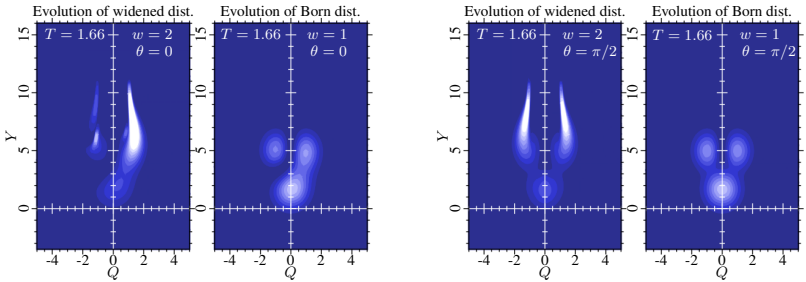

\footnotesize{\resizebox{\linewidth}{!}{
\import{Chapter_4_figures/}{super_th0.eps_tex}
\import{Chapter_4_figures/}{super_thpi2.eps_tex}
}}
\caption[Measurement of a superposition]{\label{super_sims}Evolution of joint distributions $\rho(Q,Y,T)$ under energy measurements of a nonequilibrium field mode in a superposition of a vacuum and a one-particle state with $c_0=e^{i\theta}/\sqrt{2}$ and $c_1=1/\sqrt{2}$ (Eq.\ \eqref{super1}). On the left we have taken $\theta=0$. On the right we have taken $\theta=\pi/2$. Both cases have widened distributions with $w=2$, and snapshots are taken at $T=1.66$. (For comparison, Born distributions are also shown in both cases.)}
\end{figure}
\begin{figure}
\centering
\hspace{0mm}
\import{Chapter_4_figures/}{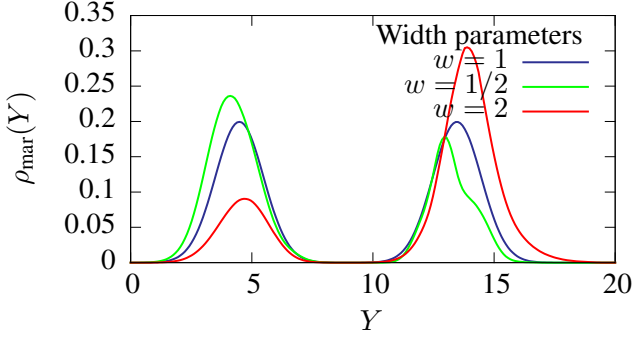}
\caption[Pointer distributions for measurement of superposition (Nonequilibrium alteration of discrete spectrum)]{\label{super_marginal}Marginal pointer distributions $\rho_{\text{mar}}(Y)$ for $c_0=c_1=1/\sqrt{2}$ and $w=1/2,1,2$, taken at $T=4.5$. Nonequilibrium is seen to cause anomalous spectra as observed by an experimenter. Similar results are obtained for other relative phases.}
\end{figure}
Under the energy measurement process, the effective wave function becomes Eq.\ \eqref{sp_measurement} and the trajectories satisfy the guidance equations \eqref{e1} and \eqref{e2}. The Born distribution evolves in the $Y$ direction at a rescaled velocity $dY/dT=3$. Since now the $Y$ velocity (Eq.\ \eqref{e2}) has terms proportional to $Q^2$ and $1/Q^2$, we might expect some increased pointer movement both for the widened and narrowed nonequilibrium cases.
In fact, our simulations show that a narrowed distribution yields relatively less pointer movement than the widened distribution (as we had for the case of the vacuum). Plots of the evolution of $\rho(Q,Y,T)$ are shown in figure \ref{excited_sims}, and marginal pointer distributions are shown in figure \ref{excited_marginal}. As in the case of the vacuum mode measurement, there is a significant probability of detecting an extra excitation or of finding the pointer in a position disallowed by standard quantum mechanics for any superposition being measured.

\subsubsection{Measurement of a superposition}
\begin{figure}
\resizebox{\linewidth}{!}{
\import{Chapter_4_figures/}{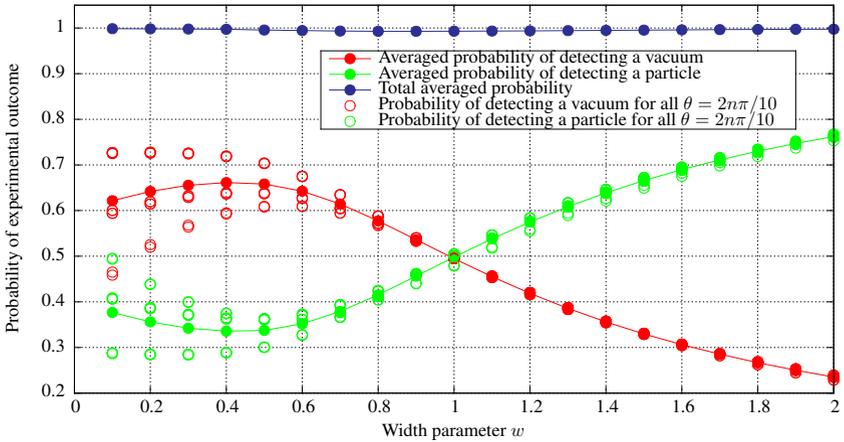}
}
\caption[Probability of nonequilibrium altering outcome of measurement]{\label{full_results}Ensemble probabilities of energy measurements for an equal superposition of particle and vacuum states as affected by quantum nonequilibrium of varying width $w$ (with results averaged over the relative phase $\theta$ in the superposition). As $|c_0|=|c_1|=1/\sqrt{2}$, there should be a 50\% probability of detecting a particle. However, widened nonequilibria give probabilities larger than 50\% for particle detection, while narrowed non-equilibria give probabilities less than 50\% for particle detection when averaged over $\theta$. (Hollow markers represent results for individual relative phases $\theta$, whilst solid markers represent averages over $\theta=2n\pi/10$, $n=1,2,\dots,10$. Dependence on the relative phase is seen to affect the outcomes only for $w\lesssim 1$.)}
\end{figure}

Quantum nonequilibrium would in general cause anomalous results for the spectra of energy measurements. To illustrate this, we take the simple example of an equal superposition of vacuum and one-particle states. Quantum mechanically, an experimenter would observe a $50\%$ probability of detecting a particle. We take $c_0=e^{i\theta}/\sqrt{2}$ and $c_1=1/\sqrt{2}$, with the wave function and velocity fields specified in Eqs.\ (\ref{super1}-\ref{super3}). The dynamics of the measurement depends strongly on the initial relative phase $\theta$ of the superposition. This is seen in figure \ref{super_sims}, where we show the time evolution of joint distributions $\rho(Q,Y,T)$. Examples of the marginal pointer distributions produced in the energy measurement process are shown in figure \ref{super_marginal}. After about $T=3.5$, all marginal pointer distributions display two distinct areas of support, meaning that an experimenter would unambiguously obtain either $\frac12\omega$ or $\frac32\omega$ in each individual energy measurement, regardless of whether nonequilibrium is present or not. However, a widened nonequilibrium distribution would cause a larger probability of obtaining the outcome $\frac32\omega_\mathbf{k}$ (`detecting a particle'), while a narrowed nonequilibrium distribution would cause the opposite effect. Although the trajectories are strongly dependent on the initial phase, the marginal pointer distributions are only weakly dependent on this. 

In practice, one might not know the initial relative phase of the superposition. To make contact with what an experimenter might actually measure (albeit in the context of our simplified field-theoretical model), we have taken an average over $10$ phases: $\theta=2\pi n/10,\,n=0,1,\dots, 9$. We run each simulation up to time $T = 4.5$ and calculate the proportion of the distribution $\rho$ that lies beyond $Y=9.0$. This is the probability of observing an excitation, whilst the proportion of $\rho$ before $Y=9.0$ is the probability of observing the vacuum. (These numbers are clear from figure \ref{super_marginal}.) 
Figure \ref{full_results} illustrates the results of this averaging process for 20 separate width parameters $w$. We find a remarkable correlation. For example, for nonequilibrium close to the Born distribution, widening the distribution will proportionally increase the ensemble probability of `detecting a particle'. Clearly, nonequilibrium would generate an incorrect energy spectrum.

\section{Conclusion}\label{Chapter4_conclusion}
We have considered the possibility that our Universe contains quantum nonequilibrium systems -- in effect a new form or phase of matter (including the vacuum) that violates the Born probability rule and which is theoretically possible in the de Broglie-Bohm formulation of quantum theory. While the practical likelihood of detecting such systems remains difficult to evaluate, we have argued that at least in principle they could exist today as relics from the very early Universe. We have provided simple field-theoretical models illustrating the effects of quantum nonequilibrium in a particle-physics context. In particular, we have seen that quantum nonequilibrium would generate anomalous spectra for standard measurements of energy, as well as generating corrections to particle-physics processes generally.

The possibility of detecting relic nonequilibrium systems today depends on uncertain features of high-energy physics and cosmology. Dark matter, which is thought to make up approximately 25\% of the mass-energy of the Universe, may consist of relic particles (such as gravitinos) that were created in the very early Universe and which have propagated essentially freely ever since. (For reviews see, for example, refs. \cite{BHS05,EO13}.) As we have seen, such particles are plausible candidates for carriers of primordial quantum nonequilibrium and we expect that particle-physics processes involving them -- for example, decay or annihilation -- would display energetic anomalies.

On the experimental front, an especially promising development would be the detection of photons from dark matter decay or annihilation. These are expected to form a sharp spectral line, probably in the gamma-ray region. Recent interest has focussed on reports of a sharp line from the Galactic centre at $\sim130\ \mathrm{GeV}$ in data from the \textit{Fermi} Large Area Telescope (LAT) \cite{B12,W12}. While the line might be a dark matter signal, its significance (and even its existence)\ is controversial. The line could be caused by a number of scenarios involving dark matter annihilations \cite{BH12}. It might also be due to decaying dark matter \cite{ITW13}, for example the decay of relic gravitinos \cite{L13,Aetal14}. (In a supersymmetric extension of the Standard Model with violation of R-parity, the gravitino is unstable and can decay into a photon and a neutrino \cite{TY00}.) On the other hand, a recent analysis of the data by the Fermi-LAT team casts doubt on the interpretation of the line as a real dark matter signal \cite{FLAT13}.

Should dark matter consist (if only partially) of relic nonequilibrium systems, we may expect to find energetic anomalies for decay and annihilation processes. However, to distinguish these from more conventional effects would require more detailed modelling than we have provided here. There is also the question of whether the anomalies are likely to be large enough to observe in practice. These are matters for future work.

In principle, it would be of interest to test dark matter decay photons for possible deviations from the Born rule (perhaps via their polarisation probabilities \cite{AV04a}). We have seen that simple perturbative couplings will transfer nonequilibrium from one field to another, leading us to expect that in general a decaying nonequilibrium particle will transfer nonequilibrium to its decay products. Another open question, however, is the degree to which the nonequilibrium might be degraded during this process. In a realistic model of a particle decay we might expect some degree of relaxation. It would be useful to study this in pilot-wave models of specific decay processes.

As a general point of principle, one might also be concerned that in the
scenario discussed in this paper the probability distribution for delocalised
field modes in the early Universe -- where the probability distribution is
presumably defined for a theoretical ensemble -- appears to have measurable
implications for decay particles in our one Universe. How can this be? A
similar point arises in the standard account of how the power spectrum for
primordial perturbations has measurable implications for our one CMB sky. In
inflationary theory, the probability distribution for a single mode
$\phi_{\mathbf{k}}$ of the inflaton field does have measurable implications in
our single Universe. As we discussed in Section \ref{IIA}, the variance
$\left\langle |\phi_{\mathbf{k}}|^{2}\right\rangle $ of the primordial
inflaton distribution appears in the power spectrum $\mathcal{P}_{\mathcal{R}%
}(k)\propto\left\langle |\phi_{\mathbf{k}}|^{2}\right\rangle $ for primordial
curvature perturbations $\mathcal{R}_{\mathbf{k}}\propto\phi_{\mathbf{k}}$ at
wave number $k$. The power spectrum $\mathcal{P}_{\mathcal{R}}(k)$ in turn
appears in the angular power spectrum $C_{l}$ (equation \eqref{Cl2}), which may be
accurately measured for our single CMB sky provided $l$ is not too small. In
the standard analysis it is assumed that the underlying `theoretical ensemble'
of Universes is statistically isotropic, which implies that the ensemble
variance $C_{l}\equiv\left\langle \left\vert a_{lm}\right\vert ^{2}%
\right\rangle $ is independent of $m$ -- where $a_{lm}$ are the harmonic
coefficients for the observed temperature anisotropy. We then in effect have
$2l+1$ measured quantities $a_{lm}$ with the same theoretical variance.
Provided $l$ is sufficiently large, one can perform meaningful statistical
tests for our single CMB sky and compare with theoretical predictions for
$C_{l}$. Statistical homogeneity also plays a role in relating the $C_{l}$'s
for a single sky to the power spectrum $\mathcal{P}_{\mathcal{R}}(k)$ for the
theoretical ensemble \cite{AV10,Muk05}. To understand how the
theoretical ensemble probability has measurable implications in a single
Universe, it is common to speak of the CMB sky as divided up into patches --
thereby providing an effective ensemble in one sky. This works if $l$ is
sufficiently large, so that the patches are sufficiently small in angular
scale and therefore sufficiently numerous. Similar reasoning applies to
particles (or field excitations) generated by inflaton decay. In this context
it is important to note that realistic particle states, as observed for
example in the laboratory, are represented by field modes defined with respect
to finite spatial volumes $V$. Almost all of the particles in our Universe
were created by inflaton decay, and in practice their states are in effect
defined with respect to finite spatial regions. By measuring particle
excitations in different spatial regions, it is possible to gather statistics
for outcomes of (for example)\ energy measurements. (One might also consider a
time ensemble in one region, but a space ensemble seems more relevant in the
case of relic decay particles.) The resulting statistical distribution of
outcomes for the decay particles will depend on the original probability
distribution for the decaying inflaton field -- just as the statistics for
patches of the CMB sky depend on the probability distribution for the inflaton
during the inflationary era. A full account would require an analysis of
inflaton decay more precise than is currently available. In particular, one
would like to understand how this process yields particle states that are
confined to finite spatial regions. It is generally understood that the decay
products form as excitations of sub-Hubble modes, with wave functions confined
to sub-Hubble distances. Depending on the details, this can correspond to
relatively small spatial distances today.
Of course,
particle wave packets will also spread out since their creation, but still we
may expect them to occupy finite spatial regions. Further elaboration of this
point lies outside the scope of this paper.

Even if there exist localised sources or spatial regions containing particles in a state of quantum nonequilibrium, it might be difficult in practice to locate those regions. In particular, if a given detector registers particles belonging to different regions without distinguishing between them, then it is possible that even if nonequilibrium is present in the individual regions it will not be visible in the data because of averaging effects. How one might guard against this in practice remains to be studied.

Finally, we have seen that the likelihood of nonequilibrium surviving until today for relic particles depends on the fact that a nonequilibrium residue can exist in the long-time limit for systems containing a small number of superposed energy states \cite{ACV14}. While this may certainly occur in principle, its detailed implementation for realistic scenarios requires further study. On the other hand, no such question arises in our scenario for relic nonequilibrium vacuum modes, since the simplicity of vacuum wave functionals guarantees that further relaxation will not occur at late times. Long-wavelength vacuum modes may be carriers of primordial quantum nonequilibrium, untouched by the violent astrophysical history that (according to our hypotheses) long ago drove the matter we see to the quantum equilibrium state that we observe today. It remains to be seen if, in realistic scenarios, the effects on particle-physics processes taking place in a nonequilibrium vacuum could be large enough to be detectable.

\section{Numerical methodology}\label{code}
Most studies of relaxation in de Broglie-Bohm theory have used the back-tracking method of ref.\ \cite{VW05} (see for instance \cite{VW05,TRV12,SC12,CV13}).
This method uses the fact that the ratio $f=\rho(\mathbf{x},t)/|\psi(\mathbf{x},t)|^2$ is conserved along trajectories. A uniform grid of final positions is evolved backwards from the final time $t_f$ to the initial time $t_i$. The final distribution is constructed from the conserved function $f$. Although this method has been successful in producing accurate results, it has the disadvantage that backtracking to $t_i$ must be carried out for each desired final time $t_f$.

We have instead chosen to integrate the continuity equation \eqref{ant_cont} directly using a finite-volume method. The method used is a variant of the corner transport upwind method detailed in sections 20.5 and 20.6 of \cite{leve02}, modified so as to apply to the conservative form of the advection equation. This algorithm has the advantage that different `high resolution limiters' may be switched off and on with ease, so that one may compare results. (We use a monotonised central (MC) limiter throughout.) The main disadvantage of this approach is a consequence of the velocity field \eqref{guidance_a} and \eqref{guidance_b} diverging at nodes (where $|\psi|\rightarrow 0$). Since such an algorithm is required to satisfy a Courant-Friedrichs-Lewy condition to maintain stability, without velocity field smoothing the algorithm is inherently unstable.
We have found that a simple way to implement a smoothing is to impose a maximum absolute value on the velocities. The maximum is taken throughout to be $1/10^{\text{th}}$ of the ratio of grid spacing to time step.

We have found that the finite-volume method is less efficient than the backtracking method over larger time scales. 
In fact, the long-time simulations shown in figure \ref{long_time_marginals} were produced using a fifth-order Runge-Kutta algorithm to evolve trajectories directly.
 However for short time scales -- the prime focus of this work -- the finite-volume method is a useful tool.

\chapter*{PREFACE TO CHAPTER 5}\addcontentsline{toc}{chapter}{PREFACE TO CHAPTER 5}
Chapter \ref{5} continues the focus on relic quantum nonequilibrium in particles, but takes a different line of approach.
The chief concern of chapter \ref{4} was the possibility that relic quantum nonequilibrium could exist in species of cosmological particles. 
It is an intriguing possibility, but a reliable estimate of whether this may be the case is too dependent upon various uncertain factors in contemporary physics.
The exact details of cosmology in the very early Universe.
Beyond the standard model particle physics (in particular that part which accounts for dark matter). 
And of course the speed and generality of reaching equilibrium is still an open topic in de Broglie-Bohm, especially when considered upon cosmological scales. 
So it seems unlikely that a reliable estimate of the likelihood of nonequilibrium persisting in relic particles can be made until at least some of these factors are ironed out.
So why not take a different approach?
The purpose of chapter \ref{5} is to consider observable consequences of quantum nonequilibrium in the context of ongoing experiment.
If the dark matter does indeed possess nonequilibrium statistics, then there is a distinct possibility that this could turn up in experiment.
But at present, there is little idea of what kind of signatures could betray the presence of nonequilibrium.
And without a phenomenology of quantum nonequilibrium in experiment, it is likely that such signatures would be overlooked or misinterpreted. 

In order to consider nonequilibrium in experiment, a context must first be selected.
And ideally such a context would satisfy two criteria. 
Firstly, it must concern favorable conditions for nonequilibrium to persist.
Secondly, the expected output in the case of equilibrium should be relatively `clean' and understood, and without too much scope for systematic or instrumental error, so that the nonequilibrium signatures may be clearly distinguished from the expected output.
Fortunately, the indirect search for dark matter provides such a favorable context.
In many dark matter models, particles may annihilate or decay into mono-energetic photons \cite{BS88,Rudaz89,B97,B04,Creview16,DW94,Abazajian01}. 
If detected by a telescope, these hypothetical mono-energetic photons would produce a line-spectrum that is commonly prefixed with the words `smoking-gun', for the reason that it would stand out clearly against the background, making it difficult to explain through other means.
A significant part of the indirect search for dark matter concerns the detection of such smoking-gun lines with space telescopes \emph{capable of single photon detection}. To date, many line searches have taken place \cite{fermi3.7,fermi5.8,Pull07,Mack08,W12,SF12-1,Albert14,HESS16,Bulbul14,Boyarsky14,A15,Hitomi16}.
Moreover, as discussed in section \ref{5.4}, the field is not unaccustomed to contentious and controversial claims concerning anomalous lines. 

For the present purposes, this is a particularly appealing scenario for the detection of quantum nonequilibrium.
For if the dark matter particles were in a state of quantum nonequilibrium prior to their decay/annihilation, then by the arguments of chapter \ref{4} this nonequilibrium would be transferred onto the mono-energetic photons.
And if these photons were sufficiently free streaming prior to their arrival at the telescope, then the nonequilibrium could be preserved until the photons arrive at the telescope. 
Hence, these telescopes have the potential to observe a nonequilibrium photon signal directly. 

That the signal has a very predictable shape in the absence of nonequilibrium is key to the results for the following reason. 
Each time a telescope measures the energy $E_\gamma$ of a photon, it produces a value $E$ according to the telescope's energy dispersion function, $D(E|E_\gamma)$.
If the telescope were subject to a true spectrum $\rho(E_\gamma)$, it would expect to record the spectrum
\begin{align*}
\rho_\text{obs}(E)=\int D(E|E_\gamma)\rho_\text{true}(E_\gamma)\mathrm{d}E_\gamma.
\end{align*}
So there are two key ingredients to the spectrum observed by the telescope.
First there is the true spectrum $\rho(E_\gamma)$, which for the present purposes may be taken to be some sort of localized bump around the line energy $E_\text{line}$.
Conventional sources of spectral effects, for instance Doppler broadening, affect the energy of each photon and hence $\rho(E_\gamma)$.
Second there is the energy dispersion function, $D(E|E_\gamma)$, which may be thought to be a function of the interaction between the photon and the telescope. 
It is this that quantum nonequilibrium affects. 
A key point to the argument is that for the current generation of telescopes, the width of $D(E|E_\gamma)$ is significantly larger than the Doppler broadened $\rho_\text{true}(E_\gamma)$.
In other words, current telescopes are incapable of resolving the expected width of the hypothetical line.
Counter-intuitively perhaps, this produces the ideal conditions to observe nonequilibrium, for it means that the true spectrum may be approximated $\rho(E_\gamma)\approx \delta(E_\gamma-E_\text{line})$, and so the observed spectrum reduces to  
\begin{align*}
\rho_\text{obs}(E)\approx D(E|E_\text{line}).
\end{align*}
Hence, under such conditions, telescopes are expected to observe their own energy dispersion function, which is directly affected by quantum nonequilibrium. 
So in principle at least, the effects of quantum nonequilibrium have the potential to be very conspicuous in the observed spectrum.

The actual spectrum observed is contextual, in that it is not just a function of the initial nonequilibrium, but also of the manner in which the telescope interacts with it. 
And as different instruments may be expected to react to the nonequilibrium in different ways, it is reasonable to expect different telescopes to disagree on the observed spectrum.
At first glance, this fact may seem to preclude any hope of experimental reproducibility. 
But it could be this very fact that distinguishes signatures of quantum nonequilibrium from other more conventional effects. 
Other remarkable signatures include the potential to record lines that appear narrower than the telescope should be capable of resolving. 
The chapter is concluded in section \ref{5.4}, with a speculative discussion on how a discovery of nonequilibrium may play out, with reference to some recent anomalous lines in the X-ray and $\gamma$-ray range.

\chapter{ANOMALOUS SPECTRAL LINES AND RELIC QUANTUM NONEQUILIBRIUM}\label{5}
\vspace{-15mm}
\begin{center}
\textit{Nicolas G. Underwood\hyperlink{address_chap_4}{$^\dagger$} and Antony Valentini\hyperlink{address_chap_4}{$^\dagger$}}
\end{center}
\begin{center}
\textit{Submitted to Phys. Rev. D} \cite{UV16}
\end{center}
\begin{center}
\hypertarget{address_chap_4}{$^\dagger$}Kinard Laboratory, Clemson University, Clemson, 29634, SC, United States of America
\end{center}

\section*{Abstract}
We describe features that could be observed in the line spectra of relic cosmological particles should quantum nonequilibrium be preserved in their statistics. According to our arguments, these features would represent a significant departure from those of a conventional origin. Among other features, we find a possible spectral broadening that is proportional to the energy resolution of the recording telescope (and so could be much larger than any conventional broadening). Notably, for a range of possible initial conditions we find the possibility of spectral line `narrowing', whereby a telescope could observe a line that is narrower than it is conventionally able to resolve. We discuss implications for the indirect search for dark matter, with particular reference to some recent controversial spectral lines.

%%%%%%%%%%%%%%%%%%%%%%%%%%%%%%%%%%%%%%%%%%%%%%%%%%%%%%%%%%%%%%%%%%%%%%%%%%%%%%%%%%%%%%%%%%%%%%%%%%%%%%%%%%%%%%
\section{Introduction}\label{5.1}
%%%%%%%%%%%%%%%%%%%%%%%%%%%%%%%%%%%%%%%%%%%%%%%%%%%%%%%%%%%%%%%%%%%%%%%%%%%%%%%%%%%%%%%%%%%%%%%%%%%%%%%%%%%%%%
In the de Broglie-Bohm pilot-wave interpretation of quantum theory \cite{deB28,BV09,B52a,B52b,Holl93}, the Born probability rule has been shown to arise dynamically in much the same way as thermal equilibrium arises in classical physics \cite{AV91a,AV92,AV01,VW05,SC12,EC06,TRV12,ACV14}.
States that obey the Born rule are called `quantum equilibrium', states that violate the Born rule are called `quantum nonequilibrium', and the process by which the former emerge from the latter is called `quantum relaxation'.
Since quantum nonequilibrium is by definition observably distinct from conventional quantum theory, and so provides a means by which the de Broglie-Bohm theory could be discriminated from other interpretations of quantum theory, attempts have been made in recent years to predict how quantum nonequilibrium might reveal itself in contemporary experimentation \cite{AV07,AV08,AV09,AV10,CV13,UV15,AV01,CV15,VPV19}.
A discovery of quantum nonequilibrium would not only demonstrate the need to re-evaluate the canonical quantum formalism, but also generate new phenomena \cite{AV91b,AV04a,AV02,PV06}, potentially opening up a large field of investigation.

It has been conjectured that the Universe could have begun in a state of quantum nonequilibrium \cite{AV91a,AV92,AV91b,AV96,AV01,AV02A}, or that exotic gravitational effects may even generate nonequilibrium \cite{AV10,AV07,AV04b,AV14}.
However, since relaxation to quantum equilibrium occurs remarkably quickly for sufficiently interacting systems \cite{VW05,TRV12,CS10}, it is expected that in all but the most exceptional circumstances, any quantum nonequilibrium that was present in the early Universe will have subsequently decayed \cite{AV91a,AV91b,AV92,AV96}. It may be said that, if de Broglie-Bohm theory is correct, then the Universe has already undergone a sub-quantum analogue of the classical heat death \cite{AV92,AV01,AV91b}. Nevertheless, data from the cosmic microwave background (CMB) do provide a possible hint for the past existence of quantum nonequilibrium, in the form of a primordial power deficit at large scales (as initially reported \cite{PlanckXV} and subsequently confirmed \cite{PlanckXI} by the Planck team). It has been argued that such a deficit is a natural prediction of de Broglie-Bohm theory \cite{AV10}, though of course the observed deficit may be caused by something else (or be a mere statistical fluctuation as some argue). Recently, predictions regarding the shape of the power deficit have been made \cite{CV15} and are currently being compared with CMB data \cite{VPV19}. Primordial quantum nonequilibrium also offers a single mechanism that could account for both the CMB power deficit and the CMB statistical anisotropies \cite{AV15}. At present these are of course only hints.

Potentially it is also possible that quantum nonequilibrium may have been preserved for some species of relic particle \cite{AV01,AV07,AV08,UV15}.
If, for instance, a relic particle species decoupled sufficiently early in the primordial Universe, and if it were sufficiently minimally interacting thereafter, there is a possibility that it may still retain nonequilibrium statistics to this day. 
A previous article \cite{UV15} explored various different means by which quantum nonequilibrium may be preserved for relic particles. An assessment of the likelihood of such scenarios requires, however, knowledge of assorted unknown contributing factors--the properties of the specific relic particle species, the correct primordial cosmology, and the extent of the speculated initial nonequilibrium.
Even so, an actual detection of relic nonequilibrium could occur if it had clear and unmistakable experimental signatures.
The purpose of this paper is to describe such signatures in an astrophysical context.

As discussed in \cite{UV15}, it is natural to consider signatures of relic nonequilibrium in the context of the indirect search for dark matter. For the reasons described in section \ref{5.2} we focus on the search for a `smoking-gun' spectral line. 
We present a field-theoretical model of spectral measurement intended to act as an analogue of a telescopic photon detector (for example a calorimeter or CCD). Whilst the model is admittedly simple, and is certainly not a realistic representation of an actual instrument, it does capture some key characteristics and demonstrates the essential difference between conventional spectral effects and those caused by quantum nonequilibrium. 

Due to the exotic nature of quantum nonequilibrium, the spectral effects we describe are something of a departure from those of a classical origin. For instance, according to our arguments the amount of spectral line broadening depends on the energy resolution of the telescope's photon detector.
Lines may acquire double or triple bumps, or as we shall discuss, more exotic profiles.
Notably, there also exists the possibility of spectral line `narrowing'--the spectral line appears narrower than the telescope should conventionally be capable of resolving.
Aside from experimental error, we are unaware of any other possible cause of this final signature, which could (if observed) constitute strong evidence for quantum nonequilibrium and the de Broglie-Bohm theory. If a source of such a nonequilibrium signal were detected and could be reliably measured, a final definitive proof could be arrived at by subjecting the signal to a specifically quantum-mechanical experiment as, for example, is described in Ref.\ \cite{AV04a}.

Our paper is organised as follows. In section \ref{5.2} we explain why we have chosen to focus on spectral lines. We summarise what appears to be the essential difference between conventional spectral effects and those caused by quantum nonequilibrium, and why this is particularly relevant for spectral lines. We also provide some background helpful to our analysis. Finally, we present an idealised and parameter-free field-theoretical model of a spectral measurement of the electromagnetic field, which is used to represent a telescope photon detector. In section \ref{5.3} we present the pilot-wave description of the model and provide explicit calculations showing the result of introducing quantum nonequilibrium.
In our concluding section \ref{5.4} we outline the phenomena that we judge the most likely to betray the presence of quantum nonequilibrium in spectral lines, and we discuss possible implications for the indirect search for dark matter with reference to three recent controversial spectral features.

\section{Modelling a telescope X/$\gamma$-ray photon detector}\label{5.2}
Relaxation to quantum equilibrium is thought to proceed efficiently for systems with sufficiently complicated quantum states \cite{VW05,TRV12,CS10}. 
For everyday matter therefore, with its long and violent astrophysical history, it is expected that quantum nonequilibrium will have long since decayed away. 
As discussed in Ref.\ \cite{UV15}, however, for more exotic particle species that decouple at very early times there exist windows of opportunity whereby quantum nonequilibrium could have been preserved. 
It is then not inconceivable that dark matter may still exhibit nonequilibrium statistics today, should it take the form of particles that indeed decoupled very early.
The search for dark matter is therefore a natural choice of context for the discussion of relic quantum nonequilibrium. For the reasons detailed below, we focus our discussion on the search for a `smoking-gun' spectral line by X/$\gamma$-ray space telescopes.
Such a line could, for example, be produced in the $\gamma$-ray range by the $XX\rightarrow\gamma\gamma$ annihilation of various WIMP dark matter candidates \cite{BS88,Rudaz89,B97,B04,Creview16} or in the X-ray range by the decay of sterile neutrinos \cite{DW94,Abazajian01}.  Recent searches have been carried out in the $\gamma$-ray range by Refs.\ \cite{fermi3.7,fermi5.8,Pull07,Mack08,W12,SF12-1,Albert14,HESS16} and in the X-ray range by Refs.\ \cite{Bulbul14,Boyarsky14,A15,Hitomi16}.
 Although dark matter does not interact directly with the electromagnetic field and so these processes are typically suppressed, it has long been argued that the detection of such a line may be among the most promising methods available to discover the nature of dark matter \cite{BS88}. 
Primarily this is because WIMP $XX\rightarrow\gamma\gamma$ annihilation would create two photons of energy $E_\gamma=m_\text{WIMP}$ and single-photon sterile neutrino decay would create photons of energy $E_\gamma=m_s/2$. The resulting spectral lines would hence yield the mass of the dark matter particle in addition to its spatial location. The hypothetical lines would furthermore appear with only very minimal (mostly Doppler) broadening ($\sim$0.1\% of $E_\gamma$) helping them to be distinguished from the background \cite{Berg12,Abazajian01}. 

The search for such dark matter lines is also, arguably, a promising context in which to consider possible signatures of relic quantum nonequilibrium.
As discussed in Ref.\ \cite{UV15}, field interactions have the effect of transferring nonequilibrium from one quantum field to another.
This means that, if a nonequilibrium ensemble of dark matter particles annihilates or decays in the manner described, we may reasonably expect some of the nonequilibrium to be transferred to the created photons.
If these photons subsequently travel to a telescope without scattering significantly, they could retain the nonequilibrium until their arrival at the telescope.  
As a result, we may focus on modelling the interaction of the detector with nonequilibrium photons, rather than considering interactions of (as yet unknown) dark matter particles. 
That the putative photons are in the X/$\gamma$-ray range is relevant as modern X/$\gamma$-ray telescopes are capable of single photon detection--an inherently quantum process.
(In particular, $\gamma$-ray telescope calorimeters are designed more like particle physics experiments than traditional telescopes.)
A model of such a detector need only describe the measurement of individual photons, affording a useful simplification.

\subsection{Role of the energy dispersion function}\label{role_of_D}
The hypothetical photon signals are close to mono-energetic, suffering only $\sim$0.1\% broadening from conventional sources \cite{Berg12,Abazajian01}. This is significant for our discussion, as we shall now explain.
When a photon of energy $E_\gamma$ arrives at a real telescope, the telescope photon detector may record a range of possible energies. 
These possible energy readings are distributed according to the detector's energy dispersion function, commonly denoted $D(E|E_\gamma)$ (see for example Refs. \cite{LAT12,XMM_MOS01}). 
This is approximately Gaussian, centred on the true energy, and with a spread quantified by the detector's energy resolution $\Delta E/E_\gamma$. (See for instance section 7 and figure 67 in Ref.\ \cite{LAT12}.) If a telescope receives photons distributed with a true spectrum $\rho_\text{true}(E_\gamma)$, it will observe a spectrum
\begin{align}
\rho_\text{obs}(E)=\int D(E|E_\gamma)\rho_\text{true}(E_\gamma)dE_\gamma\label{eq1}
\end{align}
that is convolved by $D(E|E_\gamma)$.
Conventional spectral effects, such as the $\sim$0.1\% Doppler broadening expected in the annihilation/decay lines, alter the energy of the signal photons and hence the true spectrum $\rho_\text{true}(E_\gamma)$.
This is not true of quantum nonequilibrium however.
In de Broglie-Bohm theory, the system configuration does not affect the standard Schr\"{o}dinger evolution of the quantum state \cite{Holl93}. Photons will hence arrive at the telescope with the same quantum states (and the same associated energy eigenvalues) as they would have in the absence of quantum nonequilibrium. 
Instead, as we shall see, quantum nonequilibrium affects the statistical outcomes of quantum interactions between the photons and the telescope. Thus, quantum nonequilibrium alters the detector's energy dispersion $D(E|E_\gamma)$ and not $\rho_\text{true}(E_\gamma)$. This is the essential difference between conventional spectral effects and those caused by quantum nonequilibrium. Of course, since both $\rho_\text{true}(E_\gamma)$ and $D(E|E_\gamma)$ enter into the integrand of Eq.\ \eqref{eq1}, both kinds of effects contribute to the observed spectrum $\rho_\text{obs}(E)$. The relative size of these contributions may, however, depend strongly on the context.

To illustrate this last point, consider the observation of a spectral line $\rho_\text{true}(E_\gamma)=\rho_\text{line}(E_\gamma)$ according to Eq.\ \eqref{eq1} in two separate regimes. First, consider a `high' resolution instrument in which $\Delta E/E_\gamma$ (the width of $D(E|E_\gamma)$) is significantly smaller than the width of the signal line $\rho_\text{line}(E_\gamma)$. In this case it is appropriate to make the approximation $D(E|E_\gamma)\approx \delta(E-E_\gamma)$, and so the observed spectrum \eqref{eq1} closely approximates the true signal spectrum, $\rho_\text{obs}(E)\approx \rho_\text{line}(E)$. Thus a high resolution telescope may resolve the profile of the signal line. In this regime, any moderate (order unity) alterations that quantum nonequilibrium makes to $D(E|E_\gamma)$ will be sub-dominant in the observed spectrum. 
Second, consider a `low' resolution telescope for which the width $\Delta E/E_\gamma$ of the energy dispersion function $D(E|E_\gamma)$ is significantly larger than the width of the signal line.
For this case the appropriate approximation is instead $\rho_\text{line}(E_\gamma)\approx\delta(E_\gamma-E_\text{line})$, and so the observed spectrum \eqref{eq1} closely approximates the instrument's own energy dispersion distribution, $\rho_\text{obs}(E)\approx D(E|E_\text{line})$.
In this second regime conventional broadening is not resolved, whereas moderate changes in $D(E|E_\text{line})$ (perhaps caused by quantum nonequilibrium) would be directly observed.

We may draw the remarkable conclusion that quantum nonequilibrium will be more evident in telescopes of low energy resolution.
With regards to the hypothetical WIMP and sterile neutrino lines with $\sim$0.1\% conventional (Doppler) broadening mentioned above, currently operational telescopes are within the low resolution range. 
For example, the Fermi Large Area Telescope (LAT) has a resolution of $\Delta E/E_\gamma\sim 10\%$ \cite{LAT12}\footnote{The earlier (1990s) EGRET $\gamma$-ray telescope had a resolution of $\sim$20\% \cite{EGRET93} and the 2015-operational CALET and DAMPE $\gamma$-ray telescopes achieve $\sim$2\% \cite{CALET,DAMPE}. These are also in the low resolution regime for a possible WIMP line with 0.1\% broadening.}. 
Indeed, if quantum nonequilibrium produced moderate, order unity changes in the Fermi-LAT energy dispersion distribution $D(E|E_\gamma)$, these changes would appear $\sim$100 times larger than the expected Doppler broadening of the hypothetical annihilation line in the observed spectrum.

\subsection{Idealised model of a photon detector}

To explore the potential consequences of quantum nonequilibrium, we now introduce an idealised, field-theoretical, and parameter-free model of a telescope photon detector.  
We base our model on the standard de Broglie-Bohm pilot-wave description of von Neumann measurements \cite{B52b}, which is more commonly applied to discrete spectra. 
As we shall explain, when applied to a continuous (energy) spectrum, a dispersion distribution $D(E|E_\gamma)$ and a resolution $\Delta E/E_\gamma$ naturally arise.
To model a photon detector, we model a measurement of the electromagnetic field, in which we assume only one photon is present at a time. 
To avoid complications associated with the localisability of photons, we take electromagnetic field to be quantised within a region that loosely corresponds to the body of the instrument.
For each photon, a single energy is recorded.
Over the course of the experiment, many such readings are taken and the resulting set of values may be compared with model spectra for the purposes of hypothesis testing. 

For the measurement of an observable $\mathcal{A}$ with a discrete and non-degenerate spectrum, a standard de Broglie-Bohm measurement proceeds as follows. A system with wave function $\psi(q)$ is coupled to a pointer with wave function $\phi(y)$. A commonly used interaction Hamiltonian is 
\begin{align}\label{basic_interaction}
H_\text{I}=g\mathcal{A}p_y,
\end{align}
where $p_y$ is the conjugate momentum operator of the pointer and $g$ is a coupling constant. The measurement process is taken to begin at $t=0$ and, prior to this, the coupling constant $g$ is taken to be zero.
Thereafter, $g$ is taken to be large enough to ensure that the subsequent evolution is dominated by the interaction Hamiltonian.
With this stipulation, the Schr\"{o}dinger equation takes the simple form
\begin{align}
\partial_t\Psi=-g\mathcal{A}\partial_y\Psi
\end{align}
for the duration of the measurement.
Since the spectrum of $\mathcal{A}$ is assumed discrete and non-degenerate, the system wave function may be decomposed as $\psi(q)=\sum_{n}c_n\psi_n(q)$, where $c_n$ are arbitrary coefficients and $\psi_n(q)$ are eigenstates of $\mathcal{A}$ with the eigenvalues $a_n$. The system evolves as
\begin{align}
\psi(q)\phi(y)\rightarrow \sum_n c_n \psi_n(q)\phi(y-ga_nt).\label{eq:evo}
\end{align}
The outcome probability of the experiment is determined by the effective distribution of the pointer $y$--the marginal Born distribution (hereby called the measured distribution), $\rho_{\text{meas}}(y):=\int |\Psi(q,y)|^2dq$.
Initially the different terms in the sum \eqref{eq:evo} overlap in configuration space, producing interference in the measured distribution. 
Over time each component moves at a speed proportional to its eigenvalue $a_n$. 
If the pointer is prepared in a localised state (perhaps a Gaussian), then after a sufficient time (deemed the duration of the measurement) the measured distribution will become a sum of non-overlapping terms,
\begin{align}\label{disjoint}
\rho_{\text{meas}}(y)=\sum_n|c_n|^2|\phi(y-ga_nt)|^2,
\end{align}
and will no longer exhibit interference.
Hence, at the end of each measurement, the pointer is found in one of a number of disjoint regions that correspond to the non-overlapping terms in Eq.\ \eqref{disjoint}.
In the absence of quantum nonequilibrium, there is a $|c_m|^2$ chance of finding the pointer in the $m$th region.
If a single measurement concludes with the pointer in the $m$th region, this implies a measurement outcome of the corresponding eigenvalue, $a_m$.
The discrete spectrum, $|c_n|^2$, may then be reconstructed by repeated measurements over an ensemble.

According to textbook quantum mechanics, quantum state collapse occurs at the end of each measurement in order to ensure that the pointer is found in a single one of the disjoint regions. 
In the de Broglie-Bohm account, the system (which always occupies a definite position in configuration space) is simply found in one of the regions, with no need for any non-unitary evolution.
Instead an `effective collapse' occurs as, once the components of the wave function \eqref{eq:evo} have properly separated, subsequent evolution of the system configuration is determined solely by the component that contains the configuration. 

This formulation accounts only for the measurement of observables with discrete spectra. One cannot associate disjoint regions in $y$ to eigenvalues on a continuous (energy) scale. Instead, given a particular pointer position, the energy of the incident photon must be estimated. The model photon detector measures the total (normal-ordered) Hamiltonian of the free-space electromagnetic field $\mathopen{:}H_{\text{EM}}\mathclose{:}$, so that in place of Eq.\ \eqref{basic_interaction} the interaction Hamiltonian is
\begin{align}
H_\text{I}=g\mathopen{:}H_{\text{EM}}\mathclose{:} p_y.
\end{align}
For an initial single-photon state $\ket{E_\gamma}$, the state evolution becomes simply
\begin{align}
\ket{E_\gamma}\ket{\phi(y)}\rightarrow \ket{E_\gamma}\ket{\phi(y-gE_\gamma t)}.\label{eq:state_vector}
\end{align}
Although the photon has an exact energy, the quantum uncertainty in the initial position of the pointer produces uncertainty in the pointer position at any later time.
The probability density of finding the pointer at a position $y$ is given by
\begin{align}
\rho_\text{meas}(y,t)=|\phi(y-gE_\gamma t)|^2.\label{eq:meas_point}
\end{align}
If this `measured distribution' were known, then one could correctly infer the true photon energy $E_\gamma$. 
In a single measurement, however, the pointer is found at a single position that is distributed as $\rho_\text{meas}(y,t)$.
The best estimate of the true energy is that which was most likely to have caused the observed pointer position. 
Taking the initial pointer wave packet $|\phi(y)|^2$ to be Gaussian and centred at $y=0$, this amounts to assigning the energy
\begin{align}
E=y/gt.\label{y_to_E}
\end{align}
The energy dispersion function $D(E|E_\gamma)$ is the distribution of possible energy values recorded by the instrument given the true energy $E_\gamma$. If the pointer packet has initial variance $\sigma_y^2$, the distribution of measured pointer positions \eqref{eq:meas_point} may be translated into the distribution of recorded energies \eqref{y_to_E}, giving
\begin{align}
D(E|E_\gamma) = \frac{1}{\sqrt{2\pi}}\frac{gt}{\sigma_y}e^{-\frac{1}{2}\left(\frac{gt}{\sigma_y}\right)^2(E-E_\gamma)^2}.\label{eq:e_dis_PDF}
\end{align}
Since this is a Gaussian, the energy resolution (the half-width of the fractional 68\% containment window) is simply the fractional standard deviation,
\begin{align}
\frac{\Delta E}{E_\gamma}=\frac{\sigma_y}{gtE_\gamma}.\label{eq:e_dis}
\end{align}
For hypothesis testing, we need to know the spectrum we expect to observe for each potential true spectrum $\rho_\text{true}(E_\gamma)$ . This is given by the convolution \eqref{eq1}, that for our model becomes
\begin{align}
\rho_\text{obs}(E)=\int_0^\infty\rho_\text{true}(E_\gamma)\frac{1}{\sqrt{2\pi}}\frac{gt}{\sigma_y}e^{-\frac{1}{2}\left(\frac{gt}{\sigma_y}\right)^2(E-E_\gamma)^2}\,dE_\gamma,\label{recon_simple}
\end{align}
which is a simple `Gaussian blur' (or Weierstrass transform) of the true spectrum.

The duration $t$ of the measurement appears in the denominator of Eq.\ \eqref{eq:e_dis}. Thus, in a sense the precision of the energy reading improves with the run time.
But $t$ appears only as a factor in the quantity $gt$, so a larger coupling constant would also improve the precision.
In section \ref{5.3} it will be useful to rescale the pointer variable $y$ in terms of its initial standard deviation $\sigma_y$, where it turns out that $\sigma_y$ appears only in the quantity $gt/\sigma_y$.
Thus a narrower pointer packet will also improve the precision.
With this in mind, we define the rescaled time variable
\begin{align} 
T= \frac{gtE_\gamma}{\sigma_y}=\left(\frac{\Delta E}{E_\gamma}\right)^{-1},\label{eq:T}
\end{align}
where $\Delta E/E_\gamma=1/T$ is the resolution of the telescope.

The variables $g$, $\sigma_y$ and $t$ are the only free parameters in this model. The definition \eqref{eq:T} allows their replacement with the single easily interpreted quantity $T$.
Thus, for instance, the model may reproduce a roughly EGRET resolution of $20\%$ at $T=5$, a roughly Fermi-LAT resolution of $\sim$10\% at $T=10$ or a roughly CALET/DAMPE resolution of $\sim$2\% at $T=50$.

The true energy $E_\gamma$ is included in the definition \eqref{eq:T} so that the rescaled time $T$ is exactly the reciprocal of the energy resolution \eqref{eq:e_dis}--no matter the true energy of the incident photon. 
Consequently, the effects of quantum nonequilibrium we describe in section \ref{5.3} are independent of the energy of the spectral line.

\section{Nonequilibrium spectral lines}\label{5.3}
As discussed in section \ref{5.2}, we consider a `low' resolution telescope in which the detector energy resolution $\Delta E/E_\gamma$ is significantly larger than the width of the spectral line. 
The signal photons may then be taken to be approximately mono-energetic, $\rho_\text{true}(E_\gamma)=\delta(E_\gamma-E_\text{line})$, and hence the telescope is expected to record photon energies distributed according to the detector energy dispersion function at the line energy, $\rho_\text{obs}(E)=D(E|E_\text{line})$.
In the presence of quantum nonequilibrium, we expect to observe deviations from $D(E|E_\text{line})$.
\subsection{De Broglie-Bohm description of model}
To calculate how these deviations may appear in practice, we now provide a de Broglie-Bohm description of the photon detector model introduced in section \ref{5.2}. For this we require a coordinate representation of the electromagnetic field. We work in the Coulomb gauge, $\nabla\cdot \mathbf{A}(\mathbf{x},t)=0$, with the field expansion
\begin{align}
\mathbf{A}(\mathbf{x},t)=\sum_{\mathbf{k}s}\left[A_{\mathbf{k}s}(t)\mathbf{u}_{\mathbf{k}s}(\mathbf{x})+A_{\mathbf{k}s}^*(t)\mathbf{u}_{\mathbf{k}s}^*(\mathbf{x})\right],\label{u_expansion}
\end{align}
where the mode functions
\begin{align}
\mathbf{u}_{\mathbf{k}s}(\mathbf{x})=\frac{\bm{\varepsilon}_{\mathbf{k}s}}{\sqrt{2\varepsilon_0 V}}e^{i\mathbf{k.x}}
\end{align}
and their complex conjugates define a basis, and $V$ is a normalisation volume that may be thought to correspond loosely to the volume of the instrument.
To avoid duplication of basis elements $\mathbf{u}^*_{\mathbf{k}s}$ with $\mathbf{u}_{\mathbf{-k}s}$, the summation \eqref{u_expansion} should be understood to extend over only half the possible wave vectors $\mathbf{k}$. (See for instance Ref.\ \cite{schiff}.)
This expansion allows the energy of the electromagnetic field to be written as
\begin{align}
U&=\frac12 \int_V\mathrm{d}^3x\left(\varepsilon_0 \mathbf{E}^2+\frac{1}{\mu_0}\mathbf{B}^2\right)\\
&=\sum_{\mathbf{k}s}\frac{1}{2}\left(\dot{A}_{\mathbf{k}s}\dot{A}_{\mathbf{k}s}^*+\omega_\mathbf{k}^2A_{\mathbf{k}s}A_{\mathbf{k}s}^*\right),\label{comp_HOs}
\end{align}
where $\omega_\mathbf{k}=c|\mathbf{k}|$.
Eq.\ \eqref{comp_HOs} corresponds to a decoupled set of complex harmonic oscillators of unit mass. We prefer instead to work with real variables and so decompose $A_{\mathbf{k}s}$ into its real and imaginary parts 
\begin{align}
A_{\mathbf{k}s}=q_{\mathbf{k}s1} +iq_{\mathbf{k}s2}.
\end{align}
The free field Hamiltonian may then be written as
\begin{align}
H_0=\sum_{\mathbf{k}sr}H_{\mathbf{k}sr}\label{ham_decoupled}
\end{align}
with $r=1,2$, where
\begin{align}
H_{\mathbf{k}sr}=\frac{1}{2}\left(p_{\mathbf{k}sr}^2+\omega_{\mathbf{k}}^2q_{\mathbf{k}sr}^2\right)
\end{align}
and where $p_{\mathbf{k}sr}$ is the momentum conjugate to $q_{\mathbf{k}sr}$.

The variables $q_{\mathbf{k}sr}$ and y are the configuration-space `beables'. Together they specify the configuration of the field and pointer. By rescaling the beable coordinates,
\begin{align}
Q_{\mathbf{k}sr}&=\sqrt{\frac{\omega_{\mathbf{k}}}{\hbar}}q_{\mathbf{k}sr},\quad Y=\frac{y}{\sigma_y},
\end{align}
and using the rescaled time parameter \eqref{eq:T}, the Schr\"{o}dinger equation may be written as\footnote{In Eqs.\ \eqref{schro}--\eqref{gen_guidance_2}, $E_\gamma$ should be understood to be a reference energy that will later refer to the energy of the incident photon.}
\begin{align}
\partial_T\Psi +\frac12 \sum_{\mathbf{k}sr}\frac{E_\mathbf{k}}{E_\gamma}\left(-\partial_{Q_{\mathbf{k}sr}}^2 + Q_{\mathbf{k}sr}^2-1\right)\partial_Y\Psi=0.\label{schro}
\end{align}
The following de Broglie-Bohm guidance equations may then be derived by using a similar method to that used in \cite{UV15} (based on general expressions derived in \cite{SV08}):
\begin{align}
\partial_T Q_{\mathbf{k}sr}=&\frac{E_\mathbf{k}}{E_\gamma}\left(-\frac13 \Psi\partial_{Q_{\mathbf{k}sr}}\partial_Y \Psi^* +\frac16 \partial_Y\Psi\partial_{Q_{\mathbf{k}sr}}\Psi^*\right.\nonumber\\
&\left.+\frac16 \partial_{Q_{\mathbf{k}sr}}\Psi\partial_Y\Psi^* -\frac13 \Psi^*\partial_{Q_{\mathbf{k}sr}}^2\Psi \right)/|\Psi|^2,\label{gen_guidance_1}\\
\partial_TY=&\sum_{\mathbf{k}sr}\frac{E_\mathbf{k}}{E_\gamma}\left(-\frac16 \Psi\partial_{Q_{\mathbf{k}sr}}^2\Psi^*+\frac16 \partial_{Q_{\mathbf{k}sr}}\Psi\partial_{Q_{\mathbf{k}sr}}\Psi^*\right.\nonumber\\
&\left.-\frac16 \Psi^*\partial_{Q_{\mathbf{k}sr}}^2\Psi + \frac12 \left(Q_{\mathbf{k}sr}^2 -1\right)|\Psi|^2\right)/|\Psi|^2.\label{gen_guidance_2}
\end{align}
These are the equations of motion for the $(\{Q_{\mathbf{k}sr}\},Y)$ configuration under a general quantum state. 
Prior to the spectral measurement, one mode of the field `contains' a nonequilibrium photon of energy $E_\gamma$. The beable associated with this photon-carrying mode will be referred to as $Q$.
Henceforth, all summations or products over $\mathbf{k}sr$ should be understood to exclude the mode that contains the photon.
With this in mind, the wave function(al) of the pointer-field system may be written as
\begin{align}
\Psi=&\underbrace{(2\pi)^{-\frac14}\exp\left[-\frac14\left(Y-T\right)^2\right]}_{\phi}\nonumber\\
&\times\underbrace{2^{\frac12}\pi^{-\frac14}Q\exp\left[-\frac12 Q^2\right]}_{\chi_1}
 \prod_{\mathbf{k}sr}\underbrace{\pi^{-\frac14} \exp\left[-\frac12 Q_{\mathbf{k}sr}^2\right]}_{\chi_0}.\label{wvfn_all_modes}
\end{align}
Here, $\phi$ is the pointer packet while $\chi_0$ and $\chi_1$ refer to harmonic oscillator ground and first excited states respectively. For this specific state the guidance Eqs.\ \eqref{gen_guidance_1} and \eqref{gen_guidance_2} become
\begin{align}
\partial_TQ&=\frac16\left(\frac1Q-Q\right)(Y-T),\nonumber\\
\partial_TQ_{\mathbf{k}sr}&=-\frac16 \frac{E_{\mathbf{k}}}{E_\gamma}Q_{\mathbf{k}sr}(Y-T),\nonumber\\
\partial_TY&=\frac{1}{6Q^2}+\frac13 Q^2 +\frac16 +\sum_{\mathbf{k}sr}\frac{E_\mathbf{k}}{E_\gamma}\left(\frac13 Q_{\mathbf{k}sr}^2 -\frac16\right)\label{spec_guidance_full}.
\end{align}

The pointer is coupled to the total energy of the field.
Quantum-mechanically the vacuum modes are effectively uncoupled from the pointer, as is evident from the simple Schr\"{o}dinger evolution \eqref{wvfn_all_modes}.
But in the de Broglie-Bohm treatment the guidance Eqs.\ \eqref{spec_guidance_full} describe a system in which the beables of each vacuum mode, $Q_{\mathbf{k}sr}$, are coupled directly to the pointer, and through their interaction with the pointer they are coupled indirectly to each other.
(For more details on the energy measurement of a vacuum mode see \cite{UV15}.)
This is accordingly a very complex high-dimensional system. Since our purpose is not to provide an accurate description of a real telescope photon detector, but merely to provide illustrative, qualitative examples of possible phenomena, we now truncate the model.

As a first approximation to the full de Broglie-Bohm model of Eqs.\ \eqref{spec_guidance_full}, we consider a system in which the pointer beable is decoupled from the vacuum mode beables. In this reduced system, only the beable $Q$ of the excited mode affects the evolution of the pointer, and so the system is effectively two dimensional. 

\subsection{Lengthscale of nonequilibrium spectral anomalies}
\begin{figure}
\centering
\hspace{0mm}\includegraphics{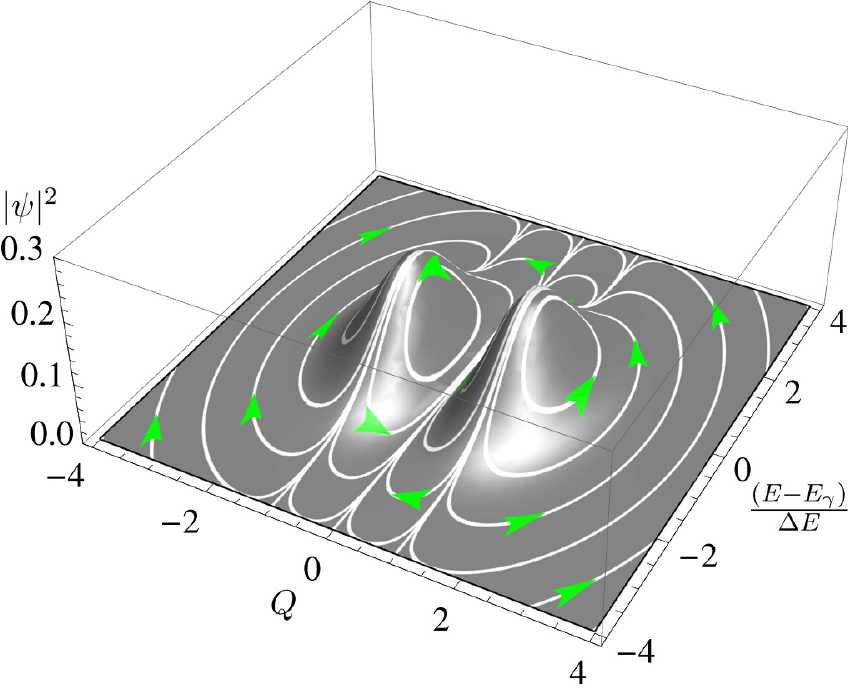}
\caption[Periodic field trajectories]{\label{fig:trajectories}Periodic orbits in the configuration (Eq.\ \eqref{paths}) produced by the guidance Eqs.\ \eqref{eq:reduced_guidance} contrasted with the quantum equilibrium distribution \eqref{reduced_equilibrium}.}
\end{figure}
Rather than using the variable $Y$ to track the evolution of the pointer, it is convenient to translate this directly into an energy reading. To this end, define the variable $\text{dev}E:=\left(E-E_\gamma\right)/\Delta E$. This is the deviation from a perfect energy reading in units of the detector energy resolution $\Delta E$, and is related to $Y$ by $\text{dev}E=Y-T$. The truncated configuration space may then be spanned by coordinates $(Q,\text{dev}E)$. In terms of these coordinates, the quantum equilibrium distribution
\begin{align}\label{reduced_equilibrium}
|\Psi|^2= \frac{1}{\sqrt{2\pi}}\exp\left(-\frac12 \text{dev}E^2\right)\frac{2}{\sqrt{\pi}}Q^2\exp\left(-Q^2\right),
\end{align}
and the guidance Eqs.\ 
\begin{align}
\partial_TQ=&\frac{1}{6}\text{dev}E\left(\frac1Q -Q\right),\nonumber\\
\partial_T\text{dev}E=&\frac{1}{6Q^2}+\frac{1}{3}Q^2-\frac{5}{6},\label{eq:reduced_guidance}
\end{align}
are both time-independent. 
The guidance equations specify stationary orbits around points at \\\mbox{$(\pm\sqrt{5\pm\sqrt{17}}/2,0)$}, with paths given by
\begin{align}
\text{constant}=\text{dev}E^{2}/2+Q^2-\ln|Q|-\ln|Q^2-1|.\label{paths}
\end{align}
These orbits are contrasted with quantum equilibrium in figure \ref{fig:trajectories}. 
Each of the orbits corresponds to an energy reading that oscillates between an overestimation and an underestimation of the true photon energy $E_\gamma$.  
The extent of these oscillations is stationary with respect to the energy scale $\Delta E$ (the detector energy resolution), rather than any fixed energy scale. 
If considered on a fixed energy scale, the oscillations would appear to shrink as the system evolves and the model, in effect, describes the outcome of an increasingly high resolution telescope (according to Eq. \eqref{eq:T}).  
As a result, the model reflects the discussion in section \ref{role_of_D}--spectral anomalies caused by quantum nonequilibrium should be expected to be observed on the scale of the telescope energy resolution, rather than any scale independent of the telescope.

\begin{figure}
\includegraphics[width=0.33\linewidth]{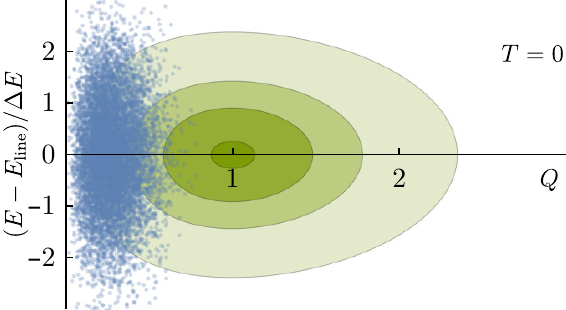}%\hspace{3pt}
\includegraphics[width=0.33\linewidth]{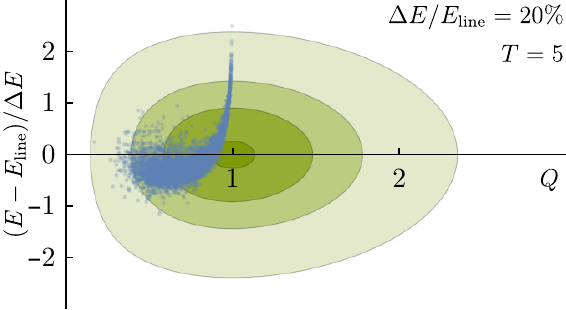}%\hspace{3pt}
\includegraphics[width=0.33\linewidth]{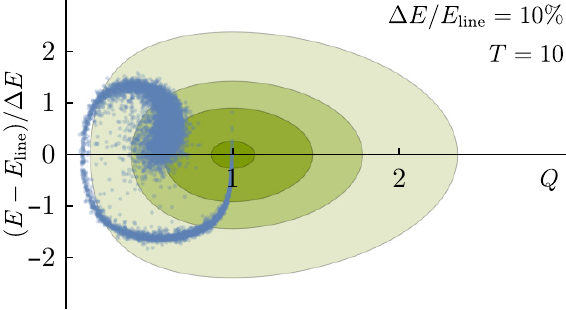}%

\includegraphics[width=0.33\linewidth]{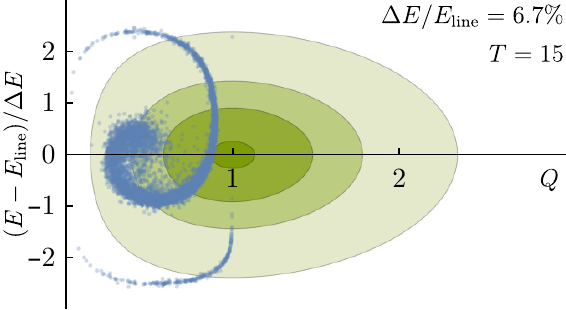}%\hspace{3pt}
\includegraphics[width=0.33\linewidth]{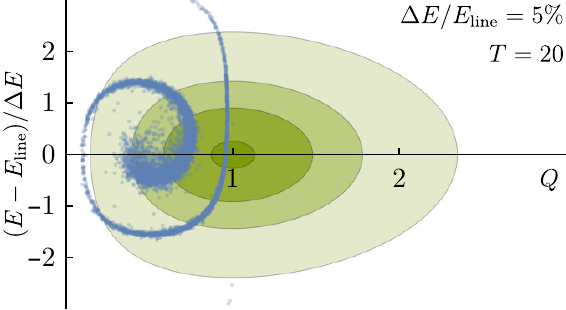}%\hspace{3pt}
\includegraphics[width=0.33\linewidth]{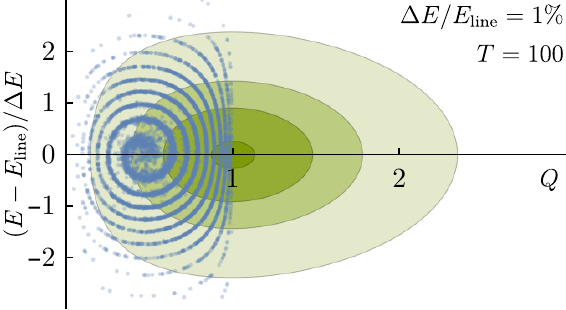}%
\caption[Spectral measurement of $w=1/4$ nonequilibrium photon]{\label{w_quarter_relaxation}The evolution of a nonequilibrium ensemble of model spectral measurements (represented by 10,000 points). The individual trajectories that compose the nonequilibrium distribution are displayed in blue and the equilibrium distribution is displayed in green. Only $Q>0$ is shown as the behaviour is identical for $Q<0$. Each frame displays the state at particular time $T$ in the evolution. Since $\Delta E/E_\text{line}=T^{-1}$, each frame is also the final system state produced at a particular energy resolution.
The corresponding spectral line profiles (at the corresponding telescope energy resolutions) are shown in figure \ref{w_quarter_marginal}.  
The initial nonequilibrium distribution for the field mode used in this figure is equal to the equilibrium distribution narrowed by a factor of 4 ($w=1/4$). 
In the model, much of the support of such $w<1$ nonequilibrium distributions is confined to the region $|Q|<1$ with the result of producing sharp, narrowed spectral lines (see Fig.\ \ref{w_quarter_marginal}). Frame 6 shows the formation of the fine structure that is a hallmark of quantum relaxation.}
\end{figure}
\begin{figure}
\includegraphics[width=0.33\linewidth]{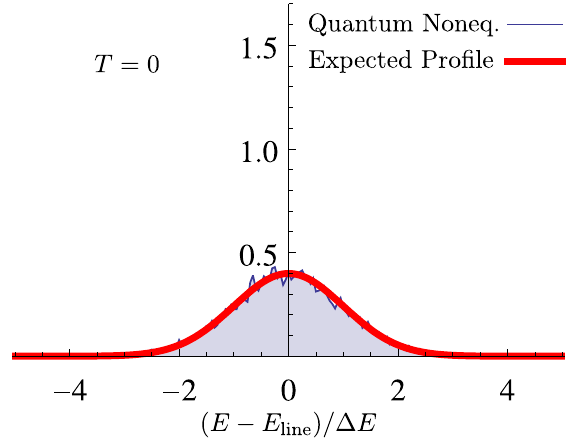}%\hspace{-2mm}
\includegraphics[width=0.33\linewidth]{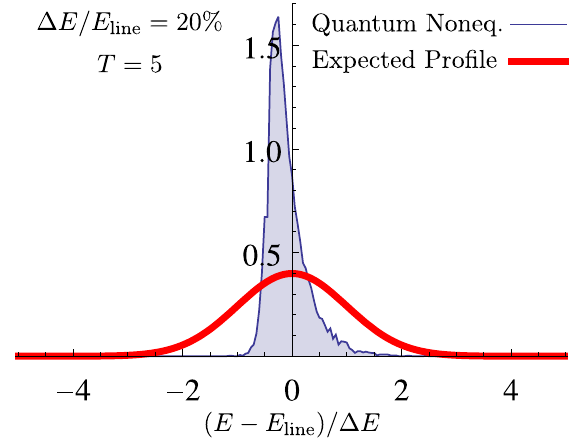}%\hspace{-2mm}
\includegraphics[width=0.33\linewidth]{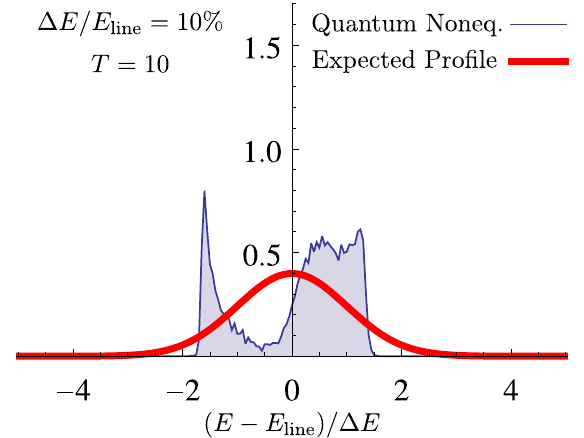}%

\includegraphics[width=0.33\linewidth]{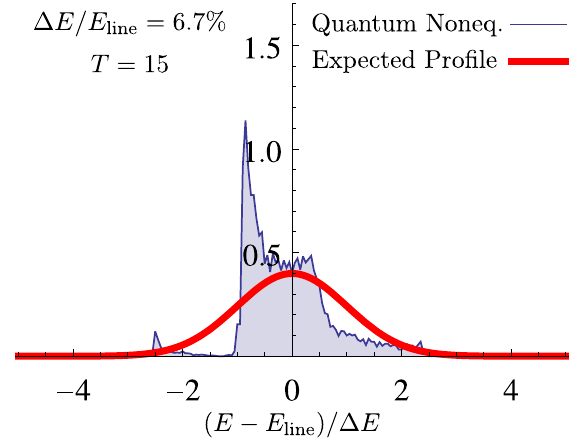}%\hspace{-2mm}
\includegraphics[width=0.33\linewidth]{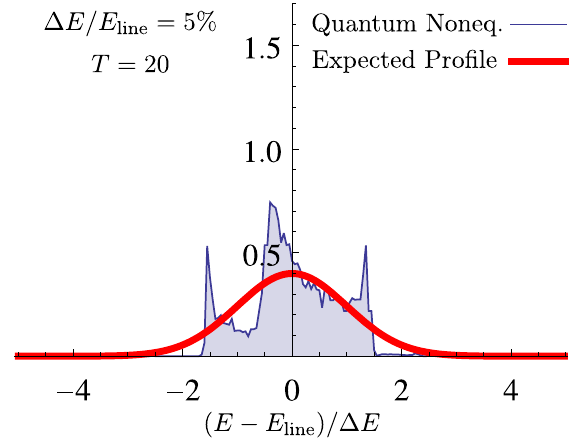}%\hspace{-2mm}
\includegraphics[width=0.33\linewidth]{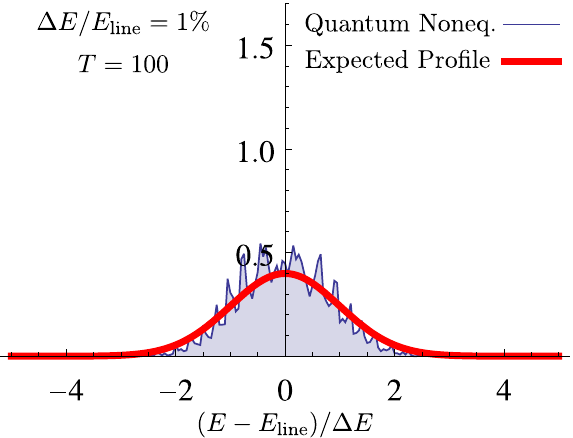}%
\caption[Outcome of measurement of $w=1/4$ nonequilibrium photon]{\label{w_quarter_marginal}Model spectral line profiles produced by a $w=1/4$ quantum nonequilibrium (corresponding to the frames shown in figure \ref{w_quarter_relaxation}), contrasted with the expected line profile, $D(E|E_\text{line})$. Each frame corresponds to a particular telescope energy resolution $\Delta E/E_\text{line}$. (The third frame, for instance, shows a resolution of $\Delta E/E_\text{line}=10\%$, approximately that of the Fermi-LAT.) The lines are given by the marginal distribution in variable $\text{dev}E:=(E-E_\text{line})/\Delta E$ of the frames in figure \ref{w_quarter_relaxation}. The plots are histograms that have been normalised to represent a probability distribution (plotted on the vertical axis) and hence there is a small amount of statistical fluctuation due to the finite sample size of 10,000. 
Frames 2-4 display distinct signatures of quantum nonequilibrium that could be searched for in experimental data. 
These show features that are too narrow to be conventionally resolved at the corresponding instrument resolution.
Frame 3 also shows a clear double bump that could not be resolved conventionally. These features are commonly produced by the model for $w<1$ nonequilibria.}
\end{figure}

\begin{figure}
\includegraphics[width=0.33\linewidth]{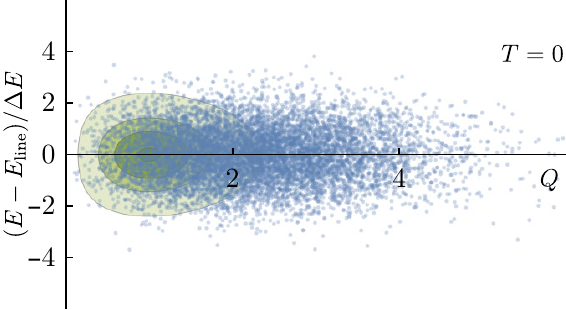}%\hspace{3pt}
\includegraphics[width=0.33\linewidth]{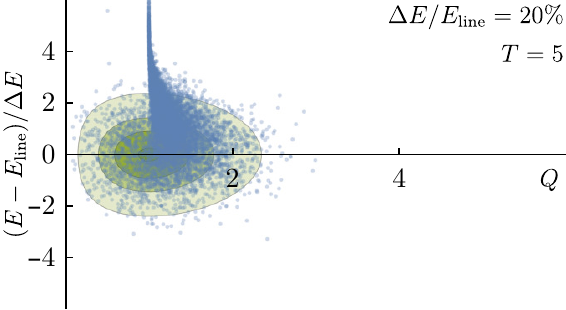}%\hspace{3pt}
\includegraphics[width=0.33\linewidth]{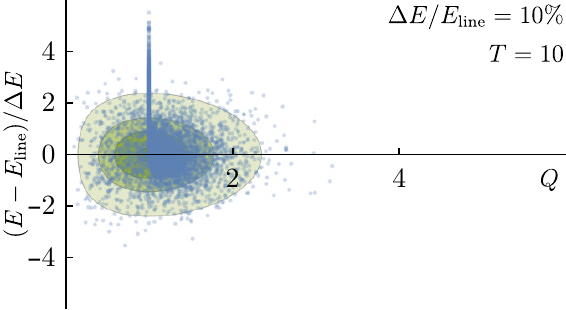}%

\includegraphics[width=0.33\linewidth]{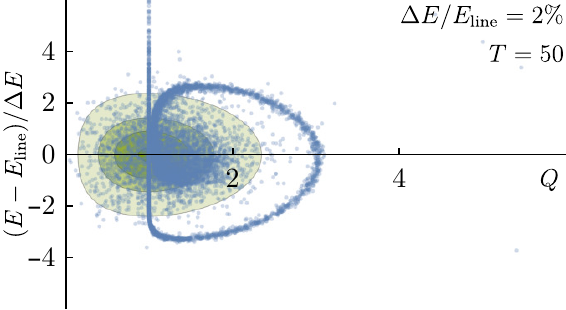}%\hspace{3pt}
\includegraphics[width=0.33\linewidth]{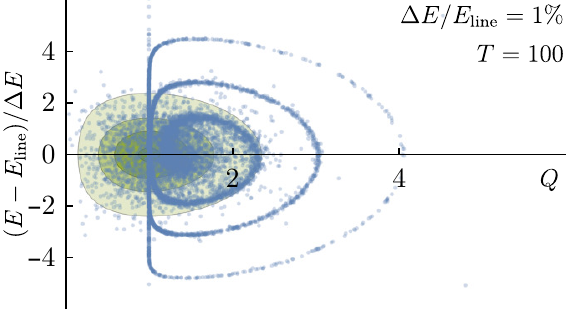}%\hspace{3pt}
\includegraphics[width=0.33\linewidth]{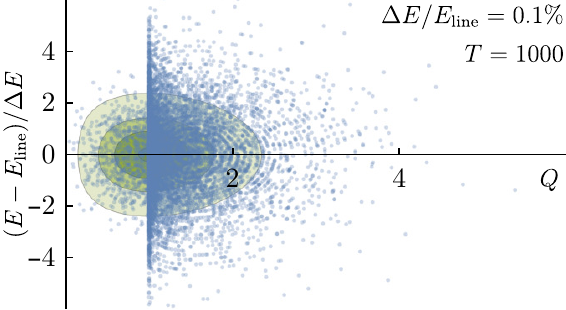}%
\caption[Spectral measurement of $w=2$ nonequilibrium photon]{\label{fig:w=2_relaxation} The evolution of an ensemble of model spectral measurements (represented by 10,000 points) subject to a $w=2$ quantum nonequilibrium. The individual beable configurations that compose the nonequilibrium ensemble are displayed in blue and the equilibrium distribution is displayed in green. Each frame displays the ensemble at a particular time $T$ in the evolution. Since $\Delta E/E_\text{line}=T^{-1}$, these are also the ensembles found upon the completion of a measurement at the corresponding energy resolution. Only $Q>0$ is shown as the behaviour is identical for $Q<0$. In the model, $w>1$ quantum nonequilibrium distributions find more of their bulk confined to the regions $|Q|>1$ (see figure \ref{fig:trajectories}). For lower resolution measurements (frames 2 and 3) this results in an overestimation of the energy and a broadened line. Towards higher resolutions, fine structure develops, leading to a partial quantum relaxation. Spectral line profiles at the corresponding energy resolutions are shown for $w>1$ nonequilibrium in figure \ref{w4_marginal} (we have used $w=4$ for clarity).} 
\end{figure}

\subsection{Introducing quantum nonequilibrium}
As the quantum nonequilibrium evolves (and in effect the system models an increasingly high resolution instrument) we expect to see a dynamic relaxation that will contribute transients.
(For an introduction to the mechanics of quantum relaxation, see Ref.\ \cite{AV91a} and chapter \ref{2} of this thesis.)
In order to provide explicit calculations displaying the line profiles produced by quantum nonequilibrium in the photon detector model, we must first specify the initial nonequilibrium photon distributions.
There are, however, as yet no a priori indications on the nature of the nonequilibrium distributions that could be present in the photon statistics--the possible or likely shape and extent of the deviations from the Born distribution remain an open question.
With this in mind, in order to provide a simple parameterisation of sample nonequilibrium distributions we use a `widening parameter' $w$ that acts to specify the initial distribution as  
\begin{align}
\rho_0(Q,\text{dev}E)=|\psi_0(Q/w,\text{dev}E)|^2/w.
\end{align}
To investigate the consequences of introducing such photon quantum nonequilibrium into the model detector, ensembles of 10,000 configurations were evolved for each of the widening parameters $w=1/16,1/8,1/4,1/2,2,4,8$ using a standard Runge-Kutta (Cash-Karp) algorithm. 
Representative snapshots of the resulting nonequilibrium distributions and spectral lines are illustrated in figures \ref{w_quarter_relaxation} through \ref{w4_marginal}.

Figure \ref{w_quarter_relaxation} shows snapshots in the evolution of the quantum nonequilibrium for $w=1/4$. (Only the $Q>0$ half of the configuration space is shown as the behaviour is identical for $Q<0$.) 
Figure \ref{w_quarter_marginal} contrasts the resulting spectral lines (at the corresponding detector energy resolutions) with the expected line profile, $D(E|E_\text{line})$.

In the model, trajectories are confined to fixed orbits (displayed in figure \eqref{fig:trajectories}) that oscillate between $\text{dev}E>0$ where the photon energy is overestimated and $\text{dev}E<0$ where it is underestimated. Since larger orbits have a larger period, trajectories that begin close to each other (and so initially oscillate in phase) gradually desynchronise. This produces swirling patterns in the configuration space distributions that grow into fine structure, a hallmark of the process of relaxation in classical mechanics and also in de Broglie-Bohm quantum theory. 

\begin{figure}
\includegraphics[width=0.33\linewidth]{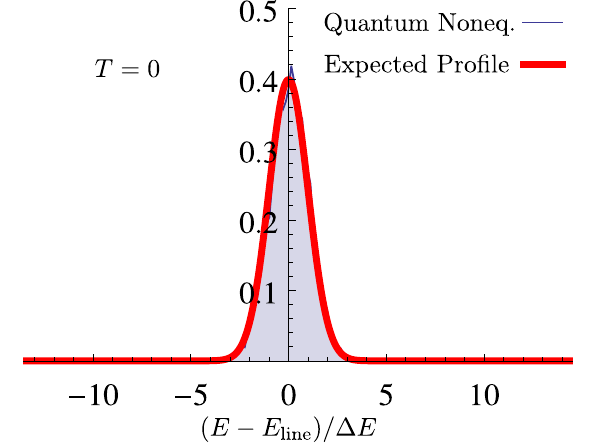}%\hspace{-2mm}
\includegraphics[width=0.33\linewidth]{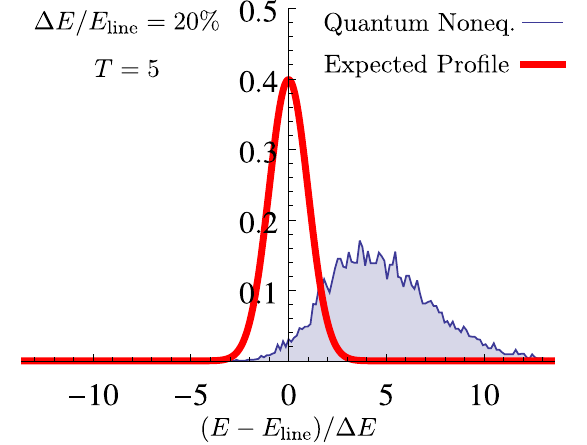}%\hspace{-2mm}
\includegraphics[width=0.33\linewidth]{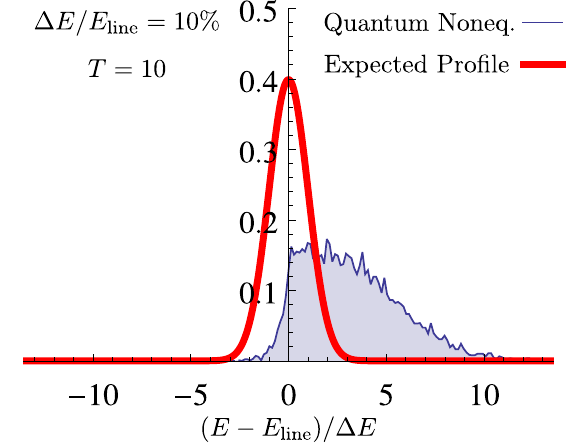}%

\includegraphics[width=0.33\linewidth]{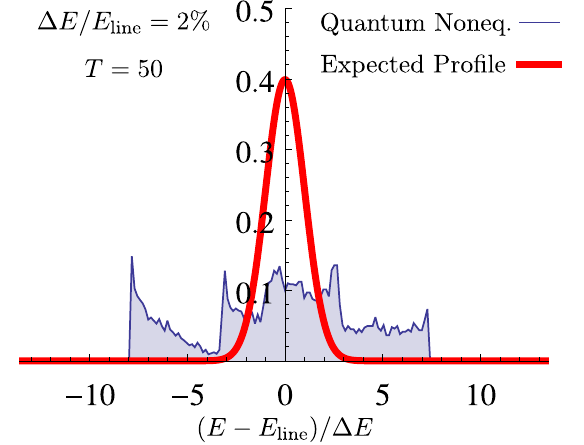}%\hspace{-2mm}
\includegraphics[width=0.33\linewidth]{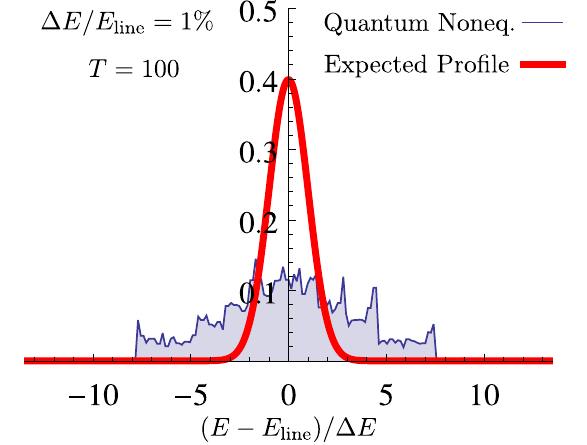}%\hspace{-2mm}
\includegraphics[width=0.33\linewidth]{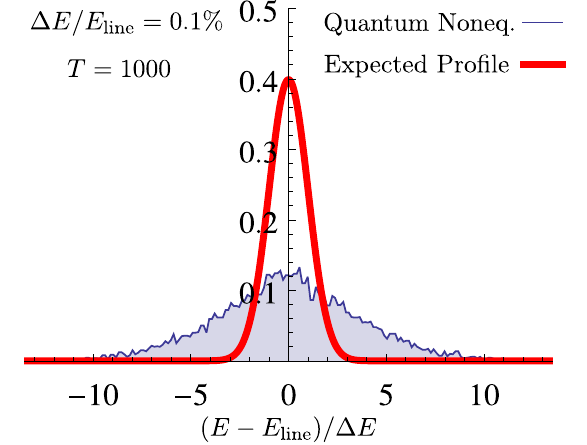}%
\caption[Outcome of measurement of $w=4$ nonequilibrium photon measurement]{\label{w4_marginal}Spectral line profiles produced by a $w=4$ quantum nonequilibrium, contrasted with the expected line profile, $D(E|E_\text{line})$. Each frame corresponds to a model telescope of a particular energy resolution. (The third frame, for instance, shows a line produced by a model telescope with a roughly Fermi-LAT resolution of $\Delta E/E_\text{line}=10\%$.) The horizontal axis denotes the deviation from the true line energy in units of the telescope energy resolution $\Delta E$. (For example, in absolute energy units the expected profile of the spectral line in frame 3 is 100 times the width of that in frame 6.) The plots are histograms that have been normalised to represent a probability distribution (plotted on the vertical axis) and hence there is a small amount of statistical fluctuation due to the finite sample size of 10,000. }
\end{figure}

At relatively low detector energy resolutions ($\Delta E/E_\gamma\gtrsim5\%$) trajectories are still mostly in phase and so mostly grouped together. The corresponding spectral line profiles appear narrower and sharper than could conventionally occur. Since the instrument should not be capable of resolving such a narrow line, lines of this type, if observed, could represent strong evidence for the presence of quantum nonequilibrium. See for instance the second frame of figure \ref{w_quarter_marginal}, which displays such a narrowed line at an energy resolution that approximately matches that of the EGRET instrument. At an energy resolution of $\Delta E/E_\gamma=10\%$, which approximately matches that of the Fermi-LAT, the distribution begins to desynchronise creating a double-line. (See frame 3 of figure \ref{w_quarter_marginal}.) These narrow and double line profiles are common to all initial nonequilibrium distributions that are sufficiently narrow in $Q$. At resolutions beyond about $\Delta E/E_\gamma=5\%$ the model detector produces a line profile that increasingly resembles the standard profile. In principle, observation of the fine structure (that is for instance shown frame 6 of figure \ref{w_quarter_marginal}) could betray the presence of quantum nonequilibrium, though it is unlikely that this fine structure will be evident in data without a very large sample of readings (signal strength).

Figure \ref{fig:w=2_relaxation} shows snapshots in the evolution when a widened ($w=2$) nonequilibrium photon distribution is introduced. Again, only the $Q>0$ half of the configuration is shown.
Representative spectral line profiles generated by a widened ($w=4$ for clarity) nonequilibrium distribution are shown in figure \ref{w4_marginal}. 
For the $w>1$ cases, more of the bulk of the ensemble density is situated at $|Q|\gtrsim2$. In these regions, the guidance equations \eqref{eq:reduced_guidance} cause early pointer movement in the positive $\text{dev}E$ direction (see figure \ref{fig:trajectories}). Accordingly, `low' resolution telescope photon detectors produce an overestimated energy. This is shown in frames 2 and 3 of figure \ref{w4_marginal}, although the extent of the overestimations may not be immediately evident. To illustrate, in frame 2 the nonequilibrium line profile is centred on approximately $\text{dev}E=5$ at an approximately EGRET energy resolution of $\Delta E/E_\gamma=20\%$. Consequently, if such a line profile were discovered, it would be likely attributed an energy of approximately double the correct energy. 

At increasing resolutions (as in the narrowed case) the widened distributions begin to display structure. Double and triple lines appear until at `high' resolutions the structure becomes fine enough that the line profile approaches an approximately stationary state with respect to the instrument energy resolution. Frame 6 of figure \ref{w4_marginal} displays the profile of the $w=4$ spectral line once it has reached this stage. Note that the amount of line broadening may appear comparable at all resolutions in figure \ref{w4_marginal}, though this is of course only with respect to the energy resolution. For instance, in frame 6 the energy resolution is $0.1\%$, which is equal to the expected conventional broadening so that the effects of quantum nonequilibrium and of conventional broadening would appear approximately the same size. In frame 3, on the other hand, the energy resolution is $10\%$ so that the broadening caused by quantum nonequilibrium is approximately 2 orders of magnitude larger than the conventional broadening.

\section{Discussion and Conclusions}
\label{5.4}
In this paper we have described possible experimental signatures of quantum nonequilibrium in the context of hypothetical dark matter line spectra. 
Quantum nonequilibrium (violation of the Born probability rule) is a prediction of the de Broglie-Bohm pilot-wave interpretation of quantum theory. The signatures we describe are not possible in any of the other major interpretations.
Among the major interpretations of quantum mechanics, de Broglie-Bohm is the only one that allows (in principle) arbitrary deviations from the Born probability rule. In fact, most interpretations are observationally equivalent to textbook quantum mechanics. The only other exception is collapse models. These allow for anomalous energy measurements, but for microscopic systems the effects are tiny. The potential effect of collapse models on spectral lines has in fact been studied \cite{BBU14}. It is found that lines can be broadened and shifted. However, the effects scale with the mass of the system and are far too small to be observed for microscopic systems (such as a hydrogen atom or an elementary particle, though the effects might be measurable for macromolecules). In contrast, our effects have no particular dependence on mass and in principle can be of order unity even for microscopic systems. 
Thus, a confirmed discovery would serve to distinguish de Broglie-Bohm theory from the other interpretations (including textbook quantum theory).

An accessible source of quantum nonequilibrium could also lead to the discovery of other new phenomena \cite{AV07,AV91b,AV04a,AV02,PV06}, potentially opening a large field of investigation.
We have argued in section \ref{5.2} that the indirect search for dark matter through spectral lines is an experimental field that is particularly well placed to observe experimental signatures of quantum nonequilibrium (should it exist). 
To support and illustrate our arguments, we have introduced an idealised model of spectral measurement intended to represent an X/$\gamma$-ray telescope photon detector.
Of course, this model has limited and debatable applicability to actual experiment.
The inner workings of contemporary X/$\gamma$-ray telescopes are very complicated, typically involving a great many interactions before an estimate of the photon energy may be made.
Our intention, however, is not to provide an accurate description of such telescopes, but only to illustrate spectral features that we argue are to be generally expected in the presence of quantum nonequilibrium. Moreover, since the notion of quantum nonequilibrium is commonly regarded as highly speculative, our intention is to highlight features that are sufficiently exceptional that they might be difficult to account for with conventional physics.

The two complementary spectral effects that we regard as the primary outcome of this work may be summarised as follows.
Firstly, as described in section \ref{5.2} and illustrated in figures \ref{w_quarter_marginal} and \ref{w4_marginal}, the effect of quantum nonequilibrium is not to alter the true energy of the signal photons, and hence the true energy spectrum $\rho_\text{true}(E_\gamma)$, but rather to alter the statistics of the interaction of the detector with the photons, and hence the telescope energy dispersion function $D(E|E_\gamma)$. 
The spectral effects of quantum nonequilibrium therefore appear on the lengthscale of the resolution of the detector ($\Delta E/E_\gamma$) rather than any fixed energy scale as is the case with standard spectral effects.
We regard this as the essential difference between conventional spectral effects and those caused by quantum nonequilibrium.
Telescopes that cannot resolve conventional sources of spectral broadening instead effectively observe their own dispersion function at the line ($\rho_\text{obs}(E)=D(E|E_\text{line})$), and so quantum nonequilibrium will produce the largest anomalies in these cases.
(Indeed it is interesting to note that the majority of currently operational X/$\gamma$-ray telescopes fall into this category when considering the hypothetical 0.1\% Doppler broadened WIMP and sterile neutrino lines that we have used as examples.)
Secondly, as illustrated in section \ref{5.3}, in the presence of quantum nonequilibrium a dynamical evolution is expected during the measurement process and this will contribute deviations from the expected line profile $D(E|E_\text{line})$. (See figures \ref{w_quarter_marginal} and \ref{w4_marginal}.) Such deviations will be maximal in telescopes that do not disturb the nonequilibrium too much. (In the model this corresponded to `lower' energy resolution.) It may be the case that real telescopes are sufficiently complicated that they degrade nonequilibrium entirely, making it effectively unobservable, although to discount this possibility would require a much more complicated (and potentially intractable) model.

Taken together, these two effects mean that a number of noteworthy profiles are possible. Lines may appear narrower than the instrument is canonically capable of resolving (see frames 2 and 3 of figure \ref{w_quarter_marginal}). Such an occurrence could only otherwise be due to instrumental error or statistical aberration. It should in practice be possible to exclude these two alternate explanations and so such a signal could represent strong evidence of quantum nonequilibrium. Lines may appear split into two or three or otherwise display structure (see frames 3-5 of figure \ref{w_quarter_marginal} and frames 3 and 4 of figure \ref{w4_marginal}). Should it occur, line broadening will appear on the lengthscale of the telescope energy resolution, which is commonly orders of magnitude larger than that expected in the physical energy spectrum (see frames 2 and 3 of figure \ref{w4_marginal}). For this reason, narrow lines could be mistaken for broad sources, confused with other nearby features or simply lost in the background. Finally, it is important to note that the line profile observed is as much a property of the telescope photon detector as it is a property of the quantum nonequilibrium. As such, two telescopes with large sets of otherwise reliable data may disagree entirely on the profile (or existence) of a line.

In the absence of any nonequilibrium distributions that are more extreme than those we have considered, it seems highly likely that the confirmed discovery of a dark matter spectral line would need to precede any investigation into the cause of an anomalous line shape. 
After all, signal statistics that are significant enough to prove a spectral line has an anomalous shape would surely be significant enough to prove the existence of the line in the first place. 
Nevertheless, the presence of quantum nonequilibrium could have important consequences for the indirect each for dark matter, and this is particularly true for telescopes of lower spectral resolution, where nonequilibrium could obfuscate the detection process, producing ambiguous and anomalous results, making the discovery of a spectral line more elusive. 

In recent years there have been a number of claims of line detections (all controversial), and these may provide context for this last point.
One such controversial line was reported in the galactic dark matter halo at $\sim$133 GeV in Fermi-LAT data in 2012 by Refs.\ \cite{Bring12,W12}. This was originally thought to be a possible WIMP annihilation line \cite{W12}, although its existence is now widely discredited. Indications of instrumental error have been found \cite{Creview16,fermi3.7}, searches with access to more LAT data have reported reduced significance \cite{fermi3.7,fermi5.8}, and a recent search with the ground-based H.E.S.S.\ II Cherenkov telescope ruled out the line with 95\% confidence \cite{HESS16}.
It is interesting to note, however, that one of the original reasons why this feature was discredited as an actual signal was that it seemed too narrow.
In 2013 the Fermi collaboration disputed the existence of the line, arguing that `the feature is narrower than the LAT energy resolution at the level of 2 to 3 standard deviations, which somewhat disfavours the interpretation of the 133 GeV feature as a real WIMP signal' \cite{fermi3.7}.
As we have seen, quantum nonequilibrium could result in features which are narrower than the telescope could canonically be capable of resolving.
Indeed, in our model this is commonly the case for telescope energy resolutions near 10\% (frames 2-4 figure \ref{w_quarter_marginal}).
For such cases, our model also tends to produce multiple lines at these resolutions (frame 3 figure \ref{w_quarter_marginal}) and, remarkably, two prominent early proponents of the feature found it to be marginally better fit by two lines at $\sim$111 GeV and $\sim$129 GeV (Refs.\ \cite{SF12-1} and \cite{SF12-2}).

A more recent (and yet to be resolved) controversy surrounds a 3.5 keV line-like feature in data from X-ray telescopes. 
The bulk of the original study that reported the line (Ref.\ \cite{Bulbul14}) concerned the stacked spectrum of a large number of nearby galaxies and galaxy clusters\footnote{The advantages of a stacked study are two-fold.
The available sample size is enlarged and, as data from many different redshifts are combined in the rest frame of the source, spectral features with an instrumental origin are smeared over whilst features with a physical origin are superposed; instrumental error is suppressed and physical features are emphasised.} with the CCD instruments aboard the XXM-Newton satellite.
In this stacked study, Ref.\ \cite{Bulbul14} found a faint line (a 1\% bump above the background continuum) with an unresolved profile, consistent with the decay of a 7 keV sterile neutrino. 
Like $\gamma$-ray calorimeters, the CCD instruments employed are capable of single photon energy measurements, albeit at a relatively low $\sim$100 eV full-width-half-maximum (FWHM) resolution in this range.
This is considerably larger than the $\sim$$0.1\%\times3.5$ keV width of the hypothetical sterile neutrino line, and so the effects of quantum nonequilibrium could also be significant in this context.
Whilst the authors reported the discovery with relatively large ($>3\sigma$) significance, they also noted some apparent anomalies in their data.
Chief amongst these was the fact that the two separate XXM-Newton CCD instruments used for the stacked study (the MOS and the pn) found best-fit energies at $3.57\pm0.02$ keV and $3.51\pm0.03$ keV respectively, a disagreement at a 2.8$\sigma$ level \cite{Bulbul14}.
Also, the best fit flux from the Perseus Cluster appeared anomalously large in data sets from both XMM-Newton and Chandra instruments.
(This high Perseus flux was confirmed by an independent study using different XMM-Newton data \cite{Boyarsky14}.)
Both of these anomalies were however mitigated if the flux of nearby K and Ar atomic transition lines were allowed to vary from their theoretical values, albeit by at least an order of magnitude.
Since then, a number of related studies have taken place using the CCD instruments aboard the XMM-Newton, Chandra and Suzaku satellites, and the Cadmium-Zinc-Telluride instruments aboard the NuStar satellite.
These paint a confusing picture, with varying levels of detection certainty\footnote{Recently, Ref.\ \cite{NME16} reported an 11$\sigma$ detection in NuStar data from the COSMOS and ECDFS empty sky fields, where a signal from the Milky-Way dark matter halo is expected.} and notable non-detections as for instance in Refs.\ \cite{A15,MNE14,JP16,R16}.
For comprehensive reviews of these studies see Refs.\ \cite{I16,White_Paper17}.
Much of the uncertainty stems from two nearby K XVIII lines at 3.47 keV and 3.51 keV \cite{White_Paper17}.
These lines are difficult to distinguish from the putative sterile neutrino line with $\sim$100 eV resolution instruments.
This matter should have been put to rest with the launch of the Hitomi satellite in February 2016, which carried aboard it the high-resolution SXS instrument, capable of a resolution of $\Delta E\simeq 5$ eV FWHM \cite{Hitomi16}, close to that of the theoretical $0.1\%\times 3.5$ keV width of the sterile neutrino line. 
Unfortunately, the Hitomi satellite was lost a month into its mission. Before this, however, it did manage to collect some preliminary data on the Perseus cluster \cite{Hitomi_Perseus,Hitomi16}. Hitomi did not detect the 3.5 keV line at the anomalously high flux level detected by Refs.\ \cite{Bulbul14,Boyarsky14}, ruling it out with $>99\%$ confidence \cite{Hitomi16}.
The Hitomi study did not rule out the lower flux level of the stacked sample in Ref. \cite{Bulbul14} however, so it is still possible this exists. It also gave no explanation of the reason behind the high flux observations in Refs. \cite{Bulbul14,Boyarsky14}.
The next generation of high spectral resolution telescopes (XRISM, Lynx, Athena) are due to become operational from the early 2020s onwards.

For our purposes, the discourse that has surrounded the 3.5 keV feature raises a number of relevant points. 
For example, since stacked studies combine the notional line at a range of redshifts, and telescope response is dependent upon photon energy, these studies experience an averaging effect in the observable consequences of quantum nonequilibrium (similar to the smearing of instrumental error). 
Furthermore, since quantum nonequilibrium is an ensemble property, it is dependent upon how the ensemble is defined. 
There is no a priori reason, for instance, why two galaxy clusters would produce the same line profile. 
If galactic nonequilibrium distributions did indeed differ, then a stacked study would also average over these, further suppressing nonequilibrium signatures. 
In principle then, quantum nonequilibrium could explain why the single Perseus source produced an anomalous signal that was not observed in the stacked sample. 
Since the observed nonequilibrium signatures are dependent upon the instrument used, it could also explain why the MOS and pn instruments disagreed about the line energy in the stacked spectrum and also why the higher resolution Hitomi instrument did not observe the anomalous Perseus signal.

Indeed, since quantum nonequilibrium is consistent such a wide range of different anomalies\footnote{Even the history of observations of the 511 keV electron-positron annihilation line in the galactic bulge (whose existence is not controversial but which has at times been argued to be of dark matter origin \cite{BHSC03,HW04}) can be argued to be consistent with quantum nonequilibrium. Although high resolution ($\Delta E\sim2$ keV FWHM) INTEGRAL-SPI observations now place the line at energy $510.954\pm0.075$ keV with a width of $2.37\pm0.25$ keV FWHM \cite{Churazov04}, early observations of the line were conflicting. Low resolution balloon observations of the galactic centre taken in the 1970s placed the line at $476\pm24$ keV \cite{Johnson73} and at $530\pm11$ keV \cite{Haymes75}, and contrary reports of the line flux produced suggestions of time-variation \cite{Leventhal82}--all in principle possible as a result quantum nonequilibrium affecting low resolution spectral measurements.}, it may never be as credible as competing explanations with more restricted predictions. 
Even the strongest evidence we have described could be explained in terms of instrument error, and would differ from telescope to telescope. 
This is not an insurmountable barrier to the discovery of particle quantum nonequilibrium, however. 
As we have noted, we would expect the discovery of a line to be confirmed before any significant inquiry into the cause of an anomalous profile.
Indeed, for the reasons we have outlined, it may be the case that a confirmed line discovery will be delayed until a telescope that is capable of resolving the line width is available (in order to minimise the effects of nonequilibrium).
Before such a time, if consistently anomalous signals were detected in lower resolution telescopes, the presence of quantum nonequilibrium could begin to be speculated upon and this might inspire a more direct approach to its detection. 
A definitive proof of the existence of quantum nonequilibrium could be arrived at by subjecting the signal photons to a specifically quantum-mechanical experiment.
One suggestion along these lines is to look for deviations from Malus' law in the polarization probabilities of the signal photons \cite{AV04a}.
There are in fact many possible such experiments that could work. Even a simple double slit experiment, for instance, would show anomalous results--a blurring of the interference pattern--in the presence of quantum nonequilibrium. 
If the incoming photons show anomalies under such tests then there could be little doubt that the Born rule has been violated and that other interpretations of quantum mechanics have failed.

There are clearly many ways in which our effects could manifest in the search for dark matter, and there are many practical reasons why our effects could turn out to be obscured even if they exist. More work remains to be done on these matters.
Only time will tell whether any of the scenarios we have outlined prove to be informative.

%\bibliographystyle{apsrev4-1_modified}
%\bibliography{../citations}
%\end{document}

\chapter{SUGGESTED FUTURE RESEARCH DIRECTIONS}\label{6}
Conclusions for chapters \ref{2} through \ref{5} are drawn in sections \ref{2.6}, \ref{sec:conclusion}, \ref{Chapter4_conclusion}, and \ref{5.4} respectively.
The purpose of this chapter is to briefly highlight promising future research projects. 

\section{For Chapter \ref{2}}
Chapter \ref{2} develops classical and quantum relax theories from a principle of (information) entropy conservation. 
Entropy conservation is shown to guarantee that arbitrary ensembles of trajectories relax towards an equilibrium distribution.
This equilibrium distribution is shown to correspond to a density of states within the state space. 
The imposition of entropy conservation is shown to lead to a theoretical framework for physical theories dubbed the iRelax framework.
Classical mechanics and de Broglie-Bohm quantum mechanics are arrived at through separate considerations of the iRelax framework.
To the author's mind, the issues covered in chapter \ref{2} point to three main topics for further research.\\
{\bf\color{red!50!black} i.....}
A discovery of quantum nonequilibrium would be unquestionably momentous. So predictions of ways in which quantum nonequilibrium could realistically be observed would be of great value. This task is taken up by chapters \ref{3}, \ref{4}, and \ref{5}.\\
{\bf\color{red!50!black} ii.....}
Chapter \ref{2} shows how entropy conservation ensures the thermodynamic relaxation of trajectories in the state space. It also, \emph{almost}, shows that entropy conservation implies the law of evolution must specify trajectories in the first place. Indeed it does so for the discrete case. 
But for continuous state spaces there is a missing link in the proof that appears to arise as a result of the divergence of entropy for $\rho(x)$s with Dirac deltas. See sections \hyperref[sec:state_propagators]{2.3.2} and \hyperref[sec:modified_state_propagators]{2.4.2} for discussion. A clarification of this point be helpful.\\
{\bf\color{red!50!black} iii.....}
De Broglie-Bohm theory is only a partial fit for the iRelax framework. In order to fully determine future evolution in a de Broglie-Bohm theory, both the configuration and the quantum state must be specified.
But at present, relaxation is only understood on the level of the configuration. 
The quantum state is generally assumed to be shared by all ensemble members.
This suggests that a unified approach to de Broglie-Bohm may be in order, in which the quantum state and the configuration are treated on the same footing.

\section{For Chapter \ref{3}}
Chapter \ref{3} considers whether full relaxation may be impeded or prevented for sufficiently `extreme' forms of quantum nonequilibrium.
It introduces the method of the `drift field', which is tailored to exploit the regular behavior of trajectories far away from the chaos-inducing nodal vortices.
It also contains a systematic treatment of nodes and how their properties may be deduced from the quantum state parameters.
It defines categories of quantum states, including one in particular that inhibits (or possibly even prevents) total relaxation.
Chapter \ref{3} prompts several relatively easy topics to investigate.\\
{\bf\color{red!50!black} i.....}
The relaxation-retarding type-0 drift fields certainly slow relaxation, but it is not yet clear whether relaxation will eventually complete, or whether it is truly prevented.
If relaxation was fully prevented, this would lend weight to the prospect of the discovery of quantum nonequilibrium.\\
{\bf\color{red!50!black} ii.....}
The drift field is a tool that may be used to analyze a great many quantum states without resorting to high performance computing. It is hoped that researchers may use the method to further probe quantum systems of their choosing. The author would be happy to supply the Fortran and Mathematica codes developed for chapter \ref{3} to those interested in pursuing this line of research.\\
{\bf\color{red!50!black} iii.....}
In order to restrict the field of investigation, it was helpful to attach a condition of `no fine-tuning' to those quantum states under investigation. 
However most quantum states so far studied in the literature do not satisfy this condition.
In the course of investigation it was found that quantum states featuring fine-tuning along the lines described can produce highly `unusual' or `pathological' drift fields. (The four-state superposition used in reference \cite{ACV14} produces such behavior for instance.) It would be informative to relax or replace the no fine-tuning condition in order to study such states. It may also be informative to study systems with a suitable truncated thermal spectrum, $\sim e^{E_i/kT}$, so that the temperature dependence of relaxation may be studied.  \\
{\bf\color{red!50!black} iv.....}
Chapter \ref{2} lists and describes a number of properties of nodes of the wave function(al), including some that are previously known. But it does so under the dual assumptions of no fine-tuning and the system in question being the two-dimensional isotropic oscillator. It would be interesting and likely useful to attempt to generalize these results. 

\section{For Chapter \ref{4}}
Chapter \ref{4} is the most speculative, but also the most intriguing chapter, as it considers possible ways in which quantum nonequilibrium could be preserved in the statistics of relic cosmological particles. 
A reliable estimate, however, relies on several uncertain areas in contemporary physics--beyond the standard model particle physics, early Universe cosmology, and of course uncertain factors in the generality of quantum relaxation (especially on cosmological scales).
Some possible ways to develop these ideas may include the following.\\
{\bf\color{red!50!black} i.....}
Since for the time being uncertain factors make a prediction of nonequilibrium surviving in relic particles difficult, it may be wise instead to develop a realistic phenomenology of quantum nonequilibrium in experiments. This way nonequilibrium may be easier to recognize if it does show up in experiment. This is the route taken by chapter \ref{5}.\\
{\bf\color{red!50!black} ii.....}
Chapter \ref{4} considers both nonequilibrium particle production by inflaton decay during reheating and the possibility of conformal vacua retaining nonequilibrium. But there may be other means by which systems could remain un-relaxed in the early Universe.
During preheating, for example, the classical part of the inflaton field creates matter through parametric resonance.
If particles are created in a field that is initially in a nonequilibrium vacuum, then these particles may `pick up' some nonequilibrium from the vacuum. Of course, such particles would have to satisfy various other criteria so that they avoid subsequent relaxation, are numerous enough to be observable, and are stable enough to avoid decay before recombination etc.\\
{\bf\color{red!50!black} iii.....}
It would be interesting to see how a single complicated system (the Universe) may reproduce nonequilibrium and equilibrium statistics in the sub-systems from which it is composed.
To the author's knowledge there is yet to be a single study to consider this topic, and so even a toy-model may highlight some unanticipated behavior. 
Such a study would also help to pin-down how nonequilibrium in the early Universe could be frozen-in, a potentiality that is yet to be fully explored.

\vspace{-5mm}
\section{For Chapter \ref{5}}
\vspace{-3mm}
Chapter \ref{5} describes how the (hypothetical) smoking-gun line spectra of dark matter annihilation/decay may be altered if the dark matter exists in a state of quantum nonequilibrium.
This work is intended to inform the indirect search for dark matter. 
If a nonequilibrium signal were detected, then a number of things could happen. 
Lines could appear widened, shifted, or feature structure and shape on a scale that should not be resolvable. 
More to the point, telescopes would likely disagree with each other on the position and shape of the line. 
A point that is related to contextuality in quantum measurements.
If systematic error could be ruled out, then such a disagreement between telescopes could be statistically significant enough to warrant further investigation.
Fortunately there are specifically quantum mechanical experiments that could be performed to settle the matter unequivocally.
These should be relatively simple to perform once a suspected source of nonequilibrium is identified.
Another nonequilibrium signature that could prompt such suspicion is an anomalous width. 
There appears to be a distinct possibility to record the width of a line to be smaller than the telescope should be capable of resolving.  
This possibility of nonequilibrium, again, has no competing explanation apart from systematic error, and so has the potential to raise significant suspicions.

\bibliography{../main_citations}

\begin{thebibliography}{100}

\bibitem{NU_In_preparation}
N.~G. Underwood.
\newblock In preparation.

\bibitem{NU18}
Nicolas~G. Underwood.
\newblock {Extreme quantum nonequilibrium, nodes, vorticity, drift and
  relaxation retarding states}.
\newblock {\em J. Phys.},
  \href{http://dx.doi.org/10.1088/1751-8121/aa9e97}{A51(5):055301}, 2018,
  \href{https://arxiv.org/abs/1705.06757}{1705.06757}.

\bibitem{UV15}
N.~G. Underwood and A.~Valentini.
\newblock {Quantum field theory of relic nonequilibrium systems}.
\newblock {\em Phys. Rev. D},
  \href{http://dx.doi.org/10.1103/PhysRevD.92.063531}{92(6):063531}, 2015,
  \href{https://arxiv.org/abs/1409.6817}{1409.6817}.

\bibitem{UV16}
N.~G. Underwood and A.~Valentini.
\newblock {Anomalous spectral lines and relic quantum nonequilibrium}.
\newblock 2016, \href{https://arxiv.org/abs/1609.04576}{1609.04576}.

\bibitem{H67}
D.T. Haar.
\newblock {\em The Old Quantum Theory}.
\newblock Pergamon Press, 1967.

\bibitem{Planck}
M.~Planck.
\newblock {Zur Theorie des Gesetzes der Energieverteilung im Normalspectrum}.
\newblock {\em Verhand. Dtsch. Phys. Ges.}, page 237, 1900.
\newblock English translation in ref. \cite{H67}.

\bibitem{E1905_light_quanta}
A.~Einstein.
\newblock \"{U}ber einen die erzeugung und verwandlung des lichtes betreffenden
  heuristischen gesichtspunkt.
\newblock {\em Annal. der Phys.},
  \href{http://dx.doi.org/10.1002/andp.19053220607}{322(6):132--148}, 1905.
\newblock English translation in ref. \cite{H67}.

\bibitem{B1913}
N.~Bohr.
\newblock {LXXIII. On the constitution of atoms and molecules}.
\newblock {\em Philos. Mag.},
  \href{http://dx.doi.org/10.1080/14786441308635031}{26(155):857--875}, 1913.

\bibitem{H25}
W.~Heisenberg.
\newblock {{\"U}ber quantentheoretische Umdeutung kinematischer und
  mechanischer Beziehungen.}
\newblock {\em {Zeitschrift f{\"u}r Physik}},
  \href{http://dx.doi.org/10.1007/BF01328377}{33(1):879--893}, 1925.
\newblock (English translation in ref.\ \cite{english_translation_of_german}).

\bibitem{Dirac_biography}
H.~Kragh.
\newblock {Dirac, Paul Adrien Maurice (1902--1984), theoretical physicist}.
\newblock \href{http://dx.doi.org/10.1093/ref:odnb/31032}{{\em {Oxford
  Dictionary of National Biography}}}, 2006.

\bibitem{D25}
P.~A.~M. Dirac.
\newblock The fundamental equations of quantum mechanics.
\newblock \href{http://dx.doi.org/10.1098/rspa.1925.0150}{{\em Proc. R. Soc.
  Lond. Series A}}, 1925.

\bibitem{BV09}
G.~Bacciagaluppi and A.~Valentini.
\newblock {\em {Quantum Theory at the Crossroads: Reconsidering the 1927 Solvay
  Conference}}.
\newblock Cambridge University Press, Cambridge, 2009,
  \href{https://arxiv.org/abs/quant-ph/0609184}{quant-ph/0609184}.

\bibitem{Beller}
M.~Beller.
\newblock {\em {Quantum Dialogue: The Making of a Revolution}}.
\newblock Science and Its Conceptual Foundations S. University of Chicago
  Press, 1999.

\bibitem{Dirac}
P.~A.~M. Dirac.
\newblock {\em {The Principles of Quantum Mechanics}}.
\newblock Clarendon, 1930.

\bibitem{BHJ26}
M.~Born, W.~Heisenberg, and P.~Jordan.
\newblock {Zur Quantenmechanik. II.}
\newblock {\em Zeitschrift f{\"u}r Physik},
  \href{http://dx.doi.org/10.1007/BF01379806}{35(8):557--615}, 1926.

\bibitem{D27a}
P.~A.~M. Dirac.
\newblock The quantum theory of the emission and absorption of radiation.
\newblock {\em Proc. Roy. Soc. Lond. Series A},
  \href{http://dx.doi.org/10.1098/rspa.1927.0039}{114}, 1927.

\bibitem{W77}
Steven Weinberg.
\newblock The search for unity: Notes for a history of quantum field theory.
\newblock {\em Daedalus}, 106(4):17--35, 1977.

\bibitem{G26}
W.~Gordon.
\newblock {Der Comptoneffekt nach der Schr{\"o}dingerschen Theorie}.
\newblock {\em Zeitschrift f{\"u}r Physik},
  \href{http://dx.doi.org/10.1007/BF01390840}{40(1):117--133}, 1926.

\bibitem{K27}
O.~Klein.
\newblock Elektrodynamik und wellenmechanik vom standpunkt des
  korrespondenzprinzips.
\newblock {\em Zeitschrift f{\"u}r Physik A Hadrons and nuclei},
  \href{http://dx.doi.org/10.1007/BF01400205}{41(6):407--442}, 1927.

\bibitem{D28a}
P.~A.~M. Dirac.
\newblock The quantum theory of the electron.
\newblock {\em Proceedings of the Royal Society of London. Series A, Containing
  Papers of a Mathematical and Physical Character},
  \href{http://dx.doi.org/10.1098/rspa.1928.0023}{117(778):610--624}, 1928.

\bibitem{D28b}
P.~A.~M. Dirac.
\newblock {The quantum theory of the electron. Part II}.
\newblock {\em Proceedings of the Royal Society of London. Series A, Containing
  Papers of a Mathematical and Physical Character},
  \href{http://dx.doi.org/10.1098/rspa.1928.0056}{118(779):351--361}, 1928.

\bibitem{D30}
Paul Adrien~Maurice Dirac and Ralph~Howard Fowler.
\newblock A theory of electrons and protons.
\newblock {\em Proceedings of the Royal Society of London. Series A, Containing
  Papers of a Mathematical and Physical Character},
  \href{http://dx.doi.org/10.1098/rspa.1930.0013}{126(801):360--365}, 1930.

\bibitem{A32}
C.~D. Anderson.
\newblock The apparent existence of easily deflectable positives.
\newblock {\em Science},
  \href{http://dx.doi.org/10.1126/science.76.1967.238}{76(1967):238--239},
  1932.

\bibitem{JW28}
P.~Jordan and E.~Wigner.
\newblock {{\"U}ber das Paulische {\"A}quivalenzverbot}.
\newblock {\em Zeitschrift f{\"u}r Physik},
  \href{http://dx.doi.org/10.1007/BF01331938}{47(9):631--651}, 1928.

\bibitem{HP29}
W.~Heisenberg and W.~Pauli.
\newblock {Zur Quantendynamik der Wellenfelder}.
\newblock {\em Zeitschrift f{\"u}r Physik},
  \href{http://dx.doi.org/10.1007/BF01340129}{56(1):1--61}, 1929.

\bibitem{HP30}
W.~Heisenberg and W.~Pauli.
\newblock {Zur Quantentheorie der Wellenfelder. II}.
\newblock {\em Zeitschrift f{\"u}r Physik},
  \href{http://dx.doi.org/10.1007/BF01341423}{59(3):168--190}, 1930.

\bibitem{F29}
E.~Fermi.
\newblock Sopra l'elettrodinamica quantistica.
\newblock {\em Rend. Lincei}, 5:881--997, 1929.
\newblock (may be found on page 305 of ref. \cite{Fermi_collected}).

\bibitem{FO34}
W.~H. Furry and J.~R. Oppenheimer.
\newblock On the theory of the electron and positive.
\newblock {\em Phys. Rev.},
  \href{http://dx.doi.org/10.1103/PhysRev.45.245}{45:245--262}, 1934.

\bibitem{PW34}
W.~Pauli and V.~Weisskopf.
\newblock {\"{U}ber die Quantisierung der skalaren relativistischen
  Wellengleichung}.
\newblock {\em Helvetica Physica Acta},
  \href{http://dx.doi.org/10.5169/seals-110395}{7:709}, 1934.

\bibitem{dB25}
L.~V. P.~R. de~Broglie.
\newblock {Recherches sur la théorie des quanta}.
\newblock {\em Annals Phys.}, 2:22--128, 1925.

\bibitem{S35}
E.~Schr{\"o}dinger.
\newblock {Die gegenw{\"a}rtige Situation in der Quantenmechanik}.
\newblock {\em Naturwissenschaften},
  \href{http://dx.doi.org/10.1007/BF01491891}{23(48):807--812}, Nov 1935.

\bibitem{L_quote}
M.~S. Leifer.
\newblock {Foundations of Quantum Mechanics: Lecture 2}.
\newblock \href{http://pirsa.org/17010034}{PIRSA:17010034}.

\bibitem{WZ82}
W.~K. Wootters and W.~H. Zurek.
\newblock {A single quantum cannot be cloned}.
\newblock {\em Nature},
  \href{http://dx.doi.org/10.1038/299802a0}{299:802--803}, 1982.

\bibitem{D82}
D.~Dieks.
\newblock {Communication by EPR Devices}.
\newblock {\em Phys. Lett.},
  \href{http://dx.doi.org/10.1016/0375-9601(82)90084-6}{A92:271--272}, 1982.

\bibitem{quantum_teleportation}
C.~H. Bennett, G.~Brassard, C.~Cr\'epeau, R.~Jozsa, A.~Peres, and W.~K.
  Wootters.
\newblock Teleporting an unknown quantum state via dual classical and
  einstein-podolsky-rosen channels.
\newblock {\em Phys. Rev. Lett.},
  \href{http://dx.doi.org/10.1103/PhysRevLett.70.1895}{70:1895--1899}, Mar
  1993.

\bibitem{first_experimental_quantum_teleportation}
D.~{Bouwmeester}, J.-W. {Pan}, K.~{Mattle}, M.~{Eibl}, H.~{Weinfurter}, and
  A.~{Zeilinger}.
\newblock {Experimental quantum teleportation}.
\newblock {\em Nature}, \href{http://dx.doi.org/10.1038/37539}{390:575--579},
  1997, \href{https://arxiv.org/abs/1901.11004}{1901.11004}.

\bibitem{NAS_quantum_computing_progress_report}
Consensus~Study Report.
\newblock Quantum computing: Progress and prospects.
\newblock \href{http://dx.doi.org/10.17226/25196}{{\em National Academies
  Press}}, 2018.

\bibitem{vN32}
J.~von Neumann.
\newblock {\em Mathematische Grundlagen der Quantenmechanik}.
\newblock Springer, 1932.
\newblock (translated into English in ref. \cite{vN55}).

\bibitem{B66}
J.~S. Bell.
\newblock On the problem of hidden variables in quantum mechanics.
\newblock {\em Rev. Mod. Phys.},
  \href{http://dx.doi.org/10.1103/RevModPhys.38.447}{38:447--452}, Jul 1966.

\bibitem{KS67}
S.~Kochen and E.~P. Specker.
\newblock The problem of hidden variables in quantum mechanics.
\newblock {\em Indiana Univ. Math. J.}, 17:59--87, 1967.

\bibitem{B52a}
D.~Bohm.
\newblock {A Suggested Interpretation of the Quantum Theory in Terms of
  ``Hidden'' Variables. I}.
\newblock {\em Phys. Rev.},
  \href{http://dx.doi.org/10.1103/PhysRev.85.166}{85:166--179}, 1952.

\bibitem{B52b}
D.~Bohm.
\newblock {A Suggested Interpretation of the Quantum Theory in Terms of
  ``Hidden'' Variables. {II}}.
\newblock {\em Phys. Rev.},
  \href{http://dx.doi.org/10.1103/PhysRev.85.180}{85:180--193}, 1952.

\bibitem{B64}
J.~S. Bell.
\newblock {On the Einstein Podolsky Rosen paradox}.
\newblock {\em Physics Physique Fizika},
  \href{http://dx.doi.org/10.1103/PhysicsPhysiqueFizika.1.195}{1:195--200},
  1964.

\bibitem{CHSH69}
J.~F. Clauser, M.~A. Horne, A.~Shimony, and R.~A. Holt.
\newblock Proposed experiment to test local hidden-variable theories.
\newblock {\em Phys. Rev. Lett.},
  \href{http://dx.doi.org/10.1103/PhysRevLett.23.880}{23:880--884}, 1969.

\bibitem{S83}
A.~Shimony.
\newblock Controllable and uncontrollable non-locality.
\newblock In S.~Kamefuchi, editor, {\em Foundations of quantum mechanics in the
  light of new technology}. 1983.

\bibitem{W05}
H.~M. {Wiseman}.
\newblock {From Einstein's theorem to Bell's theorem: a history of quantum
  non-locality}.
\newblock {\em Contemporary Physics},
  \href{http://dx.doi.org/10.1080/00107510600581011}{47:79--88}, Apr 2006,
  \href{https://arxiv.org/abs/quant-ph/0509061}{quant-ph/0509061}.

\bibitem{BCT99}
G.~Brassard, R.~Cleve, and A.~Tapp.
\newblock {The cost of exactly simulating quantum entanglement with classical
  communication}.
\newblock {\em Phys. Rev. Lett.},
  \href{http://dx.doi.org/10.1103/PhysRevLett.83.1874}{83:1874--1877}, 1999,
  \href{https://arxiv.org/abs/quant-ph/9901035}{quant-ph/9901035}.

\bibitem{BBT05}
G.~{Brassard}, A.~{Broadbent}, and A.~{Tapp}.
\newblock {Quantum Pseudo-Telepathy}.
\newblock {\em Foundations of Physics},
  \href{http://dx.doi.org/10.1007/s10701-005-7353-4}{35:1877--1907}, Nov 2005,
  \href{https://arxiv.org/abs/quant-ph/0407221}{quant-ph/0407221}.

\bibitem{JW05}
Kurt {Jacobs} and Howard~M. {Wiseman}.
\newblock {An entangled web of crime: Bell's theorem as a short story}.
\newblock {\em American Journal of Physics},
  \href{http://dx.doi.org/10.1119/1.1979499}{73:932--937}, Oct 2005,
  \href{https://arxiv.org/abs/quant-ph/0504192}{quant-ph/0504192}.

\bibitem{AV92}
A.~Valentini.
\newblock {\em {On the Pilot-Wave Theory of Classical, Quantum and Subquantum
  Physics}}.
\newblock PhD thesis, International School for Advanced Studies, Trieste,
  Italy, 1992.

\bibitem{NU_straw_men_and_covert_assumptions}
N.~G. Underwood.
\newblock In preparation.

\bibitem{ADR82}
A.~Aspect, J.~Dalibard, and G.~Roger.
\newblock {Experimental Test of Bell's Inequalities Using Time-Varying
  Analyzers}.
\newblock {\em Phys. Rev. Lett.},
  \href{http://dx.doi.org/10.1103/PhysRevLett.49.1804}{49:1804--1807}, Dec
  1982.

\bibitem{BCP14}
Nicolas {Brunner}, Daniel {Cavalcanti}, Stefano {Pironio}, Valerio {Scarani},
  and Stephanie {Wehner}.
\newblock {Bell nonlocality}.
\newblock {\em Reviews of Modern Physics},
  \href{http://dx.doi.org/10.1103/RevModPhys.86.419}{86:419--478}, Apr 2014,
  \href{https://arxiv.org/abs/1303.2849}{1303.2849}.

\bibitem{Jaynesbook}
E.T. Jaynes.
\newblock {\em {Probability Theory: The Logic of Science}}.
\newblock Cambridge University Press, 2003.

\bibitem{P53}
W.~Pauli.
\newblock Les savants du monde / dirig{\'e}e par Andr{\'e} George. Albin
  Michel, 1953.

\bibitem{K53}
J.~B. Keller.
\newblock {Bohm's Interpretation of the Quantum Theory in Terms of ``Hidden''
  Variables}.
\newblock {\em Phys. Rev.},
  \href{http://dx.doi.org/10.1103/PhysRev.89.1040}{89:1040--1041}, 1953.

\bibitem{B53}
D.~Bohm.
\newblock {Proof That Probability Density Approaches
  ${|\ensuremath{\psi}|}^{2}$ in Causal Interpretation of the Quantum Theory}.
\newblock {\em Phys. Rev.},
  \href{http://dx.doi.org/10.1103/PhysRev.89.458}{89:458--466}, 1953.

\bibitem{BV54}
D.~Bohm and J.~P. Vigier.
\newblock {Model of the Causal Interpretation of Quantum Theory in Terms of a
  Fluid with Irregular Fluctuations}.
\newblock {\em Phys. Rev.},
  \href{http://dx.doi.org/10.1103/PhysRev.96.208}{96:208--216}, 1954.

\bibitem{Holl93}
P.~R. Holland.
\newblock {\em {The Quantum Theory of Motion: an Account of the de Broglie-Bohm
  Causal Interpretation of Quantum Mechanics}}.
\newblock Cambridge University Press, Cambridge, 1993.

\bibitem{AV91a}
A.~Valentini.
\newblock {Signal-locality, uncertainty, and the subquantum {H}-theorem. {I}}.
\newblock {\em Phys. Lett. A},
  \href{http://dx.doi.org/10.1016/0375-9601(91)90116-P}{156(1–2):5--11},
  1991.

\bibitem{AV91b}
A.~Valentini.
\newblock {Signal-locality, uncertainty, and the subquantum {H}-theorem. {II}}.
\newblock {\em Phys. Lett. A},
  \href{http://dx.doi.org/10.1016/0375-9601(91)90330-B}{158(1–2):1 -- 8},
  1991.

\bibitem{G18}
G.~Ghirardi.
\newblock Collapse theories.
\newblock In Edward~N. Zalta, editor, {\em The Stanford Encyclopedia of
  Philosophy}. Metaphysics Research Lab, Stanford University, fall 2018
  edition, 2018.

\bibitem{GRW86}
G.~C. Ghirardi, A.~Rimini, and T.~Weber.
\newblock Unified dynamics for microscopic and macroscopic systems.
\newblock {\em Phys. Rev. D},
  \href{http://dx.doi.org/10.1103/PhysRevD.34.470}{34:470--491}, Jul 1986.

\bibitem{P96}
R.~Penrose.
\newblock On gravity's role in quantum state reduction.
\newblock {\em General Relativity and Gravitation},
  \href{http://dx.doi.org/10.1007/BF02105068}{28(5):581--600}, May 1996.

\bibitem{deB28}
L.~de~Broglie.
\newblock In {\em \'{E}lectrons et Photons: Rapports et Discussions du
  Cinqui\`{e}me Conseil de Physique}. Gauthier-Villars, Paris, 1928.
\newblock [{E}nglish translation in ref. \cite{BV09}].

\bibitem{VW05}
A.~Valentini and H.~Westman.
\newblock {Dynamical origin of quantum probabilities}.
\newblock {\em Proc. R. Soc. A},
  \href{http://dx.doi.org/10.1098/rspa.2004.1394}{461(2053):253--272}, 2005,
  \href{https://arxiv.org/abs/quant-ph/0403034}{quant-ph/0403034}.

\bibitem{EC06}
C.~{Efthymiopoulos} and G.~{Contopoulos}.
\newblock {Chaos in Bohmian quantum mechanics}.
\newblock {\em J. Phys. A},
  \href{http://dx.doi.org/10.1088/0305-4470/39/8/004}{39:1819--1852}, 2006.

\bibitem{CS10}
S.~{Colin} and W.~{Struyve}.
\newblock {Quantum non-equilibrium and relaxation to equilibrium for a class of
  de Broglie-Bohm-type theories}.
\newblock {\em New J. Phys.},
  \href{http://dx.doi.org/10.1088/1367-2630/12/4/043008}{12(4):043008}, 2010,
  \href{https://arxiv.org/abs/0911.2823}{0911.2823}.

\bibitem{SC12}
S.~Colin.
\newblock {Relaxation to quantum equilibrium for {D}irac fermions in the de
  {Broglie-Bohm pilot-wave} theory}.
\newblock {\em Proc. R. Soc. A},
  \href{http://dx.doi.org/10.1098/rspa.2011.0549}{468(2140):1116--1135}, 2012,
  \href{https://arxiv.org/abs/1108.5496}{1108.5496}.

\bibitem{TRV12}
M.~D. Towler, N.~J. Russell, and A.~Valentini.
\newblock {Time scales for dynamical relaxation to the {B}orn rule}.
\newblock {\em Proc. R. Soc. A},
  \href{http://dx.doi.org/10.1098/rspa.2011.0598}{468(2140):990--1013}, 2012,
  \href{https://arxiv.org/abs/1103.1589}{1103.1589}.

\bibitem{ACV14}
E.~Abraham, S.~Colin, and A.~Valentini.
\newblock {Long-time relaxation in pilot-wave theory}.
\newblock {\em J. Phys. A},
  \href{http://dx.doi.org/10.1088/1751-8113/47/39/395306}{47(39):395306}, 2014,
  \href{https://arxiv.org/abs/1310.1899}{1310.1899}.

\bibitem{AV01}
A.~Valentini.
\newblock In {\em {Chance in Physics: Foundations and Perspectives}}. Springer,
  Berlin, 2001,
  \href{https://arxiv.org/abs/quant-ph/0104067}{quant-ph/0104067}.

\bibitem{KV16}
A.~Kandhadai and A.~Valentini.
\newblock {Perturbations and quantum relaxation}.
\newblock {\em Found. Phys.},
  \href{http://dx.doi.org/10.1007/s10701-018-0227-3}{49(1):1--23}, 2019,
  \href{https://arxiv.org/abs/1609.04485}{1609.04485}.

\bibitem{AV07}
A.~Valentini.
\newblock {Astrophysical and cosmological tests of quantum theory}.
\newblock {\em J. Phys. A},
  \href{http://dx.doi.org/10.1088/1751-8113/40/12/S24}{40:3285--3303}, 2007,
  \href{https://arxiv.org/abs/hep-th/0610032}{hep-th/0610032}.

\bibitem{AV08}
A.~Valentini.
\newblock {De Broglie-Bohm prediction of quantum violations for cosmological
  super-Hubble modes}.
\newblock 2008, \href{https://arxiv.org/abs/0804.4656}{0804.4656}.

\bibitem{CV15}
S.~Colin and A.~Valentini.
\newblock {Primordial quantum nonequilibrium and large-scale cosmic anomalies}.
\newblock {\em Phys. Rev. D},
  \href{http://dx.doi.org/10.1103/PhysRevD.92.043520}{92(4):043520}, 2015,
  \href{https://arxiv.org/abs/1407.8262}{1407.8262}.

\bibitem{EKC07}
C.~{Efthymiopoulos}, C.~{Kalapotharakos}, and G.~{Contopoulos}.
\newblock {Nodal points and the transition from ordered to chaotic Bohmian
  trajectories}.
\newblock {\em J. Phys. A},
  \href{http://dx.doi.org/10.1088/1751-8113/40/43/008}{40:12945--12972}, 2007,
  \href{https://arxiv.org/abs/0709.2038}{0709.2038}.

\bibitem{ECT17}
C.~{Efthymiopoulos}, G.~{Contopoulos}, and A.C. {Tzemos}.
\newblock {Chaos in de Broglie-Bohm quantum mechanics and the dynamics of
  quantum relaxation}.
\newblock {\em Annales de la Fondation Louis de Broglie}, 42(133), 2017,
  \href{https://arxiv.org/abs/1703.09810}{1703.09810}.

\bibitem{AV10}
A.~Valentini.
\newblock {Inflationary cosmology as a probe of primordial quantum mechanics}.
\newblock {\em Phys. Rev. D},
  \href{http://dx.doi.org/10.1103/PhysRevD.82.063513}{82:063513}, 2010,
  \href{https://arxiv.org/abs/0805.0163}{0805.0163}.

\bibitem{CV13}
S.~Colin and A.~Valentini.
\newblock {Mechanism for the suppression of quantum noise at large scales on
  expanding space}.
\newblock {\em Phys. Rev. D},
  \href{http://dx.doi.org/10.1103/PhysRevD.88.103515}{88:103515}, 2013,
  \href{https://arxiv.org/abs/1306.1579}{1306.1579}.

\bibitem{AVbook}
A.~Valentini.
\newblock {\em {Hidden Variables in Modern Physics and Beyond}}.
\newblock Cambridge University Press, Cambridge, forthcoming.

\bibitem{VF82}
A.~Vilenkin and L.~H. Ford.
\newblock {Gravitational effects upon cosmological phase transitions}.
\newblock {\em Phys. Rev. D},
  \href{http://dx.doi.org/10.1103/PhysRevD.26.1231}{26:1231--1241}, 1982.

\bibitem{L82}
A.~D. Linde.
\newblock {Scalar field fluctuations in the expanding universe and the new
  inflationary universe scenario}.
\newblock {\em Phys. Lett. B},
  \href{http://dx.doi.org/10.1016/0370-2693(82)90293-3}{116(5):335 -- 339},
  1982.

\bibitem{S82}
A.~A. Starobinsky.
\newblock {Dynamics of phase transition in the new inflationary universe
  scenario and generation of perturbations }.
\newblock {\em Phys. Lett. B},
  \href{http://dx.doi.org/10.1016/0370-2693(82)90541-X}{117(3–4):175 -- 178},
  1982.

\bibitem{PK07}
B.~A. Powell and W.~H. Kinney.
\newblock {Pre-inflationary vacuum in the cosmic microwave background}.
\newblock {\em Phys. Rev. D},
  \href{http://dx.doi.org/10.1103/PhysRevD.76.063512}{76:063512}, 2007,
  \href{https://arxiv.org/abs/astro-ph/0612006}{astro-ph/0612006}.

\bibitem{WN08}
I.-C. Wang and K.-W. Ng.
\newblock {Effects of a preinflation radiation-dominated epoch to CMB
  anisotropy}.
\newblock {\em Phys. Rev. D},
  \href{http://dx.doi.org/10.1103/PhysRevD.77.083501}{77:083501}, 2008,
  \href{https://arxiv.org/abs/0704.2095}{0704.2095}.

\bibitem{VPV19}
S.~D.~P. Vitenti, P.~Peter, and A.~Valentini.
\newblock {Modeling the large-scale power deficit}.
\newblock 2019, \href{https://arxiv.org/abs/1901.08885}{1901.08885}.

\bibitem{PlanckXV}
P.~A.~R. Ade et~al.
\newblock {Planck 2013 results. XV. CMB power spectra and likelihood}.
\newblock {\em Astron. Astrophys.},
  \href{http://dx.doi.org/10.1051/0004-6361/201321573}{571:A15}, 2014,
  \href{https://arxiv.org/abs/1303.5075}{1303.5075}.

\bibitem{PlanckXI}
N.~Aghanim et~al.
\newblock {Planck 2015 results. XI. CMB power spectra, likelihoods, and
  robustness of parameters}.
\newblock {\em Astron. Astrophys.},
  \href{http://dx.doi.org/10.1051/0004-6361/201526926}{594:A11}, 2016,
  \href{https://arxiv.org/abs/1507.02704}{1507.02704}.

\bibitem{AV15}
A.~Valentini.
\newblock {Statistical anisotropy and cosmological quantum relaxation}.
\newblock 2015, \href{https://arxiv.org/abs/1510.02523}{1510.02523}.

\bibitem{EHT06}
M.~Endo, K.~Hamaguchi, and F.~Takahashi.
\newblock {Moduli-Induced Gravitino Problem}.
\newblock {\em Phys. Rev. Lett.},
  \href{http://dx.doi.org/10.1103/PhysRevLett.96.211301}{96:211301}, 2006,
  \href{https://arxiv.org/abs/hep-ph/0602061}{hep-ph/0602061}.

\bibitem{KTY06}
M.~Kawasaki, F.~Takahashi, and T.~T. Yanagida.
\newblock {Gravitino-overproduction problem in an inflationary universe}.
\newblock {\em Phys. Rev. D},
  \href{http://dx.doi.org/10.1103/PhysRevD.74.043519}{74:043519}, 2006,
  \href{https://arxiv.org/abs/hep-ph/0605297}{hep-ph/0605297}.

\bibitem{ETY07}
M.~Endo, F.~Takahashi, and T.~T. Yanagida.
\newblock {Inflaton decay in supergravity}.
\newblock {\em Phys. Rev. D},
  \href{http://dx.doi.org/10.1103/PhysRevD.76.083509}{76:083509}, 2007,
  \href{https://arxiv.org/abs/0706.0986}{0706.0986}.

\bibitem{T08}
F.~Takahashi.
\newblock {Gravitino dark matter from inflaton decay }.
\newblock {\em Phys. Lett. B},
  \href{http://dx.doi.org/10.1016/j.physletb.2007.12.048}{660(3):100 -- 106},
  2008, \href{https://arxiv.org/abs/0705.0579}{0705.0579}.

\bibitem{EO13}
J.~Ellis and K.~A. Olive.
\newblock {Supersymmetric DM candidates}.
\newblock In Gianfranco Bertone, editor, {\em Particle Dark Matter:
  Observations, Models and Searches}. Cambridge University Press, Cambridge,
  2013, \href{https://arxiv.org/abs/1001.3651}{1001.3651}.

\bibitem{C14}
L.~Covi.
\newblock {Gravitino Dark Matter confronts LHC}.
\newblock {\em J. Phys: Conf. Series}, 485(1):012002, 2014.

\bibitem{BS88}
L.~Bergstr\"{o}m and H.~Snellman.
\newblock {Observable monochromatic photons from cosmic photino annihilation}.
\newblock {\em Phys. Rev. D},
  \href{http://dx.doi.org/10.1103/PhysRevD.37.3737}{37:3737--3741}, 1988.

\bibitem{Rudaz89}
S.~Rudaz.
\newblock {On the Annihilation of Heavy Neutral Fermion Pairs Into
  Monochromatic gamma-rays and Its Astrophysical Implications}.
\newblock {\em Phys. Rev.},
  \href{http://dx.doi.org/10.1103/PhysRevD.39.3549}{D39:3549}, 1989.

\bibitem{B97}
L.~Bergstr\"{o}m and P.~Ullio.
\newblock {Full one-loop calculation of neutralino annihilation into two
  photons}.
\newblock {\em Nucl. Phys.},
  \href{http://dx.doi.org/http://dx.doi.org/10.1016/S0550-3213(97)00530-0}{504(1):27
  -- 44}, 1997.

\bibitem{B04}
L.~Bergstr\"{o}m, T.~Bringmann, M.~Eriksson, and M.~Gustafsson.
\newblock {Two photon annihilation of Kaluza-Klein dark matter}.
\newblock {\em J. Cosmol. Astropart. Phys.},
  \href{http://dx.doi.org/10.1088/1475-7516/2005/04/004}{0504:004}, 2005,
  \href{https://arxiv.org/abs/hep-ph/0412001}{hep-ph/0412001}.

\bibitem{Creview16}
J.~{Conrad}, J.~{Cohen-Tanugi}, and L.~E. {Strigari}.
\newblock {Wimp searches with gamma rays in the Fermi era: Challenges, methods
  and results}.
\newblock {\em J. Exp. Theor. Phys.},
  \href{http://dx.doi.org/10.1134/S1063776115130099}{121:1104--1135}, 2015,
  \href{https://arxiv.org/abs/1503.06348}{1503.06348}.

\bibitem{DW94}
S.~{Dodelson} and L.~M. {Widrow}.
\newblock {Sterile neutrinos as dark matter}.
\newblock {\em Phys. Rev. Lett.},
  \href{http://dx.doi.org/10.1103/PhysRevLett.72.17}{72:17--20}, 1994,
  \href{https://arxiv.org/abs/hep-ph/9303287}{hep-ph/9303287}.

\bibitem{Abazajian01}
K.~{Abazajian}, G.M. {Fuller}, and W.H. {Tucker}.
\newblock {Direct Detection of Warm Dark Matter in the X-Ray}.
\newblock {\em Astrophys. J.},
  \href{http://dx.doi.org/10.1086/323867}{562:593--604}, 2001,
  \href{https://arxiv.org/abs/astro-ph/0106002}{astro-ph/0106002}.

\bibitem{fermi3.7}
M.~Ackermann et~al.
\newblock {Search for gamma-ray spectral lines with the Fermi Large Area
  Telescope and dark matter implications}.
\newblock {\em Phys. Rev. D},
  \href{http://dx.doi.org/10.1103/PhysRevD.88.082002}{88:082002}, 2013,
  \href{https://arxiv.org/abs/1305.5597}{1305.5597}.

\bibitem{fermi5.8}
M.~{Ackermann} et~al.
\newblock {Updated search for spectral lines from Galactic dark matter
  interactions with pass 8 data from the Fermi Large Area Telescope}.
\newblock {\em Phys. Rev. D},
  \href{http://dx.doi.org/10.1103/PhysRevD.91.122002}{91(12):122002}, 2015,
  \href{https://arxiv.org/abs/1506.00013}{1506.00013}.

\bibitem{Pull07}
A.R. {Pullen}, R.-R. {Chary}, and M.~{Kamionkowski}.
\newblock {Search with EGRET for a gamma ray line from the Galactic center}.
\newblock {\em Phys. Rev. D},
  \href{http://dx.doi.org/10.1103/PhysRevD.76.063006}{76(6):063006}, 2007,
  \href{https://arxiv.org/abs/astro-ph/0610295}{astro-ph/0610295}.

\bibitem{Mack08}
G.D. {Mack}, T.D. {Jacques}, J.F. {Beacom}, N.F. {Bell}, and H.~{Y{\"u}ksel}.
\newblock {Conservative constraints on dark matter annihilation into gamma
  rays}.
\newblock {\em Phys. Rev. D},
  \href{http://dx.doi.org/10.1103/PhysRevD.78.063542}{78(6):063542}, 2008,
  \href{https://arxiv.org/abs/0803.0157}{0803.0157}.

\bibitem{W12}
C.~Weniger.
\newblock {A tentative gamma-ray line from Dark Matter annihilation at the
  Fermi Large Area Telescope}.
\newblock {\em J. Cosmol. Astropart. Phys.}, 1208(08):007, 2012,
  \href{https://arxiv.org/abs/1204.2797}{1204.2797}.

\bibitem{SF12-1}
M.~Su and D.~P. Finkbeiner.
\newblock {Strong Evidence for Gamma-ray Line Emission from the Inner Galaxy}.
\newblock 2012, \href{https://arxiv.org/abs/1206.1616}{1206.1616}.

\bibitem{Albert14}
A.~{Albert}, G.A. {G{\'o}mez-Vargas}, M.~{Grefe}, C.~{Mu{\~n}oz}, C.~{Weniger},
  E.D. {Bloom}, E.~{Charles}, M.N. {Mazziotta}, and A.~{Morselli}.
\newblock {Search for 100 MeV to 10 GeV {$\gamma$}-ray lines in the Fermi-LAT
  data and implications for gravitino dark matter in the {$\mu$}{$\nu$}SSM}.
\newblock {\em J. Cosmol. Astropart. Phys.},
  \href{http://dx.doi.org/10.1088/1475-7516/2014/10/023}{10:023}, 2014,
  \href{https://arxiv.org/abs/1406.3430}{1406.3430}.

\bibitem{HESS16}
H.~{Abdalla} et~al.
\newblock {H.E.S.S. Limits on Linelike Dark Matter Signatures in the 100 GeV to
  2 TeV Energy Range Close to the Galactic Center}.
\newblock {\em Phys. Rev. Lett.},
  \href{http://dx.doi.org/10.1103/PhysRevLett.117.151302}{117(15):151302},
  2016, \href{https://arxiv.org/abs/1609.08091}{1609.08091}.

\bibitem{Bulbul14}
E.~{Bulbul}, M.~{Markevitch}, A.~{Foster}, R.~K. {Smith}, M.~{Loewenstein}, and
  S.~W. {Randall}.
\newblock {Detection of an Unidentified Emission Line in the Stacked X-Ray
  Spectrum of Galaxy Clusters}.
\newblock {\em Astrophys. J.},
  \href{http://dx.doi.org/10.1088/0004-637X/789/1/13}{789:13}, 2014,
  \href{https://arxiv.org/abs/1402.2301}{1402.2301}.

\bibitem{Boyarsky14}
A.~Boyarsky, O.~Ruchayskiy, D.~Iakubovskyi, and J.~Franse.
\newblock {Unidentified Line in X-Ray Spectra of the Andromeda Galaxy and
  Perseus Galaxy Cluster}.
\newblock {\em Phys. Rev. Lett.},
  \href{http://dx.doi.org/10.1103/PhysRevLett.113.251301}{113:251301}, 2014,
  \href{https://arxiv.org/abs/1402.4119}{1402.4119}.

\bibitem{A15}
M.~E. {Anderson}, E.~{Churazov}, and J.~N. {Bregman}.
\newblock {Non-detection of X-ray emission from sterile neutrinos in stacked
  galaxy spectra}.
\newblock {\em Mon. Notices Royal Astron. Soc.},
  \href{http://dx.doi.org/10.1093/mnras/stv1559}{452:3905--3923}, 2015,
  \href{https://arxiv.org/abs/1408.4115}{1408.4115}.

\bibitem{Hitomi16}
F.~A. Aharonian et~al.
\newblock {$Hitomi$ constraints on the 3.5 keV line in the Perseus galaxy
  cluster}.
\newblock {\em Astrophys. J.},
  \href{http://dx.doi.org/10.3847/2041-8213/aa61fa}{837(1):L15}, 2017,
  \href{https://arxiv.org/abs/1607.07420}{1607.07420}.

\bibitem{Ehrenfests}
P.~Ehrenfest and T.~Ehrenfest.
\newblock {\em {The conceptual foundations of the statistical approach in
  mechanics}}.
\newblock Cornell University Press, 1959.

\bibitem{Daviesbook}
P.~C.~W. Davies.
\newblock {\em {The Physics of Time Asymmetry}}.
\newblock University of California Press, 1977.

\bibitem{BB}
S.J. Blundell and K.M. Blundell.
\newblock {\em {Concepts in Thermal Physics}}.
\newblock Oxford University Press, 2006.

\bibitem{B1889}
J.~Bertrand.
\newblock {Calcul des probabilit\'es}.
\newblock {\em Gauthier-Villars}, 1889.

\bibitem{J73}
E.~T. Jaynes.
\newblock {The well-posed problem}.
\newblock {\em Found. Phys.},
  \href{http://dx.doi.org/10.1007/BF00709116}{156:5--11}, 1973.

\bibitem{BH89}
D.~{Bohm} and B.~J. {Hiley}.
\newblock {Non-locality and locality in the stochastic interpretation of
  quantum mechanics}.
\newblock {\em Phys. Rep.},
  \href{http://dx.doi.org/10.1016/0370-1573(89)90160-9}{172:93--122}, 1989.

\bibitem{Shannon}
C.~E. Shannon.
\newblock {A mathematical theory of communication}.
\newblock {\em The Bell System Technical Journal},
  \href{http://dx.doi.org/10.1002/j.1538-7305.1948.tb01338.x}{27(3):379--423},
  1948.

\bibitem{Gibbsbook}
J.~W. Gibbs.
\newblock {\em {Elementary Principles in Statistical Mechanics}}.
\newblock Charles Scriber's Sons, 1902.

\bibitem{Nakahara}
M.~Nakahara.
\newblock {\em {Geometry, Topology and Physics}}.
\newblock IOP Publishing, Bristol, UK, 2003.

\bibitem{HHD}
H.~Bhatia, G.~Norgard, V.~Pascucci, and P.~T. Bremer.
\newblock {The Helmholtz-Hodge Decomposition--A Survey}.
\newblock {\em IEEE Transactions on Visualization and Computer Graphics},
  \href{http://dx.doi.org/10.1109/TVCG.2012.316}{19(8):1386--1404}, 2013.

\bibitem{Penrosebook}
R.~Penrose.
\newblock {\em {The Road to Reality: A Complete Guide to the Laws of the
  Universe}}.
\newblock Random House, 2016.

\bibitem{R61}
L.~F. Richardson.
\newblock {\em {The problem of contiguity: an appendix of statistics of deadly
  quarrels}}.
\newblock 1961.

\bibitem{M67}
B.~Mandelbrot.
\newblock {How Long Is the Coast of Britain? Statistical Self-Similarity and
  Fractional Dimension}.
\newblock {\em Science}, 156(3775):636--638, 1967.

\bibitem{Goldstein}
H.~Goldstein, C.P. Poole, and J.L. Safko.
\newblock {\em {Classical Mechanics}}.
\newblock Addison Wesley, 2002.

\bibitem{Mandl_Shaw}
F.~Mandl and G.~Shaw.
\newblock {\em {Quantum Field Theory}}.
\newblock Wiley, 2010.

\bibitem{Peskin_Schroeder}
M.~E. Peskin and D.~V. Schroeder.
\newblock {\em {An Introduction to quantum field theory}}.
\newblock Addison-Wesley, Reading, USA, 1995.

\bibitem{Bertrand}
J.~Bertrand.
\newblock {\em {Calcul Des Probabilit\'{e}s}}.
\newblock Paris, Gauthier-Villars, 1889.

\bibitem{encyclopedia_of_mathematics}
G.~S. Hall.
\newblock Weyl tensor.
\newblock Encyclopedia of Mathematics.
\newblock Springer.

\bibitem{R15}
P.~{Roser}.
\newblock {Quantum mechanics as the dynamical geometry of trajectories}.
\newblock 2015, \href{https://arxiv.org/abs/1507.08975}{1507.08975}.

\bibitem{CS07}
S.~Colin and W.~Struyve.
\newblock {A Dirac sea pilot-wave model for quantum field theory}.
\newblock {\em J. Phys. A},
  \href{http://dx.doi.org/10.1088/1751-8113/40/26/015}{40:7309--7342}, 2007,
  \href{https://arxiv.org/abs/quant-ph/0701085}{quant-ph/0701085}.

\bibitem{CW11}
S.~Colin and H.~M. Wiseman.
\newblock {The Zig-zag road to reality}.
\newblock {\em J. Phys. A},
  \href{http://dx.doi.org/10.1088/1751-8113/44/34/345304}{44:345304}, 2011,
  \href{https://arxiv.org/abs/1107.4909}{1107.4909}.

\bibitem{Hatfield}
B.~Hatfield.
\newblock {\em {Quantum Field Theory Of Point Particles And Strings}}.
\newblock Frontiers in Physics. Avalon Publishing, 1998.

\bibitem{deP95}
Gonzalo Garc\'{\i}a~de Polavieja.
\newblock Exponential divergence of neighboring quantal trajectories.
\newblock {\em Phys. Rev. A},
  \href{http://dx.doi.org/10.1103/PhysRevA.53.2059}{53:2059--2061}, Apr 1996.

\bibitem{FS95}
F.~H.~M. Faisal and U.~Schwengelbeck.
\newblock Unified theory of lyapunov exponents and a positive example of
  deterministic quantum chaos.
\newblock {\em Physics Letters A},
  \href{http://dx.doi.org/https://doi.org/10.1016/0375-9601(95)00645-J}{207(1):31
  -- 36}, 1995.

\bibitem{PV95}
R.~H. {Parmenter} and R.~W. {Valentine}.
\newblock {Deterministic chaos and the causal interpretation of quantum
  mechanics (Physics Letters A 201 (1995) 1)}.
\newblock {\em Physics Letters A},
  \href{http://dx.doi.org/10.1016/0375-9601(96)00096-5}{213:319--319}, February
  1996.

\bibitem{DM96}
R.~H. {Parmenter} and R.~W. {Valentine}.
\newblock {Deterministic chaos and the causal interpretation of quantum
  mechanics (Physics Letters A 201 (1995) 1)}.
\newblock {\em Physics Letters A},
  \href{http://dx.doi.org/10.1016/0375-9601(96)00096-5}{213:319--319}, February
  1996.

\bibitem{SC96}
S.~Sengupta and P.~K. Chattaraj.
\newblock The quantum theory of motion and signatures of chaos in the quantum
  behaviour of a classically chaotic system.
\newblock {\em Physics Letters A},
  \href{http://dx.doi.org/https://doi.org/10.1016/0375-9601(96)00240-X}{215(3):119
  -- 127}, 1996.

\bibitem{C00}
J.~T.~Cushing.
\newblock Bohmian insights into quantum chaos.
\newblock {\em Philosophy of Science},
  \href{http://dx.doi.org/10.1086/392836}{67}, 09 2000.

\bibitem{MPD00}
A.~J. Makowski, P.~Pepłowski, and S.~T. Dembiński.
\newblock Chaotic causal trajectories: the role of the phase of stationary
  states.
\newblock {\em Physics Letters A},
  \href{http://dx.doi.org/https://doi.org/10.1016/S0375-9601(00)00047-5}{266(4):241
  -- 248}, 2000.

\bibitem{SF08}
K.~G. {Schlegel} and S.~{F{\"o}rster}.
\newblock {Deterministic chaos in entangled eigenstates}.
\newblock {\em Physics Letters A},
  \href{http://dx.doi.org/10.1016/j.physleta.2008.02.044}{372:3620--3631}, May
  2008.

\bibitem{F97}
H.~Frisk.
\newblock {Properties of the trajectories in {B}ohmian mechanics}.
\newblock {\em Phys. Lett. A},
  \href{http://dx.doi.org/10.1016/S0375-9601(97)00044-3}{227(3):139--142},
  1997.

\bibitem{WS99}
H.~Wu and D.~W.~L. Sprung.
\newblock {Quantum chaos in terms of {B}ohm trajectories}.
\newblock {\em Phys. Lett. A},
  \href{http://dx.doi.org/10.1016/S0375-9601(99)00629-5}{261(3–4):150 --
  157}, 1999.

\bibitem{FF03}
P.~{Falsaperla} and G.~{Fonte}.
\newblock {On the motion of a single particle near a nodal line in the de
  Broglie-Bohm interpretation of quantum mechanics}.
\newblock {\em Physics Letters A},
  \href{http://dx.doi.org/10.1016/j.physleta.2003.08.010}{316:382--390},
  October 2003.

\bibitem{WP05}
D.~A. {Wisniacki} and E.~R. {Pujals}.
\newblock {Motion of vortices implies chaos in Bohmian mechanics}.
\newblock {\em EPL},
  \href{http://dx.doi.org/10.1209/epl/i2005-10085-3}{71:159--165}, 2005,
  \href{https://arxiv.org/abs/quant-ph/0502108}{quant-ph/0502108}.

\bibitem{WPB06}
D.~A. {Wisniacki}, E.~R. {Pujals}, and F.~{Borondo}.
\newblock {Vortex interaction, chaos and quantum probabilities}.
\newblock {\em EPL},
  \href{http://dx.doi.org/10.1209/epl/i2005-10467-5}{73:671--676}, 2006,
  \href{https://arxiv.org/abs/nlin/0507015}{nlin/0507015}.

\bibitem{CE08}
G.~Contopoulos and C.~Efthymiopoulos.
\newblock Ordered and chaotic bohmian trajectories.
\newblock {\em Celestial Mechanics and Dynamical Astronomy},
  \href{http://dx.doi.org/10.1007/s10569-008-9127-8}{102(1):219--239}, Sep
  2008.

\bibitem{EKC09}
C.~{Efthymiopoulos}, C.~{Kalapotharakos}, and G.~{Contopoulos}.
\newblock {Origin of chaos near critical points of quantum flow}.
\newblock {\em Phys. Rev. E},
  \href{http://dx.doi.org/10.1103/PhysRevE.79.036203}{79:036203}, Mar 2009,
  \href{https://arxiv.org/abs/0903.2655}{0903.2655}.

\bibitem{CMS16}
A.~{Cesa}, J.~{Martin}, and W.~{Struyve}.
\newblock {Chaotic Bohmian trajectories for stationary states}.
\newblock {\em Journal of Physics A Mathematical General},
  \href{http://dx.doi.org/10.1088/1751-8113/49/39/395301}{49:395301}, Sep 2016,
  \href{https://arxiv.org/abs/1603.01387}{1603.01387}.

\bibitem{TCE16}
A.~C. Tzemos, G.~Contopoulos, and C.~Efthymiopoulos.
\newblock {Origin of chaos in 3-d Bohmian trajectories }.
\newblock {\em Phys. Lett. A},
  \href{http://dx.doi.org/10.1016/j.physleta.2016.09.016}{380(45):3796 --
  3802}, 2016, \href{https://arxiv.org/abs/1609.07069}{1609.07069}.

\bibitem{B10}
A.~F. Bennett.
\newblock Relative dispersion and quantum thermal equilibrium in de
  broglie{\textendash}bohm mechanics.
\newblock {\em Journal of Physics A: Mathematical and Theoretical},
  \href{http://dx.doi.org/10.1088/1751-8113/43/19/195304}{43(19):195304}, apr
  2010.

\bibitem{CDE12}
G.~{Contopoulos}, N.~{Delis}, and C.~{Efthymiopoulos}.
\newblock {Order in de Broglie-Bohm quantum mechanics}.
\newblock {\em Journal of Physics A Mathematical General},
  \href{http://dx.doi.org/10.1088/1751-8113/45/16/165301}{45:165301}, Apr 2012,
  \href{https://arxiv.org/abs/1203.5598}{1203.5598}.

\bibitem{W07}
H.~M. {Wiseman}.
\newblock {Grounding Bohmian mechanics in weak values and Bayesianism}.
\newblock {\em New J. Phys.},
  \href{http://dx.doi.org/10.1088/1367-2630/9/6/165}{9:165}, 2007,
  \href{https://arxiv.org/abs/0706.2522}{0706.2522}.

\bibitem{DG97}
E.~Deotto and G.~C. Ghirardi.
\newblock {Bohmian mechanics revisited}.
\newblock {\em Found. Phys.},
  \href{http://dx.doi.org/10.1023/A:1018752202576}{28:1--30}, 1998,
  \href{https://arxiv.org/abs/quant-ph/9704021}{quant-ph/9704021}.

\bibitem{DGZ92}
D.~{D{\"u}rr}, S.~{Goldstein}, and N.~{Zangh{\'{\i}}}.
\newblock {Quantum equilibrium and the origin of absolute uncertainty}.
\newblock {\em J. Stat. Phys.},
  \href{http://dx.doi.org/10.1007/BF01049004}{67:843--907}, 1992,
  \href{https://arxiv.org/abs/quant-ph/0308039}{quant-ph/0308039}.

\bibitem{DT09}
D.~D{\"u}rr and S.~Teufel.
\newblock {\em {Bohmian Mechanics: The Physics and Mathematics of Quantum
  Theory}}.
\newblock Springer, Berlin, 2009.

\bibitem{AV96}
A.~Valentini.
\newblock In {\em Bohmian Mechanics and Quantum Theory: An Appraisal}. Kluwer,
  Dordrecht, 1996.

\bibitem{AV04b}
A.~Valentini.
\newblock {Black Holes, Information Loss, and Hidden Variables}.
\newblock 2004, \href{https://arxiv.org/abs/hep-th/0407032}{hep-th/0407032}.

\bibitem{AV02}
A.~Valentini.
\newblock {Subquantum information and computation}.
\newblock {\em Pramana J. Phys.},
  \href{http://dx.doi.org/10.1007/s12043-002-0117-1}{59(2):269--277}, 2002,
  \href{https://arxiv.org/abs/quant-ph/0203049}{quant-ph/0203049}.

\bibitem{PV06}
P.~Pearle and A.~Valentini.
\newblock In J.-P. Fran\c{c}oise et~al., editors, {\em Encyclopaedia of
  Mathematical Physics}. Elsevier, North-Holland, 2006,
  \href{https://arxiv.org/abs/quant-ph/0506115}{quant-ph/0506115}.

\bibitem{CV16}
S.~Colin and A.~Valentini.
\newblock {Robust predictions for the large-scale cosmological power deficit
  from primordial quantum nonequilibrium}.
\newblock {\em Int. J. Mod. Phys. D},
  \href{http://dx.doi.org/10.1142/S0218271816500681}{25(06):1650068}, 2016,
  \href{https://arxiv.org/abs/1510.03508}{1510.03508}.

\bibitem{CT}
C.~Cohen-Tannoudji, B.~Diu, and F.~Lalo\"{e}.
\newblock {\em {Quantum Mechanics}}.
\newblock Wiley, 1977.

\bibitem{Dirac31}
P.~A.~M. Dirac.
\newblock {Quantised singularities in the electromagnetic field}.
\newblock {\em Proc R. Soc. A},
  \href{http://dx.doi.org/10.1098/rspa.1931.0130}{133(821):60--72}, 1931.

\bibitem{HCP74}
J.~O. Hirschfelder, A.~C. Christoph, and W.~E. Palke.
\newblock {Quantum mechanical streamlines. I. Square potential barrier}.
\newblock {\em J. Chem. Phys.},
  \href{http://dx.doi.org/10.1063/1.1681899}{61(12):5435--5455}, 1974.

\bibitem{HGB74}
J.~O. Hirschfelder, C.~J. Goebel, and L.~W. Bruch.
\newblock {Quantized vortices around wavefunction nodes. II}.
\newblock {\em J. Chem. Phys.},
  \href{http://dx.doi.org/10.1063/1.1681900}{61(12):5456--5459}, 1974.

\bibitem{W06}
R.E. Wyatt.
\newblock {\em {Quantum Dynamics with Trajectories: Introduction to Quantum
  Hydrodynamics}}.
\newblock Springer, New York, 2006.

\bibitem{H77}
J.~O. Hirschfelder.
\newblock {The angular momentum, creation, and significance of quantized
  vortices}.
\newblock {\em J. Chem. Phys.},
  \href{http://dx.doi.org/10.1063/1.434769}{67(12):5477--5483}, 1977.

\bibitem{Conway}
J.~B. Conway.
\newblock {\em {Functions of One Complex Variable}}.
\newblock Springer-Verlag, Berlin, 1973.

\bibitem{AV09}
A.~Valentini.
\newblock {Beyond the Quantum}.
\newblock {\em Phys. World}, 22N11:32--37, 2009,
  \href{https://arxiv.org/abs/1001.2758}{1001.2758}.

\bibitem{AVPwtMw}
A.~Valentini.
\newblock In Simon Saunders et~al., editors, {\em {Many Worlds? {E}verett,
  Quantum Theory, \& Reality}}. Oxford University Press, Oxford, 2010,
  \href{https://arxiv.org/abs/0811.0810}{0811.0810}.

\bibitem{LL00}
A.~R. Liddle and D.~H. Lyth.
\newblock {\em {Cosmological Inflation and Large-Scale Structure}}.
\newblock Cambridge University Press, Cambridge, 2000.

\bibitem{Muk05}
V.~Mukhanov.
\newblock {\em {Physical Foundations of Cosmology}}.
\newblock Cambridge University Press, Cambridge, 2005.

\bibitem{W08}
S.~Weinberg.
\newblock {\em {Cosmology}}.
\newblock Oxford University Press, Oxford, 2008.

\bibitem{PU09}
P.~Peter and J.-P. Uzan.
\newblock {\em {Primordial Cosmology}}.
\newblock Oxford University Press, Oxford, 2009.

\bibitem{AV04a}
A.~Valentini.
\newblock {Universal signature of non-quantum systems}.
\newblock {\em Phys. Lett. A},
  \href{http://dx.doi.org/10.1016/j.physleta.2004.10.002}{332(3–4):187 --
  193}, 2004, \href{https://arxiv.org/abs/quant-ph/0309107}{quant-ph/0309107}.

\bibitem{SV08}
W.~Struyve and A.~Valentini.
\newblock {De {B}roglie-{B}ohm guidance equations for arbitrary
  {H}amiltonians}.
\newblock {\em J. Phys. A: Math. Theor.},
  \href{http://dx.doi.org/10.1088/1751-8113/42/3/035301}{42(3):035301}, 2009,
  \href{https://arxiv.org/abs/0808.0290}{0808.0290}.

\bibitem{AV08a}
A.~Valentini.
\newblock {Hidden variables and the large-scale structure of space-time}.
\newblock In W.~L. Craig and Q.~Smith, editors, {\em Einstein, Relativity and
  Absolute Simultaneity}. Routledge, London, 2008,
  \href{https://arxiv.org/abs/quant-ph/0504011}{quant-ph/0504011}.

\bibitem{BM01}
R.~H. Brandenberger and J.~Martin.
\newblock {The robustness of inflation to changes in super-{P}lanck-scale
  physics}.
\newblock {\em Mod. Phys. Lett. A},
  \href{http://dx.doi.org/10.1142/S0217732301004170}{16(15):999--1006}, 2001,
  \href{https://arxiv.org/abs/astro-ph/0005432}{astro-ph/0005432}.

\bibitem{MB01}
J.~Martin and R.~H. Brandenberger.
\newblock {Trans-Planckian problem of inflationary cosmology}.
\newblock {\em Phys. Rev. D},
  \href{http://dx.doi.org/10.1103/PhysRevD.63.123501}{63:123501}, 2001,
  \href{https://arxiv.org/abs/hep-th/0005209}{hep-th/0005209}.

\bibitem{BM13}
R.~H. Brandenberger and J.~Martin.
\newblock {Trans-Planckian Issues for Inflationary Cosmology}.
\newblock {\em Class. Quant. Grav.},
  \href{http://dx.doi.org/10.1088/0264-9381/30/11/113001}{30:113001}, 2013,
  \href{https://arxiv.org/abs/1211.6753}{1211.6753}.

\bibitem{BH93}
D.~Bohm and B.~J. Hiley.
\newblock {\em {The Undivided Universe: an Ontological Interpretation of
  Quantum Theory}}.
\newblock Routledge, London, 1993.

\bibitem{C03}
S.~Colin.
\newblock {A deterministic Bell model}.
\newblock {\em Phys. Lett. A},
  \href{http://dx.doi.org/10.1016/j.physleta.2003.09.006}{317(5–6):349 --
  358}, 2003, \href{https://arxiv.org/abs/quant-ph/0310055}{quant-ph/0310055}.

\bibitem{S10}
W.~Struyve.
\newblock {Pilot-wave theory and quantum fields}.
\newblock {\em Rep. Prog. Phys.}, 73(10):106001, 2010,
  \href{https://arxiv.org/abs/0707.3685}{0707.3685}.

\bibitem{BTW06}
B.~A. Bassett, S.~Tsujikawa, and D.~Wands.
\newblock {Inflation dynamics and reheating}.
\newblock {\em Rev. Mod. Phys.},
  \href{http://dx.doi.org/10.1103/RevModPhys.78.537}{78:537--589}, 2006,
  \href{https://arxiv.org/abs/astro-ph/0507632}{astro-ph/0507632}.

\bibitem{TB90}
J.~H. Traschen and R.~H. Brandenberger.
\newblock {Particle production during out-of-equilibrium phase transitions}.
\newblock {\em Phys. Rev. D},
  \href{http://dx.doi.org/10.1103/PhysRevD.42.2491}{42:2491--2504}, 1990.

\bibitem{ABCM10}
R.~Allahverdi, R.~Brandenberger, F.-Y. Cyr-Racine, and A.~Mazumdar.
\newblock {Reheating in Inflationary Cosmology: Theory and Applications}.
\newblock {\em Ann. Rev. Nucl. Part. Sci.},
  \href{http://dx.doi.org/10.1146/annurev.nucl.012809.104511}{60:27--51}, 2010,
  \href{https://arxiv.org/abs/1001.2600}{1001.2600}.

\bibitem{FY02}
M.~Fujii and T.~Yanagida.
\newblock {Natural gravitino dark matter and thermal leptogenesis in
  gauge-mediated supersymmetry-breaking models }.
\newblock {\em Phys. Lett. B},
  \href{http://dx.doi.org/10.1016/S0370-2693(02)02958-1}{549(3–4):273 --
  283}, 2002, \href{https://arxiv.org/abs/hep-ph/0208191}{hep-ph/0208191}.

\bibitem{NY06}
S.~Nakamura and M.~Yamaguchi.
\newblock {Gravitino production from heavy moduli decay and cosmological moduli
  problem revived }.
\newblock {\em Phys. Lett. B},
  \href{http://dx.doi.org/10.1016/j.physletb.2006.05.078}{638(5–6):389 --
  395}, 2006.

\bibitem{SC13}
S.~Colin.
\newblock {Sub-Compton quantum non-equilibrium and Majorana systems}.
\newblock 2013, \href{https://arxiv.org/abs/1306.0967}{1306.0967}.

\bibitem{BD82}
N.~D. Birrell and P.~C.~W. Davies.
\newblock {\em {Quantum Fields in Curved Space}}.
\newblock Cambridge University Press, Cambridge, 1982.

\bibitem{PT09}
L.~Parker and D.~Toms.
\newblock {\em {Quantum Field Theory in Curved Spacetime: Quantized Fields and
  Gravity}}.
\newblock Cambrdge University Press, Cambridge, 2009.

\bibitem{Knight}
C.~Gerry and P.~Knight.
\newblock {\em {Introductory Quantum Optics}}.
\newblock Cambridge University Press, 2004.

\bibitem{jc}
E.T. Jaynes and F.W. Cummings.
\newblock {Comparison of quantum and semiclassical radiation theories with
  application to the beam maser}.
\newblock {\em Proc. IEEE},
  \href{http://dx.doi.org/10.1109/PROC.1963.1664}{51(1):89--109}, 1963.

\bibitem{schiff}
L.~I. Schiff.
\newblock {\em {Quantum Mechanics}}.
\newblock McGraw-Hill, 1968.

\bibitem{BHS05}
Gianfranco B., Dan Hooper, and J.~S.
\newblock {Particle dark matter: evidence, candidates and constraints}.
\newblock {\em Phys. Rep.},
  \href{http://dx.doi.org/10.1016/j.physrep.2004.08.031}{405(5–6):279 --
  390}, 2005, \href{https://arxiv.org/abs/hep-ph/0404175}{hep-ph/0404175}.

\bibitem{B12}
T.~Bringmann, X.~Huang, A.~Ibarra, S.~Vogl, and C.~Weniger.
\newblock {{Fermi LAT} Search for Internal Bremsstrahlung Signatures from Dark
  Matter Annihilation}.
\newblock {\em J. Cosmol. Astropart. Phys.},
  \href{http://dx.doi.org/10.1088/1475-7516/2012/07/054}{1207:054}, 2012,
  \href{https://arxiv.org/abs/1203.1312}{1203.1312}.

\bibitem{BH12}
M.~R. Buckley and D.~Hooper.
\newblock {Implications of a 130 GeV gamma-ray line for dark matter}.
\newblock {\em Phys. Rev. D},
  \href{http://dx.doi.org/10.1103/PhysRevD.86.043524}{86:043524}, 2012,
  \href{https://arxiv.org/abs/1205.6811}{1205.6811}.

\bibitem{ITW13}
A.~Ibarra, D.~Tran, and C.~Weniger.
\newblock {Indirect seaches for decaying dark matter}.
\newblock {\em Int. J. Mod. Phys. A},
  \href{http://dx.doi.org/10.1142/S0217751X13300408}{28(27):1330040}, 2013,
  \href{https://arxiv.org/abs/1307.6434}{1307.6434}.

\bibitem{L13}
S.P. Liew.
\newblock {Gamma-ray line from radiative decay of gravitino dark matter}.
\newblock {\em Phys. Lett. B},
  \href{http://dx.doi.org/10.1016/j.physletb.2013.06.006}{724(1–3):88 -- 91},
  2013, \href{https://arxiv.org/abs/1304.1992}{1304.1992}.

\bibitem{Aetal14}
A.~{Albert}, G.~A. {G{\'o}mez-Vargas}, M.~{Grefe}, C.~{Mu{\~n}oz},
  C.~{Weniger}, E.~D. {Bloom}, E.~{Charles}, M.~N. {Mazziotta}, and
  A.~{Morselli}.
\newblock {Search for 100 MeV to 10 GeV {$\gamma$}-ray lines in the Fermi-LAT
  data and implications for gravitino dark matter in the {$\mu$}{$\nu$}SSM}.
\newblock {\em J. Cosmol. Astropart. Phys.},
  \href{http://dx.doi.org/10.1088/1475-7516/2014/10/023}{10:023}, 2014,
  \href{https://arxiv.org/abs/1406.3430}{1406.3430}.

\bibitem{TY00}
F.~Takayama and M.~Yamaguchi.
\newblock {Gravitino dark matter without R-parity}.
\newblock {\em Phys. Lett. B},
  \href{http://dx.doi.org/10.1016/S0370-2693(00)00726-7}{485(4):388 -- 392},
  2000, \href{https://arxiv.org/abs/hep-ph/0005214}{hep-ph/0005214}.

\bibitem{FLAT13}
M.~Ackermann et~al.
\newblock {Search for Gamma-ray Spectral Lines with the Fermi Large Area
  Telescope and Dark Matter Implications}.
\newblock {\em Phys. Rev.},
  \href{http://dx.doi.org/10.1103/PhysRevD.88.082002}{D88:082002}, 2013,
  \href{https://arxiv.org/abs/1305.5597}{1305.5597}.

\bibitem{leve02}
R.~J. LeVeque.
\newblock {\em {Finite Volume Methods for Hyperbolic Problems}}.
\newblock Cambridge University Press, Cambridge, 2002.

\bibitem{AV02A}
A.~{Valentini}.
\newblock {Signal-locality in hidden-variables theories}.
\newblock {\em Physics Letters A},
  \href{http://dx.doi.org/10.1016/S0375-9601(02)00438-3}{297:273--278}, 2002,
  \href{https://arxiv.org/abs/quant-ph/0106098}{quant-ph/0106098}.

\bibitem{AV14}
A.~Valentini.
\newblock {Trans-Planckian fluctuations and the stability of quantum
  mechanics}.
\newblock 2014, \href{https://arxiv.org/abs/1409.7467}{1409.7467}.

\bibitem{Berg12}
L.~Bergstr\"{o}m.
\newblock {Dark Matter Evidence, Particle Physics Candidates and Detection
  Methods}.
\newblock {\em Annalen Phys.},
  \href{http://dx.doi.org/10.1002/andp.201200116}{524:479--496}, 2012,
  \href{https://arxiv.org/abs/1205.4882}{1205.4882}.

\bibitem{LAT12}
M.~{Ackermann} et~al.
\newblock {The Fermi Large Area Telescope on Orbit: Event Classification,
  Instrument Response Functions, and Calibration}.
\newblock {\em Astrophys. J. Suppl.Ser.},
  \href{http://dx.doi.org/10.1088/0067-0049/203/1/4}{203:4}, 2012,
  \href{https://arxiv.org/abs/1206.1896}{1206.1896}.

\bibitem{XMM_MOS01}
M.J.L. {Turner} et~al.
\newblock {The European Photon Imaging Camera on XMM-Newton: The MOS cameras :
  The MOS cameras}.
\newblock {\em Astron. Astrophys.},
  \href{http://dx.doi.org/10.1051/0004-6361:20000087}{365:L27--L35}, 2001,
  \href{https://arxiv.org/abs/astro-ph/0011498}{astro-ph/0011498}.

\bibitem{EGRET93}
D.~J. Thompson et~al.
\newblock {Calibration of the Energetic Gamma-Ray Experiment Telescope (EGRET)
  for the Compton Gamma-Ray Observatory}.
\newblock {\em Astrophys. J. Suppl. Ser.},
  \href{http://dx.doi.org/10.1086/191793}{86:629--656}, 1993.

\bibitem{CALET}
J.~{Krizmanic}.
\newblock {The CALorimetric Electron Telescope (CALET): a High-Energy
  Astroparticle Physics Observatory on the International Space Station}.
\newblock In {\em APS Meeting Abstracts}, 2013.

\bibitem{DAMPE}
J.~Chang.
\newblock {{Dark Matter Particle Explorer}: The First {C}hinese Cosmic Ray and
  Hard $\gamma$-ray Detector in Space}.
\newblock {\em {Chin. J. Space Sci.}},
  \href{http://dx.doi.org/10.11728/cjss2014.05.550}{34(5):550}, 2014.

\bibitem{BBU14}
M.~{Bahrami}, A.~{Bassi}, and H.~{Ulbricht}.
\newblock {Testing the quantum superposition principle in the frequency
  domain}.
\newblock {\em Phys. Rev. A},
  \href{http://dx.doi.org/10.1103/PhysRevA.89.032127}{89(3):032127}, 2014,
  \href{https://arxiv.org/abs/1309.5889}{1309.5889}.

\bibitem{Bring12}
T.~Bringmann, X.~Huang, A.~Ibarra, S.~Vogl, and C.~Weniger.
\newblock {{Fermi LAT} Search for Internal Bremsstrahlung Signatures from Dark
  Matter Annihilation}.
\newblock {\em J. Cosmol. Astropart. Phys.},
  \href{http://dx.doi.org/10.1088/1475-7516/2012/07/054}{1207:054}, 2012,
  \href{https://arxiv.org/abs/1203.1312}{1203.1312}.

\bibitem{SF12-2}
M.~Su and D.~P. Finkbeiner.
\newblock {Double Gamma-ray Lines from Unassociated {F}ermi-{LAT} Sources}.
\newblock 2012, \href{https://arxiv.org/abs/1207.7060}{1207.7060}.

\bibitem{NME16}
A.~Neronov, D.~Malyshev, and D.~Eckert.
\newblock {Decaying dark matter search with NuSTAR deep sky observations}.
\newblock {\em Phys. Rev. D},
  \href{http://dx.doi.org/10.1103/PhysRevD.94.123504}{D94(12):123504}, 2016,
  \href{https://arxiv.org/abs/1607.07328}{1607.07328}.

\bibitem{MNE14}
D.~Malyshev, A.~Neronov, and D.~Eckert.
\newblock {Constraints on 3.55 keV line emission from stacked observations of
  dwarf spheroidal galaxies}.
\newblock {\em Phys. Rev. D},
  \href{http://dx.doi.org/10.1103/PhysRevD.90.103506}{D90:103506}, 2014,
  \href{https://arxiv.org/abs/1408.3531}{1408.3531}.

\bibitem{JP16}
T.E. Jeltema and S.~Profumo.
\newblock {Deep XMM Observations of Draco rule out at the 99\% Confidence Level
  a Dark Matter Decay Origin for the 3.5 keV Line}.
\newblock {\em Mon. Not. Roy. Astron. Soc.},
  \href{http://dx.doi.org/10.1093/mnras/stw578}{458(4):3592--3596}, 2016,
  \href{https://arxiv.org/abs/1512.01239}{1512.01239}.

\bibitem{R16}
O.~Ruchayskiy, A.~Boyarsky, D.~Iakubovskyi, E.~Bulbul, D.~Eckert, J.~Franse,
  D.~Malyshev, M.~Markevitch, and A.~Neronov.
\newblock {Searching for decaying dark matter in deep XMM-Newton observation of
  the Draco dwarf spheroidal}.
\newblock {\em Mon. Not. Roy. Astron. Soc.},
  \href{http://dx.doi.org/10.1093/mnras/stw1026}{460(2):1390--1398}, 2016,
  \href{https://arxiv.org/abs/1512.07217}{1512.07217}.

\bibitem{I16}
D.~Iakubovskyi.
\newblock {Observation of the new emission line at 3.5 keV in X-ray spectra of
  galaxies and galaxy clusters}.
\newblock {\em Adv. Astron. Space Phys.},
  \href{http://dx.doi.org/10.17721/2227-1481.6.3-15}{6(1):3--15}, 2016,
  \href{https://arxiv.org/abs/1510.00358}{1510.00358}.

\bibitem{White_Paper17}
M.~Drewes et~al.
\newblock {A White Paper on keV Sterile Neutrino Dark Matter}.
\newblock {\em J. Cosmol. Astropart. Phys.},
  \href{http://dx.doi.org/10.1088/1475-7516/2017/01/025}{1701(01):025}, 2017,
  \href{https://arxiv.org/abs/1602.04816}{1602.04816}.

\bibitem{Hitomi_Perseus}
F.~Aharonian et~al.
\newblock {The Quiescent Intracluster Medium in the Core of the Perseus
  Cluster}.
\newblock {\em Nature},
  \href{http://dx.doi.org/10.1038/nature18627}{535:117--121}, 2016,
  \href{https://arxiv.org/abs/1607.04487}{1607.04487}.

\bibitem{BHSC03}
C.~Boehm, D.~Hooper, J.~Silk, M.~Casse, and J.~Paul.
\newblock {MeV dark matter: Has it been detected?}
\newblock {\em Phys. Rev. Lett.},
  \href{http://dx.doi.org/10.1103/PhysRevLett.92.101301}{92:101301}, 2004,
  \href{https://arxiv.org/abs/astro-ph/0309686}{astro-ph/0309686}.

\bibitem{HW04}
D.~Hooper and L.-T. Wang.
\newblock {Possible evidence for axino dark matter in the galactic bulge}.
\newblock {\em Phys. Rev.},
  \href{http://dx.doi.org/10.1103/PhysRevD.70.063506}{D70:063506}, 2004,
  \href{https://arxiv.org/abs/hep-ph/0402220}{hep-ph/0402220}.

\bibitem{Churazov04}
E.~Churazov, R.~Sunyaev, S.~Sazonov, M.~Revnivtsev, and D.~Varshalovich.
\newblock {Positron annihilation spectrum from the Galactic Center region
  observed by SPI/INTEGRAL}.
\newblock {\em Mon. Not. Roy. Astron. Soc.},
  \href{http://dx.doi.org/10.1111/j.1365-2966.2005.08757.x}{357:1377--1386},
  2005, \href{https://arxiv.org/abs/astro-ph/0411351}{astro-ph/0411351}.

\bibitem{Johnson73}
W.N. {Johnson}, III and R.C. {Haymes}.
\newblock {Detection of a Gamma-Ray Spectral Line from the Galactic-Center
  Region}.
\newblock {\em Astrophys. J.},
  \href{http://dx.doi.org/10.1086/152309}{184:103--126}, 1973.

\bibitem{Haymes75}
R.C. {Haymes}, G.D. {Walraven}, C.A. {Meegan}, R.D. {Hall}, F.T. {Djuth}, and
  D.H. {Shelton}.
\newblock {Detection of nuclear gamma rays from the galactic center region}.
\newblock {\em Astrophys. J.},
  \href{http://dx.doi.org/10.1086/153925}{201:593--602}, 1975.

\bibitem{Leventhal82}
M.~{Leventhal}, C.J. {MacCallum}, A.F. {Huters}, and P.D. {Stang}.
\newblock {Time-variable positron annihilation radiation from the galactic
  center direction}.
\newblock {\em Astrophys. J. Lett.},
  \href{http://dx.doi.org/10.1086/183858}{260:L1--L5}, 1982.

\bibitem{english_translation_of_german}
B.L. van~der Waerden.
\newblock Sources of quantum mechanics.
\newblock \href{http://n2t.net/ark:/13960/t18k8gw97}{ark:13960/t18k8gw97}.

\bibitem{Fermi_collected}
Emilo Segr\`{e}, editor.
\newblock {\em The Collected Papers of Enrico Fermi: Vol. 1 Italy 1921-38}.
\newblock The University of Chicago Press and Accademia Nazionale dei Lincei,
  Roma, 1962.

\bibitem{vN55}
J.~von Neumann.
\newblock {\em {Mathematical Foundations of Quantum Mechanics}}.
\newblock 1955.
\newblock (translated into English by R. T. Beyer).

\end{thebibliography}
\end{document}